\def\year{2022}\relax
\documentclass[letterpaper]{article} 
\usepackage{aaai22}  
\usepackage{times}  
\usepackage{helvet}  
\usepackage{courier}  
\usepackage[hyphens]{url}  
\usepackage{graphicx} 
\usepackage{natbib}  
\usepackage{caption} 
\DeclareCaptionStyle{ruled}{labelfont=normalfont,labelsep=colon,strut=off} 
\usepackage{float}
\newcommand*{\Scale}[2][4]{\scalebox{#1}{$#2$}}%

\usepackage{newfloat}
\usepackage{listings}
\floatstyle{ruled}
\newfloat{listing}{tb}{lst}{}
\floatname{listing}{Listing}
\usepackage{amsfonts}
\usepackage{amsthm}
\usepackage{amsmath}
\usepackage{subcaption}     
\usepackage{adjustbox}
\usepackage{multirow}
\usepackage{booktabs}       
 \usepackage[ruled,linesnumbered]{algorithm2e}

\usepackage{thmtools,thm-restate}

\newcommand{\ignore}[1]{}

\usepackage{tcolorbox}
\title{Calibrated Nonparametric Scan Statistics for Anomalous Pattern Detection in Graphs}

\author {
    Chunpai Wang,\textsuperscript{\rm 1}
    Daniel B. Neill, \textsuperscript{\rm 2}
    Feng Chen \textsuperscript{\rm 3}
}
\affiliations {
    \textsuperscript{\rm 1} University at Albany - SUNY\\
    \textsuperscript{\rm 2} New York University\\
    \textsuperscript{\rm 3} The University of Texas at Dallas\\
    cwang25@albany.edu, daniel.neill@nyu.edu, feng.chen@utdallas.edu
}

\usepackage{bibentry}

\begin{document}

\maketitle

\begin{abstract}
We propose a new approach, the calibrated nonparametric scan statistic (\texttt{CNSS}), for more accurate detection of anomalous patterns in large-scale, real-world graphs.
Scan statistics identify connected subgraphs that are interesting or unexpected through maximization of a likelihood ratio statistic; in particular, nonparametric scan statistics (NPSSs) identify subgraphs with a higher than expected proportion of individually significant nodes.  However, we show that recently proposed NPSS methods are \emph{miscalibrated}, failing to account for the maximization of the statistic over the multiplicity of subgraphs.  This results in both reduced detection power for subtle signals, and low precision of the detected subgraph even for stronger signals. Thus we develop a new statistical approach to recalibrate NPSSs, correctly adjusting for multiple hypothesis testing and taking the underlying graph structure into account.  While the recalibration, based on randomization testing, is computationally expensive, we propose both an efficient (approximate) algorithm and new, closed-form lower bounds (on the expected maximum proportion of significant nodes for subgraphs of a given size, under the null hypothesis of no anomalous patterns). These advances, along with the integration of recent core-tree decomposition methods, enable \texttt{CNSS} to scale to large real-world graphs, with substantial improvement in the accuracy of detected subgraphs.  Extensive experiments on both semi-synthetic and real-world datasets are demonstrated to validate the effectiveness of our proposed methods, in comparison with state-of-the-art counterparts. 

\end{abstract}

\section{Introduction}
\label{sec:intro}

Detecting ``hotspots'' or anomalous patterns in graphs is an important but challenging problem, with numerous critical applications in areas such as epidemiology, law enforcement, finance, and security. Among the powerful and widely used methods, the paradigm of scan statistics is one of the few that has a sound and general statistical basis (for related surveys see~\citet{glaz2009scan,akoglu2015graph}; and~\citet{cadena2018graph}). Graph-based scan statistics~\cite{speakman2015scalable, speakman2013dynamic, chen2014non,cadena2019near} identify connected subgraphs that are interesting or unexpected through maximization of a likelihood ratio statistic. The connectivity constraint is important because it ensures that subgraphs reflect changes due to localized anomalous processes (e.g., disease outbreaks, water pollution events). In particular, nonparametric scan statistic (NPSS) methods~\cite{neill2007nonparametric,mcfowland2013fgss,chen2014non} are designed without assuming any known background process on the graph. These approaches use historical data (assuming no anomalous patterns are present) to compute an empirical p-value for each graph node, and then compare the observed and expected number of significantly low p-values contained in connected subgraphs. Those with the largest scores are returned as the most anomalous subgraphs. However, as we show, NPSSs fail to account for the multiple hypothesis testing effects of searching over the huge space of connected subgraphs, reducing detection performance. In this work, we conduct a systematic study of this challenging problem and make the following \textbf{key contributions:}

\begin{itemize}
    \item We show that recently proposed NPSS methods are \emph{miscalibrated}, failing to account for the maximization of the statistic over the multiplicity of subgraphs.  This results in both reduced detection power for subtle signals, and low precision of the detected subgraph. 
    \item We develop a new statistical approach to recalibrate NPSS, correctly adjusting for multiple hypothesis testing and taking the underlying graph structure into account, substantially improving detection performance.  
    \item We propose an efficient (approximate) algorithm and new, closed-form lower bounds on the expected maximum proportion of significant nodes for subgraphs of a given size, under the null hypothesis of no anomalous patterns. These advances, along with integration of recent core-tree decomposition methods, enable the \texttt{CNSS} approach to scale to large real-world graphs, with  substantial improvement in the accuracy of detected subgraphs.
    \item Extensive experiments on semi-synthetic and real-world datasets show that our methods can detect anomalous subgraphs more accurately than state-of-the-art counterparts, while maintaining comparable time efficiency. 
\end{itemize}

\section{Related Work}
\label{sec:related}
As anomaly detection in graphs has a large literature, we refer to~\citet{akoglu2015graph} and~\citet{cadena2018graph} for comprehensive surveys on this topic. For brevity, we will discuss the methods based on scan statistics for detecting anomalous connected subgraphs, including those based on parametric scan statistics and NPSSs. 

\textit{Parametric scan statistics} are defined as the likelihood ratio statistics of the hypothesis test, where under the null hypothesis $\mathcal{H}_0$, the attribute data of nodes within a candidate connected subgraph $\mathcal{S}$ are generated by a parameterized background process. Under the alternative hypothesis $\mathcal{H}_1(\mathcal{S})$, the attribute data are generated by a different parameterized distribution (a localized anomalous process). Depending on the assumptions on these two distributions, a variety of scan statistics are formulated, such as Positive Elevated
Mean~\cite{qian2014connected} and Expectation-based Poisson and Gaussian~\cite{neill2009expectation}, in addition to the Kulldorff Scan Statistic~\cite{kulldorff1997spatial}. While these methods have been shown to achieve high detection power across a variety of spatio-temporal graph datasets, they make strong parametric model assumptions, and performance degrades when these models are incorrect. 

In comparison, \textit{NPSSs} use historical data with no anomalous patterns to calibrate an empirical p-value for each node and are defined as likelihood ratio statistics of the nonparametric hypothesis test. Under the null hypothesis of no anomalous patterns ($\mathcal{H}_0$), the empirical p-values of nodes within a candidate connected subgraph ($\mathcal{S}$) follow a uniform distribution between 0 and 1.  Under the alternative hypothesis ($\mathcal{H}_1$), the empirical p-values follow a different distribution. Depending on the specific form of the distribution under $\mathcal{H}_1$, different NPSSs are formulated, such as the Berk-Jones~\cite{berk1979goodness}, Higher Criticism~\cite{donoho2004higher}, Kolmogorov–Smirnov~\cite{massey1951kolmogorov}, and Anderson-Darling scan statistics~\cite{eicker1979asymptotic}. 

Optimizing scan statistics is challenging in the presence of connectivity constraints. A number of heuristic algorithms have been proposed for parametric scan statistics, such as additive GraphScan based on shortest paths in graphs~\cite{speakman2013dynamic}, Steiner tree heuristics~\cite{rozenshtein2014event}, and simulated annealing~\cite{duczmal2004simulated}. \citet{qian2014connected} used linear matrix inequalities as a way to characterize the connectivity constraint and designed efficient iterative algorithms with convergence guarantees to optimize scan statistics (e.g., Positive Elevated Mean) that are convex after relaxation. \citet{sharpnack2015detecting} proposed a computationally tractable algorithm with consistency guarantees. 
Several heuristics have been proposed for NPSSs, such as greedy growth~\cite{chen2014non} and Steiner tree heuristics based on approximation of the underlying
graph with trees~\cite{wu2016efficient}. A depth-first-search based algorithm, named GraphScan, was proposed to exactly identify the most anomalous connected subgraphs for scan statistics (e.g., Kulldorff, Berk-Jones) that satisfy the ``linear time subset scanning'' (LTSS) property~\cite{neill2012fast}, 
but has an exponential time complexity in the worst case~\cite{speakman2015scalable}. An approximate algorithm built based on the color-coding technique~\cite{alon1995color} was designed for a large class of scan statistics with rigorous guarantees~\cite{cadena2019near}. Although it has the performance bound of $1-\epsilon$, its run time scales exponentially with the size of the most anomalous connected subgraphs. 

Recent work by~\citet{reyna2021} and~\citet{chitra2021} demonstrates the miscalibration of the Gaussian scan statistic and presents a Gaussian mixture modeling approach to reduce this bias.  As discussed in Appendix~\ref{app:npss-differences}, the nonparametric scan statistics that we consider here differ fundamentally from the Gaussian scan, both in their assumptions about the true signal (distribution of p-values under $\mathcal{H}_1$) and in their maximization over a range of significance levels $\alpha$.

\section{Limitations of Nonparametric Scan}
\label{sec:npss}

\begin{figure*}[!h]
    \centering
    \begin{subfigure}[t]{0.245\textwidth}
        \centering
        \includegraphics[width=\linewidth]{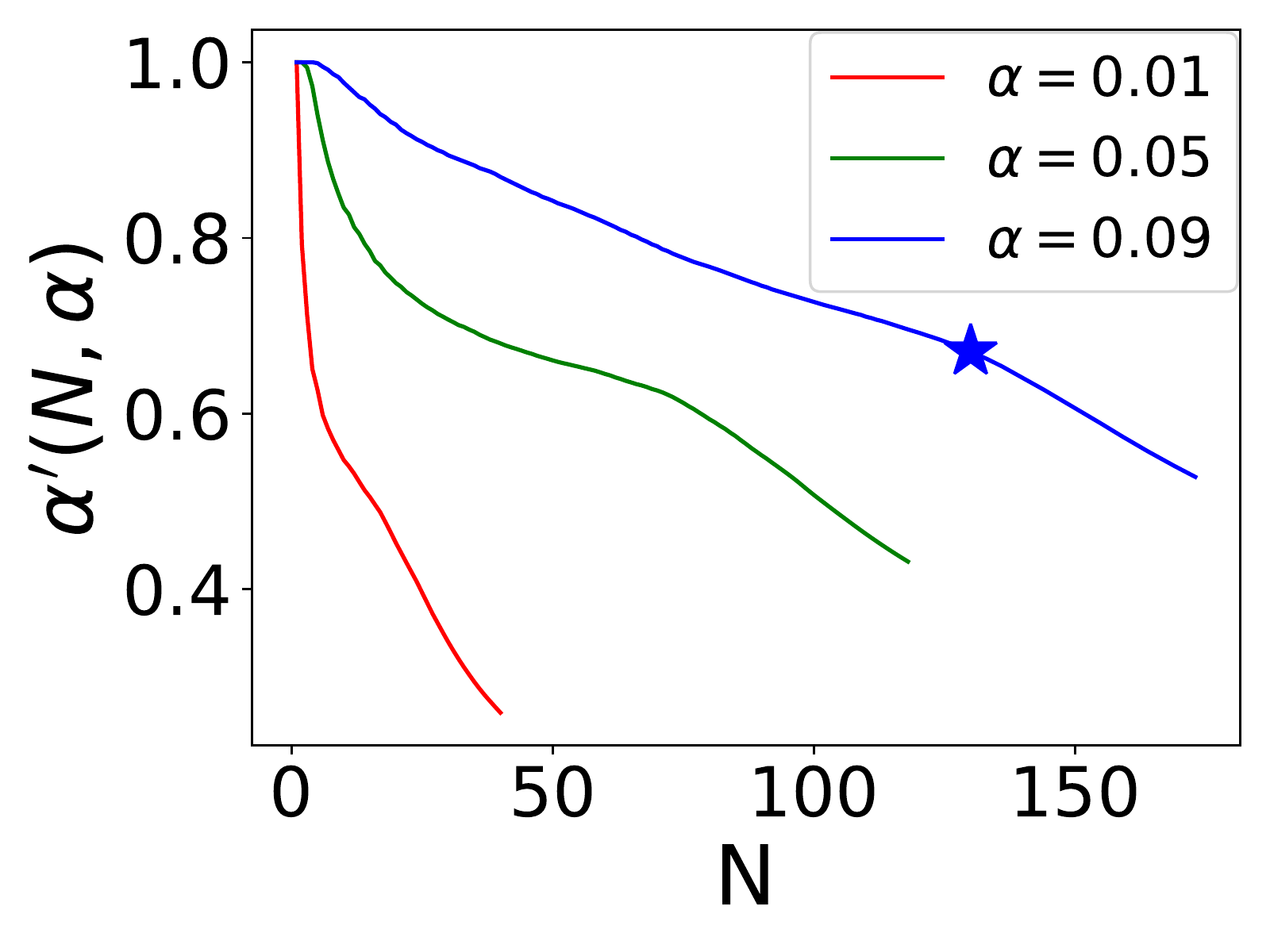}
        \vspace{-6mm}
        \caption{$|\mathcal{V}|=1000, p=0.01$. }
    \end{subfigure}
    \begin{subfigure}[t]{0.245\textwidth}
        \centering
        \includegraphics[width=\linewidth]{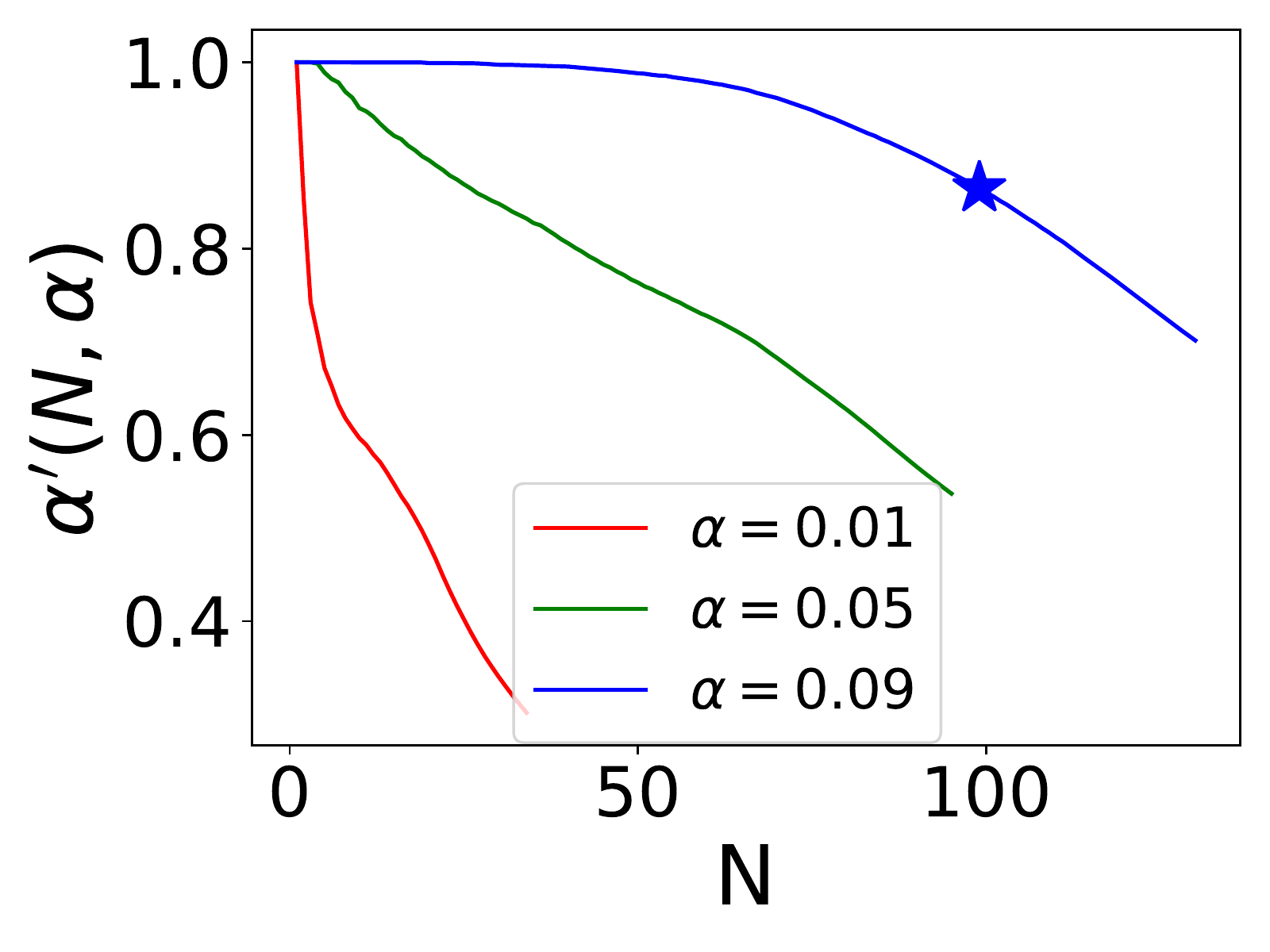}
        \vspace{-6mm}
        \caption{$|\mathcal{V}|=1000, p=0.02$. }
    \end{subfigure}
    \begin{subfigure}[t]{0.245\textwidth}
        \centering
        \includegraphics[width=\linewidth]{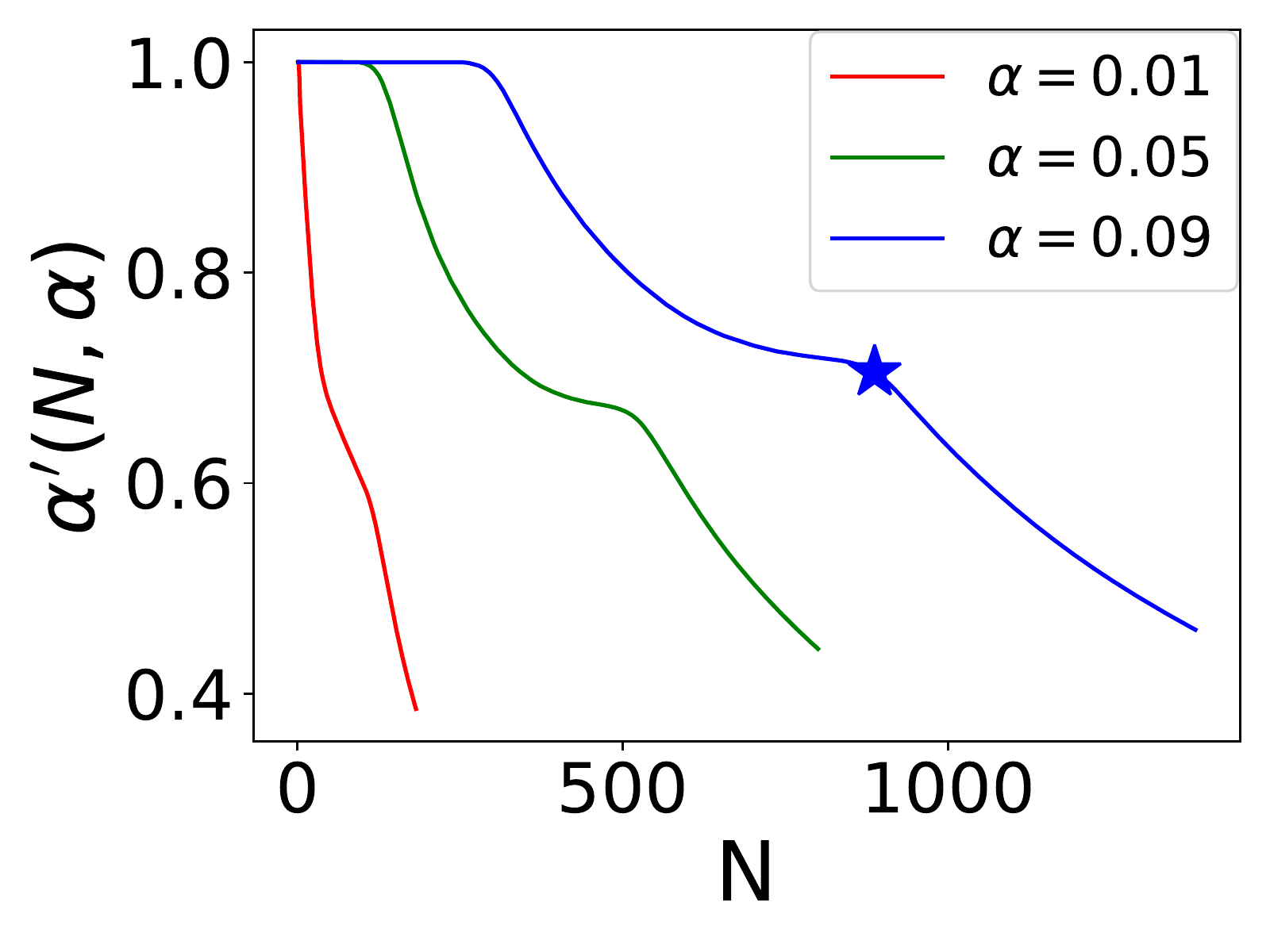}
        \vspace{-6mm}
        \caption{WikiVote.}
    \end{subfigure}
    \begin{subfigure}[t]{0.245\textwidth}
        \centering
        \includegraphics[width=\linewidth]{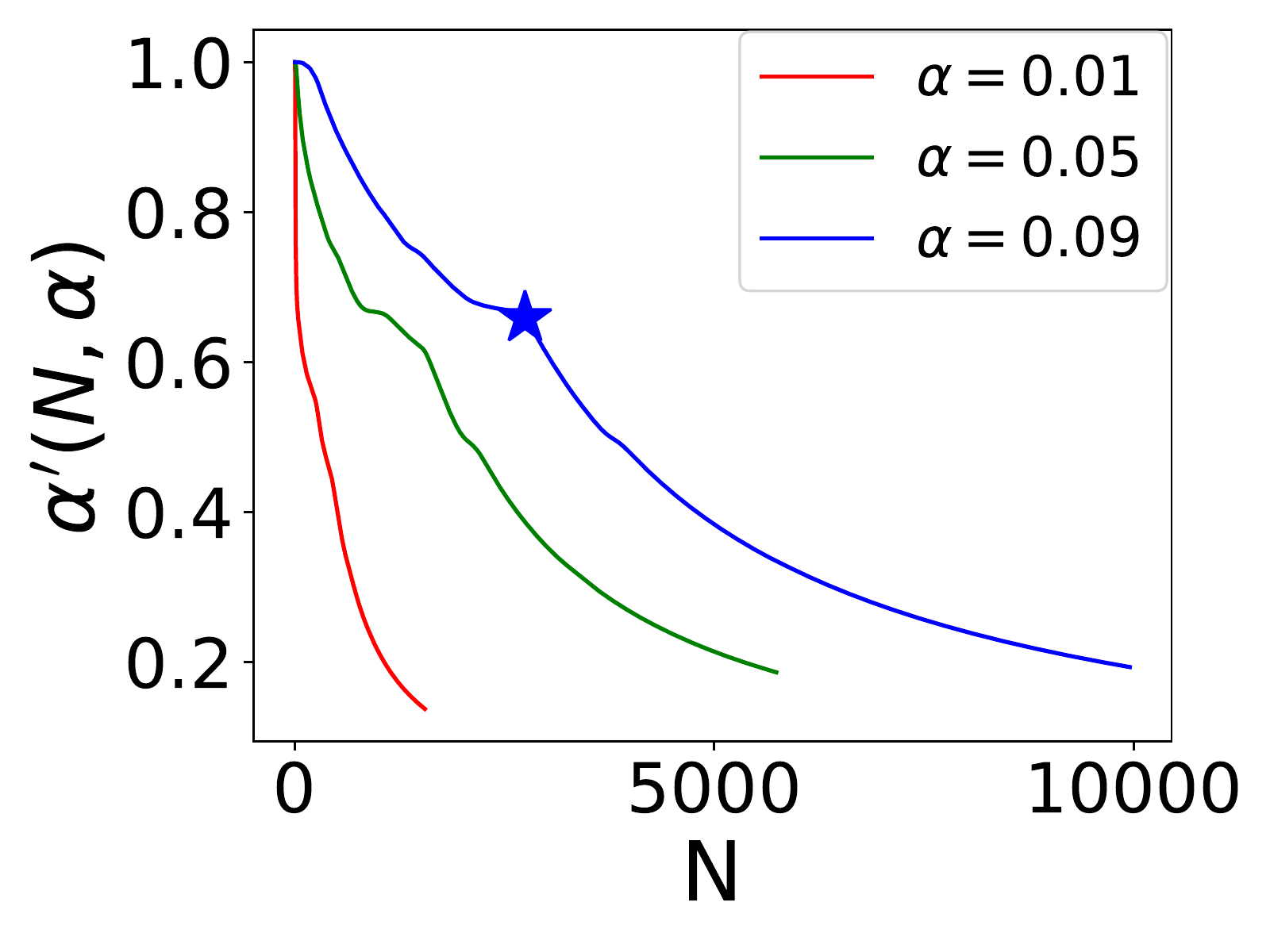}
        \vspace{-6mm}
        \caption{CondMat.}
    \end{subfigure}
    \caption{Simulation results (a) and (b) on Erdos-Renyi graphs and (c) and (d) on real world graphs under $\mathcal{H}_0$. Each curve ends at the value of $N$ such that all significant nodes in the graph are included. The starred point is the combination of $N$ and $\alpha$ for which $N \times \texttt{KL}(\alpha', \alpha)$ is maximized. Descriptive statistics of the \texttt{WikiVote} and \texttt{CondMat} datasets can be viewed in Appendix~\ref{app:datasets}.
    }
    \label{fig:erdos-renyi-simulations}
    \vspace{-1mm}
\end{figure*}

Given a graph $\mathbb{G} = (\mathcal{V}, \mathcal{E})$, with $\mathcal{V}$ being a set of $n$ vertices and $\mathcal{E}$ being a set of $m$ edges. 
Each node $v_i \in \mathcal{V}$ is associated with a feature vector ${\bf x}_i \in \mathbb{R}^{N}$ and its historical observations $\{\mathbf{x}_i^{(1)}, \cdots, \mathbf{x}_i^{(T)}\}$. We use the historical observations to convert the feature vector of each node ($\mathbf{x}_i$) to a single empirical p-value ($p_i$), using the two-stage empirical calibration procedure introduced in~\citet{chen2014non}.  Additional details on the computation of empirical p-values are provided in Appendix~\ref{app:npss-pvals}. Critically, under the null hypothesis $\mathcal{H}_0$, the current observations are assumed to be exchangeable with the null distribution of interest, and thus the p-values (computed by ranking the current observation against the historical observations and then normalizing the ranks) are asymptotically uniform on [0,1] under the null.

For instance, the graph could be a geospatial network, in which each node represents a county, two nodes are connected via an edge if they are spatially adjacent, and each node has a single feature, $x_i \in \mathbb{R}$, that is the number of confirmed Covid-19 disease cases for the current week. The goal is to detect the most anomalous cluster or connected subgraph (representing a potential Covid-19 outbreak). In this case, the empirical p-value $p_i$ is simply the proportion of the historical observations with case counts that are greater than or equal to the current observation. \ignore{We assume one-tailed p-values, aiming to answer the question: under the null hypothesis ($\mathcal{H}_0$) of no anomalous pattern, what is the probability that an observed value of a randomly selected sample would be greater than or equal to the current observation.}

We denote by $\mathbb{G}_{\mathcal{S}} = (\mathcal{S}, \mathcal{E}_{\mathcal{S}})$ the subgraph induced by the subset $\mathcal{S} \subseteq \mathcal{V}$ and $\mathbb{M} = \{\mathcal{S} \:|\: \mathcal{S} \subseteq \mathcal{V}, \mathbb{G}_{\mathcal{S}} \text{ is connected in } \mathbb{G}\}$ the set of all possible connected subsets. The problem of NPSS-based anomalous pattern detection is defined as the connected subgraph optimization problem: 
\begin{equation}
    \begin{split}
    \max_{\mathcal{S} \in \mathbb{M}} F(\mathcal{S})&=\max_{\mathcal{S} \in \mathbb{M}}\max_{\alpha \leq \alpha_{\max }} \Phi\left(\alpha, N_{\alpha}(\mathcal{S}), N(\mathcal{S})\right)\\
    &= \max_{\alpha \leq \alpha_{\max }} \max_{\mathcal{S} \in \mathbb{M}} \Phi\left(\alpha, N_{\alpha}(\mathcal{S}), N(\mathcal{S})\right)
    \end{split}
\end{equation}
where $F(\mathcal{S}):=\max_{\alpha \leq \alpha_{\max }} \Phi\left(\alpha, N_{\alpha}(\mathcal{S}), N(\mathcal{S})\right)$ refers to the general form of NPSS defined by~\citet{mcfowland2013fgss}, $\mathcal{S} \subseteq \mathcal{V}$ is a connected set of nodes, $N_{\alpha}(\mathcal{S})= \sum_{v\in \mathcal{S}} \mathbf{1}\{p(v)\leq \alpha\}$ refers to the number of p-values in subset $\mathcal{S}$ that are significant at level $\alpha$, and $N(\mathcal{S}) = \sum_{v\in \mathcal{S}} 1$ refers to the total number of p-values in subset $\mathcal{S}$. The function $\Phi\left(\alpha, N_{\alpha}(\mathcal{S}), N(\mathcal{S})\right)$ compares the observed number of significant p-values $N_\alpha(\mathcal{S})$ at level $\alpha$ to the expected number of significant p-values $\mathbb{E} \left[ N_{\alpha}(\mathcal{S}) \right] = \alpha N(\mathcal{S})$ under the null hypothesis  $\mathcal{H}_0$. Critically, NPSSs optimize the significance level $\alpha$ between $0$ and some constant $\alpha_{\max} < 1$. Maximization over a range of $\alpha$ values allows accurate detection of signals with either a small number of highly significant p-values or a larger number of moderately significant p-values. In practice, rather than considering all $\alpha \le \alpha_{\max}$, we consider a discrete set of $\alpha$ values, $\mathcal{L} = \{0.001, \cdots, 0.009, 0.01, \cdots, 0.09\}$, and solve the constrained optimization $\max_{S\in \mathbb{M}} \Phi\left(\alpha, N_{\alpha}(\mathcal{S}), N(\mathcal{S})\right)$ for each $\alpha$ in $\mathcal{L}$ to find the most anomalous subgraphs.

Here we focus on the Berk-Jones (BJ) nonparametric scan statistic, without loss of generalization to other NPSSs (see Appendix~\ref{app:npss-variants}). The BJ statistic is defined as
\begin{equation}
    \Phi_{BJ}\left(\alpha, N_{\alpha}(\mathcal{S}), N(\mathcal{S})\right)=N(\mathcal{S}) \times \verb+KL+\left(\frac{N_{\alpha}(\mathcal{S})}{N(\mathcal{S})}, \alpha\right),
\end{equation}
where \verb+KL+ is the Kullback-Liebler divergence between the observed and expected proportions of significant p-values: $\verb+KL+(a, b)=a \log \left(a/b\right)+(1-a) \log \left((1-a)/(1-b)\right)$.
The BJ statistic is the log-likelihood ratio statistic for testing whether the empirical p-values follow the \texttt{Uniform}[0, 1] distribution or a piecewise constant distribution where $\text{Pr}(p < \alpha) > \alpha$. 
Please see Appendix~\ref{app:npss-assumptions} for details.

Despite their effectiveness for anomalous pattern detection in graphs, NPSSs were originally designed without taking into consideration the multiple hypothesis testing effect resulting from the multiplicity of subgraphs. In particular, it follows from the assumption of uniform p-values under $\mathcal{H}_0$ made by NPSSs that the expected proportion of individually significant nodes within a connected subset $\mathcal{S}$ under $\mathcal{H}_0$ is $\mathbb{E}[N_\alpha(\mathcal{S})/N(\mathcal{S})] = \alpha$. However, this is true for a randomly selected connected subset, but not for connected subsets $\mathcal{S}$ that are identified by maximizing the NPSS score.  Even when the null hypothesis $\mathcal{H}_0$ holds, and p-values are uniform on [0,1], we find that the expected proportion of individually significant nodes within the highest-scoring connected subsets, denoted as $\alpha^\prime$, is typically much larger than $\alpha$, which we refer to as \emph{miscalibration}.  More precisely, we define $\alpha^\prime$ as the expected maximum proportion of significant nodes for all connected subgraphs of a given size $N$: $\alpha^\prime (N, \alpha) =\mathbb{E}[\max_{\mathcal{S} \in \mathbb{M}, |\mathcal{S}| = N} N_\alpha(\mathcal{S}) / N]$.
 
To illustrate the relationship between $\alpha^\prime$ and $\alpha$, we conduct simulations on both Erdos-Renyi (ER) random graphs and two real world graphs: WikiVote ($|\mathcal{V}| = 7066, 
|\mathcal{E}| = 100736, \text{density} = 0.004$) and CondMat ($|\mathcal{V}| = 21363, 
|\mathcal{E}| = 91286, \text{density} = 0.0004$).  First, we generate $100$ random graphs with $|\mathcal{V}|=1,000$ and edge probabilities $p\in \{0.01, 0.02\}$, respectively. We simulate p-values under $\mathcal{H}_0$ for each ER random graph, and calculate the average $\alpha^\prime$ for each $N\in \{1, 2, ..., |\mathcal{V}|\}$ and $\alpha \in \{.01,.05,.09\}$ among all the 100 random graphs, as shown in Figure~\ref{fig:erdos-renyi-simulations}(a)-(b).
We simulate p-values under $\mathcal{H}_0$ on the graphs WikiVote and CondMat for 100 times and calculate the average $\alpha^\prime$ for each $N$ and $\alpha$, as shown in Figure~\ref{fig:erdos-renyi-simulations}(c)-(d). 
The results indicate that the expected maximum proportion $\alpha^\prime(N, \alpha)$ for given values of $N$ and $\alpha$ is much higher than the expected proportion $\mathbb{E}[N_\alpha(\mathcal{S}) / N]=\alpha$. The implication is that, even when the null hypothesis $\mathcal{H}_0$ holds and there are no true subgraphs of interest, there exist subgraphs $\mathcal{S}$ with $N_{\alpha}(\mathcal{S}) \gg \alpha N(\mathcal{S})$, and thus very high NPSS scores. The amount of difference between $\alpha^\prime(N, \alpha)$ and $\alpha$ is a function of $N$, $\alpha$, and the graph structure.  
We observe that $\alpha^\prime(N, \alpha)$ decreases with $N$ but remains much higher than $\alpha$ for the entire range of $N$.

The results of this discrepancy between $\alpha^\prime$ and $\alpha$ are threefold.  First, the maximum NPSS score under the null hypothesis $\mathcal{H}_0$ will be large.  To see this, we compute the expected score $N \times \texttt{KL}(\alpha',\alpha)$ for each combination of $\alpha$ and $N$ for each of Figure~\ref{fig:erdos-renyi-simulations}(a)-(d), and show the highest-scoring combination on each graph as a star icon.  The corresponding scores range from 131.7 for Figure~\ref{fig:erdos-renyi-simulations}(a) to 2677 for Figure~\ref{fig:erdos-renyi-simulations}(d).  These large scores under the null result in \emph{reduced detection power}, since the NPSS score of the true anomalous subgraph must exceed a larger threshold (i.e., the 95th percentile of the maximum NPSS scores under $\mathcal{H}_0$) to be considered significant.  Second, NPSS will be \emph{biased toward detecting clusters at larger $\alpha$ thresholds}, even if the true signal is for a much smaller $\alpha$.  We observe that, for all four graphs in Figure~\ref{fig:erdos-renyi-simulations}, the null score is maximized at the largest of the three $\alpha$ values considered.  Third, NPSS will identify overly large clusters which include many nodes that have significant p-values just by chance, resulting in \emph{reduced precision} of the detected cluster.
We observe that, for all four graphs in Figure~\ref{fig:erdos-renyi-simulations}, the null score is maximized for a large value of $N$, where almost all of the significant nodes in the graph are included in the detected cluster. This observation is also supported by low precision (and low $F$-scores) for all uncalibrated scan methods in our evaluation results below.  An additional, concrete example of miscalibration for the (uncalibrated) BJ statistic is provided in Appendix~\ref{app:correctness}.

\section{Calibrated Nonparametric Scan Statistics}
\label{sec:cnss}

Thus we have shown that uncalibrated NPSS methods discover 
large, high-scoring anomalous connected subsets even under the null hypothesis $\mathcal{H}_0$, resulting in reduced detection power and precision. This observation motivates us to develop a new approach to recalibrating the non-parametric scan statistic that accounts for multiple testing (and the resulting, large differences between $\alpha^\prime$ and $\alpha$), for improved detection performance. Hence, we propose \textit{Calibrated Nonparametric Scan Statistics (\texttt{CNSS})}, where $F(\mathcal{S}) =  \max_{\alpha \le \alpha_{\max}} \Phi\left(\alpha, N_{\alpha}(\mathcal{S}), N(\mathcal{S})\right)$ as above, 
but the expected proportion of significant p-values ($\mathbb{E}[N_\alpha(\mathcal{S})/N(\mathcal{S})] = \alpha$) used in $\Phi(\cdot)$
is replaced with the expected \emph{maximum} proportion of significant p-values $\alpha^\prime(N,\alpha)$ over all subgraphs of size $N$ under the null hypothesis $\mathcal{H}_0$.
For example, our proposed 
 Calibrated Berk Jones (CBJ) statistic is defined as
\begin{equation}
    \Scale[0.95]{
    \Phi_{\texttt{CBJ}}\left(\alpha, N_{\alpha}(\mathcal{S}), N(\mathcal{S})\right)=N(\mathcal{S}) \times \texttt{KL}\left(\frac{N_{\alpha}(\mathcal{S})}{N(\mathcal{S})}, \alpha^{\prime}(N(\mathcal{S}),\alpha)\right)
    }
\end{equation}
where 
\begin{equation}
    \alpha^{\prime}(N,\alpha) = \frac{\mathbb{E}\left[\max _{\mathcal{S}\in \mathbb{M}, |\mathcal{S}|=N} N_{\alpha}(\mathcal{S})\right]}{N}.
\end{equation}
Critically, this approach guarantees that, under the null hypothesis $\mathcal{H}_0$ that current and historical observations are exchangeable, for any $N$ and $\alpha$, the expected ratio $N_\alpha(\mathcal{S})/N$ for the highest-scoring subgraph $\mathcal{S}$ of size $N$ is equal to $\alpha'$, thus adjusting for the multiplicity of subgraphs and correctly calibrating across $N$ and $\alpha$.  See Appendix~\ref{app:correctness} for more explanation on the correctness of the calibration approach.

As shown in Figure~\ref{fig:erdos-renyi-simulations}, the expected maximum proportion of significant nodes $\alpha^\prime$ depends on the subgraph size $N$, the significance level $\alpha$, and the graph structure. To estimate $\alpha^\prime(N,\alpha)$ for a given graph, we use a \emph{randomization test} to
estimate $\mathbb{E}[\max _{\mathcal{S} \in \mathbb{M}, |\mathcal{S}|=N} N_{\alpha}(\mathcal{S})]$ for each $N$ and $\alpha$. In particular, we create a large number ($K=200$) of replica datasets under the null hypothesis $\mathcal{H}_0$, where each node of the input graph $\mathbb{G}$ has its p-value redrawn uniformly at random from $[0,1]$.
We then apply the efficient approximate algorithm proposed in Section~\ref{sect:efficient-algorithm} to solve the constrained set optimization problem $\max _{\mathcal{S} \in \mathbb{M}, |\mathcal{S}|=N} N_{\alpha}(\mathcal{S})$ for each combination $(N, \alpha)\in \{1, \cdots, |\mathcal{V}|\}\times \mathcal{L}$. Based on the $K$ replica datasets, for each combination of $N$ and $\alpha$, we collect $K$ samples of the maximum number of significant nodes $\max_{\mathcal{S} \in \mathbb{M}, |\mathcal{S}|=N} N_{\alpha}(\mathcal{S})$ and use the samples to estimate $\alpha^\prime(N,\alpha)$ under $\mathcal{H}_0$. The same algorithm is applied to the original dataset to identify the most significant subgraph $\max_{\mathcal{S} \in \mathbb{M}, |\mathcal{S}|=N} N_{\alpha}(\mathcal{S})$ for each $(N,\alpha) \in \{1, \cdots, |\mathcal{V}|\}\times \mathcal{L}$, and then the subgraph with the highest score 
$F(\mathcal{S}) = \max_{\alpha \le \alpha_{\max}} \Phi_{\texttt{CBJ}}(\alpha,N_\alpha(\mathcal{S}),N(\mathcal{S}))$ is returned.  More details are provided in Algorithm~\ref{alg:cnss} in Appendix~\ref{appendix:CNSS}.  

\subsection{An Efficient Approximate Algorithm}
\label{sect:efficient-algorithm}
The fundamental computational challenge of \texttt{CNSS} is to find the maximum number of significant nodes, $\max _{\mathcal{S} \in \mathbb{M}, |\mathcal{S}|=N} N_{\alpha}(\mathcal{S})$, for connected subgraphs of every size $N \in \{1, \cdots, |\mathcal{V}|\}$.  
One approach for doing so would be, for each $N$ and each $\alpha \in \mathcal{L}$, to separately run a prize-collecting Steiner tree (PCST) algorithm to identify the maximum $N_\alpha$. 
The PCST is NP-hard but can be approximated in $O(|\mathcal{V}|^2 \log |\mathcal{V}|)$  time; however, computing the PCST for each $N$ would then result in an insufficiently scalable $O(|\mathcal{V}|^3 \log |\mathcal{V}|)$ algorithm.  As an alternative, we propose a novel algorithm which approximates the maximum $N_\alpha$ for each $N \in \{1, \ldots, |\mathcal{V}|\}$, for a given value of $\alpha$, in a single, efficient run.  This process must then be repeated for each value of $\alpha$ under consideration. 

The pseudocode of estimating the maximum $N_{\alpha}$ for each $N$ under a given significance threshold $\alpha$ is described in Algorithm \ref{alg:merge} in Appendix~\ref{app:alg:merge}. 
The approach is based on repeated merging of nodes with the highest proportion of significant p-values. Given a graph with node-level p-values, we first merge all adjacent significant nodes, and maintain a list $\mathcal{Z}$ of merged nodes sorted by significance ratio $N_{\alpha}(\mathcal{S})/N(\mathcal{S})$.  We repeatedly choose the merged node with highest significance ratio and perform one of the following three merge steps: (1) add an adjacent node which contains some or all significant p-values; (2) add an adjacent non-significant node that is also adjacent to at least one other significant node; or (3) add the highest-degree non-significant neighbor.  At each merge step, our method will try all three options and utilize the one leading to a merged node with the highest $N_{\alpha}(\mathcal{S})/N(\mathcal{S})$ ratio; this is repeated until the list $\mathcal{Z}$ only contains a single merged node.  The advantage of this merging process is that we can keep track of the maximum $N_{\alpha}(\mathcal{S})$ for each $N(\mathcal{S})$ and iteratively update these values throughout the entire merging process.
In the end, we have a list of estimated max $N_{\alpha}(\mathcal{S})$ for $N(\mathcal{S}) \in \{1, \cdots, |\mathcal{V}|\}$. 

The overall time complexity of this algorithm is $\mathcal{O}(k|\mathcal{V}| + |\mathcal{V}|\log|\mathcal{V}|)$ where $k$ is the largest degree of a node in the graph. See Appendix~\ref{app:time_complexity} for a more detailed analysis.

\subsection{Lower Bounds for the Expected Maximum Proportion of Significant Nodes, $\alpha^\prime(N, \alpha)$}
\label{sec:lower_bounds}

\begin{figure*}[t]
    \centering
    \begin{subfigure}[t]{0.3\textwidth}
        \centering
        \includegraphics[width=\linewidth]{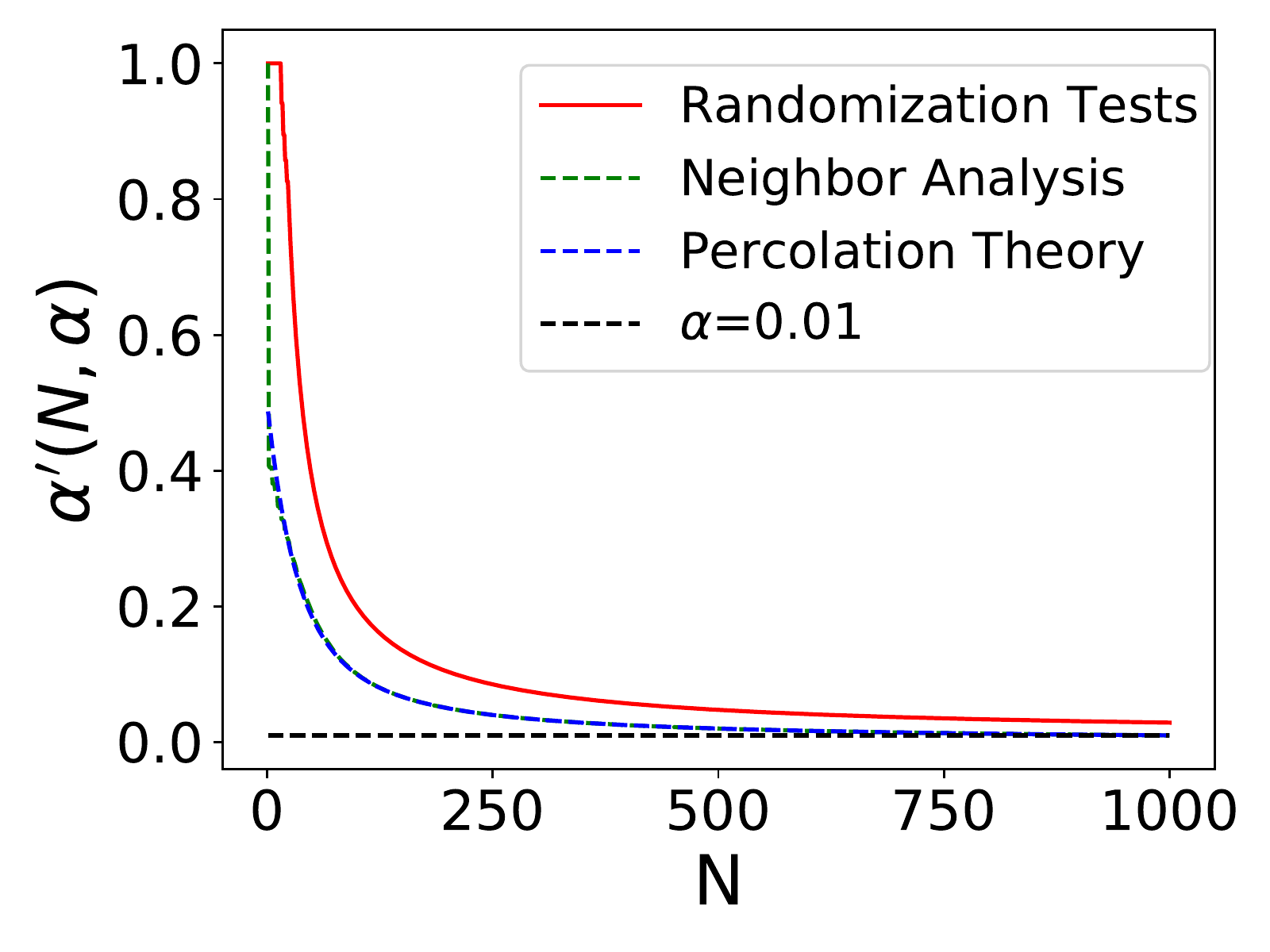}
        \vspace{-6mm}
        \caption{ER, $p=0.05, \alpha=0.01$}
    \end{subfigure}
    \begin{subfigure}[t]{0.3\textwidth}
        \centering
        \includegraphics[width=\linewidth]{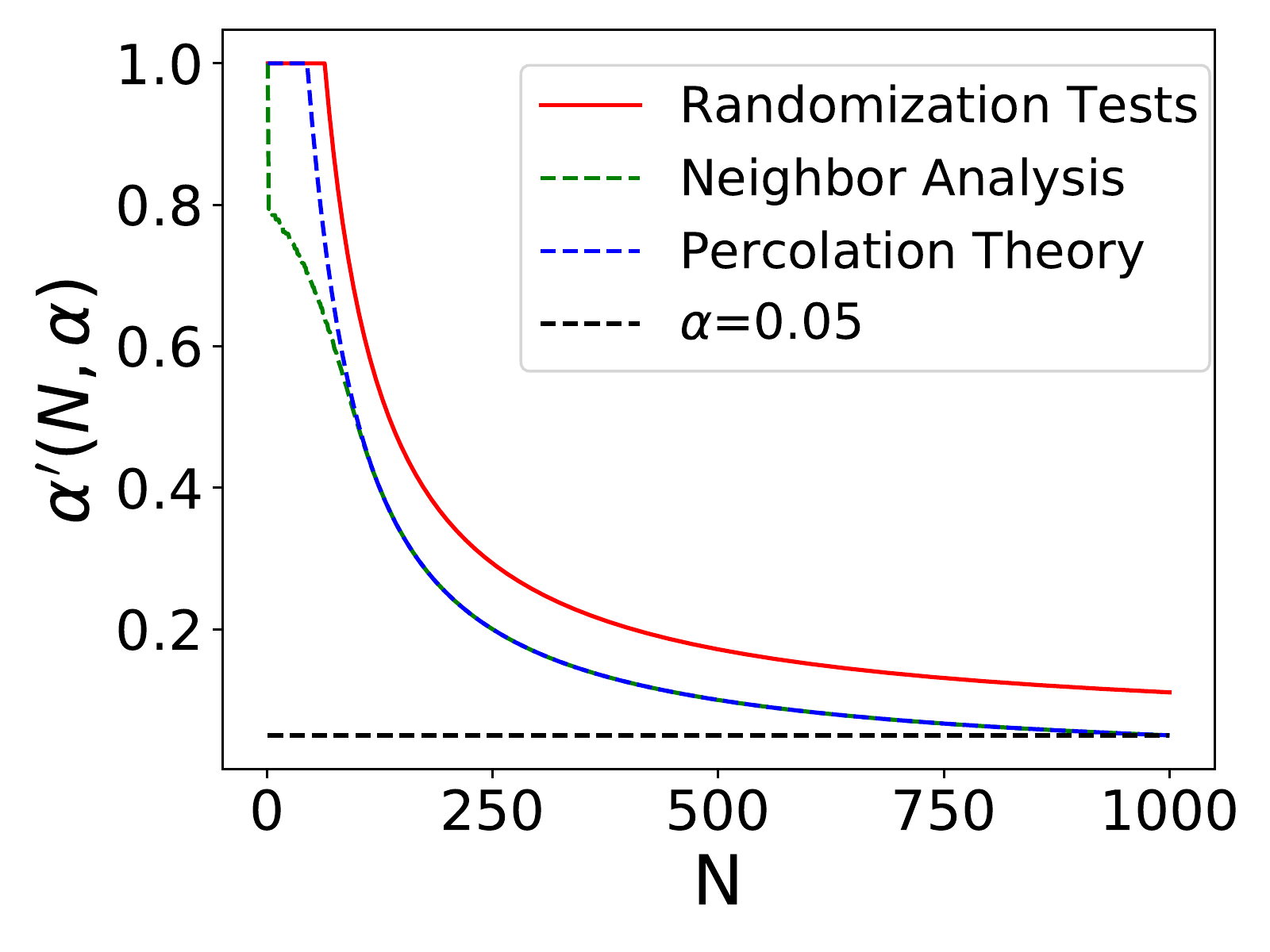}
        \vspace{-6mm}
        \caption{ER, $p=0.05, \alpha=0.05$}
    \end{subfigure}
    \begin{subfigure}[t]{0.3\textwidth}
        \centering
        \includegraphics[width=\linewidth]{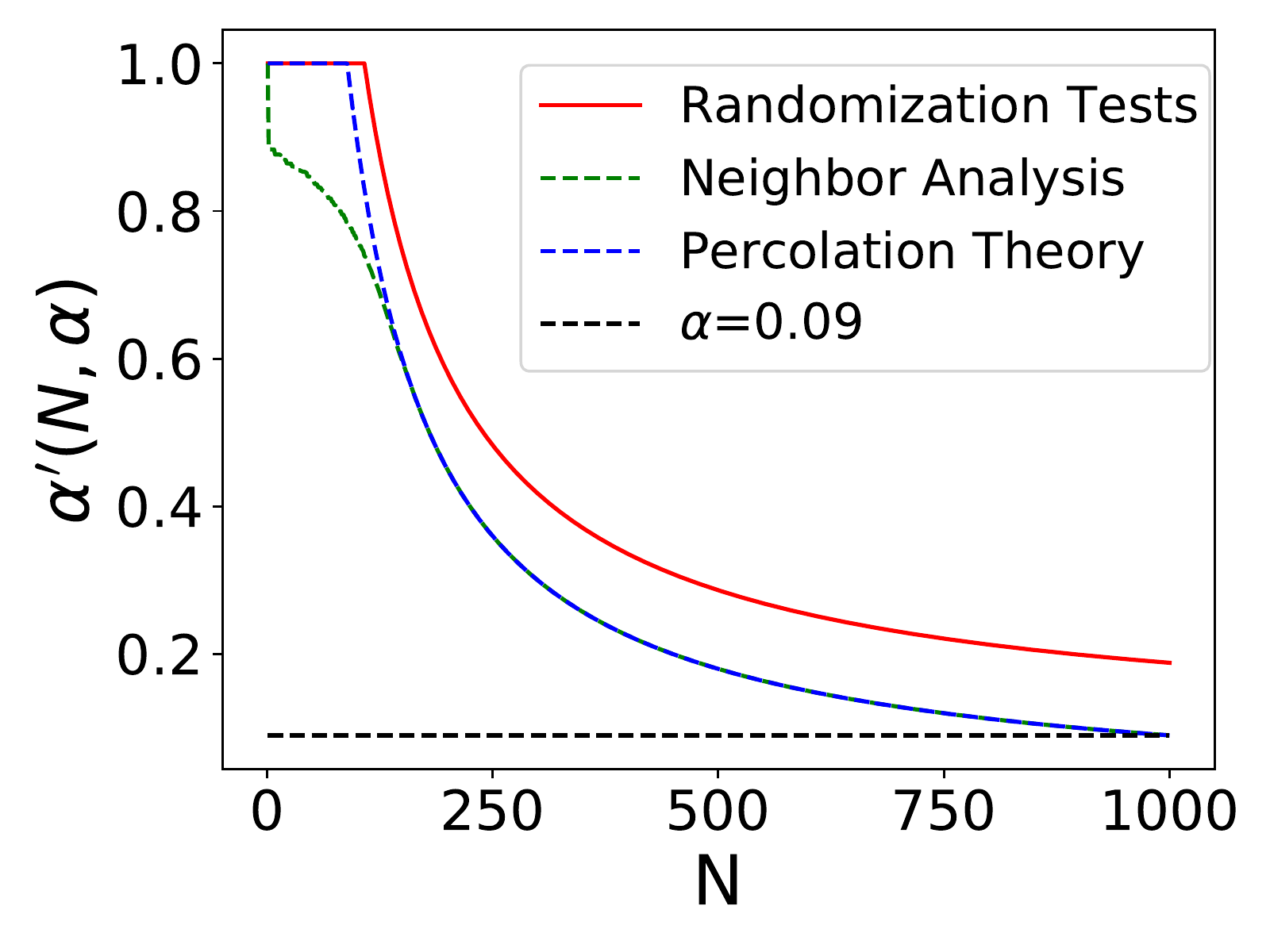}
        \vspace{-6mm}
        \caption{ER, $p=0.05, \alpha=0.09$}
    \end{subfigure}
    \caption{Lower Bounds of $\alpha'$ Compared with Empirical Distribution by Randomization Tests.}
    \label{fig:lower-bounds}
    \vspace{-2mm}
\end{figure*}

One limitation of our proposed calibration method is that it requires randomization tests to calibrate $\alpha^\prime(N, \alpha)$ which are time-consuming for large graphs. Here we explore two strategies for obtaining closed-form lower bounds of $\alpha^\prime(N, \alpha)$, thus avoiding the time-consuming randomization.

\subsubsection{Lower Bound from Network Neighborhood Analysis}
The first approach is based on neighborhood analysis, and we denote the obtained lower bound of $\alpha'$ as $\alpha_1'$. 
We lower bound the maximum number of significant nodes $N_{\alpha}$ for any given subgraph size $N$ under $\mathcal{H}_0$, by identifying a subgraph of size $N$ with expected number of significant nodes $\mathbb{E}[N_{\alpha}]$.  Given any subgraph $\mathcal{S}$, let the exterior (``ext'') degree of $\mathcal{S}$ be the number of edges between vertices $v_i\in \mathcal{S}$ and $v_i'\not\in \mathcal{S}$.

\begin{restatable}{thm}{FirstThm}
\label{thm:thm1}
For each $c \in \{1, ..., |\mathcal{V}|\}$, let $k_c$ be the largest ext-degree of a connected subgraph of size $c$.  Then for any $N\in \{1, ..., |\mathcal{V}|\}$ such that $c \le N \le c + k_c$, a lower bound for $\mathbb{E}[\max _{\mathcal{S} \in \mathbb{M}, |\mathcal{S}|=N} N_{\alpha}(\mathcal{S})]$ is: $c\alpha + \min(k_c\alpha, N-c).$
\end{restatable}

\begin{proof}
See Appendix \ref{app:proofs}.
\end{proof}
Given that high ext-degree subgraphs are more likely to connect more significant nodes, we first select the highest ext-degree node for $c=1$ and count the its neighbors as $k_c$, and continue the process by increasing $c$ and adding the highest ext-degree neighbor (i.e., with the highest number of neighbors not in $\mathcal{S}$) into $\mathcal{S}$.  This approximation potentially underestimates $k_c$ but cannot overestimate it, thus remaining a lower bound.  For each $N$, we obtain multiple lower bounds $\alpha_1'$ due to all the values of $c$ under consideration, and we choose the largest (tightest) lower bound for each $N$.

\subsubsection{Lower Bound from Percolation Theory}

The second lower bound, $\alpha_2'$, is based on percolation theory on Erdos-Renyi (ER) graphs~\cite{erdHos1960evolution, achlioptas2009explosive}. Given a large ER graph with $|\mathcal{V}|$ nodes and edge probability $p$, the average node degree is $\langle k\rangle = (|\mathcal{V}|-1)p$.  Percolation theory states that if a sufficiently large fraction of the graph nodes, $\rho > \frac{1}{\langle k \rangle}$, are ``marked'', then with high probability, there exists a connected subgraph $\mathcal{S}$ consisting of only marked nodes, with $|\mathcal{S}|$ equal to a constant fraction $P_\infty$ of $|\mathcal{V}|$.  More precisely, 
as shown by~\citet{erdHos1960evolution} and \citet{bollobas1976cliques},
$P_\infty$ is the solution to the equation, $P_\infty = \rho(1-\exp(-\langle k\rangle P_\infty)$. We apply this result by ``marking'' both significant and (as needed) non-significant graph nodes to reach the percolation threshold, allowing us to prove:
\begin{restatable}{thm}{SecondThm}
\label{thm:thm2}
For an Erdos-Renyi $(|\mathcal{V}|,p)$ graph with average degree $\langle k \rangle = (|\mathcal{V}|-1)p$, with high probability,
\[ \alpha^\prime \ge \min\left(1, \frac{\alpha|\mathcal{V}|}{N}\left(1-\exp\left(-\langle k \rangle\frac{N}{|\mathcal{V}|}\right)\right)\right).\]
\end{restatable}
\begin{proof}
See Appendix \ref{app:proofs}.
\end{proof}

We show averaged $\alpha^\prime$ lower bounds on $100$ Erdos-Renyi graphs with size $1000$ and $p=0.05$ in Figure \ref{fig:lower-bounds} for $\alpha\in [0.01, 0.05, 0.09]$. 
Compared to the true $\alpha'$ obtained from randomization testing, we observe empirically that the lower bounds $\alpha^\prime_2$ from percolation theory are tighter than the lower bounds $\alpha^\prime_1$ from neighbor analysis. However, we do not have theoretical results on the tightness of these bounds. We also note that the percolation bound is only guaranteed to be a lower bound on $\alpha'$ when the graph is Erdos-Renyi, while the neighbor analysis guarantees a lower bound for all graphs.

\subsection{Core-Tree Decomposition}
\label{sec:core_tree}

Core-whiskers (or core-periphery) structure commonly exists in many real-world networks, such as social networks, transportation networks, and the World Wide Web~\cite{rombach2014core, leskovec2009community}. 
That is, real-world networks can be viewed as a set of low tree-width periphery surrounding a core consisting of a small fraction of nodes. The core tends to be an expander graph and has similar properties to random graphs~\cite{leskovec2008statistical}. 
We first apply core-tree decomposition~\cite{maehara2014vldb} to decompose the graph into a small, dense core and a low-treewidth periphery. One benefit is that the small core keeps the general skeleton and connectivity of the entire graph, enabling adjacent, significant nodes from the whiskers to be incorporated into the detected subgraph. Thus we apply tree-node compression which 
merges the significant nodes in each single tree into an adjacent core node for follow-up optimization in a smaller core. If multiple core nodes are adjacent to a significant tree node, then we compress the significant tree node into the most significant (lowest p-value) core node. See Appendix \ref{app:core-tree} for details of the compression procedure.

\vspace{-1mm}

\section{Experiments}

\begin{figure*}[!t]
    \centering
    \includegraphics[width=1.0\linewidth]{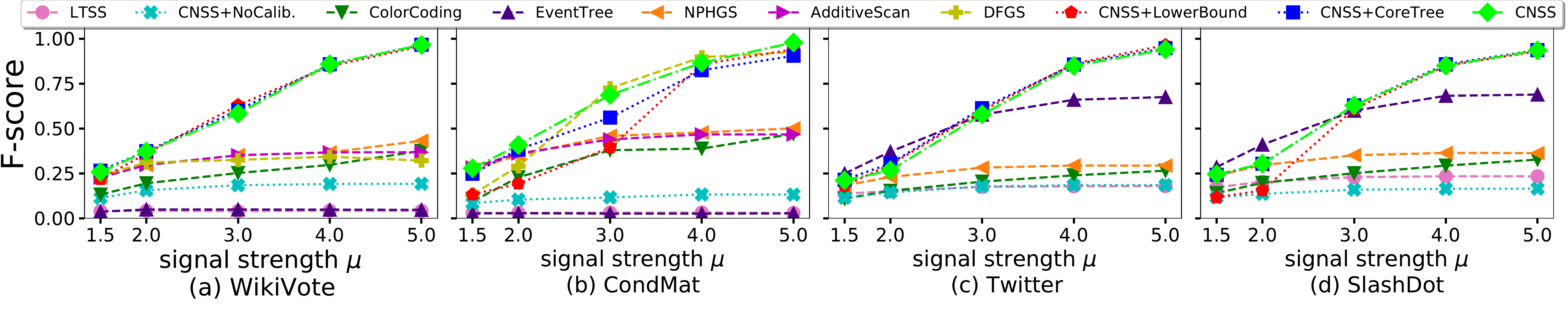}
    \vspace{-5mm}
    \caption{$F$-score of each method under different signal strengths and network structures (best viewed in color). Experiments of \texttt{AdditiveScan} and \texttt{DFGS} on \texttt{Twitter} and \texttt{SlashDot} datasets took over 2 weeks of clock time to run on $250$ CPUs, therefore we do not report them.  See Appendix~\ref{app:results}, Table~\ref{table:real_world_results}, for the corresponding numerical results with significance tests. 
    }
    \vspace{-2mm}
    \label{fig:results_fscores}
\end{figure*}

\label{sec:experiments}
In this section, we investigate four main research questions:

\textbf{Q1. Subgraph Detection:} Does our proposed \texttt{CNSS} have a better performance than state-of-the-art baselines on the task of anomalous subgraph detection?

\textbf{Q2. Calibration:} How does calibration affect detection performance, as a function of signal strength and graph structure?

\textbf{Q3. Lower Bounds:} How does the use of lower bounds of $\alpha'$, instead of $\alpha'$ obtained via randomization tests, affect detection performance?

\textbf{Q4. Core Tree Decomposition:} How does integrating core-tree decomposition into \texttt{CNSS} affect the detection performance and run time? 

\subsection{Experiment Setup}
\textbf{Datasets}: We use five semi-synthetic datasets from the Stanford Network Analysis Project (SNAP~\footnote{\url{https://snap.stanford.edu/data/}}), including 1) \texttt{WikiVote}; 2) \texttt{CondMat}; 3) \texttt{Twitter}; 4) \texttt{Slashdot}; and 5) \texttt{DBLP}. 
We leverage the graph structure of these five networks, and simulate the true subgraph $\mathcal{S}$ using a random walk with size $\approx 0.01|\mathcal{V}|$. We generate the p-value of each graph node assuming Gaussian signals, $x_i \sim N(\mu_i, 1)$ and $p_i = 1 - \texttt{CDF}(x_i)$, where $\mu_i = \mu$ for $v_i \in \mathcal{S}$, and $\mu_i = 0$ (and thus $p_i \sim \texttt{Uniform}[0,1]$) for $v_i \in \mathcal{V}\setminus \mathcal{S}$.  Here $\mu\in[1.5, 2, 3, 4, 5]$ is the signal strength and $\texttt{CDF}(\cdot)$ assumes the standard normal distribution.
We report the average performance over 50 runs of simulations of true subgraphs and p-values on each network structure.  See Appendix~\ref{app:datasets} for 
more details of datasets, and Appendix~\ref{app:piecewise} for simulation results with piecewise constant p-values rather than Gaussian signals (i.e., assuming the BJ model is correctly specified).

\noindent\textbf{Baseline Methods}: We compare \texttt{CNSS} with 6 baselines, including 1) Linear Time Subset Scanning (\texttt{LTSS})~\cite{neill2012fast}; 2) \texttt{EventTree}~\cite{rozenshtein2014event}; 3) Non-parametric Heterogeneous Graph Scan (\texttt{NPHGS})~\cite{chen2014non}; 4) \texttt{AdditiveScan}~\cite{speakman2013dynamic} 5) Depth First Graph Scan (\texttt{DFGS})~\cite{speakman2015scalable}; and 6) \texttt{ColorCoding}~\cite{cadena2019near}. 
We summarize the limitations and time complexity of each competing method in Appendix~\ref{app:baselines}.

\noindent\textbf{Ablation Study}: We also validate the effectiveness of our proposed components by comparing \texttt{CNSS} with methods: 1) \texttt{CNSS+NoCalib}, which removes calibration from \texttt{CNSS}, performing the same search but using the original $\alpha$ instead of $\alpha'$ in the score function; 2) \texttt{CNSS+LowerBound}, which replaces the randomization test with the tightest lower bound, $\max(\alpha_1',\alpha_2')$, of $\alpha'$; and 3) \texttt{CNSS+CoreTree}, which integrates core-tree decomposition into \texttt{CNSS}.

\noindent\textbf{Evaluation Metrics}: We evaluate the detection performance of the \texttt{CNSS} and competing methods on running time, detection power, precision, recall, and $F$-score (see definitions in Appendix~\ref{app:metrics}).
We report the results of $F$-score in the paper, and the remaining metrics are shown in Appendix~\ref{app:results}.

\subsection{Results}
\textbf{Subgraph Detection:} As shown in Figure~\ref{fig:results_fscores}, \texttt{CNSS} outperforms all baselines in terms of $F$-score when the event signal is strong, and it consistently has good and stable performance for different strengths of event signal and network structures.  Competing methods have low precision, and thus low $F$-score, even when the signal is strong.  In addition, we average the $F$-score over all signal strengths and network structures under consideration for each method, and we observe the following performance order: \texttt{CNSS} $>$ \texttt{CNSS+CoreTree} $>$ \texttt{CNSS+LowerBound} $>$ \texttt{DFGS} $>$ \texttt{AdditiveScan} $>$ \texttt{NPHGS} $>$ \texttt{EventTree} $>$ \texttt{ColorCoding} $>$ \texttt{CNSS+NoCalib} $>$ \texttt{LTSS}. The average $F$-score of \texttt{CNSS} is 0.603, while the best-performing baseline method DFGS has average $F$-score 0.451. See Appendix~\ref{app:results} for additional performance results. Appendix~\ref{app:piecewise} shows very similar results for piecewise constant signals.

\noindent\textbf{Calibration:}
Calibration significantly improves detection performance across different signal strengths and various network structures.
Specifically, the calibrated BJ scan statistic helps to pinpoint the true cluster as the strength of signal increases.  On the contrary, all baselines, as well as the uncalibrated version of \texttt{CNSS}, fail to achieve accurate detection (as measured by $F$-score) for all network structures under consideration.  These results demonstrate that calibration, rather than the search procedure for detecting anomalous subgraphs, is driving the difference in performance between methods.  Our proposed search procedure simply enables calibration by making it computationally feasible to find $\max_{\mathcal{S}:|\mathcal{S}|=N} N_\alpha(\mathcal{S})$ for each combination of $N$ and $\alpha$.

\noindent\textbf{Lower Bounds:}
Based on the empirical results on real-world networks, we find that our derived lower bounds provide substantial performance improvement on real-world networks, as shown as \texttt{CNSS+LowerBound} in Figure~\ref{fig:results_fscores}. 
Overall performance of the lower bound is lower than that of the randomization testing-based \texttt{CNSS} approach, particularly for low signal strengths, but \texttt{CNSS+LowerBound} substantially outperforms the baselines with respect to precision and $F$-score, particularly for stronger signals. Most importantly, computing lower bounds of $\alpha'$ is much faster than computing $\alpha'$ using randomization tests, resulting in a $400$x to $2200$x speedup for the various network structures under consideration. See Table~\ref{table:run_time_calibration} in Appendix~\ref{app:results}.

\noindent\textbf{Core Tree Decomposition:}
Core-tree decomposition substantially reduces run time for all datasets and does not significantly change detection performance. We see that \texttt{CNSS+CoreTree} is 2x faster on WikiVote dataset and 20x faster on CondMat dataset than \texttt{CNSS}. With core-tree decomposition, \texttt{CNSS} is more scalable than baseline methods including \texttt{ColorCoding}, \texttt{NPHGS}, \texttt{AdditiveScan}, and \texttt{DFGS}. While it is still more computationally expensive than \texttt{LTSS} and \texttt{EventTree}, our proposed method has much better detection performance. See Appendix~\ref{app:results} for details.

\begin{table*}[!t]
\caption{COVID-19 Case Study: Top-3 Detected Subgraphs for Each Method}\label{table:covid19_results}
\centering
\resizebox{1.0\textwidth}{!}{
\begin{tabular}{c|c|c|c|c|c|c|c}
\toprule
 & \begin{tabular}[c]{@{}c@{}}\# of weeks\\ detected\end{tabular}& \begin{tabular}[c]{@{}c@{}}avg. \# of counties\\ detected per week\end{tabular} & \begin{tabular}[c]{@{}c@{}}avg. population of\\ detected counties\end{tabular} & \begin{tabular}[c]{@{}c@{}}avg. confirmed \\cases per week\end{tabular} & \begin{tabular}[c]{@{}c@{}}avg. deaths per \\ week (2 weeks lag) \end{tabular} &\begin{tabular}[c]{@{}c@{}}avg. confirmed\\ cases rate $\times 10^{-5}$ \end{tabular} & \begin{tabular}[c]{@{}c@{}} avg. death rate\\ (2 weeks lag) $\times 10^{-5}$ \end{tabular} \\
 \cline{3-6}
 \hline
\texttt{CNSS} 1st &16  &294.19  &49369759.69 	&86596.81 	&4166.44 &175 &$\mathbf{8.44}$\\
\hline
\texttt{CNSS} 2nd &15 &60.67 &10151920.33 	&14001.60 	&520.6 &138 &5.13\\
\hline
\texttt{CNSS} 3rd &13	&7.69	&4480384.39 	&10877.31 &207 &$243$ &4.62\\
\hline\hline
\texttt{LTSS} 1st &17 &632.24 &111861408.00 &138212.47 &5986 &124 &5.35\\
\hline
\texttt{LTSS} 2nd &14 	&5.14 &802079.71 &678.43 &8.71 &85 &1.09\\
\hline
\texttt{LTSS} 3rd &4 	&9.25 &2505224.25 &1935.50 &34.25 &77 &1.37 \\
\hline\hline
\texttt{EventTree} 1st &16 &566.13 &96492336.44 &134612.50 &5739.69 &140 &5.95\\
\hline
\texttt{EventTree} 2nd &7 &2.14 &762258.57 &579.43 &32.14 &76 &4.22\\
\hline
\texttt{EventTree} 3rd &1 	&2 &299612.00 &262 &13 &87 &4.34\\
\bottomrule
\end{tabular}
}
\end{table*}

\subsection{Case Studies}
We now compare the anomalous subgraphs detected by our \texttt{CNSS} method to those identified by two of the competing methods (\texttt{LTSS} and \texttt{EventTree}) on two real-world datasets, \texttt{COVID-19} infection rates and Twitter data related to the Black Lives Matter movement.  We note that the \texttt{ColorCoding}, \texttt{NPHGS}, \texttt{AdditiveScan}, and \texttt{DFGS} approaches were not able to scale to these large real-world datasets. We show the \texttt{COVID-19} case study in the paper and \texttt{BlackLivesMatter} case study in Appendix~\ref{app:blm_case}.

\subsubsection{COVID-19 Confirmed Cases Subgraph Discovery}
We study our proposed method on \texttt{COVID-19} data\footnote{\url{https://usafacts.org/visualizations/coronavirus-covid-19-spread-map/}} to discover significant infected regions over time. This dataset contains the daily confirmed cases for 3,234 counties in the USA across over 25 weeks from January 22-July 8, 2020. We build a spatial-temporal graph with 80,850 nodes and 850,725 edges based on the weekly confirmed cases and county adjacency (see Appendix~\ref{app:covid19_case} for more details), where each node represents a county in one week. In addition to the edges that represent adjacency between counties (which are identical for each week $t$), we add an undirected temporal edge from each node $i$ in week $t$ to node $i$ in week $t+1$ as well as undirected edges from each node $i$ in week $t$ to all neighboring nodes $j$ in week $t+1$. The p-value of each node is generated based on the rank of the weekly confirmed cases to county population ratio divided by the total number of nodes in the graph. Therefore, a higher ratio of the number of weekly confirmed cases to the county population indicates a higher rank and thus a smaller p-value. 

We apply our proposed method on this spatial-temporal graph and discover three subgraphs that are significant (as identified using randomization tests on 100 runs under the null hypothesis).  The statistics of these three discovered subgraphs are shown in Table \ref{table:cnss_covid19} in Appendix~\ref{app:covid19_case}.

As shown in Table \ref{table:covid19_results}, our \texttt{CNSS} method detects a significant connected subgraph of counties that have a 42\% higher death rate two weeks later, as compared with the top-1 subgraphs detected by \texttt{LTSS} and \texttt{EventTree}. 
The use of the two-week-lagged death rate as an evaluation metric better identifies the anomaly in true \texttt{COVID-19} cases than the confirmed cases rate, which was highly affected in many areas by insufficient testing resources. (Note that death rate data is not provided to the detection algorithms.) 
The visualization of the highest-scoring subgraphs detected by different methods is shown in Appendix~\ref{app:covid19_case}. 
We see that the baseline methods cannot discover a cohesive subgraph due to the poorly calibrated objective function, instead showing a dispersed pattern across much of the country.  In contrast, our method is capable of detecting more impacted geographic regions, for better targeting of needed health resources.

\section{Limitations and Conclusions}
\label{sec:limitations}

While \texttt{CNSS} achieves state of the art performance for anomalous pattern detection on graphs, it has two main limitations. First, the randomization test-based calibration approach is time-consuming, particularly for large-scale graphs. Although our proposed closed-form lower bounds of $\alpha^\prime(N,\alpha)$ avoid the need for randomization tests and hence reduce the time cost of \texttt{CNSS} significantly, detection power is reduced when the anomalous signal strength is low, as shown in Figure~\ref{fig:results_fscores}. Second, our proposed efficient algorithm is heuristic rather than exact, and thus is not guaranteed to discover the maximum number of significant nodes $N_\alpha$ for each subgraph size $N$. However, as discussed in Section~\ref{sec:related}, subgraph detection is very challenging in the presence of connectivity constraints and no methods exist that have rigorous guarantees and at the same time are scalable to large graphs. For the calibrated scan, the computational problem is even more difficult: we must identify the subgraph with the largest number of significant p-values $N_\alpha$ for each subgraph size $N$ and significance level $\alpha$, which prevents us from using previous methods that search for a single highest-scoring subgraph.
Finally, since the problem of pattern detection in graphs is general, detection approaches could be used for negative as well as beneficial social impacts, such as monitoring of social media by an oppressive government.

In summary, we demonstrated that existing nonparametric scan statistic methods are miscalibrated for anomalous pattern detection in graphs, and developed a new statistical approach to recalibrate NPSSs to account for the multiple hypothesis testing effect of the graph structure. We proposed a more efficient algorithm and new, closed-form lower bounds, and integrated recent core-tree decomposition methods, to enable our proposed \texttt{CNSS} approach to scale to large, real-world graphs. We observed outstanding performance of our method compared with six state-of-the-art baselines on five real-world datasets under various signal strengths and network structures. 
Finally, we applied \texttt{CNSS} to two real-world applications, and found more meaningful subgraphs compared with competing methods.

\section*{Acknowledgements}
The work of Feng Chen is supported by the National Science Foundation (NSF) under Grant Number $\#1815696$ and $\#1750911$.

\bibliography{aaai22.bib}

\begin{thebibliography}{31}
\providecommand{\natexlab}[1]{#1}

\bibitem[{Achlioptas, D'Souza, and Spencer(2009)}]{achlioptas2009explosive}
Achlioptas, D.; D'Souza, R.~M.; and Spencer, J. 2009.
\newblock Explosive percolation in random networks.
\newblock \emph{Science}, 323(5920): 1453--1455.

\bibitem[{Akoglu, Tong, and Koutra(2015)}]{akoglu2015graph}
Akoglu, L.; Tong, H.; and Koutra, D. 2015.
\newblock Graph based anomaly detection and description: a survey.
\newblock \emph{Data Mining and Knowledge Discovery}, 29(3): 626--688.

\bibitem[{Alon, Yuster, and Zwick(1995)}]{alon1995color}
Alon, N.; Yuster, R.; and Zwick, U. 1995.
\newblock Color-coding.
\newblock \emph{Journal of the ACM}, 42(4): 844--856.

\bibitem[{Berk and Jones(1979)}]{berk1979goodness}
Berk, R.~H.; and Jones, D.~H. 1979.
\newblock Goodness-of-fit test statistics that dominate the Kolmogorov
  statistics.
\newblock \emph{Zeitschrift f{\"u}r Wahrscheinlichkeitstheorie und verwandte
  Gebiete}, 47(1): 47--59.

\bibitem[{Bollob{\'a}s and Erdos(1976)}]{bollobas1976cliques}
Bollob{\'a}s, B.; and Erdos, P. 1976.
\newblock Cliques in random graphs.
\newblock In \emph{Mathematical Proceedings of the Cambridge Philosophical
  Society}, volume~80, 419--427. Cambridge University Press.

\bibitem[{Cadena, Chen, and Vullikanti(2018)}]{cadena2018graph}
Cadena, J.; Chen, F.; and Vullikanti, A. 2018.
\newblock Graph anomaly detection based on Steiner connectivity and density.
\newblock \emph{Proceedings of the IEEE}, 106(5): 829--845.

\bibitem[{Cadena, Chen, and Vullikanti(2019)}]{cadena2019near}
Cadena, J.; Chen, F.; and Vullikanti, A. 2019.
\newblock Near-optimal and practical algorithms for graph scan statistics with
  connectivity constraints.
\newblock \emph{ACM Transactions on Knowledge Discovery from Data}, 13(2):
  1--33.

\bibitem[{Chen and Neill(2014)}]{chen2014non}
Chen, F.; and Neill, D.~B. 2014.
\newblock Non-parametric scan statistics for event detection and forecasting in
  heterogeneous social media graphs.
\newblock In \emph{Proceedings of the 20th ACM SIGKDD International Conference
  on Knowledge Discovery and Data Mining}, 1166--1175. ACM.

\bibitem[{Chitra et~al.(2021)Chitra, Ding, Lee, and Raphael}]{chitra2021}
Chitra, U.; Ding, K.; Lee, J. C.~H.; and Raphael, B.~J. 2021.
\newblock Quantifying and Reducing Bias in Maximum Likelihood Estimation of
  Structured Anomalies.
\newblock In \emph{Proc. 38th Intl. Conf. on Machine Learning, PMLR 139},
  1908--1919.

\bibitem[{Donoho and Jin(2004)}]{donoho2004higher}
Donoho, D.; and Jin, J. 2004.
\newblock Higher criticism for detecting sparse heterogeneous mixtures.
\newblock \emph{The Annals of Statistics}, 32(3): 962--994.

\bibitem[{Duczmal and Assuncao(2004)}]{duczmal2004simulated}
Duczmal, L.; and Assuncao, R. 2004.
\newblock A simulated annealing strategy for the detection of arbitrarily
  shaped spatial clusters.
\newblock \emph{Computational Statistics \& Data Analysis}, 45(2): 269--286.

\bibitem[{Eicker(1979)}]{eicker1979asymptotic}
Eicker, F. 1979.
\newblock The asymptotic distribution of the suprema of the standardized
  empirical processes.
\newblock \emph{The Annals of Statistics}, 116--138.

\bibitem[{Erd{\H{o}}s and R{\'e}nyi(1960)}]{erdHos1960evolution}
Erd{\H{o}}s, P.; and R{\'e}nyi, A. 1960.
\newblock On the evolution of random graphs.
\newblock \emph{Publication of the Mathematical Institute of the Hungarian
  Academy of Sciences}, 5(1): 17--60.

\bibitem[{Glaz, Pozdnyakov, and Wallenstein(2009)}]{glaz2009scan}
Glaz, J.; Pozdnyakov, V.; and Wallenstein, S. 2009.
\newblock \emph{Scan statistics: Methods and applications}.
\newblock Springer Science \& Business Media.

\bibitem[{Kulldorff(1997)}]{kulldorff1997spatial}
Kulldorff, M. 1997.
\newblock A spatial scan statistic.
\newblock \emph{Communications in Statistics-Theory and methods}, 26(6):
  1481--1496.

\bibitem[{Leskovec et~al.(2008)Leskovec, Lang, Dasgupta, and
  Mahoney}]{leskovec2008statistical}
Leskovec, J.; Lang, K.~J.; Dasgupta, A.; and Mahoney, M.~W. 2008.
\newblock Statistical properties of community structure in large social and
  information networks.
\newblock In \emph{Proceedings of the 17th international conference on World
  Wide Web}, 695--704. ACM.

\bibitem[{Leskovec et~al.(2009)Leskovec, Lang, Dasgupta, and
  Mahoney}]{leskovec2009community}
Leskovec, J.; Lang, K.~J.; Dasgupta, A.; and Mahoney, M.~W. 2009.
\newblock Community structure in large networks: Natural cluster sizes and the
  absence of large well-defined clusters.
\newblock \emph{Internet Mathematics}, 6(1): 29--123.

\bibitem[{Maehara et~al.(2014)Maehara, Akiba, Iwata, and
  Kawarabayashi}]{maehara2014vldb}
Maehara, T.; Akiba, T.; Iwata, Y.; and Kawarabayashi, K.-i. 2014.
\newblock Computing Personalized PageRank Quickly by Exploiting Graph
  Structures.
\newblock \emph{Proceedings of the VLDB Endowment}, 7(12): 1023–1034.

\bibitem[{Massey~Jr(1951)}]{massey1951kolmogorov}
Massey~Jr, F.~J. 1951.
\newblock The Kolmogorov-Smirnov test for goodness of fit.
\newblock \emph{Journal of the American Statistical Association}, 46(253):
  68--78.

\bibitem[{{McFowland III}, Speakman, and Neill(2013)}]{mcfowland2013fgss}
{McFowland III}, E.; Speakman, S.; and Neill, D.~B. 2013.
\newblock Fast generalized subset scan for anomalous pattern detection.
\newblock \emph{Journal of Machine Learning Research}, 14: 1533--1561.

\bibitem[{Neill(2009)}]{neill2009expectation}
Neill, D.~B. 2009.
\newblock Expectation-based scan statistics for monitoring spatial time series
  data.
\newblock \emph{International Journal of Forecasting}, 25(3): 498--517.

\bibitem[{Neill(2012)}]{neill2012fast}
Neill, D.~B. 2012.
\newblock Fast subset scan for spatial pattern detection.
\newblock \emph{Journal of the Royal Statistical Society: Series B (Statistical
  Methodology)}, 74(2): 337--360.

\bibitem[{Neill and Lingwall(2007)}]{neill2007nonparametric}
Neill, D.~B.; and Lingwall, J. 2007.
\newblock A nonparametric scan statistic for multivariate disease surveillance.
\newblock \emph{Advances in Disease Surveillance}, 4: 106.

\bibitem[{Qian, Saligrama, and Chen(2014)}]{qian2014connected}
Qian, J.; Saligrama, V.; and Chen, Y. 2014.
\newblock Connected sub-graph detection.
\newblock In \emph{Proceedings of the 17th International Conference on
  Artificial Intelligence and Statistics}, 796--804. PMLR.

\bibitem[{Reyna et~al.(2021)Reyna, Chitra, Elyanow, and Raphael}]{reyna2021}
Reyna, M.~A.; Chitra, U.; Elyanow, R.; and Raphael, B.~J. 2021.
\newblock NetMix: A Network-Structured Mixture Model for Reduced-Bias
  Estimation of Altered Subnetworks.
\newblock \emph{Journal of Computational Biology}, 28(5): 469--484.

\bibitem[{Rombach et~al.(2014)Rombach, Porter, Fowler, and
  Mucha}]{rombach2014core}
Rombach, M.~P.; Porter, M.~A.; Fowler, J.~H.; and Mucha, P.~J. 2014.
\newblock Core-periphery structure in networks.
\newblock \emph{SIAM Journal on Applied Mathematics}, 74(1): 167--190.

\bibitem[{Rozenshtein et~al.(2014)Rozenshtein, Anagnostopoulos, Gionis, and
  Tatti}]{rozenshtein2014event}
Rozenshtein, P.; Anagnostopoulos, A.; Gionis, A.; and Tatti, N. 2014.
\newblock Event detection in activity networks.
\newblock In \emph{Proceedings of the 20th ACM SIGKDD International Conference
  on Knowledge Discovery and Data Mining}, 1176--1185. ACM.

\bibitem[{Sharpnack, Rinaldo, and Singh(2015)}]{sharpnack2015detecting}
Sharpnack, J.; Rinaldo, A.; and Singh, A. 2015.
\newblock Detecting anomalous activity on networks with the graph Fourier scan
  statistic.
\newblock \emph{IEEE Transactions on Signal Processing}, 64(2): 364--379.

\bibitem[{Speakman, McFowland~III, and Neill(2015)}]{speakman2015scalable}
Speakman, S.; McFowland~III, E.; and Neill, D.~B. 2015.
\newblock Scalable detection of anomalous patterns with connectivity
  constraints.
\newblock \emph{Journal of Computational and Graphical Statistics}, 24(4):
  1014--1033.

\bibitem[{Speakman, Zhang, and Neill(2013)}]{speakman2013dynamic}
Speakman, S.; Zhang, Y.; and Neill, D.~B. 2013.
\newblock Dynamic pattern detection with temporal consistency and connectivity
  constraints.
\newblock In \emph{Proceedings of the 13th IEEE International Conference on
  Data Mining}, 697--706. IEEE.

\bibitem[{Wu et~al.(2016)Wu, Chen, Li, Zhou, and
  Ramakrishnan}]{wu2016efficient}
Wu, N.; Chen, F.; Li, J.; Zhou, B.; and Ramakrishnan, N. 2016.
\newblock Efficient nonparametric subgraph detection using tree shaped priors.
\newblock In \emph{Proceedings of the 30th AAAI Conference on Artificial
  Intelligence}, volume~30. The AAAI Press.

\end{thebibliography}




\clearpage
\appendix
\section*{Technical Appendix}

This Technical Appendix, a supplement to ``Calibrated Nonparametric Scan Statistics for Pattern Detection in Graphs'', consists of:
\begin{itemize}
\item Appendix~\ref{app:npss-background}, Additional Background on Nonparametric Scan Statistics
\item Appendix~\ref{app:cnss-details}, Additional Details of \texttt{CNSS} 
\item Appendix~\ref{app:exp}, Additional Experimental Details 
\item Appendix~\ref{app:case}, Case Studies
\end{itemize}

\section{Additional Background on Nonparametric Scan Statistics}
\label{app:npss-background}

Here we present additional background on nonparametric scan statistics~\cite{neill2007nonparametric,mcfowland2013fgss,chen2014non}. 

As we describe in the main paper, the fundamental problem that NPSSs solve is to find a subset of the data $\mathcal{S}$, often subject to additional constraints (such as connectedness in the graph setting), and a corresponding significance level $\alpha$, such that the proportion of significant p-values (at level $\alpha$) in $\mathcal{S}$ is significantly higher than expected. Or equivalently, if p-values are drawn uniformly on [0,1] under the null hypothesis $\mathcal{H}_0$, and under the alternative hypothesis $\mathcal{H}_1(\mathcal{S})$
the p-values in subset $\mathcal{S}$ are drawn with a higher than expected proportion of low (significant) p-values, we wish both to distinguish $\mathcal{H}_0$ from $\mathcal{H}_1$, thus detecting whether a signal is present, and if so, to correctly identify the affected subset $\mathcal{S}$. 

More precisely, NPSSs optimize an objective function 
$F(\mathcal{S})=\max_{\alpha \leq \alpha_{\max}} \Phi\left(\alpha, N_{\alpha}(\mathcal{S}),N(\mathcal{S})\right)$, where 
$\Phi(\cdot)$ compares the observed number of significant p-values $N_\alpha(\mathcal{S})$ at level $\alpha$ to the expected number of significant p-values $\mathbb{E} \left[ N_{\alpha}(\mathcal{S}) \right] = \alpha N(\mathcal{S})$ under the null hypothesis  $\mathcal{H}_0$. The expectation $\mathbb{E} \left[ N_{\alpha}(\mathcal{S}) \right] = \alpha N(\mathcal{S})$ follows because, under $\mathcal{H}_0$, the current data from which the p-values are generated is exchangeable with the historical data against which the current data values are ranked, leading to p-values that are asymptotically uniform on [0,1] under the null.  We discuss the Berk-Jones statistic $\Phi_{BJ}(\cdot)$ in detail in Appendix~\ref{app:npss-assumptions}, and other variants of NPSS in Appendix~\ref{app:npss-variants}.

Critically, NPSSs optimize the significance level $\alpha$ between $0$ and some constant $\alpha_{\max} < 1$.  As noted by~\citet{mcfowland2013fgss}, the purpose of maximizing over a range of $\alpha$ values is to ensure that the statistic can reliably detect either a small number of highly significant p-values or a larger number of moderately significant p-values.  If $\alpha$ was fixed at a high value, the statistic would have poor detection performance in the former case; if $\alpha$ was fixed at a low value, it would perform poorly in the latter case. However, maximization over $\alpha$ also presents a serious drawback for the uncalibrated scan: for large real-world graphs, NPSSs select overly large values of $\alpha$, contributing to their failure to correctly identify the affected subgraph $\mathcal{S}$.

A number of algorithms have been proposed to optimize $F(\mathcal{S})$ over connected subgraphs, including \texttt{DFGS}~\cite{speakman2015scalable}, \texttt{AdditiveScan}~\cite{speakman2013dynamic}, \texttt{NPHGS}~\cite{chen2014non}, and \texttt{ColorCoding}~\cite{cadena2019near}, but none of these approaches are both exact and scalable to large graphs.  Moreover, calibration of NPSSs requires us to identify the subgraph with the largest number of significant p-values $N_\alpha$ for \emph{each subgraph size} $N$ and \emph{each significance level} $\alpha$, rather than a single highest-scoring subgraph, thus necessitating our new (approximate) optimization algorithm described in the main paper and in Appendix~\ref{app:alg:merge} below.

Once the highest scoring subgraph, $\mathcal{S}^\ast = \arg\max_{\mathcal{S} \in \mathbb{M}} F(\mathcal{S})$, has been identified, randomization testing can be used to compute the statistical significance of $\mathcal{S}^\ast$.  To do so, a large number $R$ of replica graphs are generated under the null hypothesis, i.e., each replica graph has the same structure as the original graph but all p-values $p_i$ are drawn i.i.d.~from $\texttt{Uniform}[0,1]$.  The same search procedure is used to identify the highest scoring subgraph $\mathcal{S}^{(r)}$ for each replica graph $r=1\ldots R$, and the score $F(\mathcal{S}^\ast)$ is compared to the distribution of replica scores $F(\mathcal{S}^{(r)})$.  To be significant at the standard $0.05$ significance level, $F(\mathcal{S}^\ast)$ must exceed the 95th percentile of the null distribution.  This standard approach corrects for \emph{multiple testing}, in that the family-wise error rate (probability of detecting any false positive subgraphs if data is generated under the null) is bounded by the nominal level (e.g., 0.05).  However, we note that it does not correct for \emph{miscalibration} across subgraph sizes $N$ and significance levels $\alpha$, in that  large, high-scoring subgraphs $\mathcal{S^\ast}$ are likely to be detected even when the true subset $\mathcal{S}$ is small or no signal is present.

Finally, as is typical in the scan statistics literature, we can perform \emph{multiple cluster detection} by repeated single cluster detection.  That is, after we detect the single highest-scoring subgraph, and test it for statistical significance, we ``remove'' that subgraph from the data in one of two ways, either assigning the p-value of each detected node as 1, or deleting the detected nodes from the network structure.  (The former approach allows secondary clusters to overlap with the primary cluster, while the latter approach does not.)  In either case, we apply the same procedure to the updated network to detect the new highest-scoring subgraph, compare the score of this cluster to the significance threshold determined by randomization testing, and repeat until no further significant clusters are present.  This statistical testing approach is conservative for the secondary clusters, and several variants of the multiple cluster scan have been proposed to increase detection power for secondary clusters (Zhang, Assuncao, and Kulldorff, 2010; Li et al., 2011).

In the remainder of Appendix~\ref{app:npss-background}, we present additional details on the fundamental modeling assumptions of the Berk-Jones nonparametric scan statistic (Appendix~\ref{app:npss-assumptions}), computation of empirical p-values (Appendix~\ref{app:npss-pvals}), other variants of NPSS (Appendix~\ref{app:npss-variants}), and differences between NPSS and the Gaussian scan approach of~\citet{reyna2021} and~\citet{chitra2021} (Appendix~\ref{app:npss-differences}).

\subsection{Fundamental modeling assumptions}
\label{app:npss-assumptions}

In this section, we describe the fundamental modeling assumptions of nonparametric scan statistics, following~\citet{mcfowland2013fgss}, and focusing primarily on the Berk-Jones (BJ) likelihood ratio statistic.  Unlike parametric scan statistics such as the Poisson and Gaussian statistics~\cite{kulldorff1997spatial,neill2009expectation}, NPSSs do not assume that the raw data is drawn from any particular parametric distribution.  Instead, the data is converted to empirical p-values by ranking the current data against a reference distribution (e.g., historical values), as described in Appendix~\ref{app:npss-pvals}.  The assumption under the null hypothesis $\mathcal{H}_0$ is that the current and historical data are exchangeable, and thus, ranking the current data against the historical data (and normalizing) will result in empirical p-values that are uniformly distributed on [0,1].  This also implies that, for any significance level $\alpha$, the probability that a given p-value is significant ($p_i < \alpha$) is equal to $\alpha$.  

Under the alternative hypothesis $\mathcal{H}_1(\mathcal{S})$, NPSSs must make some assumption about how the distribution of p-values in subset $\mathcal{S}$ differs from the uniform distribution on [0,1].  For the BJ statistic, the assumption under $\mathcal{H}_1(\mathcal{S})$ is that there exist some $\alpha$ and $\beta$, where $0 < \alpha < \beta \le 1$, such that the probability that a given p-value $p_i \in \mathcal{S}$ is significant ($p_i < \alpha$) is equal to $\beta$.  This is typically framed as an assumption that the distribution of p-values $p_i$ for $v_i \in \mathcal{S}$ is piecewise constant.  More precisely, we have $\mathcal{H}_0: p_i \sim \texttt{Uniform[0,1]} \: \forall v_i \in \mathcal{V}$.  Under $\mathcal{H}_1(\mathcal{S})$, we have $p_i \sim \texttt{Uniform}[0,\alpha]\ \text{with probability}\ \beta \ \text{and} \ p_i \sim \texttt{Uniform}[\alpha,1]\ \text{with probability}\ 1-\beta, \: \forall v_i\in \mathcal{S}, \ \text{for}\ \beta > \alpha$; and $p_i \sim \texttt{Uniform}[0,1]\ \forall v_i \in \mathcal{V}\setminus\mathcal{S}$. Here the values of both $\alpha$ and $\beta$ are fit by maximum likelihood estimation.

The resulting generalized log-likelihood ratio scan statistic can be written as:
\[ F(\mathcal{S}) = \max_{\beta>\alpha} \log \frac{\beta^{N_\alpha(\mathcal{S})} (1-\beta)^{N(\mathcal{S})-N_\alpha(\mathcal{S})}}{\alpha^{N_\alpha(\mathcal{S})} (1-\alpha)^{N(\mathcal{S})-N_\alpha(\mathcal{S})}} \]
\[ = \max_{\beta>\alpha} N_\alpha(\mathcal{S}) \log \left(\frac{\beta}{\alpha}\right) + (N(\mathcal{S})-N_\alpha(\mathcal{S})) \log \left(\frac{1-\beta}{1-\alpha}\right).\]

Then plugging in the maximum likelihood estimate $\beta = N_\alpha(\mathcal{S})/N(\mathcal{S})$ and simplifying, we obtain:
\[ F(\mathcal{S}) = \max_\alpha N(\mathcal{S}) \: \texttt{KL}\left(\frac{N_\alpha(\mathcal{S})}{N(\mathcal{S})},\alpha\right), \]
where the Kullback-Liebler (\texttt{KL}) divergence is defined as $\texttt{KL}(a,b) = a \log \frac{a}{b} + (1-a) \log \frac{1-a}{1-b}$, if $a>b$, and 0 otherwise. (Note that we use a one-sided form of \texttt{KL} divergence throughout, since we care only about subgraphs with a \emph{higher} than expected proportion of significant p-values.)  Thus we can see that the score $F(\mathcal{S}) = \max_{\alpha} \Phi_{BJ}(\alpha,N_\alpha(\mathcal{S}),N(\mathcal{S}))$ is maximized over both $\alpha$ and $\mathcal{S}$, identifying a subgraph $\mathcal{S}$ and significance level $\alpha$ for which $\mathcal{S}$ has a higher than expected proportion of significant p-values at level $\alpha$.

The assumption of piecewise constant p-values under the alternative hypothesis $\mathcal{H}_1(\mathcal{S})$ is a relatively lightweight assumption, in that all significant p-values at level $\alpha$ are treated identically: for a given $\alpha$, the precise value of each p-value does not impact the score $F(\mathcal{S})$, only whether or not that p-value is less than $\alpha$.  The resulting log-likelihood ratio statistic is equivalent to a log-likelihood ratio defined in terms of the number of significant p-values at level $\alpha$.  That is, the null hypothesis $\mathcal{H}_0$
can be written as $N_\alpha(\mathcal{S}) \sim \texttt{Binomial}(N(\mathcal{S}),\alpha)$ for all $\mathcal{S}$, and the alternative hypothesis $\mathcal{H}_1(\mathcal{S})$ can be written as $N_\alpha(\mathcal{S}) \sim \texttt{Binomial}(N(\mathcal{S}),\beta)$ for $\beta > \alpha$.  Again, we must not only optimize over $\beta$, using the maximum likelihood estimate $\beta = \frac{N_\alpha(\mathcal{S})}{N(\mathcal{S})}$, but also the significance level $\alpha$, to identify the highest-scoring subgraph $\mathcal{S}$.

\subsection{Computation of empirical p-values}
\label{app:npss-pvals}

As noted in Section~\ref{sec:npss} of the main paper, empirical p-values $p_i$ for the nonparametric scan statistic are computed for each graph node $v_i$ using the two-stage empirical calibration process described by~\citet{chen2014non}.  We provide more details on this process and explain why it follows that p-values are asymptotically uniform on [0,1] under the null hypothesis $\mathcal{H}_0$. Assume that node $v_i$ has a current feature vector ${\bf x}_i \in \mathbb{R}^{N}$ and historical feature vectors $\{\mathbf{x}_i^{(1)}, \cdots, \mathbf{x}_i^{(T)}\}$.  Moreover, under the null hypothesis $\mathcal{H}_0$, we assume that the current data is exchangeable with the historical data, i.e., ${\bf x}_i$ and $\mathbf{x}_i^{(1)}, \cdots, \mathbf{x}_i^{(T)}$ are all drawn from the same (unknown) distribution.

We first consider the simplest case, in which each node $v_i$ has only a single feature $x_i$.  In this case, the empirical calibration process reduces to ranking the current feature value $x_i$ against its historical values $x_i^{(1)}, \cdots, x_i^{(T)}$ and normalizing, i.e.,
\[ p_i = \frac{1+\sum_{t=1\ldots T} \mathbf{1}\{x_i^{(t)} \ge x_i\}}{1+T} \]

Here we assume one-sided p-values (i.e., higher values of $x_i$ correspond to lower, more significant p-values), but two-sided p-values, or one-sided p-values where lower values of $x_i$ are more significant, can be easily constructed as well. See~\citet{mcfowland2013fgss} for details.

Under the null hypothesis of exchangeability, it is easy to see that $p_i$ is discrete uniform, taking on values $\frac{1}{1+T}, \frac{2}{1+T}, \cdots, 1$ with equal probabilities, and converges in distribution to \texttt{Uniform}[0,1] as the number of historical observations $T$ becomes large.  An alternative is to use p-value ranges, as proposed by~\citet{mcfowland2013fgss}, which guarantee uniform (rather than asymptotically uniform) p-values under the null.   We also note that an arbitrary reference set can be used in place of historical data for node $v_i$, in which case the assumption under $\mathcal{H}_0$ becomes exchangeability of the current observation with that reference set. See~\citet{chen2014non} for details.

In the more general case where the feature vectors ${\bf x}_i$ and ${\bf x}_i^{(t)}$ have more than one feature, the two-stage empirical calibration process first ranks each feature value $x_{ij}$ against its historical values $x_{ij}^{(t)}$, and computes a ``first-stage p-value'' corresponding to each feature value as above:
\[ p_{ij} = \frac{1+\sum_{t=1\ldots T} \mathbf{1}\{x_{ij}^{(t)} \ge x_{ij}\}}{1+T}. \]
A similar ``first-stage p-value'' is computed for each historical value:
\[ p_{ij}^{(t)} = \frac{1+\mathbf{1}\{x_{ij} \ge x_{ij}^{(t)}\} + \sum_{t'=1\ldots T,t'\ne t} \mathbf{1}\{x_{ij}^{(t')} \ge x_{ij}^{(t)}\}}{1+T}. \]

Next, the two-stage empirical calibration process computes the minimum (most significant) p-value for each feature vector:
\[ p_{i,\min} = \min_j p_{ij}, \]
\[ p_{i,\min}^{(t)} = \min_j p_{ij}^{(t)}. \]
And finally, it computes the ``second-stage p-value'', using the normalized rank of the minimum p-value $p_{i,\min}$ (here, lower is more significant):
\[ p_i = \frac{1+\sum_{t=1\ldots T} \mathbf{1}\{p_{i,\min}^{(t)} \le p_{i,\min}\}}{1+T}. \]

The exchangeability of the first-stage p-values $p_{ij}$ and $p_{ij}^{(t)}$, and the exchangeability of their minima $p_{i,\min}$ and $p_{i,\min}^{(t)}$, follow from the exchangeability of $x_{ij}$ and $x_{ij}^{(t)}$ under the null, and thus the second-stage p-values are asymptotically uniform on [0,1] under $\mathcal{H}_0$ as above. See Theorem 1 of~\citet{chen2014non} and Section 2.2 of~\citet{mcfowland2013fgss} for additional details.

The uniformity of p-values is critical since it follows that the expected proportion of significant p-values for a randomly selected connected subset under $\mathcal{H}_0$ is equal to the significance level $\alpha$. This is the basis for the NPSS approach of comparing the observed proportion of p-values that are significant at level $\alpha$ to the expected proportion of significant p-values $\alpha$.  As we discuss in detail in the main paper and in Appendix~\ref{app:correctness} below, this uncalibrated NPSS approach fails to adjust for the multiplicity of subgraphs, and thus we propose to calibrate $\alpha$ by replacing it with the expected \emph{maximum} proportion of significant p-values $\alpha^\prime(N,\alpha)$.

\subsection{Variants of nonparametric scan}
\label{app:npss-variants}

While we focused on the Berk-Jones (BJ) nonparametric scan statistic, our calibration approach can easily be applied to other NPSSs such as Higher Criticism (HC) and Kolmogorov-Smirnov (KS), since like BJ these statistics can be written as the maximum (over $\alpha$ values from 0 to $\alpha_{\max}$) of a scaled divergence $\Phi(\alpha,N_\alpha(S),N(S))$ between the observed and expected proportion of significant p-values at level $\alpha$.  Some examples are provided below: \\

\noindent\textbf{Berk-Jones}:\\
\begin{equation}
    \Phi_{BJ}\left(\alpha, N_{\alpha}(\mathcal{S}), N(\mathcal{S})\right) = N(\mathcal{S}) \: \texttt{KL}\left(\frac{N_\alpha(\mathcal{S})}{N(\mathcal{S})},\alpha\right)
\end{equation}

\noindent\textbf{Higher Criticism}:\\
\begin{equation}
    \Phi_{HC}\left(\alpha, N_{\alpha}(\mathcal{S}), N(\mathcal{S})\right) = \frac{N_{\alpha}(\mathcal{S}) - \alpha N(\mathcal{S})}{\sqrt{N(\mathcal{S}) \alpha (1-\alpha)}}
\end{equation}

\noindent\textbf{Kolmogorov–Smirnov}:\\ 
\begin{equation}
    \Phi_{KS}\left(\alpha, N_{\alpha}(\mathcal{S}), N(\mathcal{S})\right) = \sqrt{N(\mathcal{S})}\cdot \left(\frac{N_{\alpha}(\mathcal{S})}{N(\mathcal{S})} - \alpha \right) 
\end{equation}

Note that in each of these cases, we use a one-sided divergence, since we only wish to detect subgraphs where the observed proportion of significant p-values $N_\alpha(\mathcal{S})/N(\mathcal{S})$ is \emph{greater} than $\alpha$.  We have defined the one-sided \texttt{KL} divergence in Appendix~\ref{app:npss-assumptions} above, and for the other statistics we simply set them to zero whenever $N_\alpha(\mathcal{S})/N(\mathcal{S}) \le \alpha$.

Once the value of $\alpha^\prime (N, \alpha) =\mathbb{E}[\max_{\mathcal{S}: |\mathcal{S}| = N} N_\alpha(\mathcal{S}) / N]$ has been computed for each $N$ and $\alpha$, this value can be substituted for $\alpha$ in any of the above equations to obtain a calibrated nonparametric scan statistic: \\ 

\noindent\textbf{Calibrated Berk-Jones}:
\begin{equation}
    \begin{split}
    &\quad \Phi_{CBJ}\left(\alpha, N_{\alpha}(\mathcal{S}), N(\mathcal{S})\right)\\
    &=  N(\mathcal{S}) \: \texttt{KL}\left(\frac{N_\alpha(\mathcal{S})}{N(\mathcal{S})},\alpha^\prime(N(\mathcal{S}),\alpha)\right) \\
    \end{split}
\end{equation}

\noindent\textbf{Calibrated Higher Criticism}:
\begin{equation}
    \begin{split}
    &\quad \Phi_{CHC}\left(\alpha, N_{\alpha}(\mathcal{S}), N(\mathcal{S})\right)\\
    &= \frac{N_{\alpha}(\mathcal{S}) - \alpha'(N(\mathcal{S}), \alpha) N(\mathcal{S})}{\sqrt{N(\mathcal{S}) \alpha'(N(\mathcal{S}), \alpha) (1-\alpha'(N(\mathcal{S}), \alpha))}} \\
    \end{split}
\end{equation}

\noindent\textbf{Calibrated Kolmogorov–Smirnov}:
\begin{equation}
    \begin{split}
    &\quad \Phi_{CKS}\left(\alpha, N_{\alpha}(\mathcal{S}), N(\mathcal{S})\right)\\
    &= \sqrt{N(\mathcal{S})}\cdot \left(\frac{N_{\alpha}(\mathcal{S})}{N(\mathcal{S})} - \alpha'(N(\mathcal{S}), \alpha) \right)
    \end{split}
\end{equation}

Our future work will evaluate the impact of calibration on the detection performance of HC, KS, and other nonparametric scan statistics, as well as exploring whether the calibration approach can be adapted to other scan statistics outside the NPSS family.

We note that the HC nonparametric scan statistic, despite its interpretation as a Gaussian approximation of the BJ likelihood ratio statistic, is distinct from the Gaussian scan statistics described in the following subsection.  HC does not assume that individual p-values in $\mathcal{S}$ follow a (transformed) mean-shifted Gaussian distribution under the alternative hypothesis $\mathcal{H}_1(\mathcal{S})$.  The HC statistic (like BJ) is based only on the number of significant p-values at level $\alpha$, which implicitly assumes that the pdf of the p-values is piecewise constant.  Rather, it is the number of significant p-values $N_\alpha(\mathcal{S})$ that is assumed to be Gaussian, as a large sample Gaussian approximation to the Binomial distribution for $N_\alpha(\mathcal{S})$ assumed by BJ. Assuming that the number of p-values $N(\mathcal{S})$ is large, then by the Central Limit Theorem, the number of significant p-values at level $\alpha$, $N_\alpha(\mathcal{S}) \sim \texttt{Binomial}(N(\mathcal{S}), \alpha)$, converges in distribution to $\texttt{Gaussian}(\alpha N(\mathcal{S}),\alpha(1-\alpha) N(\mathcal{S}))$, and the HC statistic is the z-score of $N_\alpha(\mathcal{S})$ given this Gaussian distribution, which is similar to a Wald test.  Thus HC is a large-sample Gaussian approximation to BJ, regardless of whether the individual p-values are Gaussian.

\subsection{Differences between NPSS and Gaussian scan}
\label{app:npss-differences}

As we note in the main paper, two recent papers~\cite{reyna2021, chitra2021} investigate miscalibration of scan statistics in the Gaussian setting, demonstrating that the Gaussian scan statistic tends to identify subgraphs that are much larger than the true anomalous subgraph, and presenting an approach (based on Gaussian mixture modeling) that can reduce this bias.  
In this subsection we explain how our nonparametric scan statistic setting is fundamentally different than the Gaussian setting, resulting in a different source of bias (miscalibration of the parameter $\alpha$) and thus motivating a different approach to correcting this bias, i.e., recalibration of $\alpha$ using $\alpha^\prime(N,\alpha)$.

First, we note that the typical use of the Gaussian scan, assuming that the raw data follows a Gaussian distribution and computing the likelihood ratio statistic based on this assumption, differs from the nonparametric scan setting where the raw data is converted to p-values that (because of the assumption of exchangeability of current and historical observations) will be uniformly distributed on [0,1] under the null hypothesis.  Nonparametric scans do not rely on strong distributional assumptions (like Gaussianity) of the raw data, but rather assume that sufficient reference data (e.g., historical data) are available to convert the raw data to empirical p-values (by ranking it against the reference data and normalizing) as described in Appendix~\ref{app:npss-pvals} above.

However, the Gaussian scan approach of~\citet{reyna2021} and~\citet{chitra2021} differs from this typical use in that the raw data are first converted to p-values $p_i$ by ranking and then converted to Gaussian z-scores $z_i$ by the Gaussian probability integral transform, $z_i = \texttt{CDF}^{-1}(1-p_i)$, where $\texttt{CDF}(\cdot)$ assumes the standard normal distribution.  The null hypothesis $\mathcal{H}_0$ is that $z_i \sim \texttt{Gaussian}(0,1)$ for all vertices $v_i$, while the alternative hypothesis $\mathcal{H}_1(\mathcal{S})$ is that $z_i \sim \texttt{Gaussian}(\mu,1)$ for $v_i \in \mathcal{S}$ and $z_i \sim \texttt{Gaussian}(0,1)$ for $v_i \in \mathcal{V} \setminus \mathcal{S}$.  Converting back to p-values, we can write equivalently that $p_i \sim \texttt{Uniform}[0,1]$ under $\mathcal{H}_0$, as in the nonparametric scan, and $p_i \sim 1-\texttt{CDF}(z_i)$ where $z_i \sim \texttt{Gaussian}(\mu,1)$ for $v_i \in \mathcal{S}$ under the alternative hypothesis $\mathcal{H}_1(\mathcal{S})$, where the parameter $\mu$ is fit by maximum likelihood estimation.  

While this framing of the Gaussian scan leads to identical null hypotheses, with p-values uniformly distributed on [0,1], the NPSS alternative hypothesis $\mathcal{H}_1(\mathcal{S})$ is fundamentally different from the Gaussian setting in two ways.  First, as derived in Appendix~\ref{app:npss-assumptions} above, the nonparametric scan fits two parameters $\alpha$ (significance level of the subgraph) and $\beta$ (fraction of significant nodes in the subgraph) by maximum likelihood estimation, while the Gaussian scan fits only the parameter $\mu$.  The additional parameter $\alpha$ is critical to the NPSS setting for two reasons: (1) maximizing over a range of significance levels $\alpha$ gives the nonparametric scan high power to detect compact signals (a small number of highly significant p-values), dispersed signals (a large number of slightly significant p-values), or anything in between; and (2) as we show, miscalibration of the estimated proportion $N_\alpha/N$ across different $\alpha$ values leads to an incorrect (overly large) choice of $\alpha$, obscuring the true signal, but using the corrected $\alpha^\prime(N,\alpha) = \mathbb{E}[\max_{\mathcal{S} \in \mathbb{M}: |\mathcal{S}| = N} N_\alpha(\mathcal{S}) / N]$ in place of the uncorrected $\alpha = \mathbb{E}[N_\alpha(\mathcal{S}) / N]$ solves this problem.  The previous approaches for calibrating the Gaussian scan cannot solve this issue of miscalibration over $\alpha$, nor is a Gaussian mixture modeling approach appropriate when p-values do not follow a transformed Gaussian under the null. On the other hand, as our experimental results show, our new approach to calibrating NPSSs is effective regardless of whether the signal is a (transformed) Gaussian or piecewise constant p-values.

\begin{figure}[!ht]
\begin{subfigure}{0.23\textwidth}
         \centering
          \includegraphics[width=1.\textwidth]{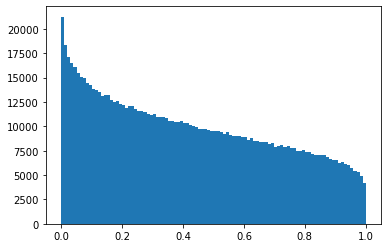}
\end{subfigure}
\begin{subfigure}{0.23\textwidth}
          \centering
          \includegraphics[width=1.\textwidth]{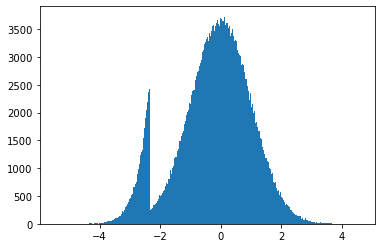}
\end{subfigure}
\caption{Examples of the difference in signal shape between piecewise constant and transformed Gaussian p-values. Left panel: histogram of p-values corresponding to a transformed Gaussian signal ($\mu=0.3$). Right panel: histogram of z-scores corresponding to a piecewise-constant p-value signal.}
\label{fig:Gaussian_mismatch_example}
\end{figure}

A second fundamental difference is that, even for a given value of $\alpha$, the nonparametric scan assumes a qualitatively different signal shape, i.e., piecewise constant rather than transformed Gaussian p-values, under the alternative hypothesis $\mathcal{H}_1(\mathcal{S})$.  As shown in the left panel of Figure~\ref{fig:Gaussian_mismatch_example}, 
if we plot the histogram of p-values corresponding to a transformed mean-shifted Gaussian (with $\mu=0.3$ in our example), we see that they are decreasing on the entire interval [0,1], as opposed to the NPSS assumption of piecewise constant p-values.  Conversely, if we plot the histogram of z-scores corresponding to the alternative hypothesis for the BJ statistic (with 10\% of p-values significant at $\alpha=.01$ in our example), as in the right panel of Figure~\ref{fig:Gaussian_mismatch_example}, we see that the distribution is neither Gaussian nor a mixture of Gaussians, as the mixture component with low p-values is heavily left-skewed, with a sharp cutoff of the right tail at the z-score corresponding to $\alpha$. 

While a performance comparison of nonparametric and (transformed) Gaussian scans is beyond the scope of this paper, we note that the differing alternative hypotheses have important implications as to what types of signals can be detected.  One might expect the transformed Gaussian scan to have somewhat higher power if its modeling assumptions are correct and the signals are in fact Gaussian. However, certain signals (such as cases where the p-value distribution is symmetric around $p=0.5$) would not be detectable in the Gaussian setting, while the nonparametric scans can detect signals as long as there exist some $(\alpha,\beta)$, where $\beta > \alpha$, such that $\mbox{Pr}(p < \alpha) = \beta$.

In summary, the nonparametric scan statistics that we consider here differ fundamentally from the (transformed) Gaussian scans considered by~\citet{reyna2021} and~\citet{chitra2021}, both in their assumptions about the true signal (distribution of p-values under $\mathcal{H}_1$) and in their maximization over a critical parameter, the significance level $\alpha$.  These differences motivate both our new empirical studies to understand and quantify the miscalibration of $\alpha$ for the (uncalibrated) nonparametric scan statistics, as well as our new approach to correctly calibrate $\alpha$. 

\section{Additional Details of CNSS}
\label{app:cnss-details}

\subsection{Correctness of the calibration approach}
\label{app:correctness}

As we discuss in detail in the main paper, the primary source of miscalibration for NPSSs is the discrepancy between $\alpha$ and $\alpha^\prime(N,\alpha)$.  That is, for a given significance level $\alpha$, the expected proportion of significant nodes (at level $\alpha$) for a randomly selected subgraph under $\mathcal{H}_0$ is equal to $\alpha$. However, for a non-random subgraph that is selected by maximizing the number of significant nodes, $\arg\max_{\mathcal{S} \in \mathbb{M}: |\mathcal{S}| = N} N_\alpha(\mathcal{S})$, the expected proportion of significant nodes $\frac{N_\alpha(\mathcal{S})}{N(\mathcal{S})}$ under $\mathcal{H}_0$, called $\alpha^\prime(N,\alpha)$, is much greater than $\alpha$, as shown in Figure~\ref{fig:erdos-renyi-simulations}.  Previous, uncalibrated NPSS approaches do not account for this discrepancy between $\alpha^\prime$ and $\alpha$, causing them to incorrectly detect large, high-scoring subgraphs even when no signal is present, or equivalently, these incorrectly detected subgraphs will obscure a true signal with lower score. This motivates our development of \emph{calibrated} NPSS score functions that compare the observed proportion of significant nodes $\frac{N_{\alpha}(\mathcal{S})}{N(\mathcal{S})}$ to $\alpha^\prime(N,\alpha)$ rather than $\alpha$.  Moreover, we observe that the amount of discrepancy between $\alpha^\prime$ and $\alpha$ varies not only with $N$ and $\alpha$, but also with the graph structure (including its size and sparsity), thus motivating our decision to empirically calibrate $\alpha^\prime(N,\alpha)$ based on randomization testing.

In the discussion below, we consider the correctness of using $\alpha^\prime (N, \alpha) = \mathbb{E}\left[\max_{\mathcal{S} \in \mathbb{M}: |\mathcal{S}| = N} \frac{N_\alpha(\mathcal{S})}{N}\right]$ in place of the expectation $\alpha = \mathbb{E}\left[\frac{N_\alpha(\mathcal{S})}{N(\mathcal{S})}\right]$ in NPSS.  For example, for the Berk-Jones score function $\Phi_{BJ}(\alpha,N_\alpha(\mathcal{S}),N(\mathcal{S})) = N(\mathcal{S}) \: \texttt{KL}\left(\frac{N_\alpha(\mathcal{S})}{N(\mathcal{S})},\alpha\right)$, we instead define the calibrated Berk-Jones score function, $\Phi_{CBJ}(\alpha,N_\alpha(\mathcal{S}),N(\mathcal{S})) = N(\mathcal{S}) \: \texttt{KL}\left(\frac{N_\alpha(\mathcal{S})}{N(\mathcal{S})},\alpha^\prime(N(\mathcal{S}),\alpha)\right)$.  

We now present a more formal argument for the correctness of calibration, followed by a numeric example showing how the uncalibrated Berk-Jones statistic fails and why our calibration approach corrects this issue.  

First, we note that the objective function $\max_{\mathcal{S} \in \mathbb{M}} F(\mathcal{S})$ can be written as $\max_{N\in \{1\ldots|\mathcal{V}|\},\alpha \le \alpha_{\max}} g(N,\alpha)$, where 
\[ g(N,\alpha) = \max_{\mathcal{S} \in \mathbb{M}: |\mathcal{S}|=N} \Phi(\alpha,N_\alpha(\mathcal{S}),N) \]
is the maximum score over all subgraphs of size $N$ at significance level $\alpha$.  Since $\Phi(\alpha,N_\alpha,N)$ is an increasing function of $N_\alpha$, we can rewrite $g(N,\alpha) =  \Phi(\alpha,\max_{\mathcal{S} \in \mathbb{M}: |\mathcal{S}|=N} N_\alpha(\mathcal{S}),N)$.  Now, we would like $g(N,\alpha)$ to satisfy two intuitive properties if the null hypothesis $\mathcal{H}_0$ is true:
(i) $g(N,\alpha)$ should be small for all $N$ and $\alpha$, and (ii) $g(N,\alpha)$ should be similar in magnitude across all values of $N$ and $\alpha$.  The first property makes it possible to differentiate $\mathcal{H}_0$ from $\mathcal{H}_1(\mathcal{S})$, since large values of $g(N,\alpha)$ would obscure the true signal $\mathcal{S}$. The second property provides similar detection power across different subgraph sizes $N$ and significance levels $\alpha$.  

However, the uncalibrated Berk-Jones statistic does not satisfy either of these properties. For brevity, let $h(N,\alpha)$ denote the maximum proportion of significant p-values in a subgraph of size $N$, 
\[ h(N,\alpha) =  \max_{\mathcal{S} \in \mathbb{M}: |\mathcal{S}|=N} \frac{N_\alpha(\mathcal{S})}{N}.
\]
Then under $\mathcal{H}_0$, we have $g(N,\alpha) = N \: \texttt{KL}(h(N,\alpha),\alpha)$.  Replacing $h(N,\alpha)$ with its expectation under the null, $\alpha^\prime(N,\alpha)$, we obtain $g(N,\alpha)
\approx N \: \texttt{KL}(\alpha^\prime(N,\alpha),\alpha)$. As shown in Figure~\ref{fig:erdos-renyi-simulations} and the numeric example below, the discrepancy between $\alpha^\prime$ and $\alpha$, and the resulting values of $g(N,\alpha)$, are large, obscuring the true signal when one is present. Moreover, $g(N,\alpha)$ is much larger for high values of the significance level $\alpha$ and subgraph size $N$, leading to the detection of overly large, incorrect subgraphs. 

On the other hand, for the calibrated Berk-Jones statistic, we have: \[g(N,\alpha) = N \: \texttt{KL}\left(\max_{\mathcal{S} \in \mathbb{M}: |\mathcal{S}|=N} \frac{N_\alpha(\mathcal{S})}{N},\alpha^\prime(N,\alpha)\right) \]
\[ = 
N \: \texttt{KL}(h(N,\alpha),
\mathbb{E}[h(N,\alpha)]),\]
where the expectation assumes that $\mathcal{H}_0$ is true, and is computed by averaging $h(N,\alpha)=\max_{\mathcal{S} \in \mathbb{M}: |\mathcal{S}|=N} \frac{N_\alpha(\mathcal{S})}{N}$ over a large number of instantiations of the graph $\mathbb{G}$ with p-values drawn from the null distribution, $p_i \sim \texttt{Uniform}[0,1] \: \forall v_i \in \mathcal{V}$.  Thus if the null hypothesis $\mathcal{H}_0$ is true, the value of 
 $h(N,\alpha)$ for the real data is drawn from the same distribution as the null data, and then compared (using one-sided \texttt{KL} divergence) to the expectation of that distribution $\alpha^\prime(N,\alpha)$ in order to compute the score $g(N,\alpha)$.  From this, it is clear that the scores $g(N,\alpha)$ will be close to zero under $\mathcal{H}_0$, diverging from zero only when the value of  $h(N,\alpha)$ happens by chance to be greater than its expectation.

To more precisely quantify the impact of the variance of $h(N,\alpha)$ on the score $g(N,\alpha)$, assuming that the null hypothesis $\mathcal{H}_0$ is true, we use a second-order Taylor expansion of the \texttt{KL} divergence to obtain:
\[ \mathbb{E}[g(N,\alpha)] \approx \frac{N \mbox{Var}[h(N,\alpha)]}{2\alpha^\prime(N,\alpha)(1-\alpha^\prime(N,\alpha))}, \]
where the variance is taken over instantiations of the graph $\mathbb{G}$ with p-values drawn under the null distribution.  We observe empirically that, for large $N$, the variance of $h(N,\alpha)$ under the null is approximated well by $\frac{\alpha'(N,\alpha)}{N}$, giving us $\mathbb{E}[g(N,\alpha)] \approx \frac{1}{2} (1-\alpha^\prime(N,\alpha))^{-1}$, which is small and slowly decreases with $N$.  For small $N$, we observe empirically that $\mbox{Var}[h(N,\alpha)]$ is much smaller than $\frac{\alpha'(N,\alpha)}{N}$. For example, $\mbox{Var}[h(N,\alpha)]=0$ when $\alpha^\prime(N,\alpha)=1$. As a result, we observe scores $g(N,\alpha)$ that are close to zero (typically peaking in the low single digits) for all $N$ and $\alpha$, as illustrated in Figure~\ref{fig:wikivote_h0_instantiation_c}.  This observation has two important implications: first, since the null scores are much lower than for the uncalibrated BJ statistic, a true signal $\mathcal{H}_1(\mathcal{S})$ can be more easily detected and the true subgraph $\mathcal{S}$ more accurately identified.  Second, we no longer observe the biases toward large $N$ and $\alpha$ which led the uncalibrated BJ statistic to detect large, incorrect subgraphs.  

Thus the calibrated BJ statistic corrects for the multiplicity of subgraphs of a given size $N$, by comparing the observed maximum value of $N_\alpha(\mathcal{S})$ across all size-$N$ subgraphs to the expectation of that maximum value under the null.  In doing so, it calibrates the statistic across all values of $N$ and $\alpha$, giving similar values of $g(N,\alpha)$ under $\mathcal{H}_0$. However, we note that calibration alone does not correct for the multiple testing resulting from maximization of $g(N,\alpha)$ over all $N\in\{1 \ldots |\mathcal{V}|\}$ and $\alpha \le \alpha_{\max}$. We must still apply the standard randomization testing approach described in Appendix~\ref{app:npss-background} to compare the maximum calibrated score $\max_{\mathcal{S} \in \mathbb{M}} F(\mathcal{S})$ to the distribution of the maximum calibrated score under $\mathcal{H}_0$, thus bounding the overall false positive rate.

Finally, we note that correcting the miscalibration of NPSSs does not require an exact solution to the optimization problem, i.e., maximizing $N_\alpha(\mathcal{S})$ over subgraphs of size $N$.  Under $\mathcal{H}_0$, the current and null data are exchangeable, so for a given $\alpha$, $\widetilde h(N,\alpha) \approx \max_{\mathcal{S} \in \mathbb{M}: |\mathcal{S}|=N}
\frac{N_{\alpha}(\mathcal{S})}{N}$ will be distributed identically for the current and null data, as long as the \emph{same} approximation algorithm is used for the current and null data.  That is, we compare the observed value of the (approximate) maximum proportion of significant p-values, $\widetilde h(N,\alpha)$, to the expectation of $\widetilde h(N,\alpha)$ under $\mathcal{H}_0$, 
$\widetilde \alpha^\prime(N,\alpha)$, so $g(N,\alpha) \approx N \: \texttt{KL}(\widetilde h(N,\alpha),\widetilde \alpha^\prime(N,\alpha))$ remains well-calibrated. The downside of a using an approximate rather than exact search is some potential loss of detection power and accuracy under $\mathcal{H}_1(\mathcal{S})$, but our experiments demonstrate that the approximate algorithm achieves high detection performance across five large real-world datasets, outperforming the uncalibrated scan and baseline methods by a wide margin.

\begin{figure*}[t]
    \centering
    \begin{subfigure}[t]{0.3\textwidth}
        \centering
        \includegraphics[width=\linewidth]{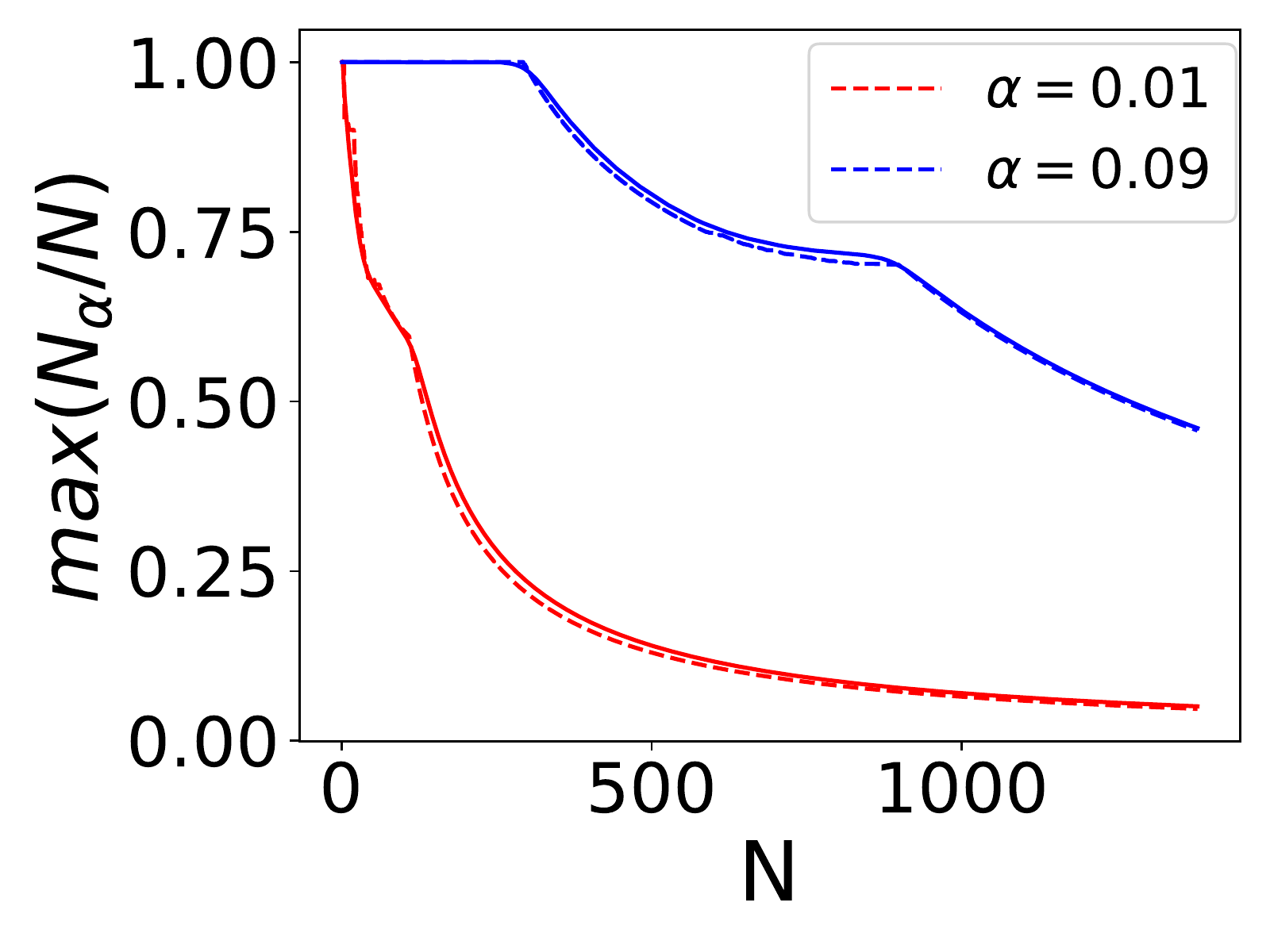}
        \vspace{-6mm}
        \caption{}
    \end{subfigure}
    \begin{subfigure}[t]{0.3\textwidth}
        \centering
        \includegraphics[width=\linewidth]{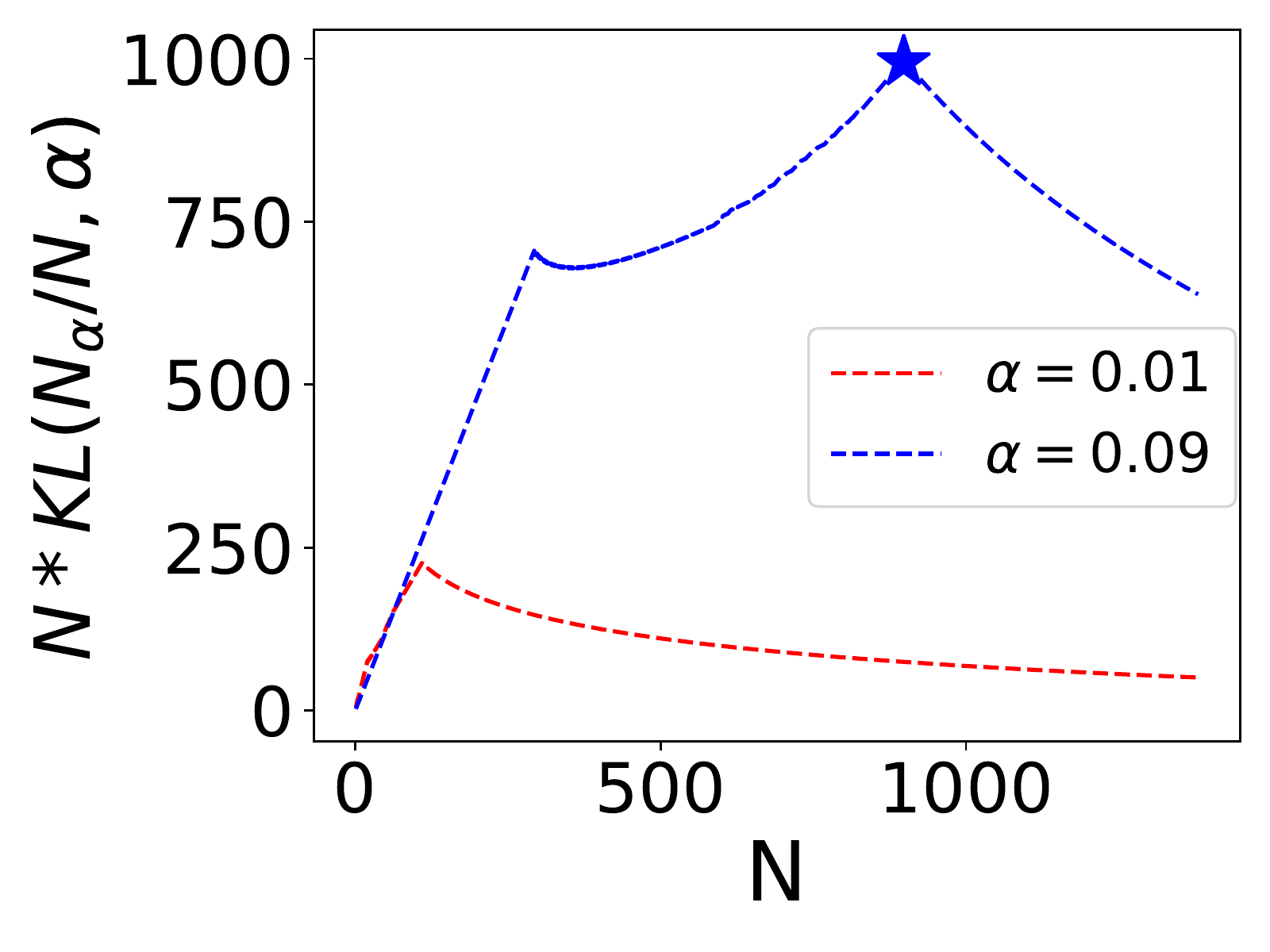}
        \vspace{-6mm}
        \caption{}
    \end{subfigure}
    \begin{subfigure}[t]{0.3\textwidth}
        \centering
        \includegraphics[width=\linewidth]{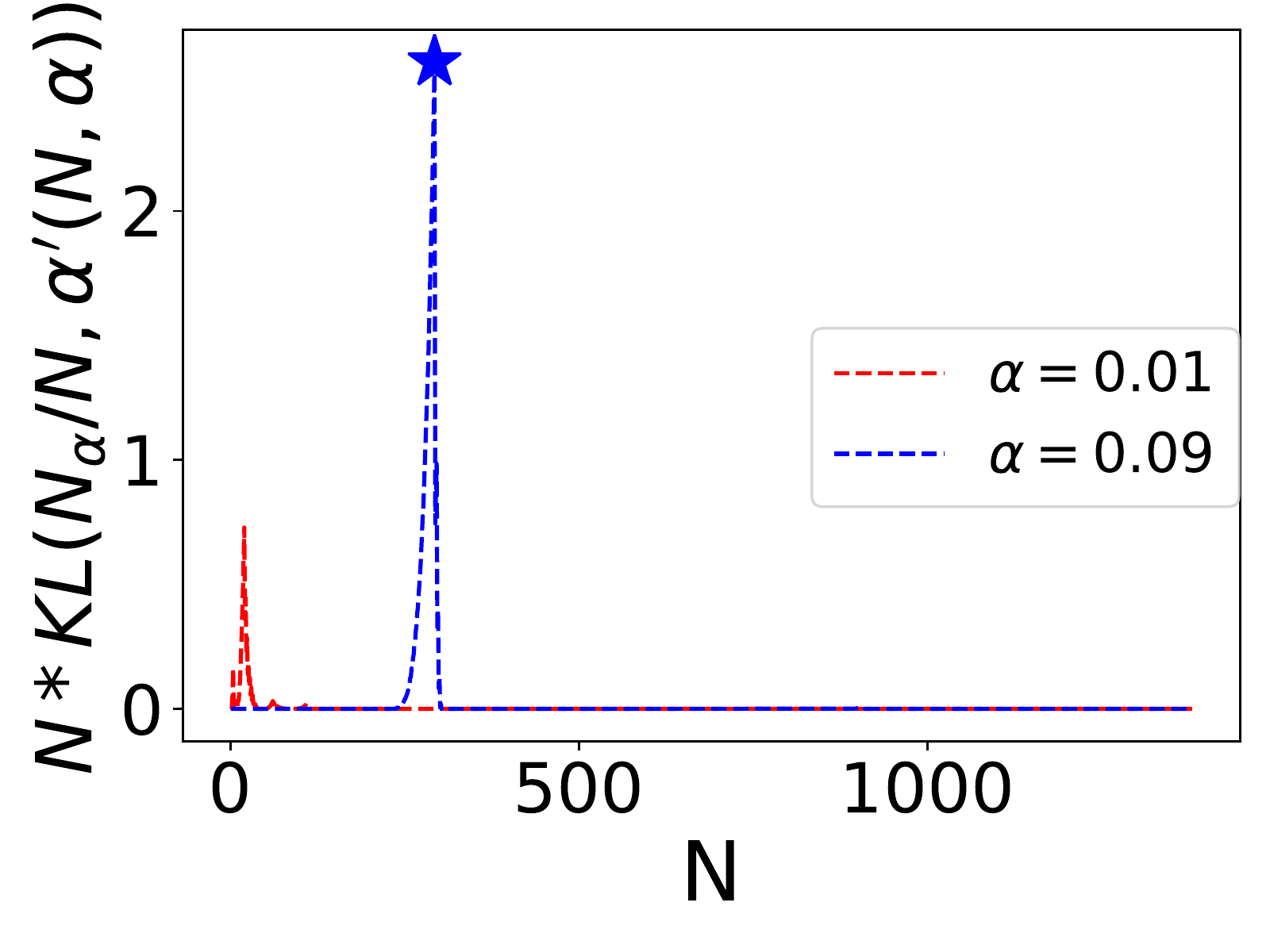}
        \vspace{-6mm}
        \caption{}
        \label{fig:wikivote_h0_instantiation_c}
    \end{subfigure}
    \caption{Example of calibrated versus uncalibrated score computation for the \texttt{WikiVote} graph with p-values generated under $\mathcal{H}_0$. Red lines: $\alpha = .01$. Blue lines: $\alpha=.09$. (a) Observed maximum value of $N_\alpha/N$ (dashed line) and expected maximum value $\alpha^\prime(N,\alpha)$ (solid line). (b) BJ score of the uncalibrated scan statistic, comparing the observed $N_\alpha/N$ to the expected value $\alpha$.
    (c) BJ score of the calibrated scan statistic, comparing the observed $N_\alpha/N$ to the expected maximum value $\alpha^\prime(N,\alpha)$.}
    \label{fig:wikivote_h0_instantiation}
\end{figure*}

\begin{figure*}[t]
    \centering
    \begin{subfigure}[t]{0.3\textwidth}
        \centering
        \includegraphics[width=\linewidth]{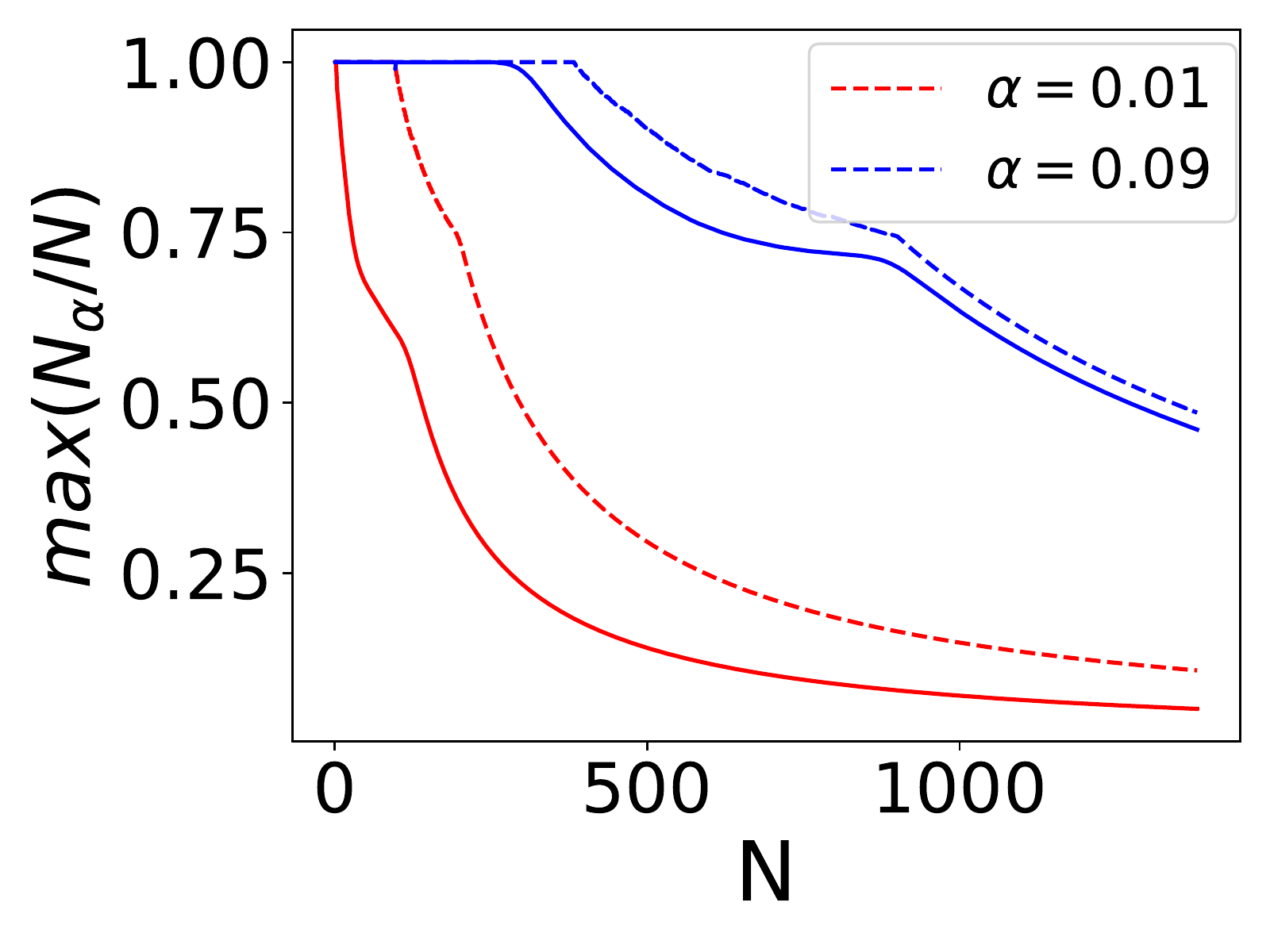}
        \vspace{-6mm}
        \caption{}
    \end{subfigure}
    \begin{subfigure}[t]{0.3\textwidth}
        \centering
        \includegraphics[width=\linewidth]{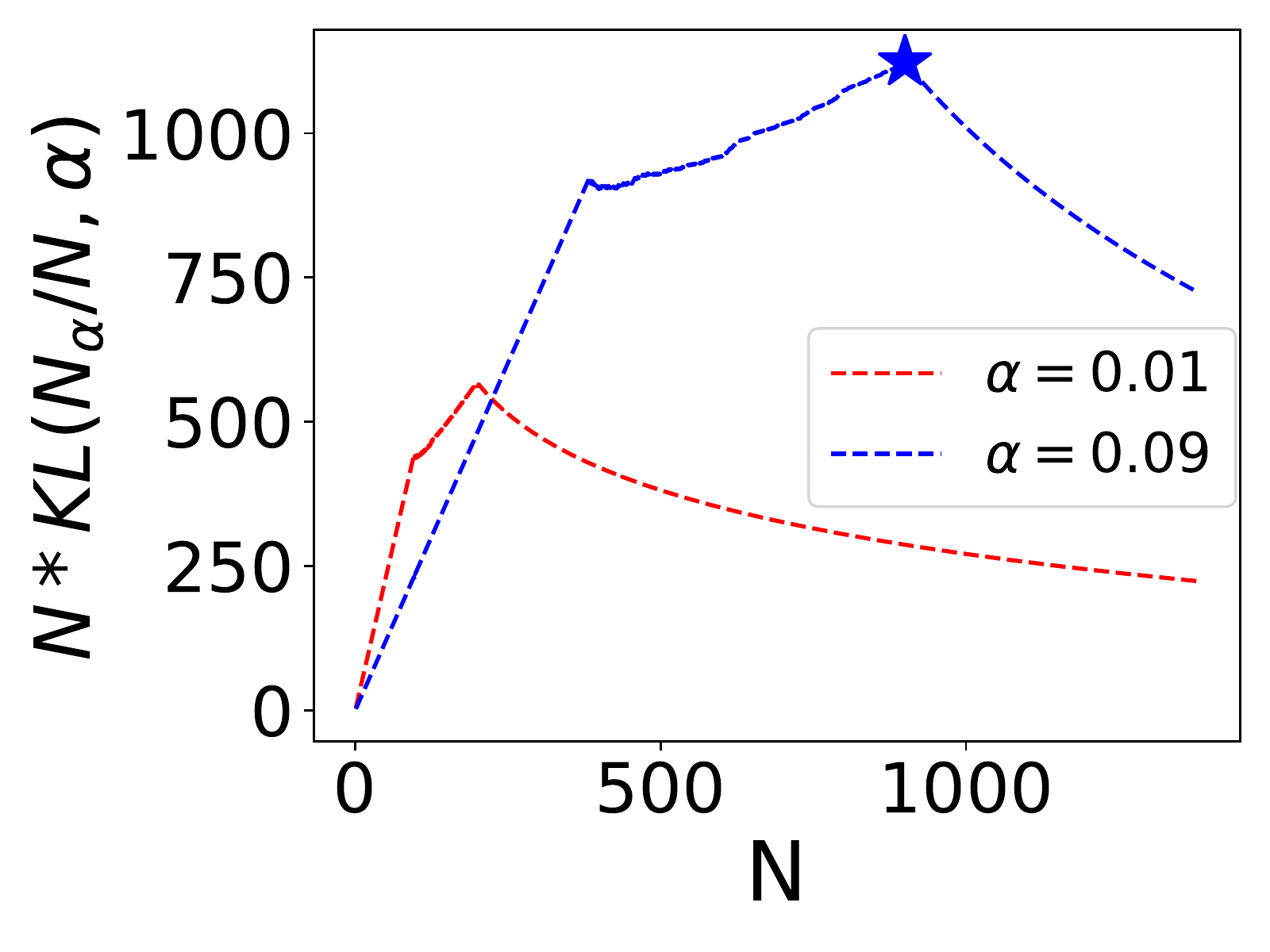}
        \vspace{-6mm}
        \caption{}
    \end{subfigure}
    \begin{subfigure}[t]{0.3\textwidth}
        \centering
        \includegraphics[width=\linewidth]{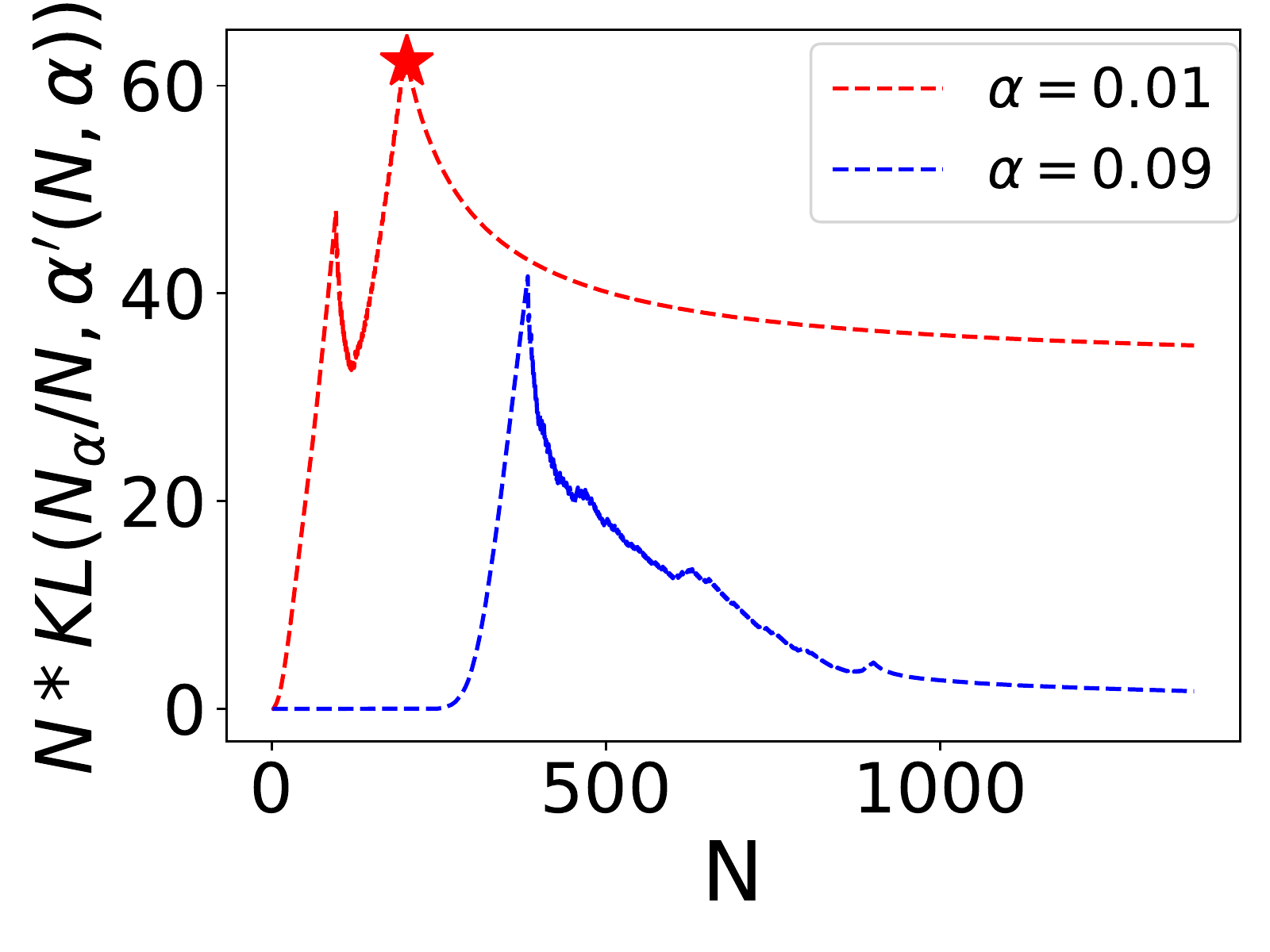}
        \vspace{-6mm}
        \caption{}
    \end{subfigure}
    \caption{Example of calibrated versus uncalibrated score computation for the \texttt{WikiVote} graph with piecewise constant p-values generated under $\mathcal{H}_1(\mathcal{S})$. Red lines: $\alpha = .01$. Blue lines: $\alpha=.09$. (a) Observed maximum value of $N_\alpha/N$ (dashed line) and expected maximum value $\alpha^\prime(N,\alpha)$ (solid line). (b) BJ score of the uncalibrated scan statistic, comparing the observed $N_\alpha/N$ to the expected value $\alpha$.
    (c) BJ score of the calibrated scan statistic, comparing the observed $N_\alpha/N$ to the expected maximum value $\alpha^\prime(N,\alpha)$.}
    \label{fig:wikivote_h1_instantiation}
\end{figure*}

\subsubsection{Motivating Example.} As a concrete example of how the uncalibrated BJ statistic fails, and how calibration solves this problem, let us consider a single instantiation of the \texttt{WikiVote} graph ($|\mathcal{V}|=7,066$) generated under $\mathcal{H}_1(\mathcal{S})$.  The true subgraph $\mathcal{S}$ was generated using a random walk on the graph structure, with $|\mathcal{S}|=100$ and a relatively strong signal injected, such that 75\% of the p-values in $\mathcal{S}$ are significant at $\alpha=.01$.  This corresponds to $q=75$ for the piecewise constant p-value simulations in Appendix~\ref{app:piecewise} below.  In this case, 
the uncalibrated BJ score of the true subgraph at the significance threshold $\alpha=0.01$ can be computed as 
$N(\mathcal{S}) \: \texttt{KL}\left(\frac{N_\alpha(\mathcal{S})}{N(\mathcal{S})}, \alpha\right) = 100 \: \texttt{KL}(0.75, 0.01) \approx 289$.
Thus the true subset $\mathcal{S}$ (into which the signal has been injected) has a high BJ score, corresponding to the true significance threshold $\alpha = 0.01$.  

However, another, much larger subset has an even higher score, corresponding to a high significance threshold $\alpha$.  More precisely, at the highest $\alpha$ value considered ($\alpha = .09$), the uncalibrated scan picks out a subgraph containing nearly all of the approximately 700 significant p-values in the graph, plus some additional nodes needed to connect them, resulting in a 900 node subgraph with $74.4\%$ significant p-values.  If this observed proportion of $74.4\%$ is compared to the $9\%$ of significant p-values that one would expect to see in a random subgraph at $\alpha=.09$, it would have an extremely high BJ score of
$N(\mathcal{S}) \: \texttt{KL}\left(\frac{N_\alpha(\mathcal{S})}{N(\mathcal{S})}, \alpha\right) = 900 \: \texttt{KL}(0.744, 0.09) \approx 1100$, and as a result, the uncalibrated scan incorrectly identifies this subgraph instead of the true subgraph.  However, under $\mathcal{H}_0$, for a graph of this size and structure, one would \emph{expect} to see some subgraph of this size $N = 900$ with about $\alpha'(900,0.09) = 69.9\%$ significant p-values. Then the calibrated score of that incorrect subgraph (comparing the observed $74.4\%$ to the expected $69.9\%$) would be close to zero: $N(\mathcal{S}) \: \texttt{KL}\left(\frac{N_\alpha(\mathcal{S})}{N(\mathcal{S})}, \alpha'(N(\mathcal{S}), \alpha)\right) = 900 \: \texttt{KL}(0.744, 0.699) = 4.47$, allowing a subgraph closer to the true subgraph (with $N=202$ and 73.3\% of p-values significant at $\alpha=.01$, as compared to $\alpha'(202,0.01) =0.347$, for a score of 62.26) to be found instead.  Comparing this subgraph to the true subgraph using the metrics in Appendix~\ref{app:metrics} below, we compute precision = 0.72, recall = 0.69, and 
F-score = 0.70, while the subgraph found by the uncalibrated scan had slightly higher recall (0.75) but much lower precision (0.08) and F-score (0.15). This is why calibration using $\alpha'(N,\alpha)$ in place of $\alpha$ improves detection performance: it prevents the scan from detecting large, incorrect subgraphs whose log-likelihood ratio score exceeds the score of the true subgraph.

For additional clarity, we show examples of score computation for the calibrated and uncalibrated BJ statistics in Figure~\ref{fig:wikivote_h0_instantiation} and~\ref{fig:wikivote_h1_instantiation}.  Figure~\ref{fig:wikivote_h0_instantiation} is for an instantiation of the \texttt{WikiVote} graph with no injected signal, such that all p-values are uniform on [0,1], while Figure~\ref{fig:wikivote_h1_instantiation} is for the instantiation of the \texttt{WikiVote} graph with injected piecewise-constant signal discussed above.  For illustration, we consider only $\alpha=.01$ (red lines) and $\alpha=.09$ (blue lines).  In each graph, the left panel shows the observed value of $h(N,\alpha) = \max_{\mathcal{S} \in \mathbb{M}: |\mathcal{S}| = N} \frac{N_\alpha(\mathcal{S})}{N}$
as a function of $N$, shown as a dashed line, as compared to the expected value $\alpha^\prime(N,\alpha) =\mathbb{E}\left[\max_{\mathcal{S} \in \mathbb{M}: |\mathcal{S}| = N} \frac{N_\alpha(\mathcal{S})}{N}\right]$ under the null hypothesis $\mathcal{H}_0$, shown as a solid line.  As expected, when there is no signal (Figure~\ref{fig:wikivote_h0_instantiation}), the observed maximum matches $\alpha^\prime$ almost exactly (very slight differences are visible when zooming in on the graph), demonstrating that $\alpha^\prime$ is correctly calibrated, and the resulting calibrated BJ score (right panel) is close to zero.  On the other hand, we note that the observed maximum is much larger than $\alpha$, and thus the resulting uncalibrated BJ score (center panel) is very high.  When the signal is present (Figure~\ref{fig:wikivote_h1_instantiation}), we see clear differences between the observed maximum and $\alpha^\prime(N,\alpha)$ (left panel), resulting in a high calibrated score (right panel) that is maximized at the true $\alpha$ value, $\alpha=.01$, for a subgraph that closely matches the true subgraph as described above.  On the other hand, for the uncalibrated scan, the score of the true subgraph at the true $\alpha$ value is exceeded by the score of the much larger subgraph described above (center panel), at the incorrect $\alpha=.09$.

In summary, this example, along with the discussion above, demonstrates how miscalibration can harm the detection performance of uncalibrated NPSSs, resulting in a very large, incorrect detected subgraph, at an incorrect $\alpha$ value, that swamps the true signal. However, when the uncalibrated $\alpha$ value is replaced with the calibrated $\alpha^\prime(N,\alpha)$, the miscalibration issue is solved, leading to substantially 
improved detection performance for the calibrated scan.   

\subsection{Proofs of Theorems 1-2}
\label{app:proofs}

\FirstThm*
\begin{proof}
We consider the size-$N$ subgraph consisting of all $c$ nodes from the subgraph and $N-c$ nodes from the neighbors.  There are two cases.  For $c \leq N \leq c+k_c\alpha$, we choose only significant nodes from the neighbors, so the expected number of significant nodes for a given $N$ is  $c\alpha + (N-c)$.  For $c + k_c\alpha \leq N \leq c+k_c$, all significant neighbors are included, and thus the expected number of significant nodes for the given $N$ is $c\alpha + k_c\alpha$. Thus we have $\mathbb{E}[N_{\alpha}] = c\alpha + \min(k_c\alpha, N-c)$.
\end{proof}

\SecondThm*
\begin{proof}
To find a lower bound on $\alpha^\prime$ given $N$ and $\alpha$, let $Z = \frac{|\mathcal{V}|}{N}(1-\exp(-\langle k \rangle\frac{N}{|\mathcal{V}|}))$. We note that $Z \le \langle k \rangle$, since $1-\exp(-x) \le x$ for all non-negative $x$.  Now there are two cases.  

\textbf{Case 1:} If $\alpha \ge \frac{1}{Z}$, then the fraction of significant nodes $\alpha$ is also greater than $\frac{1}{\langle k \rangle}$. Thus there exists w.h.p.~a giant cluster of size $|\mathcal{V}| P_\infty$, where $P_\infty = \alpha(1-\exp(-\langle k \rangle P_\infty)) \ge  \frac{1}{Z}(1-\exp(-\langle k \rangle P_\infty))$, consisting entirely of significant nodes.  
Now we can see that $P_\infty \ge \frac{N}{|\mathcal{V}|}$, 
since $\frac{N}{|\mathcal{V}|}  =\frac{1}{Z}(1-\exp(-\langle k \rangle \frac{N}{|\mathcal{V}|}))$,
and thus there exists a cluster of size $N$ consisting entirely of significant nodes, i.e., $\alpha^\prime = 1$ for the given $N$ and $\alpha$.

\textbf{Case 2:} If $\alpha < \frac{1}{Z}$, then mark all of the significant nodes and fraction $\frac{1/Z-\alpha}{1-\alpha}$ of non-significant nodes, so that the proportion of marked nodes is $1/Z$, and the probability that a marked node is significant is $Z\alpha$.  Since the fraction of marked nodes $\frac{1}{Z} > \frac{1}{\langle k \rangle}$, there exists w.h.p.~a giant cluster of size $|\mathcal{V}| P_\infty$, where $P_\infty = \frac{1}{Z}(1-\exp(-\langle k\rangle P_\infty)) = \frac{N}{|\mathcal{V}|}$, consisting entirely of marked nodes. Thus there exists a cluster of size $N$ such that the fraction of significant nodes in that cluster is $Z\alpha = \frac{\alpha|\mathcal{V}|}{N}(1-\exp(-\langle k \rangle\frac{N}{|\mathcal{V}|}))$, i.e., $\alpha^\prime \ge \frac{\alpha|\mathcal{V}|}{N}(1-\exp(-\langle k \rangle\frac{N}{|\mathcal{V}|}))$ for the given $N$ and $\alpha$.
Combining these two cases, we obtain the lower bound on $\alpha^\prime$, for all $\alpha$ and $N$.
\end{proof}

\subsection{Algorithm 1}
\label{appendix:CNSS}

As described in Section~\ref{sec:cnss}, Algorithm~\ref{alg:cnss} searches for the highest-scoring connected subgraph, $\arg\max_\mathcal{S} F(\mathcal{S}) = \arg\max_\mathcal{S} \Phi_{\texttt{CBJ}}(\alpha,N_\alpha(\mathcal{S}),N(\mathcal{S}))$.  To do so, it steps over a range of significance thresholds $\alpha \in \mathcal{L}$, calling Algorithm~\ref{alg:merge} for each $\alpha$ value, and collecting the subgraph $\mathcal{S}$ with largest $N_\alpha(\mathcal{S})$ for each $N \in \{1, \cdots, |\mathcal{V}|\}$ and each $\alpha \in \mathcal{L}$.  These subgraphs are then scored with the calibrated Berk-Jones statistic $\Phi_{\texttt{CBJ}}$, and the subgraph with the highest score is returned.

\begin{algorithm}[!ht]
\caption{Anomalous Subgraph Detection via Calibrated Non-parametric Scan Statistics}
\label{alg:cnss}
\KwIn{Graph $\mathbb{G}=(\mathcal{V},\mathcal{E})$, $\mathbf{X}$ which contains all historical feature observations of each node, and a list of $\alpha$ denoted as $\mathcal{L}$, i.e., $\mathcal{L}=[0.001, \cdots, 0.009, 0.01, \cdots, 0.09]$;}
\KwOut{A set of nodes $\mathcal{S}\subseteq \mathcal{V}$ that form a connected subgraph of $\mathbb{G}$;}
Compute empirical p-values $\mathbf{p}$ for all nodes based on $\mathbf{X}$, following the approach described in~\citet{chen2014non}.\\
\For{$\alpha \in \mathcal{L}$}{
Apply Algorithm~\ref{alg:merge} on graph $\mathbb{G}$ with empirical p-values $\mathbf{p}$ and significance threshold $\alpha$, to collect a list of candidate subgraphs $\mathcal{S}$ and corresponding $N_{\alpha}(\mathcal{S})$ and $N(\mathcal{S})$.\\
\For{each collected triplet $(\mathcal{S}, N_{\alpha}(\mathcal{S}), N(\mathcal{S}))$}{
Compute the calibrated BJ scan statistic:
\begin{equation}
    \begin{split}
    &\quad \Phi_{\texttt{CBJ}}\left(\alpha, N_{\alpha}(\mathcal{S}), N(\mathcal{S})\right)\\
    &=N(\mathcal{S}) \times \texttt{KL}\left(\frac{N_{\alpha}(\mathcal{S})}{N(\mathcal{S})}, \alpha^{\prime}(N(\mathcal{S}), \alpha)\right)
    \end{split}
\end{equation}
where \begin{equation}
    \alpha^{\prime}(N, \alpha) = \frac{\mathbb{E}\left[\max _{\mathcal{S}:|\mathcal{S}|=N} N_{\alpha}(\mathcal{S})\right]}{N}
    \label{eq:alpha_prime}
\end{equation} 
}}
\Return $\mathcal{S}$ with the highest $\Phi_{\texttt{CBJ}} \left( \alpha, N_{\alpha}(\mathcal{S}), N(\mathcal{S})\right)$
\end{algorithm}

At line 1, we use the two-stage empirical calibration procedure described in~\cite{chen2014non} to convert the observed node features into a single p-value for each node based on the historical node features. Note that the computation of these p-values is not the focus of the present work; thus, in our simulated experiments on the five datasets, we simulate the p-values directly (as discussed in Section~\ref{sec:experiments}) rather than simulating the node features and computing p-values from them.  This is not only faster, but provides a more natural way to measure the strength of the injected signal.

At line 2, we iterate over a list of significance thresholds $\alpha$. In our experiments, we used significance thresholds $\alpha \in \{0.001, \cdots, 0.009, 0.01, \cdots, 0.09\}$.

At line 5, the algorithm assumes that we have pre-computed the $\alpha^\prime(N, \alpha)$ values for each $N$ and $\alpha$ under consideration.  To do so, there are two options.  First, we could use the randomization tests that apply Algorithm~\ref{alg:merge} on $K$ replicas of datasets under $\mathcal{H}_0$ to collect $K$ number of $N_{\alpha}$ values for each $N\in\{1, \cdots, |\mathcal{V}|\}$ and $\alpha\in\mathcal{L}$. Then we use the averaged $N_{\alpha}$ to compute the $\alpha'(N, \alpha)$.
We could also replace the randomization tests with the lower bounds of $\alpha^\prime(N, \alpha)$ as discussed in Section~\ref{sec:lower_bounds}, which we describe as the method \texttt{CNSS+LowerBound}.

In addition, we can also apply the core-tree decomposition and tree compression steps described in Section~\ref{sec:core_tree} and Appendix~\ref{app:core-tree} to obtain a compressed core $\mathbb{C}$ and corresponding p-values $\mathbf{p}$, and then apply the algorithm on $\mathbb{C}$ and $\mathbf{p}$ to speed up the search. The core-tree decomposition can be done just once for a given graph (between lines 1 and 2), while the tree compression is done separately for each $\alpha$ value (between lines 2 and 3). This method is called \texttt{CNSS+CoreTree}.

\subsection{Algorithm 2}\label{app:alg:merge}

As described in Section~\ref{sect:efficient-algorithm}, for a given graph $\mathbb{G}$ with corresponding empirical p-values $\mathbf{p}$ for each node, and a given significance threshold $\alpha$, Algorithm~\ref{alg:merge} searches for the most significant subgraph
 $\max_{\mathcal{S} \in \mathbb{M}, |\mathcal{S}|=N} N_{\alpha}(\mathcal{S})$ for each $N \in \{1, \cdots, |\mathcal{V}|\}$.  The algorithm is a greedy merging approach, which enables it to scale to large graph sizes but has the drawback of not guaranteeing that the subgraph with maximum $N_\alpha$ will be found.

\begin{algorithm}[!ht]
\caption{Greedily Merge Graph and Estimate $\max N_{\alpha}$ for $N\in \{1,\cdots,|\mathcal{V}|\}$.}
\label{alg:merge}
\KwIn{graph $\mathbb{G}=(\mathcal{V},\mathcal{E})$, $\mathbf{p}\in [0,1]^{|\mathcal{V}|}$, and $\alpha \in [0, 1]$.}
\KwOut{a set of $(\mathcal{S}, N_\alpha(\mathcal{S}),N(\mathcal{S}))$ triplets}
Merge all adjacent significant nodes, and get a list of merged nodes $\mathcal{S}$, denoted as $\mathcal{Z}$.\\
Sort the list $\mathcal{Z}$ by significance ratio $N_\alpha(\mathcal{S})/N(\mathcal{S})$, from highest to lowest, and breaking ties using larger size $N(\mathcal{S})$.\\
Initialize an empty list $\mathcal{P}$ to store $(\mathcal{S}, N_\alpha(\mathcal{S}), N(\mathcal{S}))$ triplets.\\
Get the merged node $\mathcal{S}^T$ with the highest significance ratio in $\mathcal{Z}$.\\
Add $(S^T, N_\alpha(\mathcal{S}^T),N(\mathcal{S}^T))$ of the merged node $\mathcal{S}^T$ to $\mathcal{P}$.\\
\While{length$(\mathcal{Z}) > 1$}{
Select the merged node $\mathcal{S}^T$ with the highest significance ratio in $\mathcal{Z}$ as root.\\
Select and apply the best merge option among the three options described below, to merge another node into $\mathcal{S}^T$.\\
\If{the best merge option is option 1}{
Add $(\mathcal{S}^T, N_\alpha(\mathcal{S}^T),N(\mathcal{S}^T))$ of the merged node $\mathcal{S}^T$ to $\mathcal{P}$.}}
Sort the list $\mathcal{P}$ by $N(\mathcal{S})$ from highest to lowest.\\
$\texttt{previous\_ratio} \leftarrow N_\alpha(\mathcal{S})/N(\mathcal{S})$ of the first element of $\mathcal{P}$.\\
\For{each successive element $(\mathcal{S}, N_\alpha(\mathcal{S}), N(\mathcal{S}))$ in list $\mathcal{P}$}{
$\texttt{current\_ratio}\leftarrow N_\alpha(\mathcal{S})/N(\mathcal{S})$.\\
\uIf{\texttt{current\_ratio} $>$ \texttt{previous\_ratio}}{$\texttt{previous\_ratio} \leftarrow \texttt{current\_ratio}$} 
\Else{Delete element $(\mathcal{S}, N_\alpha(\mathcal{S}), N(\mathcal{S}))$ from $\mathcal{P}$.}
}
\Return $\mathcal{P}$
\end{algorithm}

At line 1: After we merge adjacent significant nodes, the merged nodes have significance ratio equal to 1. Those merged nodes $\mathcal{S}$ could be viewed as candidate detected subgraphs, and will be merged further to create larger candidate subgraphs.

At line 2: We maintain the ordered list $\mathcal{Z}$
throughout the algorithm, where $\mathcal{Z}$ is sorted
by significance ratio first, highest to lowest. If two items in $\mathcal{Z}$ have the same significance ratio, we sort them based on the merged node size $N(\mathcal{S})$, highest to lowest.

At line 4: Initially, all nodes in the list have the same significance ratio of 1, so the merged node representing the largest subgraph of significant nodes will be the initial root $\mathcal{S}^T$.

At line 8: We evaluate three merge options on the root $\mathcal{S}^T$ as follows:
\begin{enumerate}
    \item If there exists a neighbor $n$ of $\mathcal{S}^T$ which contains some or all significant p-values, 
    merge $n$ into $\mathcal{S}^T$. 
    \item If there exists a non-significant neighbor $n$ of $\mathcal{S}^T$ which is also adjacent to at least one other significant node, merge $n$ into $\mathcal{S}^T$.
    \item Merge the highest-degree non-significant neighbor $n$ into $\mathcal{S}^T$. 
\end{enumerate}
Then we select and apply the option that leads to the highest significance ratio for the merged node among these three options. If they result in same significance ratio, we use the priority order $1 > 2 > 3$.

At line 9: We also collect the $(\mathcal{S}, N_\alpha(\mathcal{S}), N(\mathcal{S}))$ triplets of the merged nodes after each merge of option 1, which produces a list of $(\mathcal{S}, N_\alpha(\mathcal{S}), N(\mathcal{S}))$ triplets where $N(\mathcal{S})\in \{1, \cdots,|\mathcal{V}|\}$. Only the largest $N_\alpha$ value for each $N$ must be kept.  Note that we do not need to record the triplets formed after option 2 or option 3 since no significant p-values have been added to $\mathcal{S}^T$.

The purpose of lines 13 - 22
is to remove sub-optimal $(\mathcal{S}, N_\alpha(\mathcal{S}), N(\mathcal{S}))$ from $\mathcal{P}$, as a subgraph with smaller significance ratio $N_\alpha(\mathcal{S})/N(\mathcal{S})$ and smaller size $N(\mathcal{S})$ is guaranteed to have lower score.

We note that, as written, the algorithm only returns $(\mathcal{S}, N_\alpha(\mathcal{S}), N(\mathcal{S}))$ for a subset of $N$ values, $\mathcal{Q} \subseteq \{1,\ldots,|\mathcal{V}|\}$.  For the original graph, these are the only values of $N(\mathcal{S})$ for which the corresponding subgraph $\mathcal{S}$ could have optimal score $F(\mathcal{S})$.  For the replica graphs used in randomization testing, we apply linear interpolation to estimate the $N_{\alpha}$ for the remaining values of $N$, that is, $N \in \{1,\ldots, |\mathcal{V}|\} \setminus \mathcal{Q}$. Also, when we apply Algorithm~\ref{alg:merge} under the null hypothesis for randomization tests, we only record the $(N_\alpha(\mathcal{S}), N(\mathcal{S}))$ pairs, without recording $\mathcal{S}$, to save memory space.

\subsection{Core-Tree Decomposition}
\label{app:core-tree}
We adopt the implementation of core-tree decomposition from~\cite{maehara2014vldb}. 
After we decompose the whole graph into the core part $\mathbb{C}=(\mathcal{V}_C, \mathcal{E}_C)$ and tree part $\mathbb{T}=(\mathcal{V}_T, \mathcal{E}_T)$, we utilize an additional tree-compression step before applying Algorithm~\ref{alg:merge} on $\mathbb{C}$.

Tree compression merges the significant nodes in each single tree into an adjacent core node. We could conceivably optimize over each single tree to identify and merge the highest scoring sub-tree, but this would be time-consuming given the large number of trees resulting from the core-tree decomposition. Instead, we use breadth-first tree search to find and merge each significant sub-tree that is adjacent to the core. If a significant tree node is adjacent to multiple core nodes, then we merge this tree node into the most significant adjacent core node. In order to achieve this, we first sort the core nodes $\mathcal{V}_C$ by p-value (lowest to highest), and also remove all non-significant tree nodes from the graph (this may disconnect some of the significant tree nodes, which are removed from the graph as well).  We then iteratively select the most significant core node in the sorted $\mathcal{V}_C$ as the root of breadth-first tree search until all remaining tree nodes $\mathcal{V}_T$ are explored.

The time complexity of core-tree decomposition is $\mathcal{O}(d|\mathcal{V}| + |\mathcal{E}|)$ where $d$ denotes a user specified tree width, and the time complexity of tree compression is mainly on the sequence of breadth-first tree search, which has $\mathcal{O}(|\mathcal{V}_T| + |\mathcal{E}_T|)$, as well as the sorting of core nodes, which has $\mathcal{O}(|\mathcal{V}_C|\log |\mathcal{V}_C|)$.

The impact of core-tree decomposition and tree compression is to substantially reduce the effective graph size down to the size of the core, while keeping many of the significant p-values in the trees.  However, we note that there is a potential trade-off to this computationally efficient approach.  Any significant tree node that is not adjacent to a core node, and is not connected to the core by a path consisting only of other significant nodes, will not be merged into the core and therefore will not be part of the detected subgraph returned by \texttt{CNSS+CoreTree}.  In practice, however, we find that the loss of accuracy from this approach is minimal, while the speedup in runtime is substantial. 

\subsection{Time Complexity Analysis of CNSS}
\label{app:time_complexity}
Since CNSS Algorithm~\ref{alg:cnss} applies Algorithm~\ref{alg:merge} for each significance threshold $\alpha$ under consideration, we focus on the time complexity of Algorithm~\ref{alg:merge} first. 
Algorithm~\ref{alg:merge} includes three main steps: 
\begin{itemize}
    \item Step 1 (line 1): merge adjacent significant nodes in the graph;
    \item Step 2 (lines 2-12): greedily merge the whole graph according to the three options as described in Appendix~\ref{app:alg:merge} and record $(\mathcal{S}, N_\alpha(\mathcal{S}), N(\mathcal{S}))$ throughout the merge process;
    \item Step 3 (lines 13-22): filter out recorded $(\mathcal{S}, N_\alpha(\mathcal{S}), N(\mathcal{S}))$ with suboptimal $N_{\alpha}(\mathcal{S})$ for the corresponding $N(\mathcal{S})$.
\end{itemize}
For step 1, the algorithm must iterate over all edges and record all significant and non-significant neighbors for each node. Therefore, the time complexity of step 1 is $\mathcal{O}(|\mathcal{V}|+|\mathcal{E}|)$.  
The time complexity of step 2 is mainly based on sorting and searching for the best merge option, which iterates over the root node's neighbors. The sorting takes $\mathcal{O}(|\mathcal{V}|\log|\mathcal{V}|)$. For the first option of merge, we randomly merge one significant neighbor. For the second and third options, they need to iterate over all neighbors of current root node.
Hence, the time complexity would be $\mathcal{O}(k|\mathcal{V}|)$ where $k$ denotes the largest degree of a node in the network. The overall time complexity of step 2 is $\mathcal{O}(k|\mathcal{V}| + |\mathcal{V}|\log |\mathcal{V}|)$.
For step 3, the time complexity is mainly on the sorting, thus time complexity is $\mathcal{O}(|\mathcal{V}|\log|\mathcal{V}|)$. 
Therefore, the overall time complexity of the Algorithm~\ref{alg:merge} is $\mathcal{O}(k|\mathcal{V}| + |\mathcal{V}|\log |\mathcal{V}|)$.


For CNSS Algorithm~\ref{alg:cnss}, it applies Algorithm~\ref{alg:merge} for each significance threshold $\alpha\in\mathcal{L}$. For obtaining $\alpha'(N,\alpha)$ for a given $N$ and $\alpha$, it requires us to apply Algorithm~\ref{alg:merge} on $K$ replicas of datasets under the null hypothesis, unless the lower bound method is used in place of randomization testing. Therefore, the time complexity is $\mathcal{O}(K|\mathcal{L}|(k|\mathcal{V}| + |\mathcal{V}|\log |\mathcal{V}|))$.  As we note below, these $K$ replicas can be run in parallel, or alternatively, the lower bound approach avoids the need for randomization testing; in either case, the time needed to compute $\alpha'(N,\alpha)$ can be reduced by a factor of $K$.

\subsection{Implementation Details and Reproducibility}
\label{app:implementation}
We performed all experiments on Linux servers with the same hardware configuration (64-bit  machines  with  Intel(R)Xeon(R) CPU E5-2680 v4 @ 2.40GHz and 251GB RAM).
We implemented the CNSS in Python, and the code is accessible via the following link: \url{https://bit.ly/2QTvDzM}.

We set the random seed of each run under the alternative hypothesis and null hypothesis as the run index, and we used significance thresholds $\alpha \in [0.001, 0.002, \cdots, 0.009, 0.01, \cdots, 0.09]$.

The implementation details of each baseline method used in our evaluation are provided in Appendix~\ref{app:baselines}.

\section{Additional Experimental Details}
\label{app:exp}

\subsection{Datasets}
\label{app:datasets}
 Five real-world networks were obtained from the Stanford Network Analysis Project (SNAP)~\footnote{\url{https://snap.stanford.edu/data/}}$^,$\footnote{License information for these datasets is found at: \url{https://snap.stanford.edu/snap/license.html}}, including
 1) \texttt{Twitter}: a social network in where every node is a user and every edge represents a relation of follower and followee, where we do not consider the edge direction and thus treat the dataset as an undirected graph; 
 2) \texttt{DBLP}: a co-authorship network where every author is represented by a node, and two authors are connected if they publish at least one paper together;
 3) \texttt{SlashDot}: a technology-related news social network in where every node is an user and every link represents the friendship between two users; 
 4) \texttt{CondMat}: Arxiv COND-MAT (Condensed Matter Physics) collaboration network is from the e-print arXiv and covers scientific collaborations between authors on papers submitted to the Condensed Matter category. If an author $i$ co-authored a paper with author $j$, the graph contains a undirected edge between $i$ and $j$. If the paper is co-authored by $k$ authors, the co-authorship graph contains a completely connected subgraph on these $k$ nodes; and 
 5) \texttt{WikiVote}: the network contains all the Wikipedia voting data from the inception of Wikipedia until January 2008. Nodes in the network represent Wikipedia users and a directed edge from node $i$ to node $j$ represents that user $i$ voted on user $j$. In our experiments, we treat it as an undirected graph. 
 The descriptive statistics of all datasets are described in Table \ref{table:semi-synthetic}. 
 
 We have also tried different sizes of true subgraphs (up to 5\% of $|\mathcal{V}|$), and we only report one of them since they have consistent relative performance of methods as our reported results in the paper.
 
\begin{table*}[!ht]
\caption{Descriptive Statistics of Real-World Networks.}
\label{table:semi-synthetic}
\begin{adjustbox}{center}
\resizebox{0.9\linewidth}{!}{
\centering
\begin{tabular}{c|c|c|c|c|c|c}
\toprule
Dataset & \multicolumn{1}{c|}{Vertices $|\mathcal{V}|$} & \multicolumn{1}{c|}{Edges $|\mathcal{E}|$} & \multicolumn{1}{c|}{Density} & \multicolumn{1}{c|}{Core Vertices $|\mathcal{V}_C|$} &\multicolumn{1}{c|}{Core Density} &\multicolumn{1}{c}{True Nodes $|\mathcal{S}|$} \\ \hline
WikiVote & 7,066 & 100,736 & 0.00403 & 1,823 & 0.0425 & 100  \\
CondMat & 21,363 & 91,286 & 0.0004 & 2,513  & 0.00487 & 200  \\
Twitter & 81,309 & 1,342,296 & 0.000406 & 17,337 & 0.0041 & 1,000 \\
SlashDot & 82,168 & 504,230 & 0.000149 & 10,599 & 0.0046 & 1,000 \\
DBLP & 317,080 & 1,049,866 & 0.0000208 & 22,354 & 0.00054 & 1,000
 \\
\bottomrule
\end{tabular}
}
\end{adjustbox}
\end{table*}

\subsection{Details of Comparison Methods}
\label{app:baselines}
We compare our proposed algorithm with six state of the art methods for event detection and anomalous subgraph detection.  These six methods are commonly used as baselines for detection of anomalous subgraphs, and include:
\begin{itemize}
    \item Linear Time Subset Scanning (\texttt{LTSS}) \cite{neill2012fast} is an efficient event detection algorithm in massive data sets, where the event detection problem could be viewed as a problem of finding the subset which maximizes some score function. For score functions satisfying the LTSS property (e.g., the Berk-Jones scan statistic), the subset of data records which maximizes $F(\mathcal{S})$ can be found by ordering the records according to some ``priority'' function and searching over groups consisting of the top-$k$ highest priority records, requiring a linear rather than exponential number of subsets to be evaluated~\cite{neill2012fast}. The time complexity of \texttt{LTSS} is 
    $O(|\mathcal{V}|\log |\mathcal{V}|)$. However, \texttt{LTSS} does not enforce the graph connectivity constraint and may produce a disconnected subset of graph nodes.
    Thus we use the largest connected component in the detected subset of nodes as the detected subgraph $S$.
    We obtained the code from the authors, and we use the BJ scan statistic as the objective to maximize. We iterate over the list of significance thresholds $\alpha \in \{0.001, \cdots, 0.009, 0.01, \cdots, 0.09\}$ to find the maximum BJ score.
    
    \item \texttt{EventTree}~\cite{rozenshtein2014event}: is an event detection algorithm, which defines an event to be a connected subgraph of nodes in the network that are close to each other and have high activity levels.  Unlike the other methods considered here, the objective function of \texttt{EventTree} is not a log-likelihood ratio statistic. 
    Rather, the objective is to maximize $Q(\mathcal{S})=\lambda W(\mathcal{S})-D(\mathcal{S})$, where $W(\mathcal{S})=\sum_{v\in \mathcal{S}} w(v)$ measures the total node weight value of a subgraph $\mathcal{S}$, $D(\mathcal{S})=\frac{1}{2}\sum_{u\in \mathcal{S}} \sum_{v\in \mathcal{S}}d(u,v)$ measures the distance value of subgraph $\mathcal{S}$, and $\lambda$ is a normalization coefficient. 
    We optimized the code provided by the authors with a more efficient PCST solver\footnote{Hegde, Chinmay, Piotr Indyk, and Ludwig Schmidt. ``A fast, adaptive variant of the Goemans-Williamson scheme for the prize-collecting Steiner tree problem.'' Workshop of the 11th DIMACS Implementation Challenge. Vol. 2, 2014.} that reduces the time complexity from $\mathcal{O}(|\mathcal{V}|^2\log|\mathcal{V}|)$ to $\mathcal{O}(|\mathcal{E}|\log|\mathcal{V}|)$ but has identical detection performance.
    Since we use the simulated p-values for other methods, we use the reciprocal of the p-value as the node weight for \texttt{EventTree}. The parameter $\lambda$ also controls the granularity of the detected event, and we tuned it over the list $[0.001, \cdots, 0.009, 0.01, \cdots, 0.09]$.
    
    \item \texttt{ColorCoding}~\cite{cadena2019near}: is an unified framework for optimizing a large class of parametric and non-parametric scan statistics for networks with connectivity constraints. It is the only baseline method other than \texttt{DFGS} that provides a rigorous solution guarantee. The time complexity of \texttt{ColorCoding} is $\mathcal{O}(2^k \cdot e^k |\mathcal{E}|\log(\frac{|\mathcal{V}|}{\epsilon}))$ for a $(1-\epsilon)$ approximate solution, where $k$ is the effective solution size~\cite{cadena2019near}. It is extremely expensive when $k$ is large. However, it provides additional, heuristic preprocessing steps to reduce the graph size such that, empirically, $k<10$ is sufficient to find good solutions.  We used the code provided by the authors and implemented the approximation refinement based on suggestions from the authors. We tuned the parameter refinement coefficient $\beta$ over $[0.1, \cdots, 0.9]$ for different signal strengths, and ended up with $\beta=0.9$ for $\mu=5$, $\beta=0.8$ for $\mu=4$, $\beta=0.7$ for $\mu=3$, $\beta=0.6$ for $\mu=2$, and $\beta=0.5$ for $\mu=1.5$ which achieve the best performance. In addition, we set $k=5$ with $300$ iterations as suggested by the authors. 
    
    \item Non-parametric Heterogeneous Graph Scan (\texttt{NPHGS}) \cite{chen2014non}: optimizes the original Berk-Jones nonparametric scan statistic over connected subgraphs, using a greedy growth heuristic.   The nonparametric scan statistics are free of distributional assumptions and can be applied to anomalous connected subgraph detection in heterogeneous graph data, in contrast to traditional parametric scan statistics (e.g., the Kulldorff statistic). The time complexity of \texttt{NPHGS} is 
    $O(|\mathcal{V}|^2\log |\mathcal{V}|)$.
    We used the code provided by the authors, and chose the BJ scan statistic as the objective score to optimize with parameter $\alpha \in \{0.001, \cdots, 0.009, 0.01, \cdots, 0.09\}$.

    \item Additive Graph Scan (\texttt{AdditiveScan})~\cite{speakman2013dynamic}: was proposed as an efficient heuristic alternative to \texttt{DFGS} which can be used to identify the high-scoring (most positive) connected subsets in a given graph structure with real-valued weights at each node. This method stems from two facts that 1) additive functions satisfy the LTSS property, which could leverage \texttt{DFGS}, however, 2) computation time of \texttt{DFGS} is exponential in the graph size. The time complexity of \texttt{AdditiveScan} is $O(|\mathcal{V}|^2\sqrt{|\mathcal{V}|})$.
    We adopted the implementation by the authors, and chose the BJ scan statistic as the objective score to optimize with parameter $\alpha \in \{0.001, \cdots, 0.009, 0.01, \cdots, 0.09\}$.
    
    \item Depth First Graph Scan (\texttt{DFGS}) \cite{speakman2015scalable}: is a graph scan method that is guaranteed to find the exact solution, however, the worst case complexity of \texttt{DFGS} is exponential in the neighborhood size $k$. If no pruning was performed, \texttt{DFGS} would evaluate all connected subsets, requiring $O(2^k)$ run time; however, it is able to rule out many connected subsets as provably suboptimal, reducing complexity to $O(q^k)$ for some constant $1 <q< 2$, where $q$ is dependent on the proportion of subsets that are pruned. 
    We adopted the implementation by the authors directly with suggested hyperparameters, and chose the BJ scan statistic as the objective score to optimize with parameter $\alpha \in \{0.001, \cdots, 0.009, 0.01, \cdots, 0.09\}$.

\end{itemize}

\subsection{Experimental Setup and Evaluation Metrics}
\label{app:metrics}
For each of the five real-world graph structures enumerated above, we simulate $200$ runs of p-values generated under the null hypothesis and $50$ runs of p-values generated under each of five different alternative hypotheses with $\mu\in[1.5, 2, 3, 4, 5]$. We run all algorithms on these $2250$ graphs and record the detected subgraphs and their corresponding scores, and then performance evaluations in terms of detection power, precision, recall, and F-score are conducted. Let $\mathcal{R}$ be the ground truth nodes in the anomalous subgraph and let $\mathcal{S}$ be the detected subgraph.
\begin{itemize}
    \item Detection power measures the ability of a method to distinguish between graphs with or without an affected subgraph. It is computed based on the following steps: 1) compute BJ score for each detected subgraph; 2) for each alternative run, we conduct a hypothesis test with significance level $\alpha=0.05$ by setting p-value as the proportion of 
    null runs that have higher BJ score than the alternative run; 3) compute the proportion of hypothesis tests (for each method, for each real-world graph, for each signal strength $\mu$) that reject the null hypothesis.
    \item Precision is the ratio of the number of detected true positive nodes divided by the total number of detected nodes, that is, 
    \begin{equation}
        \texttt{Precision} = \frac{|\mathcal{R} \cap \mathcal{S}|}{|\mathcal{S}|}.
    \end{equation}
    \item Recall is the ratio of the number of detected true positive nodes divided by the number of nodes in true subgraph, that is, 
    \begin{equation}
        \texttt{Recall} = \frac{|\mathcal{R} \cap \mathcal{S}|}{|\mathcal{R}|}.
    \end{equation}
    \item F-score is the harmonic mean of precision and recall, that is,
    \begin{equation}
        \texttt{F-score} = \frac{2\cdot \texttt{Precision} \cdot \texttt{Recall}}{\texttt{Precision} + \texttt{Recall}}.
    \end{equation}
    We note that precision, recall, and F-score are each averaged over the 50 graphs created for a given signal strength $\mu$.
\end{itemize}

\subsection{Additional Experimental Results}
\label{app:results}

We now present a detailed comparison of the \texttt{CNSS} and baseline methods with respect to run time, detection power, and detection performance (precision, recall, and F-measure). Note that, for the larger datasets, we do not report the run time and performance of some baseline methods due to extremely long clock run time (over 2 weeks on 250 CPUs) to finish all experiments over 50 replicas of datasets under the alternative hypothesis.

\textbf{Run time.} As shown in Table~\ref{table:run_time}, not including the pre-processing time needed to compute the distribution of $\alpha'(N,\alpha)$ for a given graph structure, our method has competitive run time that is faster than \texttt{NPHGS}, \texttt{AdditiveScan}, and \texttt{DFGS}, which is aligned with the time complexity analysis.  We also observe substantial speedups (ranging from 2.6x for \texttt{WikiVote} to 28x for \texttt{CondMat}) for \texttt{CNSS+CoreTree} as compared to \texttt{CNSS} without core-tree decomposition.
These timing results are not impacted by the approach used to compute $\alpha'(N,\alpha)$, i.e., calibration by randomization tests versus calibration by lower bounds versus no calibration.

However, as shown in Table~\ref{table:run_time_calibration}, the total pre-processing time needed to perform calibration by randomization testing is large because the same search must be performed on a large number $K$ of replica datasets generated under the null hypothesis $\mathcal{H}_0$, thus multiplying the run time by $K$.  We used $K = 200$ for our experiments, thus substantially increasing run time.  However, we note that these null runs are entirely independent and thus can be easily parallelized. Moreover, this preprocessing step must only be performed once for a given graph structure, and can be reused if the graph structure remains constant, for example, for daily monitoring of disease cases with a graph structure defined by zip code adjacency.
Additionally, we observe that the use of core-tree decomposition resulted in similar speedups, ranging from 2.4x for \texttt{WikiVote} to 28x on \texttt{CondMat}, on the replica datasets.  Most critically, the use of lower bounds in place of randomization testing eliminates the need to generate and search over the large number of replica datasets, resulting in huge savings in preprocessing time. These speedups ranged from 300x to 2000x as compared to randomization testing with core-tree decomposition, assuming $K=200$ and no parallelization for the randomization tests.

\begin{table*}[!t]
\caption{Run Time on All Datasets. The run time of our method is reported for 18 different $\alpha$ values using a single processor, which could be parallelized to speed up the run time. We implement all discussed  algorithms and perform the   experiments on Linux servers  with the same hardware configuration  (64-bit  machines with Intel(R)Xeon(R) CPU E5-2680 v4 @ 2.40GHz and 251GB memory).  Note that the run times for \texttt{CNSS} and \texttt{CNSS+CoreTree} do not include the pre-processing time required to estimate $\alpha'(N,\alpha)$, and thus are independent of the calibration approach (randomization test versus lower bounds versus no calibration).}
\begin{adjustbox}{center}
\centering
\resizebox{0.9\textwidth}{!}{
\begin{tabular}{|c|c|c|c|c|c|}
\hline  
\multirow{2}{*}{Methods} & \multicolumn{1}{c|}{WikiVote}& \multicolumn{1}{c|}{CondMat} & \multicolumn{1}{c|}{Twitter}  & \multicolumn{1}{c|}{SlashDot} & \multicolumn{1}{c|}{DBLP} \\ \cline{2-6}
 & \multicolumn{1}{c|}{Run Time (sec.)} & \multicolumn{1}{c|}{Run Time (sec.)} &\multicolumn{1}{c|}{Run Time (sec.)} &\multicolumn{1}{c|}{Run Time (sec.)} &\multicolumn{1}{c|}{Run Time (sec.)}
 \\ \hline
LTSS &21 &24 &619 & 243 &1425 \\
EventTree &23 &25 &179 & 186 & 1019  \\
ColorCoding &5220 &8295 & 66690 & 29790 &124956\\
NPHGS &8912 &52046 &998624 &496587 & $\times$ \\
AdditiveScan & 17950 & 123100 & $\times$ & $\times$ & $\times$\\
DFGS &22791 &$\times$  & $\times$ & $\times$ & $\times$ \\
\hline
CNSS & 1771 & 43325 & 489624 & 447800 & $\times$ \\
CNSS+CoreTree & 685 & 1544 & 128812 & 45208 & 185053 \\
\hline
\end{tabular}
}
\end{adjustbox}
\label{table:run_time}
\end{table*}

\begin{table*}[!t]
\caption{Preprocessing Time Comparison for the Computation of $\alpha^\prime(N, \alpha)$. The run time of randomization testing is the run time of a single \texttt{CNSS} run under $\mathcal{H}_0$ times the number of replica datasets $K$, and we used $K=200$ runs for our experiments.}
\begin{adjustbox}{center}
\centering
\resizebox{1.0\textwidth}{!}{
\begin{tabular}{|c|c|c|c|c|c|}
\hline
\multirow{2}{*}{Methods} & \multicolumn{1}{c|}{WikiVote}& \multicolumn{1}{c|}{CondMat} & \multicolumn{1}{c|}{Twitter}  & \multicolumn{1}{c|}{SlashDot} & \multicolumn{1}{c|}{DBLP} \\ \cline{2-6}
 & \multicolumn{1}{c|}{Run Time (sec.)} & \multicolumn{1}{c|}{Run Time (sec.)} &\multicolumn{1}{c|}{Run Time (sec.)} &\multicolumn{1}{c|}{Run Time (sec.)} &\multicolumn{1}{c|}{Run Time (sec.)}
 \\ \hline
RandomizationTest  &$1602\times K$ & $28341\times K$ &$299349\times K$ &$375999\times K$ & $\times$  \\
RandomizationTest+CoreTree &$660\times K$ &$1026\times K$ & $107192\times K$ &$40124\times K$ & $147086\times K$ \\
LowerBounds &$59$ &$504$ &$16094$ &$9073$ &$87832$\\
\hline
\end{tabular}
}
\end{adjustbox}
\label{table:run_time_calibration}
\end{table*}
 
\textbf{Detection power.}
Table~\ref{table:detection_power_real_world_results} compares the detection power (i.e., the proportion of signals detected, at a fixed false positive rate of 0.05) for \texttt{CNSS} and baseline methods across the five real-world datasets and varying signal strengths. Note that several of the slower methods were not run for the largest graphs due to excessive run times needed to perform detection on the null graphs.  

We observe that all three of the calibrated CNSS methods (\texttt{CNSS}, \texttt{CNSS+CoreTree}, and \texttt{CNSS+LowerBound}) achieve perfect detection power across all of the real-world graphs and signal strengths considered.  In contrast, the uncalibrated CNSS methods and the baseline methods have substantially reduced detection power for low signal strengths, particularly on the smaller graphs (\texttt{WikiVote} and \texttt{CondMatter}) which also had a smaller number of true nodes generated under the alternative hypothesis $\mathcal{H}_1$.  Interestingly, even the uncalibrated CNSS methods outperformed the baseline methods with respect to detection power; among the baseline methods, \texttt{EventTree} and \texttt{ColorCoding} tended to outperform \texttt{LTSS}, \texttt{NPHGS}, and \texttt{AdditiveScan}.

\begin{table*}[!htbp]
\caption{Detection Power Comparison on Five Real World Datasets with Gaussian Signals.}
\label{table:detection_power_real_world_results}
\begin{adjustbox}{center}
\centering
\resizebox{1.0\textwidth}{!}{
\begin{tabular}{c|c|c|c|c|c}
\hline
\toprule
\multirow{2}{*}{Methods} & \multicolumn{1}{c|}{WikiVote $(\mu=1.5)$} & \multicolumn{1}{c|}{WikiVote $(\mu=2)$}  & \multicolumn{1}{c|}{WikiVote $(\mu=3)$} & \multicolumn{1}{c|}{WikiVote $(\mu=4)$} & \multicolumn{1}{c}{WikiVote $(\mu=5)$}\\ \cline{2-6}
 & \multicolumn{1}{c|}{Detection\ Power} &  \multicolumn{1}{c|}{Detection\ Power} &   \multicolumn{1}{c|}{Detection\ Power}&  \multicolumn{1}{c|}{Detection\ Power} &  \multicolumn{1}{c}{Detection\ Power} 
 \\ \hline
LTSS & 0.28  &0.4  &0.96  &1.0 &1.0\\
EventTree & 0.6&0.94  &1.0  &1.0  &1.0  \\
ColorCoding & 0.8&1.0 &1.0 &1.0  &1.0\\
NPHGS  & 0.0 &0.0  &0.0  &0.0 &0.0 \\
AdditiveScan &0.0  &0.0 &0.0 &0.0 &0.0 \\
\hline
CNSS+NoCalib & 0.94 &1.0 &1.0 &1.0  &1.0 \\
CNSS+CoreTree+NoCalib  & 0.86 &0.98 &0.98 &1.0  &1.0 \\
CNSS+CoreTree &1.0  &1.0 &1.0  &1.0&1.0  \\
CNSS+LowerBound & 1.0 &1.0 &1.0 &1.0 &1.0 \\
CNSS & 1.0 &1.0  &1.0 &1.0 &1.0 \\
\bottomrule
\hline\hline
\toprule
\multirow{2}{*}{Methods} & \multicolumn{1}{c|}{CondMat $(\mu=1.5)$} & \multicolumn{1}{c|}{CondMat $(\mu=2)$}  & \multicolumn{1}{c|}{CondMat $(\mu=3)$} & \multicolumn{1}{c|}{CondMat $(\mu=4)$} & \multicolumn{1}{c}{CondMat $(\mu=5)$}\\ \cline{2-6}
 & \multicolumn{1}{c|}{Detection\ Power} &  \multicolumn{1}{c|}{Detection\ Power} &   \multicolumn{1}{c|}{Detection\ Power}&  \multicolumn{1}{c|}{Detection\ Power} &  \multicolumn{1}{c}{Detection\ Power} 
 \\ \hline
LTSS &0.3 &0.68  &1.0  &1.0  &1.0  \\
EventTree& 0.66 &0.94  &1.0  &1.0  &1.0 \\
ColorCoding &0.0 &0.8  &1.0  &1.0  &1.0  \\
NPHGS &  0.0 &0.1  &0.72  &0.96  &1.0 \\
AdditiveScan & 0.32 &0.36  &0.36  &0.36  &0.38\\
\hline
CNSS+NoCalib &  0.96 &1.0  &1.0  &1.0  &1.0 \\
CNSS+CoreTree+NoCalib &  0.86 &1.0  &1.0  &1.0  &1.0 \\
CNSS+CoreTree &1.0 &1.0  &1.0  &1.0  &1.0  \\
CNSS+LowerBound & 1.0 &1.0  &1.0  &1.0  &1.0\\
CNSS &1.0 &1.0  &1.0  &1.0  &1.0 \\
\bottomrule
\hline\hline
\toprule
\multirow{2}{*}{Methods} & \multicolumn{1}{c|}{Twitter $(\mu=1.5)$} & \multicolumn{1}{c|}{Twitter $(\mu=2)$}  & \multicolumn{1}{c|}{Twitter $(\mu=3)$} & \multicolumn{1}{c|}{Twitter $(\mu=4)$} & \multicolumn{1}{c}{Twitter $(\mu=5)$}\\ \cline{2-6}
 & \multicolumn{1}{c|}{Detection\ Power} &  \multicolumn{1}{c|}{Detection\ Power} &   \multicolumn{1}{c|}{Detection\ Power}&  \multicolumn{1}{c|}{Detection\ Power} &  \multicolumn{1}{c}{Detection\ Power} 
 \\ \hline
LTSS &  0.0 &0.0  &0.0  &1.0  &1.0\\
EventTree & 1.0 &1.0  &1.0  &1.0  &1.0 \\
ColorCoding &1.0 &1.0  &1.0  &1.0  &1.0 \\
NPHGS &0.0 &0.0  &0.0  &0.02  &0.34 \\
\hline
CNSS+NoCalib & 1.0 &1.0  &1.0  &1.0  &1.0\\
CNSS+CoreTree+NoCalib & 1.0 &1.0  &1.0  &1.0  &1.0 \\
CNSS+CoreTree &  1.0 &1.0  &1.0  &1.0  &1.0 \\
CNSS+LowerBound & 1.0 &1.0  &1.0  &1.0  &1.0\\
CNSS & 1.0 &1.0  &1.0  &1.0  &1.0\\
\bottomrule
\hline\hline
\toprule
\multirow{2}{*}{Methods} & \multicolumn{1}{c|}{SlashDot $(\mu=1.5)$} & \multicolumn{1}{c|}{SlashDot $(\mu=2)$}  & \multicolumn{1}{c|}{SlashDot $(\mu=3)$} & \multicolumn{1}{c|}{SlashDot $(\mu=4)$} & \multicolumn{1}{c}{SlashDot $(\mu=5)$}\\ \cline{2-6}
 & \multicolumn{1}{c|}{Detection\ Power} &  \multicolumn{1}{c|}{Detection\ Power} &   \multicolumn{1}{c|}{Detection\ Power}&  \multicolumn{1}{c|}{Detection\ Power} &  \multicolumn{1}{c}{Detection\ Power} 
 \\ \hline
LTSS &1.0 &1.0  &1.0  &1.0  &1.0 \\
EventTree &1.0 &1.0  &1.0  &1.0  &1.0  \\
ColorCoding &1.0 &1.0  &1.0  &1.0  &1.0\\
NPHGS & 0.0 &0.0  &0.0  &0.04  &0.38 \\
\hline
CNSS+NoCalib. &1.0 &1.0  &1.0  &1.0  &1.0\\
CNSS+CoreTree+NoCalib &1.0 &1.0  &1.0  &1.0  &1.0 \\
CNSS+CoreTree &1.0 &1.0  &1.0  &1.0  &1.0 \\
CNSS+LowerBound &1.0 &1.0  &1.0  &1.0  &1.0 \\
CNSS &1.0 &1.0  &1.0  &1.0  &1.0 \\
\bottomrule
\hline\hline
\toprule
\multirow{2}{*}{Methods} & \multicolumn{1}{c|}{DBLP $(\mu=1.5)$} & \multicolumn{1}{c|}{DBLP $(\mu=2)$}  & \multicolumn{1}{c|}{DBLP $(\mu=3)$} & \multicolumn{1}{c|}{DBLP $(\mu=4)$} & \multicolumn{1}{c}{DBLP $(\mu=5)$}\\ \cline{2-6}
 & \multicolumn{1}{c|}{Detection\ Power} &  \multicolumn{1}{c|}{Detection\ Power} &   \multicolumn{1}{c|}{Detection\ Power}&  \multicolumn{1}{c|}{Detection\ Power} &  \multicolumn{1}{c}{Detection\ Power} 
 \\ \hline
LTSS & 1.0 &1.0  &1.0  &1.0  &1.0  \\
EventTree & 0.98 &1.0  &1.0  &1.0  &1.0\\
ColorCoding & 1.0 &1.0  &1.0  &1.0  &1.0 \\
\hline
CNSS+CoreTree &1.0 &1.0  &1.0  &1.0  &1.0\\
\bottomrule
\hline
\end{tabular}
}
\end{adjustbox}
\end{table*}

\textbf{Detection performance.}
The detection performances are shown in Table \ref{table:real_world_results}. The bold number indicates that method is significantly better than other methods. 
Overall, our proposed \texttt{CNSS} outperforms baseline methods under different signal strengths $\mu$ on the various network structures. 
Specifically, the calibrated BJ score helps to precisely pinpoint the true affected subgraph as the strength of signal increases. 
The use of core-tree decomposition and lower bounds do not have substantial effects on detection performance for these five real-world datasets, while significantly reducing run time.
On the other hand, the baseline methods do not have consistent performance over different values of $\mu$ with different network structures. 
\texttt{LTSS} has the worst detection performance, because it does not enforce the graph connectivity constraints and thus picks out a subset of disconnected, individually anomalous nodes that do not accurately reflect the true affected subgraph.
\texttt{EventTree} has relatively good performance on \texttt{Twitter} and \texttt{SlashDot} datasets when the event signal is not strong. 
However, when event signal is strong, the subgraph detected by \texttt{EventTree} includes many noisy nodes around true nodes to reach a higher score, thus dramatically harming precision; this pattern is also observed for the other baseline methods, while our proposed calibrated scan approach converges to both high precision and high recall as the signal strength increases. \texttt{DFGS} has relatively good performance when the graph is small, such as \texttt{CondMat} data with a strong event signal, but does not perform well on \texttt{WikiVote} and quickly becomes computationally infeasible for the larger graphs. \texttt{ColorCoding}, \texttt{AdditiveScan}, and \texttt{NPHGS} have similar performance: all of them suffer from the similar issue as \texttt{EventTree} in that, when the event signal is strong, they include many individually significant nodes that are not part of the true affected subgraph, thus reducing precision and F-score.

In addition, we show the average performance (F-score) over various signal strengths and network structures in Figure~\ref{fig:avg_f_scores}. We can see that \texttt{CNSS}, \texttt{CNSS+CoreTree}, and \texttt{CNSS+LowerBound} have much better average performance than all baselines, while CNSS without calibration (\texttt{CNSS+NoCalib}) performs poorly.

These results demonstrate the advantage of the novel calibration approach proposed here for precisely identifying the true affected subgraph, as well as the utility of our core-tree decomposition and lower bound approaches for enabling scalability and computational feasibility for large graphs.

\begin{table*}[!htbp]
\caption{Detection Performance Results (Average Precision, Recall, and F-score) on Five Real World Datasets with Gaussian Signals. The bold number indicates that method has a significantly higher F-score than the other methods.  Statistical significance is computed using paired t-tests ($p<0.05$).}
\label{table:real_world_results}
\begin{adjustbox}{center}
\centering
\resizebox{1.0\textwidth}{!}{
\begin{tabular}{c|c|c|c|c|c|c|c|c|c|c|c|c|c|c|c}
\hline
\toprule
\multirow{2}{*}{Methods} & \multicolumn{3}{c|}{WikiVote $(\mu=1.5)$} & \multicolumn{3}{c|}{WikiVote $(\mu=2)$}  & \multicolumn{3}{c|}{WikiVote $(\mu=3)$} & \multicolumn{3}{c|}{WikiVote $(\mu=4)$} & \multicolumn{3}{c}{WikiVote $(\mu=5)$}\\ \cline{2-16}
 & \multicolumn{1}{c|}{Prec.} & \multicolumn{1}{c|}{Rec.} & F-Score & \multicolumn{1}{c|}{Prec.} & \multicolumn{1}{c|}{Rec.} & F-Score & \multicolumn{1}{c|}{Prec.} & \multicolumn{1}{c|}{Rec.} & F-Score &
 \multicolumn{1}{c|}{Prec.} & \multicolumn{1}{c|}{Rec.} & F-Score & \multicolumn{1}{c|}{Prec.} & \multicolumn{1}{c|}{Rec.} & F-Score 
 \\ \hline
LTSS & 0.023&	0.144&	0.039 & 0.024&	0.154&	0.042&0.022&	0.145&	0.039& 0.024&	0.157&	0.042& 0.024&	0.155&	0.041\\
EventTree & 0.035 & 0.041 & 0.037 & 0.041&	0.060& 0.049 & 0.036&	0.075&	0.049& 0.034&	0.083&	0.048& 0.026	&0.244&	0.047\\
ColorCoding &0.074	&0.664	&0.132 &0.111	&0.801	&0.195 &0.145	&0.953	&0.252 &0.174	&0.997	&0.297 & 0.244	&0.920	&0.376\\
NPHGS & 0.144 & 0.547 & 0.227 &0.185 & 0.751 & 0.297& 0.216 & 0.946 & 0.351& 0.225 & 0.997 & 0.367& 0.346	&0.949	&0.432  \\
AdditiveScan & 0.143&	0.547&	0.227& 0.185&	0.751&	0.296 & 0.216&	0.946&	0.351 &  0.225&	0.997&	0.367& 0.226&	1.000&	0.368\\
DFGS & 0.154&	0.523&	0.235& 0.204&	0.703&	0.311& 0.206&	0.804&	0.325& 0.219&	0.836&	0.343& 0.203&	0.800&	0.321  \\
\hline
CNSS+NoCalib & 0.068 & 0.628 & 0.123& 0.087 & 0.802 & 0.156& 0.102 & 0.959 & 0.184 & 0.106 & 0.998 & 0.191& 0.106 & 1.000 & 0.192 \\
CNSS+CoreTree+NoCalib  & 0.073 & 0.652 & 0.132 &0.091 & 0.813 & 0.163 & 0.105 & 0.960 & 0.189 & 0.109 & 0.998 & 0.197 & 0.110 & 1.000 & 0.198\\
CNSS+CoreTree & 0.233&	0.400&	\textbf{0.265}& 0.285&	0.641&	\textbf{0.373}& 0.752&	0.578&	0.601 & 0.923&	0.803&	\textbf{0.858}& 0.968&	0.965&	\textbf{0.966}\\
CNSS+LowerBound &0.213	&0.233	&0.218 & 0.360	&0.379	&0.361 & 0.737	&0.567	&\textbf{0.630} &0.891	&0.810	&0.847 &0.951 &0.967	&0.958\\
CNSS & 0.232&	0.401&	\textbf{0.257}& 0.289&	0.645&	\textbf{0.372}& 0.706&	0.604&	0.583& 0.921&	0.803&	\textbf{0.858} & 0.965&	0.965&	\textbf{0.965} \\
\bottomrule
\hline\hline
\toprule
\multirow{2}{*}{Methods} & \multicolumn{3}{c|}{CondMat  ($\mu=1.5$)} & \multicolumn{3}{c|}{CondMat $(\mu=2)$}  & \multicolumn{3}{c|}{CondMat $(\mu=3)$} & \multicolumn{3}{c|}{CondMat $(\mu=4)$} & \multicolumn{3}{c}{CondMat $(\mu=5)$}\\ \cline{2-16}
 & \multicolumn{1}{c|}{Prec.} & \multicolumn{1}{c|}{Rec.} & F-Score & \multicolumn{1}{c|}{Prec.} & \multicolumn{1}{c|}{Rec.} & F-Score & \multicolumn{1}{c|}{Prec.} & \multicolumn{1}{c|}{Rec.} & F-Score &
 \multicolumn{1}{c|}{Prec.} & \multicolumn{1}{c|}{Rec.} & F-Score & \multicolumn{1}{c|}{Prec.} & \multicolumn{1}{c|}{Rec.} & F-Score 
 \\ \hline
LTSS & 0.017&	0.083&	0.029& 0.018&	0.090&	0.030 & 0.018&	0.094&	0.031& 0.019&	0.096&	0.032& 0.017 &	0.088&	0.029 \\
EventTree& 0.014&	0.204&	0.027& 0.015&	0.216&	0.028& 0.014&	0.208&	0.026& 0.014&	0.209&	0.026& 0.014&	0.214&	0.027  \\
ColorCoding & 0.352	&0.074	&0.094 &0.169	&0.564	&0.228 &0.254	&0.899	&0.379 &0.255	&0.894	&0.388 &0.313	&1.000	&0.470\\
NPHGS &0.203 & 0.416 & 0.272 & 0.245 & 0.626 & 0.351 & 0.304 & 0.939 & 0.459 & 0.315 & 0.995 & 0.478& 0.345 & 0.999 & 0.501 \\
AdditiveScan & 0.202&	0.478&	\textbf{0.283}& 0.247&	0.675&	0.361& 0.286&	0.945&	0.439& 0.306&	0.995&	0.467& 0.304&	1.000&	0.467    \\
DFGS & 0.459&	0.078&	0.132& 0.616&	0.193&	0.290& 0.819& 0.656&	\textbf{0.725}& 0.861&	0.937&	\textbf{0.897}& 0.866&	0.995&	0.926\\
\hline
CNSS+NoCalib & 0.045 & 0.629 & 0.084 & 0.056 & 0.779 & 0.104 & 0.068 & 0.957 & 0.126 & 0.070 & 0.997 & 0.132 & 0.070 & 1.000 & 0.132 \\
CNSS+CoreTree+NoCalib & 0.057 & 0.620 & 0.105& 0.071 & 0.779 & 0.130 & 0.086 & 0.958 & 0.158 & 0.089 & 0.995 & 0.163& 0.089 & 0.999 & 0.164\\
CNSS+CoreTree & 0.235 & 0.326 & 0.248 & 0.334 & 0.484 & 0.383& 0.533 & 0.654 & 0.560& 0.883 & 0.782 & 0.824& 0.932 & 0.890 & 0.905  \\
CNSS+LowerBound &0.078	&0.451	&0.132 &0.115	&0.609	&0.193 &0.267	&0.825	&0.394 & 0.883	&0.833	&0.857 &0.909	&0.975	&0.941\\
CNSS & 0.258 & 0.361 & \textbf{0.278} & 0.340 & 0.563 & \textbf{0.408} & 0.735 & 0.685 & 0.687 & 0.916 & 0.821 & 0.866 & 0.985 & 0.974 & \textbf{0.979} \\
\bottomrule
\hline\hline
\toprule
\multirow{2}{*}{Methods} & \multicolumn{3}{c|}{ Twitter ($\mu=1.5$)} & \multicolumn{3}{c|}{Twitter $(\mu=2)$}  & \multicolumn{3}{c|}{Twitter $(\mu=3)$} & \multicolumn{3}{c|}{Twitter $(\mu=4)$} & \multicolumn{3}{c}{Twitter $(\mu=5)$}\\ \cline{2-16}
 & \multicolumn{1}{c|}{Prec.} & \multicolumn{1}{c|}{Rec.} & F-Score & \multicolumn{1}{c|}{Prec.} & \multicolumn{1}{c|}{Rec.} & F-Score & \multicolumn{1}{c|}{Prec.} & \multicolumn{1}{c|}{Rec.} & F-Score &
 \multicolumn{1}{c|}{Prec.} & \multicolumn{1}{c|}{Rec.} & F-Score & \multicolumn{1}{c|}{Prec.} & \multicolumn{1}{c|}{Rec.} & F-Score 
 \\ \hline
LTSS & 0.075 & 0.726 & 0.136 & 0.084 &  0.829 & 0.152 & 0.096 & 0.974 & 0.175 &0.098 & 0.998 & 0.178 & 0.098 & 0.999 & 0.179\\
EventTree & 0.222 & 0.283 & \textbf{0.249} & 0.318 & 0.449 & \textbf{0.371} & 0.454 & 0.790 & 0.577 & 0.502 &0.962 & 0.660 & 0.510 & 0.997 & 0.675 \\
ColorCoding &0.060	&0.661	&0.109 & 0.084	&0.783	&0.153 &0.114	&0.957	&0.204 &0.136	&0.997	&0.239 &0.153	&1.000	&0.265 \\
NPHGS & 0.110 & 0.553 & 0.183 & 0.136 & 0.737 & 0.230 & 0.166 & 0.951 & 0.282 & 0.172 & 0.996 & 0.294 & 0.172 & 1.000 & 0.293  \\
\hline
CNSS+NoCalib & 0.063 & 0.601 & 0.114 & 0.079 & 0.764 & 0.144 & 0.097 & 0.954 & 0.177 & 0.101 & 0.996 & 0.184 & 0.102 & 0.999 & 0.185  \\
CNSS+CoreTree+NoCalib & 0.071 &0.617 &0.127 & 0.088 & 0.774 & 0.158 & 0.106 & 0.957 & 0.192 & 0.111 & 0.996 & 0.199 & 0.111 & 0.999 & 0.200  \\
CNSS+CoreTree & 0.137 & 0.496 & 0.215 & 0.198 & 0.662 & 0.305 &0.788 & 0.502 & \textbf{0.613} & 0.895 & 0.821 & \textbf{0.856} & 0.923 & 0.972 & \textbf{0.947} \\
CNSS+LowerBound &0.150 &0.264 &0.181 &0.260 &0.378 &0.302 &0.725 &0.515 &0.599 &0.909 &0.819 &0.861 &0.957 &0.972 &0.964\\
CNSS & 0.137 & 0.491 & 0.211 & 0.167 & 0.701 & 0.267 & 0.728 & 0.508 & 0.579 & 0.880 & 0.820 & \textbf{0.849} & 0.910 & 0.972 & \textbf{0.940}  \\
\bottomrule
\hline\hline
\toprule
\multirow{2}{*}{Methods} & \multicolumn{3}{c|}{ Slashdot ($\mu=1.5$)} & \multicolumn{3}{c|}{Slashdot $(\mu=2)$}  & \multicolumn{3}{c|}{Slashdot $(\mu=3)$} & \multicolumn{3}{c|}{Slashdot $(\mu=4)$} & \multicolumn{3}{c}{Slashdot $(\mu=5)$}\\ \cline{2-16}
 & \multicolumn{1}{c|}{Prec.} & \multicolumn{1}{c|}{Rec.} & F-Score & \multicolumn{1}{c|}{Prec.} & \multicolumn{1}{c|}{Rec.} & F-Score & \multicolumn{1}{c|}{Prec.} & \multicolumn{1}{c|}{Rec.} & F-Score &
 \multicolumn{1}{c|}{Prec.} & \multicolumn{1}{c|}{Rec.} & F-Score & \multicolumn{1}{c|}{Prec.} & \multicolumn{1}{c|}{Rec.} & F-Score 
 \\ \hline
LTSS & 0.099 & 0.663 & 0.173 & 0.116 & 0.823 & 0.203 & 0.130 & 0.973 & 0.229 & 0.133 & 0.998 & 0.234 & 0.132 & 0.999 & 0.234 \\
EventTree & 0.264 & 0.312 & \textbf{0.285} & 0.456 & 0.482 & \textbf{0.410} & 0.476 & 0.810 & 0.599 & 0.527 & 0.967 & 0.682 &0.526 & 0.996 & 0.689 \\
ColorCoding & 0.077&	0.713	&0.140 &0.112	&0.797	&0.197 &0.145	&0.959	&0.251 & 0.172	&0.997	&0.293  & 0.196	&1.000	&0.327\\
NPHGS  &0.153 & 0.539 & 0.238 &0.186 & 0.731 & 0.296 & 0.215 & 0.949 & 0.351 & 0.223 & 0.996 & 0.364 & 0.222 & 1.000 & 0.363\\
\hline
CNSS+NoCalib. & 0.063 & 0.693 & 0.116 & 0.074 & 0.822 & 0.136 & 0.086 & 0.965 & 0.159 & 0.089 & 0.997 & 0.164 & 0.089 & 0.999 & 0.165 \\
CNSS+CoreTree+NoCalib & 0.078 & 0.709 & 0.141 & 0.091 & 0.831 & 0.164 & 0.104 & 0.968 & 0.189 & 0.108 & 0.997 & 0.195 & 0.108 & 0.999 & 0.195  \\
CNSS+CoreTree & 0.157 & 0.529 &0.242  & 0.192 & 0.720 &0.303  & 0.591 & 0.672 &\textbf{0.629} & 0.886 & 0.829 &\textbf{0.857}  & 0.905 & 0.971 & \textbf{0.937} \\
CNSS+LowerBound &0.070 &0.655 &0.127 &0.086 &0.786 &0.155 &0.716 &0.541 &0.609 &0.873 &0.829 &0.850 &0.891 &0.972 &0.930\\
CNSS &0.159 & 0.528 &0.245 &0.193 &0.722 & 0.304 & 0.587 & 0.678 & \textbf{0.629} & 0.872 & 0.829 & \textbf{0.850} & 0.901 & 0.971 & \textbf{0.935}\\
\bottomrule
\hline\hline
\toprule
\multirow{2}{*}{Methods} & \multicolumn{3}{c|}{ DBLP ($\mu=1.5$)} & \multicolumn{3}{c|}{DBLP $(\mu=2)$}  & \multicolumn{3}{c|}{DBLP $(\mu=3)$} & \multicolumn{3}{c|}{DBLP $(\mu=4)$} & \multicolumn{3}{c}{DBLP $(\mu=5)$}\\ \cline{2-16}
 & \multicolumn{1}{c|}{Prec.} & \multicolumn{1}{c|}{Rec.} & F-Score & \multicolumn{1}{c|}{Prec.} & \multicolumn{1}{c|}{Rec.} & F-Score & \multicolumn{1}{c|}{Prec.} & \multicolumn{1}{c|}{Rec.} & F-Score &
 \multicolumn{1}{c|}{Prec.} & \multicolumn{1}{c|}{Rec.} & F-Score & \multicolumn{1}{c|}{Prec.} & \multicolumn{1}{c|}{Rec.} & F-Score 
 \\ \hline
LTSS & 0.055 & 0.565 & 0.100 & 0.073 & 0.777 & 0.134 & 0.087 & 0.968 & 0.159 & 0.089 & 0.998 & 0.164 & 0.089 & 1.000 & 0.164 \\
EventTree & 0.081 & 0.357 & 0.132 & 0.119 & 0.538 & 0.194 & 0.174 & 0.834 & 0.288 & 0.198 & 0.972 & 0.329 & 0.203 & 0.998 & 0.340\\
ColorCoding &0.056	&0.435	&0.100 &0.104 &0.656	&0.180 &0.152 &0.943 &0.261 &0.179 &0.996	&0.304 &0.205 &1.000 &0.340  \\
\hline
CNSS+CoreTree+NoCalib & 0.022&0.623&0.043& 0.028&0.775&0.054 &0.034&0.955&0.066&0.035&0.994&0.068&0.035&0.995&0.068 \\
CNSS+CoreTree & 0.135 & 0.265 & \textbf{0.179} & 0.190 &0.448 &\textbf{0.267} & 0.355 & 0.560 &\textbf{0.435} & 0.550 & 0.821 & \textbf{0.659} & 0.831 & 0.949 & \textbf{0.886}\\
\bottomrule
\hline
\end{tabular}
}
\end{adjustbox}
\end{table*}

\begin{figure*}[!ht]
    \centering
    \includegraphics[width=0.7\linewidth]{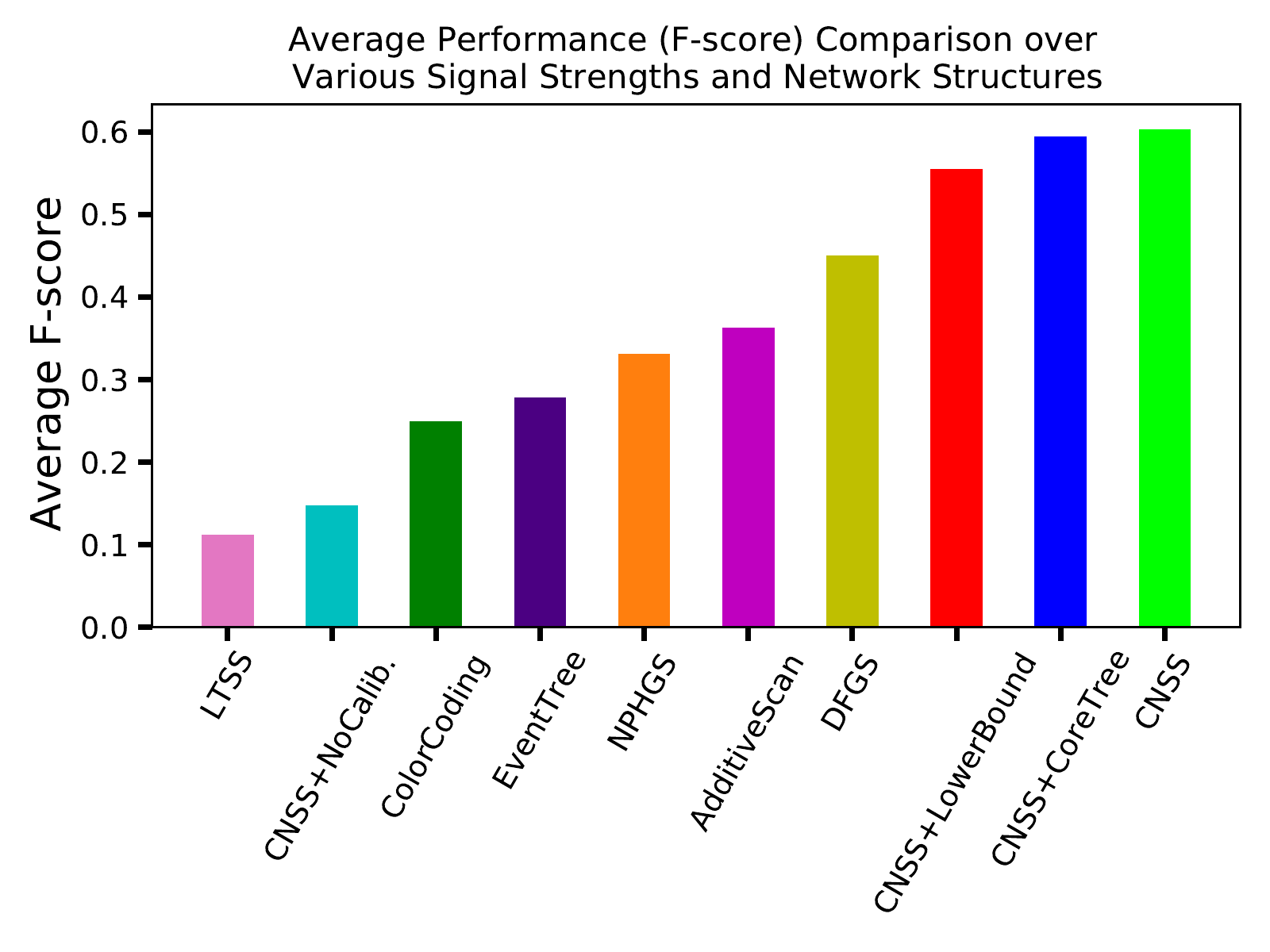}
    \caption{We compare the average performance (F-score) over various signal strengths and network structures to mimic the various real-world scenarios. We performed two-sample paired t-tests between each baseline and \texttt{CNSS}, and we find that all baselines have significantly lower performance than \texttt{CNSS} with p-value $< 0.05$.} \label{fig:avg_f_scores}
\end{figure*}

\subsection{Results with piecewise constant p-values}
\label{app:piecewise}

Though the nonparametric scan statistics (calibrated or uncalibrated) do not make any assumptions of Gaussianity, we used Gaussian signals in our simulation experiments to show that the calibrated NPSS formulation can achieve high detection performance even when the injected signal does not necessarily obey our specific modeling assumptions. Gaussian mean-shift signals are simple, have a natural way of measuring signal strength, and are frequently used in the literature, e.g., by~\citet{reyna2021} and~\citet{chitra2021}.  However, one might ask what happens when the Berk-Jones NPSS modeling assumptions are precisely correct, and the resulting p-values are piecewise constant under $\mathcal{H}_1(\mathcal{S})$.  Does the uncalibrated Berk-Jones statistic still fail to detect these signals due to miscalibration, and does the calibrated BJ statistic still outperform competing approaches by a wide margin?  

To explore these questions, we performed additional simulations using the \texttt{WikiVote} and \texttt{CondMat} datasets, comparing the calibrated scans (\texttt{CNSS}, \texttt{CNSS+CoreTree}, and \texttt{CNSS+LowerBounds}) with the uncalibrated scan (\texttt{CNSS+NoCalib}) and the various baselines (\texttt{LTSS}, \texttt{EventTree}, \texttt{NPHGS}, \texttt{AdditiveScan}, \texttt{DFGS}, and \texttt{ColorCoding}) with respect to detection power, 
precision, recall, and F-score.  As in the main evaluation, we simulated the true subgraph $\mathcal{S}$ using a random walk with size roughly $0.01|\mathcal{V}|$ (see Table~\ref{table:semi-synthetic}), and reported the average performance over 50 runs of simulations of true subgraphs and p-values, for each signal strength, on each network structure.  However, for these runs, we assumed piecewise constant p-values, where each p-value $p_i \in \mathcal{S}$ is drawn from \texttt{Uniform}[0, 0.01] with probability $q\cdot 0.01$, and from \texttt{Uniform}[0.01, 1] with probability $1-q \cdot 0.01$.  Signal strengths $q \in \{10,25,50,75,100\}$ were used for these simulations, and p-values outside subset $S$ were drawn from \texttt{Uniform}[0,1] as usual. 
The results of these simulations are shown in Tables~\ref{table:detection_power_real_world_results_piecewise_uniform} and~\ref{table:real_world_results_piecewise_uniform}.  

We observe that the results for piecewise constant p-values are highly consistent with those for Gaussian signals, demonstrating that it is miscalibration (not the shape of the signal) that is causing the uncalibrated methods to perform poorly.  As we observed for the Gaussian signals, our proposed \texttt{CNSS} and \texttt{CNSS+CoreTree} outperformed all baselines and the uncalibrated \texttt{CNSS+NoCalib} by a wide margin in terms of detection accuracy and detection power, while the uncalibrated methods suffered from low detection power for low signal strengths, and low precision (and therefore low F-score) across all signal strengths.  \texttt{CNSS+LowerBounds} consistently outperformed the uncalibrated scan and baseline methods for \texttt{WikiVote} across all signal strengths, and for \texttt{CondMat} for high signal strengths.  For low signal strengths on \texttt{CondMat}, \texttt{CNSS+LowerBounds} achieved higher detection power and recall than the uncalibrated scan and baseline methods, but had lower precision and F-score.  

Finally, we examined the mean and standard deviation of the selected value of $\alpha$ for each method across the 50 runs for each dataset and signal strength, noting that the signal was injected with a true $\alpha$ value of 0.01.  These results are shown in Table~\ref{table:chosen_alpha_piecewise_uniform}.  We observe that the uncalibrated scans and baseline NPSS methods fail to identify the true $\alpha$ value, instead consistently selecting the largest $\alpha$ value considered, i.e., $\alpha=0.09$.  In contrast, as the signal strength increases, \texttt{CNSS}, \texttt{CNSS+CoreTree}, and \texttt{CNSS+LowerBounds} are all able to reliably identify the value, $\alpha=0.01$, corresponding to the true injected signal.

\begin{table*}[!htbp]
\caption{Detection Power Comparison on \texttt{WikiVote} and \texttt{CondMat} datasets, assuming piecewise constant p-values.}
\label{table:detection_power_real_world_results_piecewise_uniform}
\begin{adjustbox}{center}
\centering
\resizebox{1.0\textwidth}{!}{
\begin{tabular}{c|c|c|c|c|c}
\hline
\toprule
\multirow{2}{*}{Methods} & \multicolumn{1}{c|}{WikiVote $(q=10)$} & \multicolumn{1}{c|}{WikiVote $(q=25)$}  & \multicolumn{1}{c|}{WikiVote $(q=50)$} & \multicolumn{1}{c|}{WikiVote $(q=75)$} & \multicolumn{1}{c}{WikiVote $(q=100)$}\\ \cline{2-6}
 & \multicolumn{1}{c|}{Detection\ Power} &  \multicolumn{1}{c|}{Detection\ Power} &   \multicolumn{1}{c|}{Detection\ Power}&  \multicolumn{1}{c|}{Detection\ Power} &  \multicolumn{1}{c}{Detection\ Power} 
 \\ \hline
LTSS & 0.12 &0.18  &0.4  &0.64  &0.86\\
EventTree & 0.12 &0.2  &0.24  &0.42  &0.4 \\
ColorCoding & 0.02 &0.12  &0.98  &1.0  &1.0 \\
NPHGS  & 0.0 &0.0  &0.0  &0.0  &0.0  \\
AdditiveScan & 0.04 &0.12  &0.52  &0.98  &1.0\\
\hline
CNSS+NoCalib & 0.12 &0.44  &0.86  &0.98  &1.0 \\
CNSS+CoreTree+NoCalib & 0.08 &0.4  &0.8  &0.98  &1.0 \\
CNSS+CoreTree & 1.0 &1.0  &1.0  &1.0  &1.0  \\
CNSS+LowerBound & 1.0 &1.0  &1.0  &1.0  &1.0\\
CNSS & 1.0 &1.0  &1.0  &1.0  &1.0\\
\bottomrule
\hline\hline
\toprule
\multirow{2}{*}{Methods} & \multicolumn{1}{c|}{CondMat $(q=10)$} & \multicolumn{1}{c|}{CondMat $(q=25)$}  & \multicolumn{1}{c|}{CondMat $(q=50)$} & \multicolumn{1}{c|}{CondMat $(q=75)$} & \multicolumn{1}{c}{CondMat $(q=100)$}\\ \cline{2-6}
 & \multicolumn{1}{c|}{Detection\ Power} &  \multicolumn{1}{c|}{Detection\ Power} &   \multicolumn{1}{c|}{Detection\ Power}&  \multicolumn{1}{c|}{Detection\ Power} &  \multicolumn{1}{c}{Detection\ Power} 
 \\ \hline
LTSS & 0.06 &0.14  &0.42  &0.76  &0.92\\
EventTree & 0.06 &0.1  &0.28  &0.38  &0.54 \\
ColorCoding & 0.08 &0.64  &1.0  &1.0  &1.0 \\
NPHGS & 0.0 &0.0  &0.0  &0.1  &0.86\\
AdditiveScan & 0.04 &0.34  &0.88  &1.0  &1.0\\
\hline
CNSS+NoCalib & 0.12 &0.36  &0.84  &0.98  &1.0\\
CNSS+CoreTree+NoCalib  & 0.44 &0.72  &0.94  &1.0  &1.0 \\
CNSS+CoreTree & 1.0 &1.0  &1.0  &1.0  &1.0  \\
CNSS+LowerBound & 1.0 &1.0  &1.0  &1.0  &1.0 \\
CNSS & 1.0 &1.0  &1.0  &1.0  &1.0\\
\bottomrule
\hline
\end{tabular}
}
\end{adjustbox}
\end{table*}

\begin{table*}[!htbp]
\caption{Detection Performance Results (Average Precision, Recall, and F-score) on \texttt{WikiVote} and \texttt{CondMat} datasets, assuming piecewise constant p-values. The bold number indicates that method has a significantly higher F-score than the other methods.  Statistical significance is computed using paired t-tests ($p<0.05$).}
\label{table:real_world_results_piecewise_uniform}
\begin{adjustbox}{center}
\centering
\resizebox{1.0\textwidth}{!}{
\begin{tabular}{c|c|c|c|c|c|c|c|c|c|c|c|c|c|c|c}
\hline
\toprule
\multirow{2}{*}{Methods} & \multicolumn{3}{c|}{WikiVote  ($q=10$)} & \multicolumn{3}{c|}{WikiVote $(q=25)$}  & \multicolumn{3}{c|}{WikiVote $(q=50)$} & \multicolumn{3}{c|}{WikiVote $(q=75)$} & \multicolumn{3}{c}{WikiVote $(q=100)$}\\ \cline{2-16}
 & \multicolumn{1}{c|}{Prec.} & \multicolumn{1}{c|}{Rec.} & F-Score & \multicolumn{1}{c|}{Prec.} & \multicolumn{1}{c|}{Rec.} & F-Score & \multicolumn{1}{c|}{Prec.} & \multicolumn{1}{c|}{Rec.} & F-Score &
 \multicolumn{1}{c|}{Prec.} & \multicolumn{1}{c|}{Rec.} & F-Score & \multicolumn{1}{c|}{Prec.} & \multicolumn{1}{c|}{Rec.} & F-Score 
 \\ \hline
LTSS & 0.016 & 0.100 & 0.027 & 0.015	&0.101	&0.027 &0.015	&0.104	&0.027 & 0.015	&0.107	&0.027 & 0.015	&0.110	&0.027\\
EventTree & 0.058	&0.006	&0.010 &0.083	&0.008	&0.015 & 0.064	&0.010	&0.017 & 0.060	&0.010	&0.017 & 0.049	&0.010	&0.016\\
ColorCoding & 0.045	&0.346	&0.080 & 0.068	&0.408	&0.117 & 0.105	&0.586	&0.178 &0.148	&0.779	&0.249 & 0.200	&1.000	&0.333\\
NPHGS & 0.051	&0.159	&0.077 & 0.089	&0.299	&0.137 &0.142	&0.528	&0.223 &0.187	&0.761	&0.300 & 0.227	&1.000	&0.369\\
AdditiveScan & 0.051	&0.160	&0.078 & 0.089	&0.299	&0.137 &0.142	&0.528	&0.223 & 0.187	&0.761	&0.300 & 0.226	&1.000	&0.369 \\
DFGS & 0.052	&0.159	&0.078 &0.092	&0.292	&0.139 & 0.150	&0.501	&0.228 & 0.184	&0.684	&0.288 &0.209	&0.820	&0.329\\
\hline
CNSS+NoCalib & 0.034	&0.303	&0.061 & 0.046	&0.422	&0.084 & 0.066	&0.612	&0.120 & 0.086	& 0.805	& 0.155 &0.105	&1.000	&0.190\\
CNSS+CoreTree+NoCalib  & 0.035	&0.301	&0.063 & 0.049	&0.423	&0.088 & 0.070	&0.614	&0.125 & 0.090	&0.806	&0.162 & 0.110	&1.000	&0.198\\
CNSS+CoreTree & 0.169	&0.164	&\textbf{0.102} & 0.467	&0.211	&0.275 &0.616	&0.471	&\textbf{0.522} &0.714	&0.730	&\textbf{0.718} & 0.651	&0.992	&\textbf{0.773}\\
CNSS+LowerBound & 0.068	& 0.174	& 0.096 & 0.127	 &0.312	 &0.177 & 0.274	 &0.534	 &0.355 & 0.404	& 0.760	 & 0.526 & 0.485	&1.000	& 0.653\\
CNSS & 0.163	&0.246	&\textbf{0.107} & 0.464	&0.218	&\textbf{0.283} & 0.608	&0.473	&\textbf{0.521} &0.679	&0.736	&0.695 &0.595	&0.992	&0.730\\
\bottomrule
\hline\hline
\toprule
\multirow{2}{*}{Methods} & \multicolumn{3}{c|}{CondMat  ($q=10$)} & \multicolumn{3}{c|}{CondMat $(q=25)$}  & \multicolumn{3}{c|}{CondMat $(q=50)$} & \multicolumn{3}{c|}{CondMat $(q=75)$} & \multicolumn{3}{c}{CondMat $(q=100)$}\\ \cline{2-16}
 & \multicolumn{1}{c|}{Prec.} & \multicolumn{1}{c|}{Rec.} & F-Score & \multicolumn{1}{c|}{Prec.} & \multicolumn{1}{c|}{Rec.} & F-Score & \multicolumn{1}{c|}{Prec.} & \multicolumn{1}{c|}{Rec.} & F-Score &
 \multicolumn{1}{c|}{Prec.} & \multicolumn{1}{c|}{Rec.} & F-Score & \multicolumn{1}{c|}{Prec.} & \multicolumn{1}{c|}{Rec.} & F-Score 
 \\ \hline
LTSS & 0.009	&0.092	&0.017 &0.009	&0.093	&0.017 &0.009	&0.095	&0.017 &0.009	&0.097	&0.017 &0.009	&0.099	&0.017 \\
EventTree& 0.055	&0.007	&0.012 & 0.052	&0.008	&0.014 &0.051	&0.010	&0.016 &0.053	&0.012	&0.019 & 0.053	&0.014	&0.021\\
ColorCoding & 0.036	&0.120	&0.055 & 0.076	&0.228	&0.113 &0.139	&0.460	&0.213 & 0.204	&0.730	&0.318 & 0.273	&1.000	&0.429\\
NPHGS & 0.052	&0.061	&0.055 & 0.103	&0.154	&0.122 &0.193	&0.398	&0.258 &0.261	&0.697	&0.379 &0.316	&1.000	&0.480\\
AdditiveScan & 0.048	&0.076	&0.059 &0.100	&0.182	&0.128 &0.184	&0.433	&0.257 &0.247	&0.715	&0.367 &0.298	&1.000	&0.459\\
DFGS &  0.046	&0.076	&0.057 & 0.093	&0.185	&0.123  &0.174	&0.425	&0.246 & 0.240	&0.709	&0.358 &0.294	&1.000	&0.454\\
\hline
CNSS+NoCalib & 0.021	&0.289	&0.039 &0.030	&0.418	&0.056 & 0.044	&0.614	&0.081 & 0.056	&0.799	&0.105 & 0.070	&1.000	&0.131\\
CNSS+CoreTree+NoCalib & 0.025	&0.275	&0.046 &0.036	&0.398	&0.066 & 0.053	& 0.592	& 0.097 & 0.069	& 0.790	& 0.127 & 0.086	&1.000	&0.159\\
CNSS+CoreTree & 0.049	&0.235	&0.052 &0.294	&0.213	&0.130 &0.584	&0.314	&0.350 &0.703	&0.530	&0.577 & 0.836	&0.836	&0.829\\
CNSS+LowerBound & 0.021	&0.286	&0.040 &0.031	&0.416	&0.057 &0.057	&0.609	&0.102 &0.218	&0.783	&0.315 &0.422	&1.000	&0.591\\
CNSS & 0.072	&0.300	&\textbf{0.064} & 0.372	&0.200	&\textbf{0.192} & 0.524	&0.394	&\textbf{0.390} &0.766	&0.631	&\textbf{0.677} & 0.893	&1.000	&\textbf{0.943}\\
\bottomrule
\hline
\end{tabular}
}
\end{adjustbox}
\end{table*}

\begin{table*}[!htbp]
\caption{Chosen $\alpha$ values for each method on \texttt{WikiVote} and \texttt{CondMat} datasets, assuming piecewise constant p-values.  Note that the \texttt{EventTree} method, unlike the nonparametric scan approaches, does not optimize over $\alpha$.}
\label{table:chosen_alpha_piecewise_uniform}
\begin{adjustbox}{center}
\centering
\resizebox{1.0\textwidth}{!}{
\begin{tabular}{c|c|c|c|c|c|c|c|c|c|c}
\hline
\toprule
\multirow{2}{*}{Methods} & \multicolumn{2}{c|}{WikiVote  ($q=10$)} & \multicolumn{2}{c|}{WikiVote $(q=25)$}  & \multicolumn{2}{c|}{WikiVote $(q=50)$} & \multicolumn{2}{c|}{WikiVote $(q=75)$} & \multicolumn{2}{c}{WikiVote $(q=100)$}\\ \cline{2-11}
 & Mean & SD  & Mean & SD & Mean & SD & Mean & SD &Mean & SD\\ \hline
LTSS & 0.09 & 0 & 0.09 & 0 & 0.09 & 0 & 0.09 & 0 & 0.09 & 0 \\
EventTree & - & - & - & - & - & - & - & - & - & - \\
ColorCoding & 0.09 & 0 & 0.09 & 0 & 0.09 & 0 & 0.09 & 0 & 0.09 & 0 \\
NPHGS & 0.09 & 0 & 0.09 & 0 & 0.09 & 0 & 0.09 & 0 & 0.09 & 0 \\
AdditiveScan & 0.09 & 0 & 0.09 & 0 & 0.09 & 0 & 0.09 & 0 & 0.09 & 0 \\
DFGS & 0.09 & 0 & 0.09 & 0 & 0.09 & 0 & 0.09 & 0 & 0.09 & 0 \\
\hline
CNSS+NoCalib & 0.09 & 0 & 0.09 & 0 & 0.09 & 0 & 0.09 & 0 & 0.09 & 0\\
CNSS+CoreTree+NoCalib & 0.09 & 0 & 0.09 & 0 & 0.09 & 0 & 0.09 & 0 & 0.09 & 0\\
CNSS+CoreTree & 0.040 & 0.038 & 0.014 & 0.010 & 0.014 & 0.012 & 0.011 & 0.006 & 0.010 & 0.001 \\
CNSS+LowerBound & 0.022 & 0.006 &0.021 & 0.007 & 0.014 & 0.007 & 0.011 & 0.003 & 0.01 & 0 \\
CNSS & 0.040 & 0.038 & 0.014 & 0.010 & 0.014 & 0.012 & 0.011 & 0.006 & 0.01 & 0 \\
\bottomrule
\hline\hline
\toprule
\multirow{2}{*}{Methods} & \multicolumn{2}{c|}{CondMat  ($q=10$)} & \multicolumn{2}{c|}{CondMat $(q=25)$}  & \multicolumn{2}{c|}{CondMat $(q=50)$} & \multicolumn{2}{c|}{CondMat $(q=75)$} & \multicolumn{2}{c}{CondMat $(q=100)$}\\ \cline{2-11}
 & Mean & SD  & Mean & SD & Mean & SD & Mean & SD &Mean & SD\\ \hline
LTSS & 0.09 & 0 & 0.09 & 0 & 0.09 & 0 & 0.09 & 0 & 0.09 & 0 \\
EventTree & - & - & - & - & - & - & - & - & - & - \\
ColorCoding & 0.09 & 0 & 0.09 & 0 & 0.09 & 0 & 0.09 & 0 & 0.09 & 0 \\
NPHGS & 0.09 & 0 & 0.09 & 0 & 0.09 & 0 & 0.09 & 0 & 0.09 & 0 \\
AdditiveScan & 0.09 & 0 & 0.09 & 0 & 0.09 & 0 & 0.09 & 0 & 0.09 & 0\\
DFGS & 0.09 & 0 & 0.09 & 0 & 0.09 & 0 & 0.09 & 0 & 0.09 & 0 \\
\hline
CNSS+NoCalib & 0.09 & 0 & 0.09 & 0 & 0.09 & 0 & 0.09 & 0 & 0.09 & 0\\
CNSS+CoreTree+NoCalib  & 0.09 & 0 & 0.09 & 0 & 0.09 & 0 & 0.09 & 0 & 0.09 & 0\\
CNSS+CoreTree & 0.077 &0.025 &0.044 & 0.034 & 0.019 &0.018 & 0.018 & 0.017 & 0.011 & 0.004 \\
CNSS+LowerBound & 0.087 & 0.006 & 0.087 & 0.006 &0.076 &0.022 &0.039 & 0.036 & 0.011 & 0.004 \\
CNSS & 0.055 & 0.037 & 0.02 & 0.02 & 0.021   & 0.019 & 0.019 & 0.017 & 0.01 & 0\\
\bottomrule
\hline
\end{tabular}
}
\end{adjustbox}
\end{table*}

\section{Case Studies}
\label{app:case}

\subsection{Black Lives Matter Event Detection in Twitter}
\label{app:blm_case}

\begin{figure*}[!ht]
    \centering
    \includegraphics[width=0.9\linewidth]{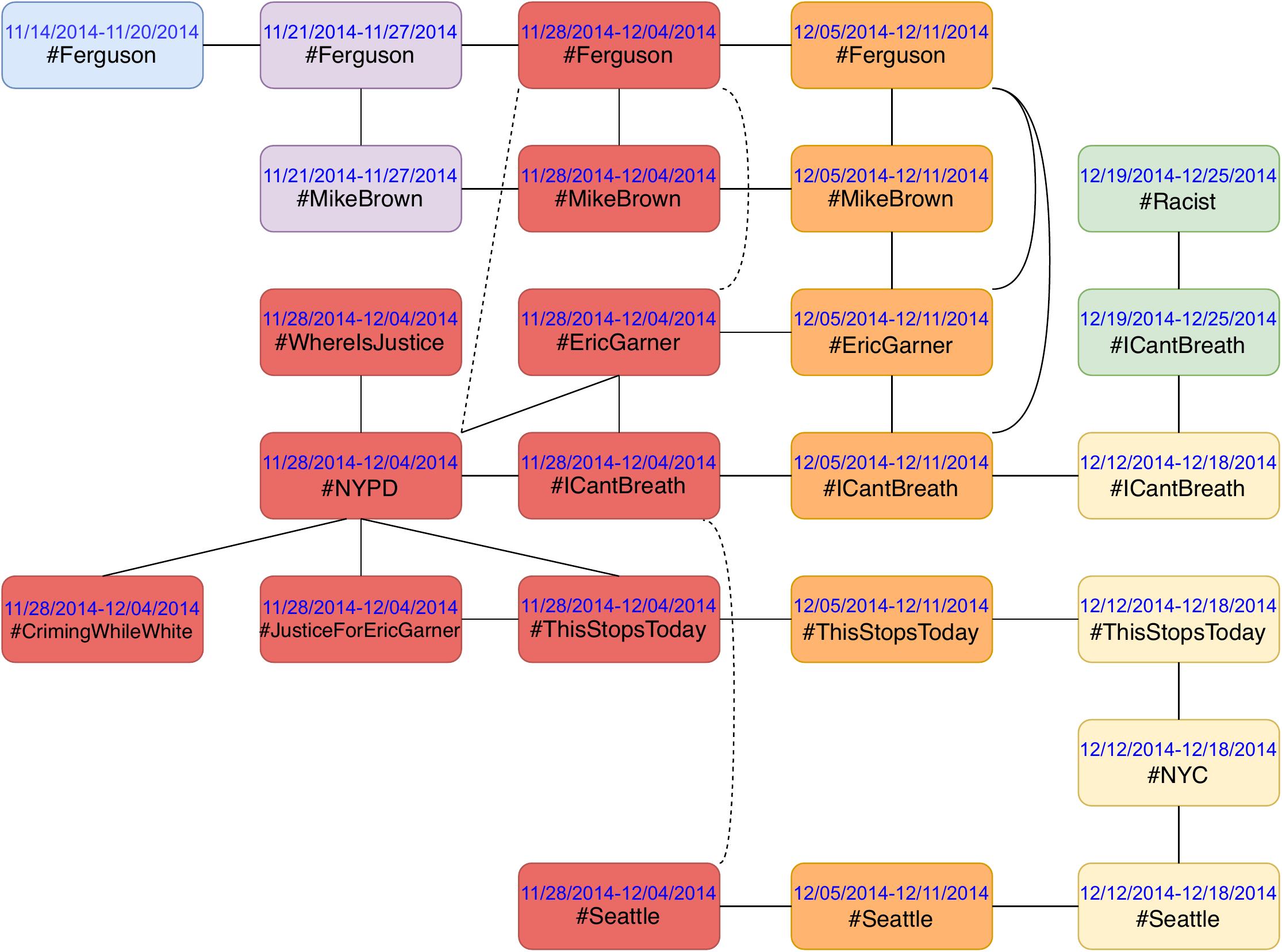}
    \caption{Detected Subgraph in \#BlackLivesMatter Tweets. Different colors indicate different weeks. Dashed lines indicate crossing edges.}
    \label{fig:cnss_blm}
\end{figure*}

In addition to the COVID-19 case study described in the main paper, we also conduct another case study using tweets with hashtag \textit{\#BlackLivesMatter}\footnote{\url{https://dataverse.harvard.edu/dataset.xhtml?persistentId=doi:10.7910/DVN/IQ525U&version=1.0}} collected from August 8th, 2014 to August 31st, 2015 to discover events 
related to this social movement for racial justice during these 56 weeks. There are 442,077 unique tweets and 42,898 unique hashtags in total. 
This Twitter dataset contains account user names which may reveal personally identifiable information. However, we pre-processed the data to remove the user names, and only used the hashtags in each tweet and corresponding creation timestamps in our research.
We generate a temporal graph with 100,123 nodes and 257,641 edges based on the mentioned hashtags in tweets during these 56 weeks, in which each node represents a mentioned hashtag in a particular week. 
We connect two nodes (two hashtags in a week) if they were co-mentioned in at least one tweet that week. 
We also add temporal edges between node $v_{i,t}$ (hashtag $i$ in week $t$) and node $v_{i,t+1}$ (hashtag $i$ in week $t+1$) if both nodes have non-zero mentioned counts in each week.
In addition, we also add dummy nodes and edges to smooth the temporal transition. For example, a dummy node $v_{i,t}$  is added into the graph if the hashtag $i$ is not mentioned in week $t$ but is mentioned in both week $t-1$ and week $t+1$. The temporal edges $(v_{i,t-1}, v_{i, t})$ and $(v_{i,t}, v_{i,t+1})$ are also added into the graph.

Some co-mentioned hashtags may not be relevant to each other. In order to detect a event with relevant hashtags, we remove the edges between any two hashtags in a week if the overlap coefficient $\rho_{ij}^{t}$ of the edge $(v_{i,t}, v_{j,t})$ is smaller than a constant (i.e., $\rho_{ij}^{t} < 0.1$).  Overlap coefficient is defined as 
\begin{equation}
    \rho_{ij}^t = \frac{\text{\#co-mention(tag\_$i$, tag\_$j$) in week $t$}}{(\text{\#tag\_$i$ + \#tag\_$j$ - \#co-mention(tag\_$i$, tag\_$j$)) in week $t$}}.
\end{equation}
The processed graph includes 100,123 nodes and 137,984 edges. The p-value of each node in the graph is computed based on the rank of the expectation-based Poisson (EBP) statistic~\cite{neill2012fast} divided by the total number of nodes. For each node, we compute the EBP score as:
\begin{equation}
    \texttt{EBP} = C \log\frac{C}{B} + B - C,
\end{equation}
if $C>B$, and $\texttt{EBP} = 0$ otherwise, where $C$ is the observed count (number of mentions of hashtag $h$ in week $w$) and $B$ is a baseline assuming independence of hashtag counts and time, i.e.,
\begin{equation}
B = \frac{(\text{\# tweets in week $w$})(\text{\# tweets mentioning hashtag $h$})}{\text{\# total tweets}}.
\end{equation}

We apply our \texttt{CNSS} method on this processed graph and discover one subgraph that consists of 3,294 nodes. 
By observing the hashtags, we find that it is a large subgraph connecting multiple events related to the Black Lives Matter social movement. 
One reason for this seems to be that hashtags related to certain events (such as the police-involved killings of Mike Brown in Ferguson, MO and Eric Garner in New York City)
are used by the BLM movement not just at the time those events occurred, but as a rallying cry throughout the temporal duration of the data, perhaps to emphasize that these abuses have persisted throughout time and are all tied to the same underlying phenomena of societal injustice, inequity, and discrimination.

In order to narrow the focus of our detection method to a specific event and validate our algorithm, we decrease the value of $\alpha_{\max}$ from $0.09$ to a smaller value $0.008$ empirically, in where we could view the choice of $\alpha_{\max}$ as influencing the granularity of a detected event. 
In the end, we are able to obtain a significant event with much smaller hashtag cluster as shown in Figure \ref{fig:cnss_blm}.
To be fair to the competing methods, we also attempt to shrink the granularity of detection by reducing $\alpha_{\max}$ for these methods as well. 
For \texttt{LTSS}, we shrink $\alpha_{\max}$ from $0.09$ to $0.008$, since it detects 6,024 nodes that maximize the BJ score with $\alpha_{\max}=0.09$. 
With $\alpha_{\max}=0.008$, \texttt{LTSS} detects $676$ nodes that spread over $54$ weeks of data. 
Therefore, we shrink the $\alpha_{\max}$ to $0.001$ for \texttt{LTSS} and still detect $88$ nodes that cross over $52$ weeks and do not represent any single, specific event.
For \texttt{EventTree}, the granularity is controlled by the normalization coefficient $\lambda$. 
We set $\lambda \in \{0.001, \cdots, 0.009,0.01, \cdots, 0.09\}$. 
The smallest detected subgraph or event has $130$ nodes with $\lambda=0.001$ that spread over $40$ weeks, again failing to identify a single event that is localized in time.


As shown in Figure \ref{fig:cnss_blm}, the hashtags detected by \texttt{CNSS} correspond to the widespread protests related to two closely occurring events: the grand juries' decisions not to indict the police officers responsible for the deaths of Eric Garner and Mike Brown. 
On July 17, 2014, Eric Garner died in the New York City borough of Staten Island after Daniel Pantaleo, a New York Police Department (NYPD) officer, put him in a prohibited chokehold while arresting him.
On August 9, 2014, Mike Brown, an 18-year-old Black man, was fatally shot by a white Ferguson police officer in the city of Ferguson, Missouri. 
On November 24, 2014, the St. Louis County grand jury decided not to indict the police officer, and on December 4, 2014, a Richmond County grand jury decided not to indict Pantaleo. 
These decisions stirred massive public protests and rallies in Ferguson, New York City, and Seattle in the following weeks.  As we can see in Figure \ref{fig:cnss_blm}, our detected subgraph clearly captures the emergence, the peak, and the end of this event using hashtags of tweets, while the lower volumes of hashtag mentions from  continued references after these events are not included in the detected subgraph.

In contrast, as shown in Table \ref{table:cnss_blm} and Figure \ref{fig:blm}, \texttt{EventTree} and \texttt{LTSS} detect multiple small events indicated by multiple peaks in the Figure \ref{fig:blm} across long periods. 
We annotate four peaks for \texttt{EventTree}. 
The first event is corresponding to the same event detected by \texttt{CNSS} in Figure \ref{fig:cnss_blm}. 
The second peak is around the Martin Luther King Jr. Day on January 19, 2015. 
The third event is about the death of Freddie Gray in Baltimore on April 19, 2015. 
The fourth event is about the death of Sandra Bland in Texas on July 13, 2015.
We also annotate five peaks for \texttt{LTSS}. 
The first event is corresponding to the same event detected by \texttt{CNSS} in Figure \ref{fig:cnss_blm}. 
The second peak contains multiple unrelated hashtags, such as \textit{\#Nigeria}, \textit{\#Grammys}, \textit{\#Oscars}, and \textit{\#ReclaimMLK}. 
The third peak corresponds to the shooting of Tony Robinson on March 6, 2015, while the fourth and fifth peaks correspond to the deaths of Freddie Gray and Sandra Bland respectively.

\begin{table*}[!ht]
\caption{BlackLivesMatter: Statistics of Detected Subgraphs by Different Methods}\label{table:cnss_blm}
\begin{adjustbox}{center}
\centering
\resizebox{0.8\textwidth}{!}{
\begin{tabular}{c|c|c|c}
\toprule
 & EventTree & LTSS & CNSS \\ 
\hline
periods &11/21/2014-08/20/2015 &08/29/2014-08/13/2015& 11/14/2014-12/25/2014\\
\hline
\# of weeks &40 &52 &6  \\ 
\hline
\# of nodes &130 &88 &25 \\
\hline
\# of hashtags &51 &73 &12  \\
\bottomrule
\end{tabular}
}
\end{adjustbox}
\end{table*}

\begin{figure*}
    \includegraphics[width=0.33\linewidth]{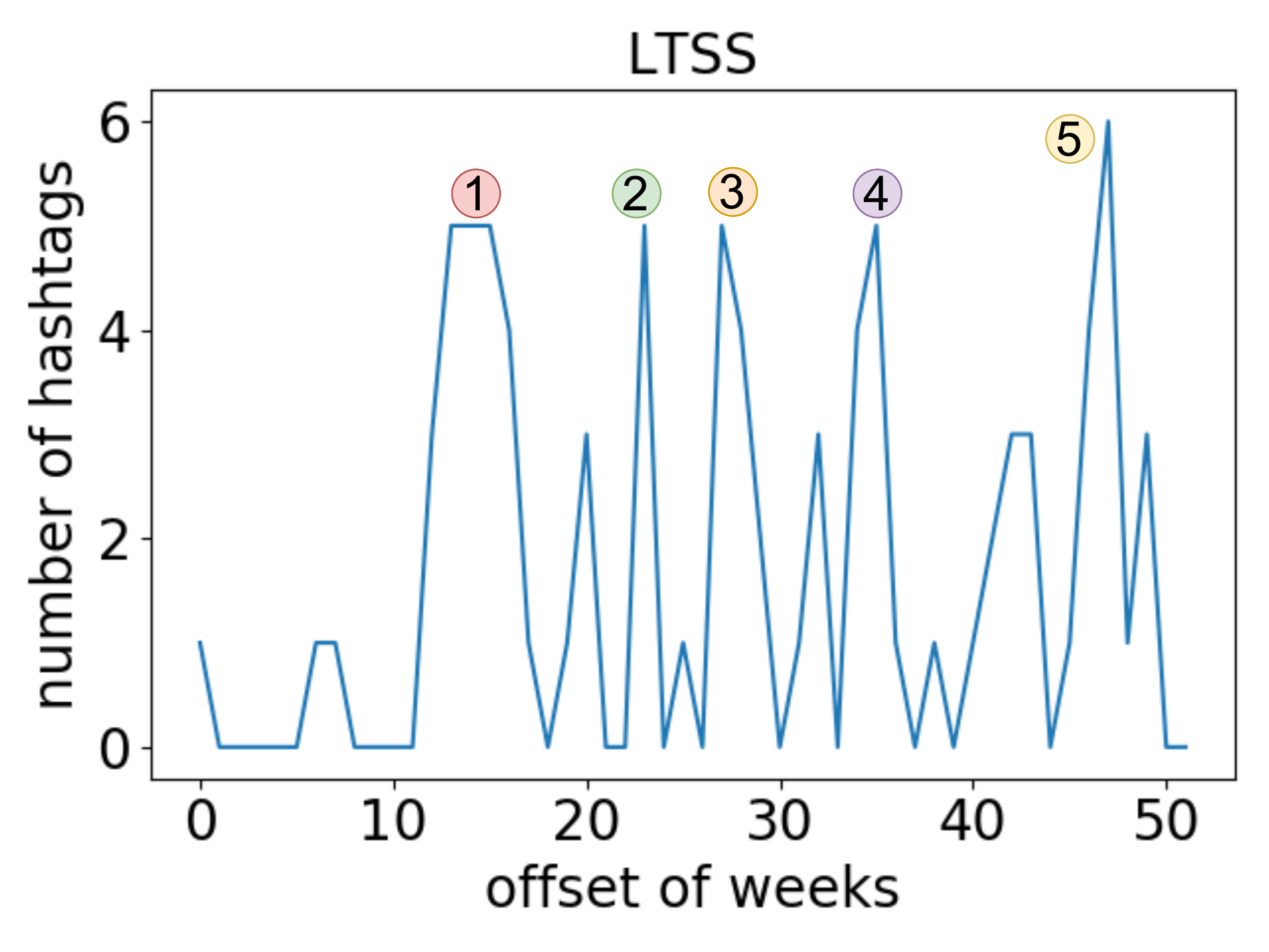}
    \includegraphics[width=0.33\linewidth]{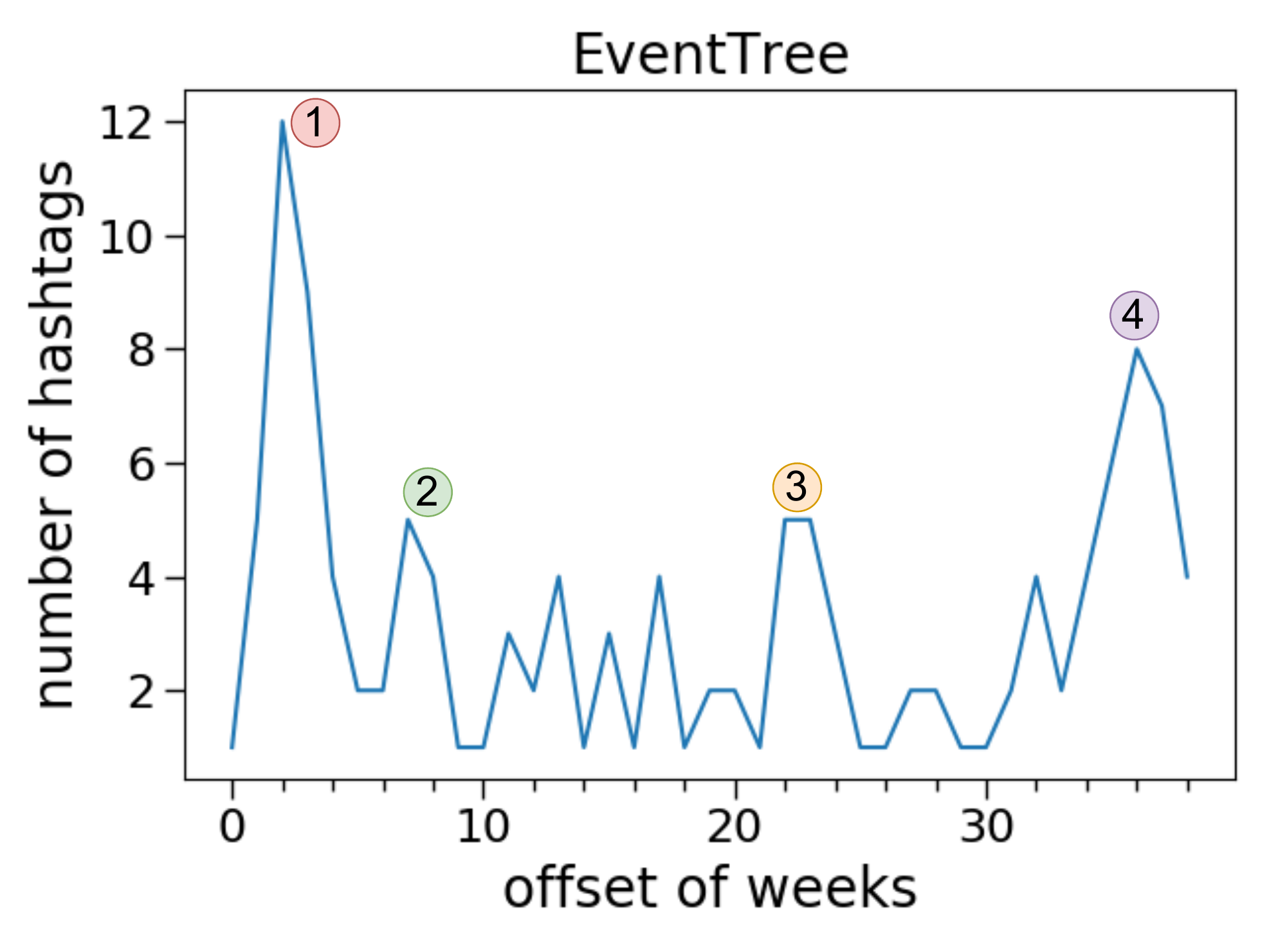}
    \includegraphics[width=0.33\linewidth]{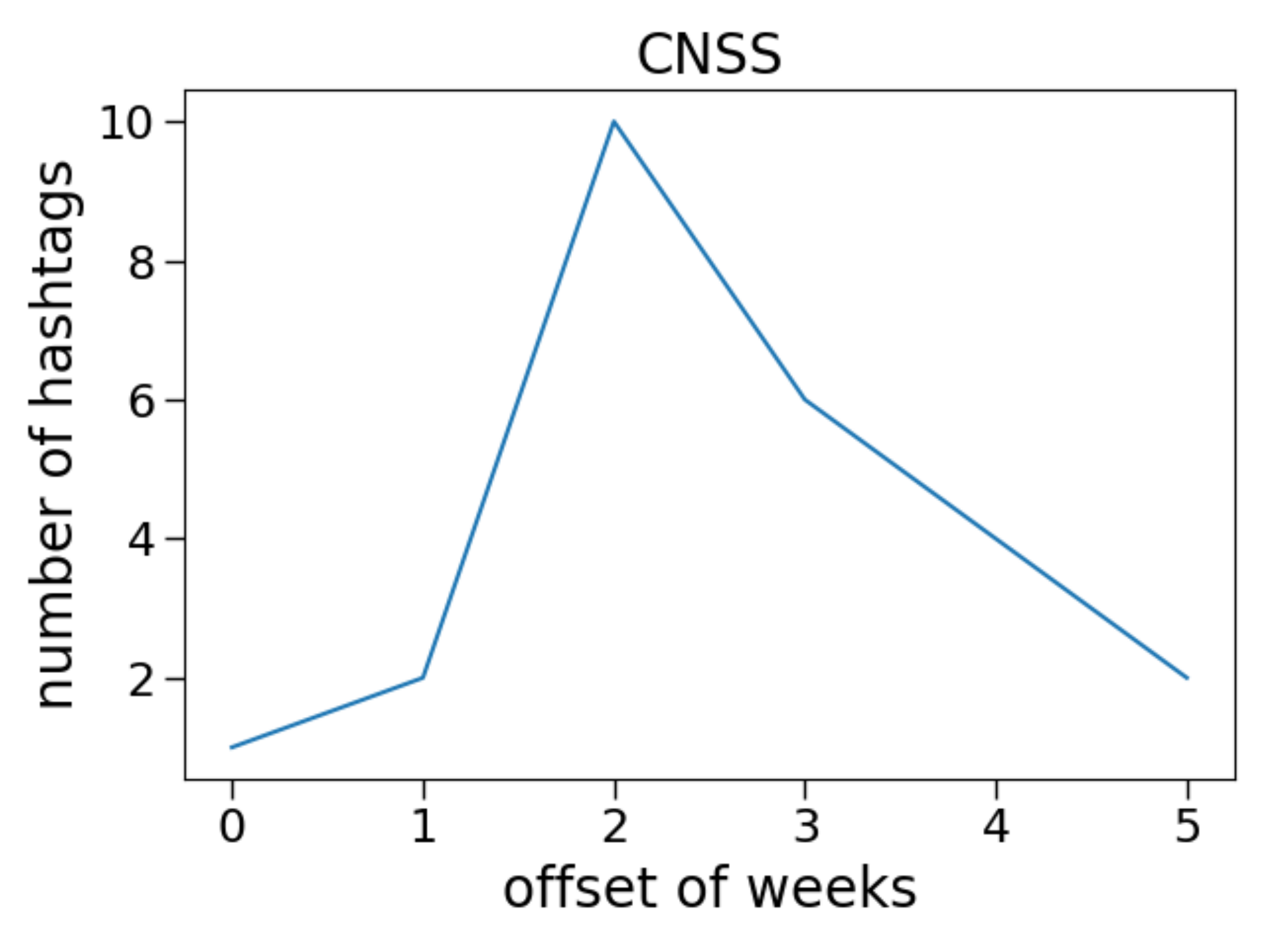}
\caption{Detected Hashtags Distributions Over Time for \texttt{LTSS}, \texttt{EventTree}, and \texttt{CNSS}.}
\label{fig:blm}
\end{figure*}

\begin{table*}[!ht]
\caption{COVID-19 Case Study: Statistics of Top-3 \texttt{CNSS} Detected Subgraphs. The calibrated BJ scores of these 3 subgraphs are higher than all $100$ calibrated BJ scores under $\mathcal{H}_0$.}\label{table:cnss_covid19}
\centering
\resizebox{0.6\textwidth}{!}{
\begin{tabular}{c|c|c|c|c|c}
\toprule
 &$N$ &$\alpha$ &$N_{\alpha}$ &$\alpha'$ & Calibrated BJ-Score\\ 
 \hline
1st Subgraph & 4707 & 0.09 & 4702 & 0.843 & 774.249 \\
\hline
2nd Subgraph & 910 & 0.09 & 898 & 0.897 & 61.187\\
\hline
3rd Subgraph & 100 & 0.03 & 100 & 0.708 & 34.595\\
\bottomrule
\end{tabular}
}
\end{table*}

These results demonstrate that, even when adjusting the significance threshold $\alpha$ for finer event granularity, the competing methods identify multiple individually anomalous hashtag-week combinations and fail to detect a coherent subgraph corresponding to a single event of interest.

\subsection{COVID-19 Case Study}
\label{app:covid19_case}

We now show visualizations of the top-1 detected subgraph for \texttt{CNSS} and two competing methods (\texttt{LTSS} and \texttt{EventTree}) in Figures \ref{fig:covid19_cnss_detected_subgraph}, \ref{fig:covid19_ltss_detected_subgraph}, and \ref{fig:covid19_eventtree_detected_subgraph}, respectively. We note that the overall spatial-temporal subgraphs identified by \texttt{CNSS} and \texttt{EventTree} are connected, though the set of spatial locations for any given time slice may not be connected. LTSS does not enforce any connectivity constraints.

The connected subgraph identified by \texttt{CNSS} clearly demonstrates the initial progression of the COVID-19 outbreak across the eastern United States between March-June 2020, with initial peaks in New York City and the northeastern U.S.~that gradually spread into the southeastern U.S. and Texas.  In contrast, the connected subgraph identified by \texttt{EventTree} and the subset identified by \texttt{LTSS} are dispersed over the entire country and do not clearly show the progression of the outbreak's peak.  

We used the following county adjacency data to build the spatial temporal network for the COVID-19 case study: 
\url{https://www.census.gov/geographies/reference-files/2010/geo/county-adjacency.html}, and statistics of the top-3 subgraphs detected by \texttt{CNSS} are presented in Table~\ref{table:cnss_covid19}.

\clearpage

\begin{figure*}[!ht]
      \begin{subfigure}{0.26\textwidth}
          \centering
         \includegraphics[width=1.\textwidth]{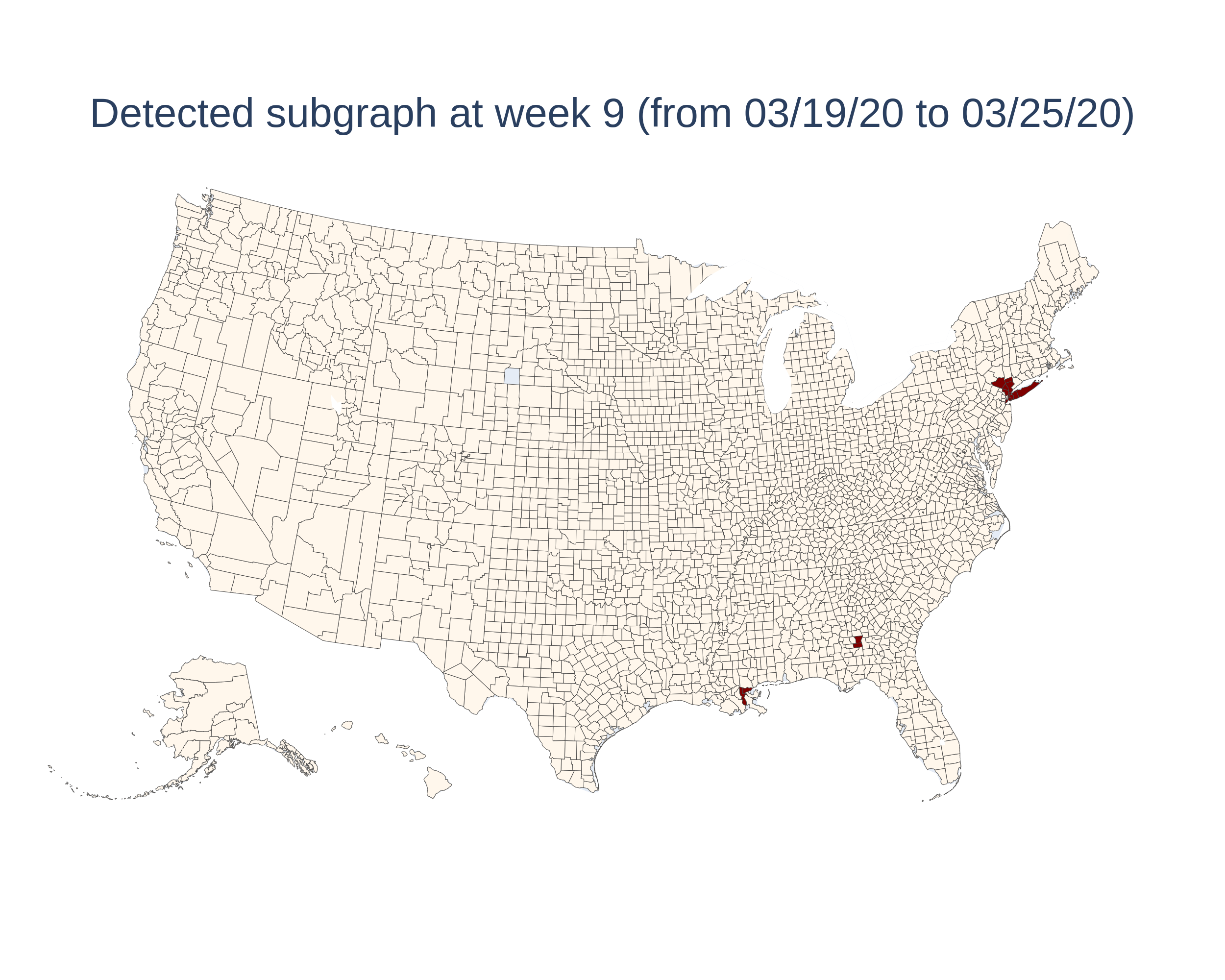}
      \end{subfigure}
      \begin{subfigure}{0.26\textwidth}
          \centering
          \includegraphics[width=1.\textwidth]{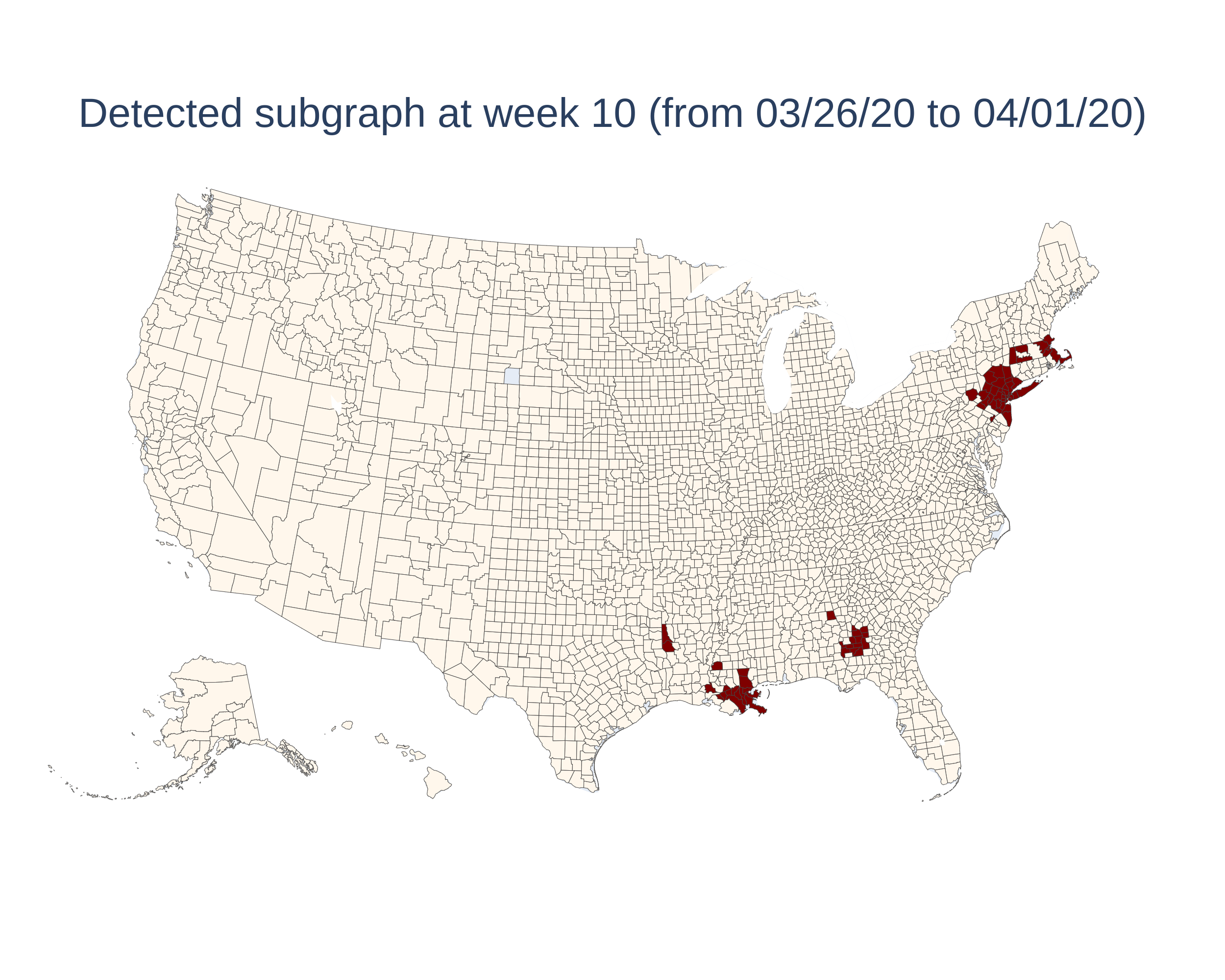}
      \end{subfigure}
      \begin{subfigure}{0.26\textwidth}
          \centering
          \includegraphics[width=1.\textwidth]{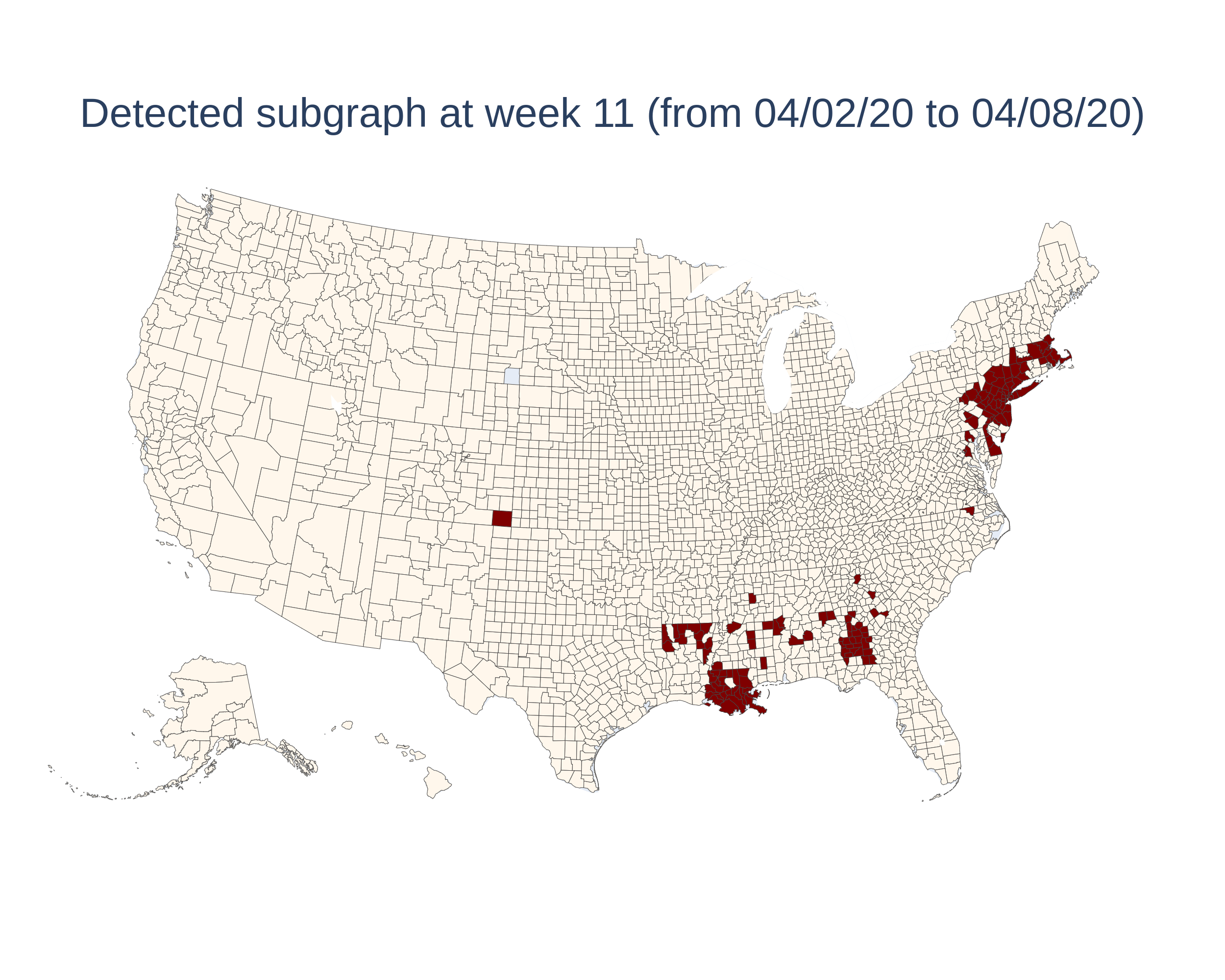}
     \end{subfigure}
      \newline
      \begin{subfigure}{0.26\textwidth}
          \centering
          \includegraphics[width=1.\textwidth]{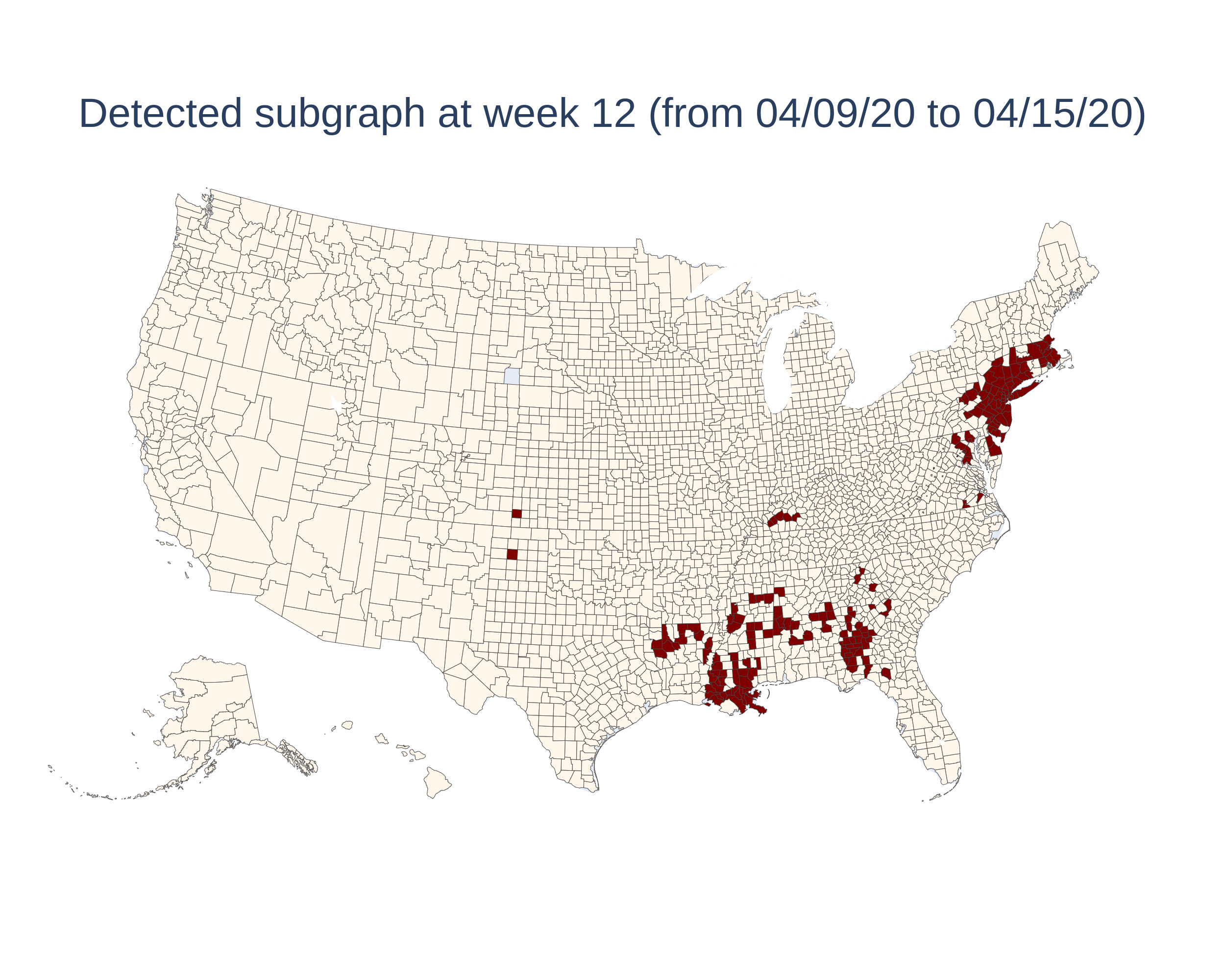}
     \end{subfigure}
     \begin{subfigure}{0.26\textwidth}
          \centering
          \includegraphics[width=1.\textwidth]{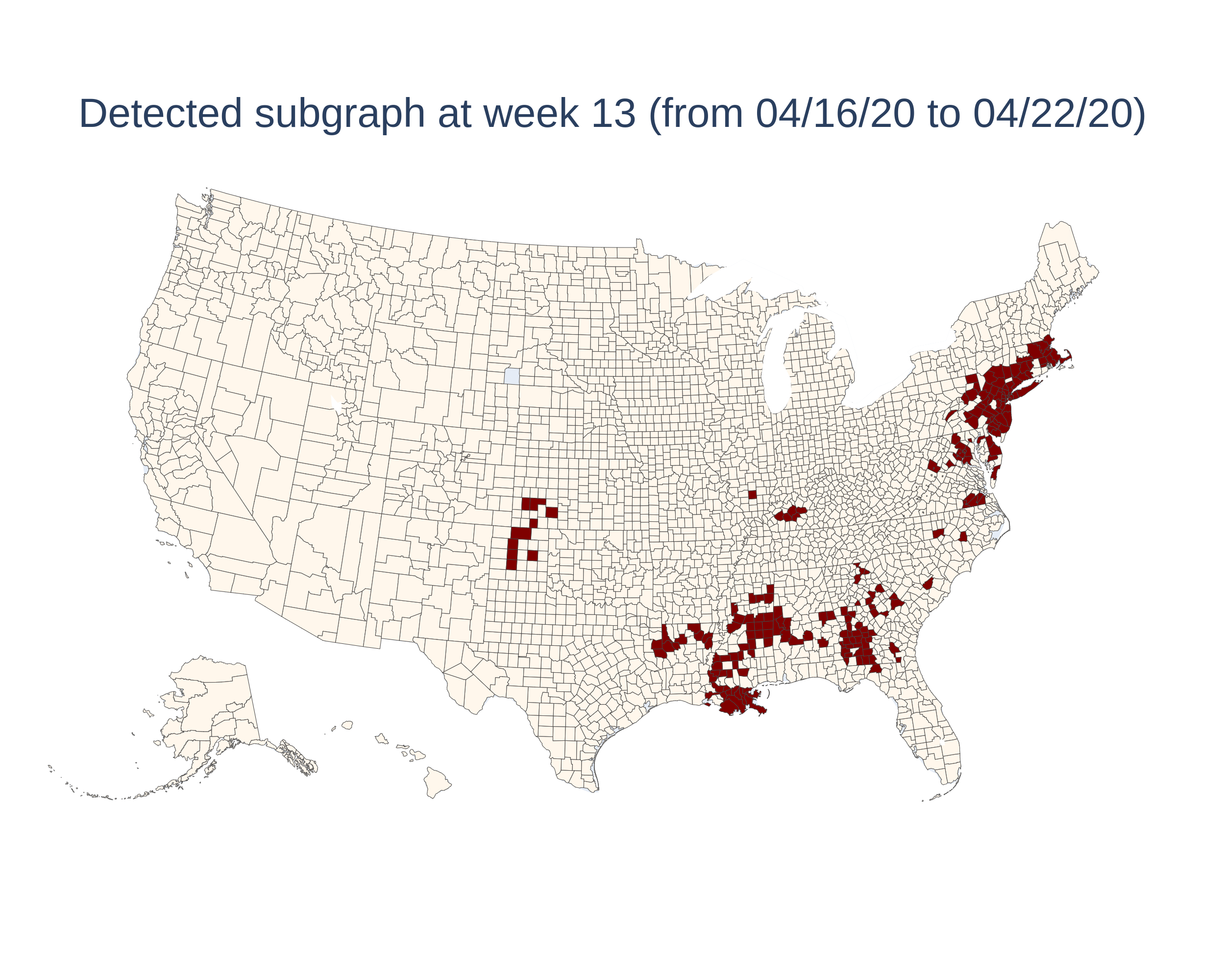}
      \end{subfigure}
      \begin{subfigure}{0.26\textwidth}
          \centering
          \includegraphics[width=1.\textwidth]{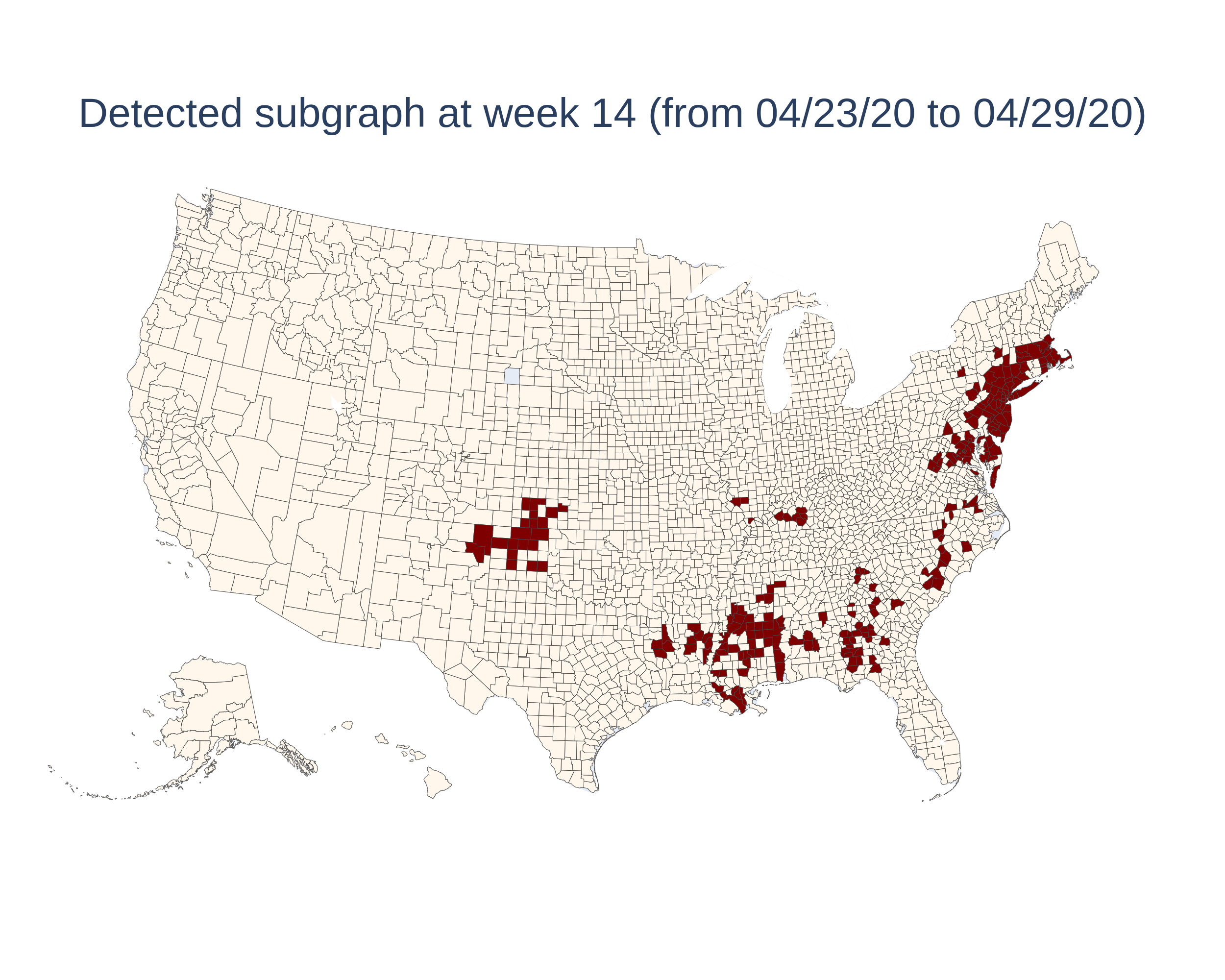}
     \end{subfigure}
      \newline
      \begin{subfigure}{0.26\textwidth}
          \centering
          \includegraphics[width=1.\textwidth]{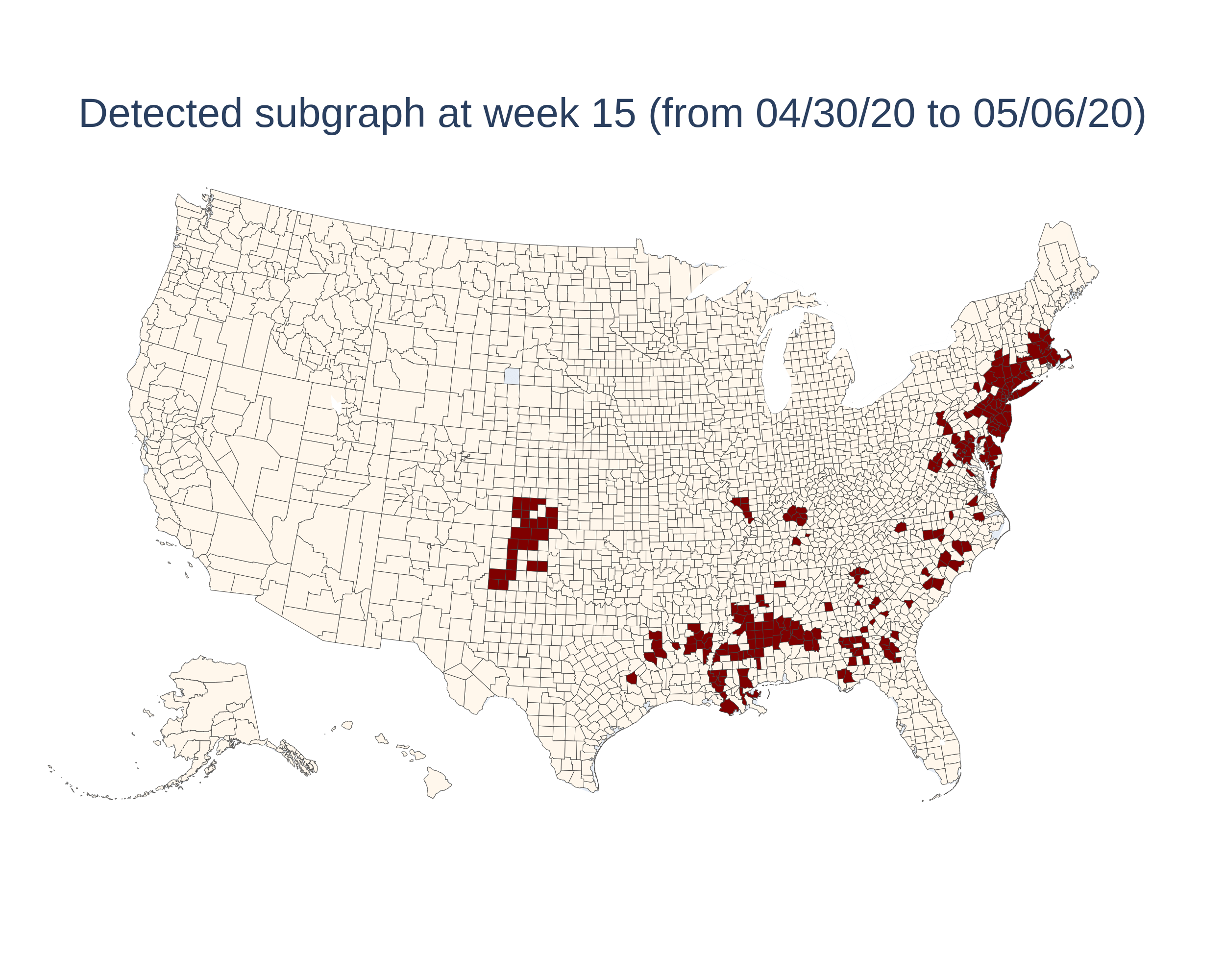}
      \end{subfigure}
      \begin{subfigure}{0.26\textwidth}
          \centering
          \includegraphics[width=1.\textwidth]{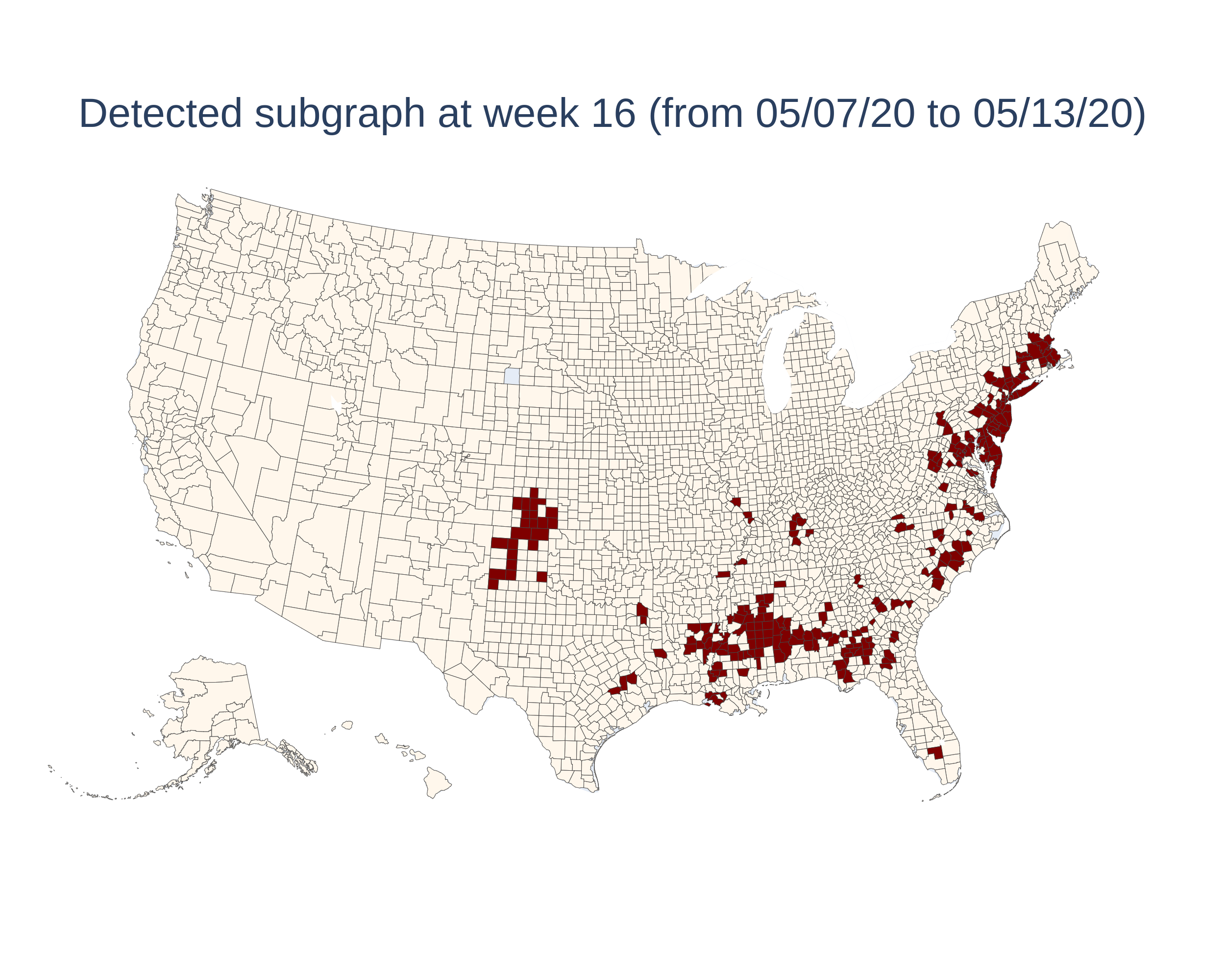}
      \end{subfigure}
      \begin{subfigure}{0.26\textwidth}
          \centering
          \includegraphics[width=1.\textwidth]{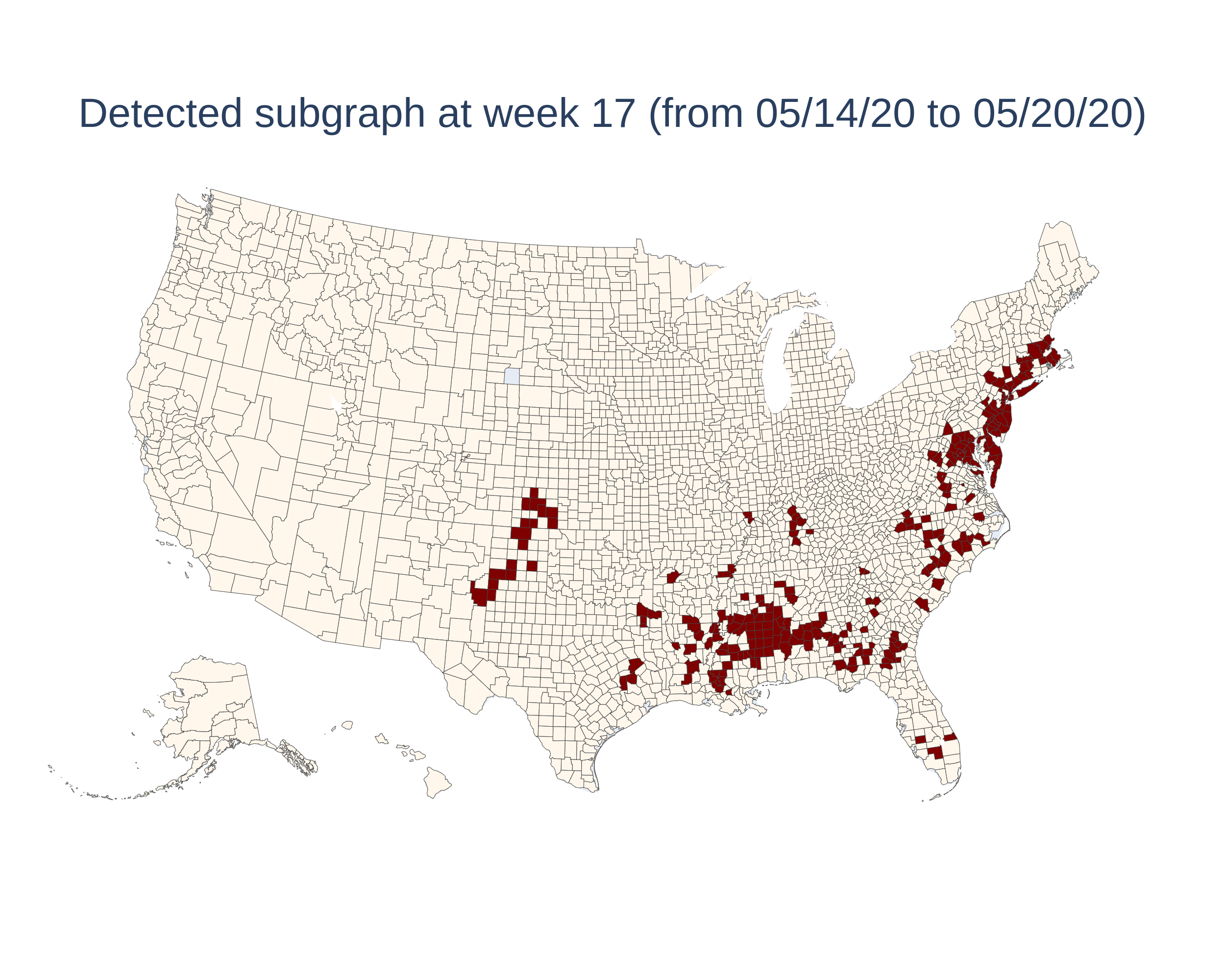}
      \end{subfigure}
      \newline
      \begin{subfigure}{0.26\textwidth}
          \centering
          \includegraphics[width=1.\textwidth]{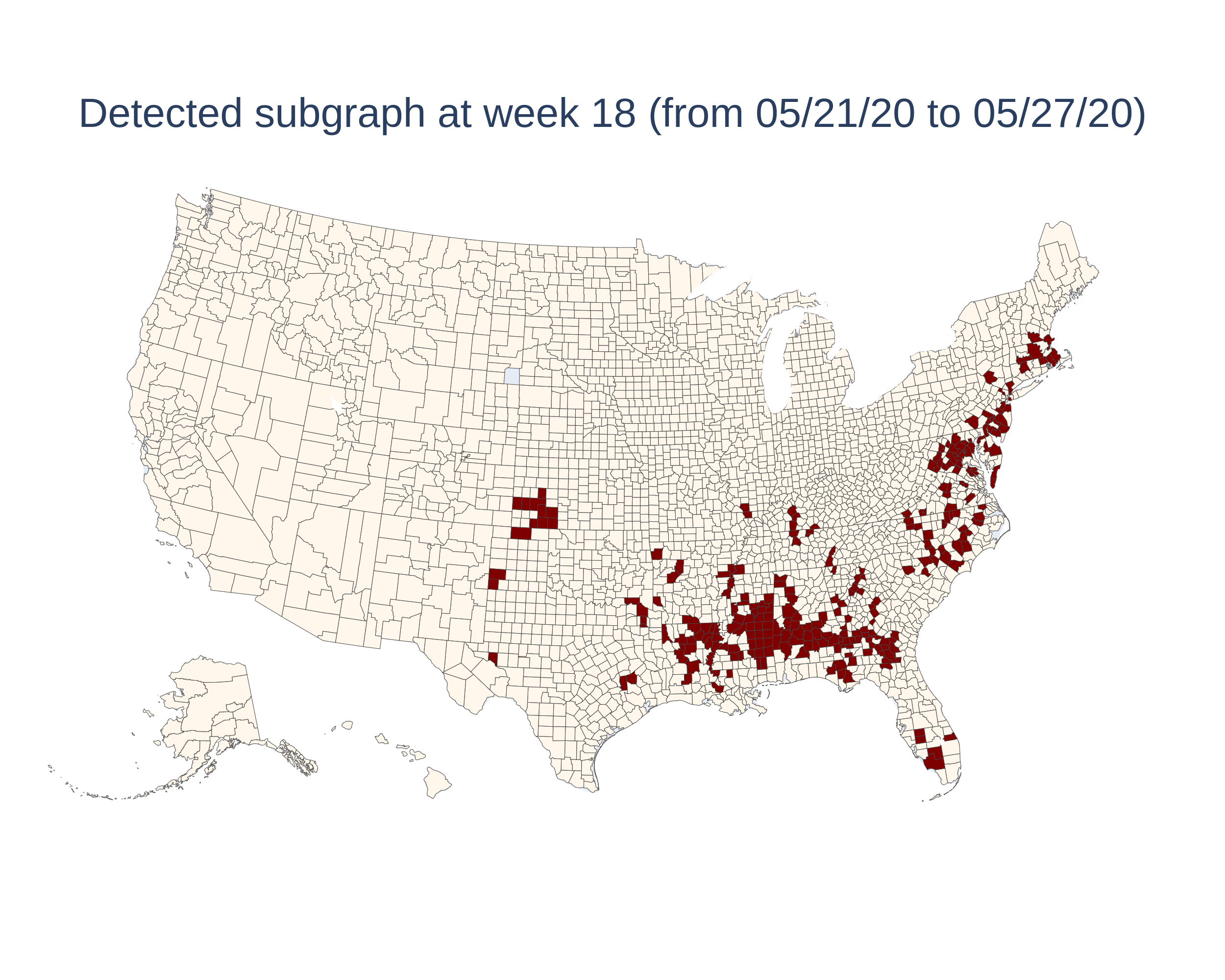}
      \end{subfigure}
      \begin{subfigure}{0.26\textwidth}
          \centering
          \includegraphics[width=1.\textwidth]{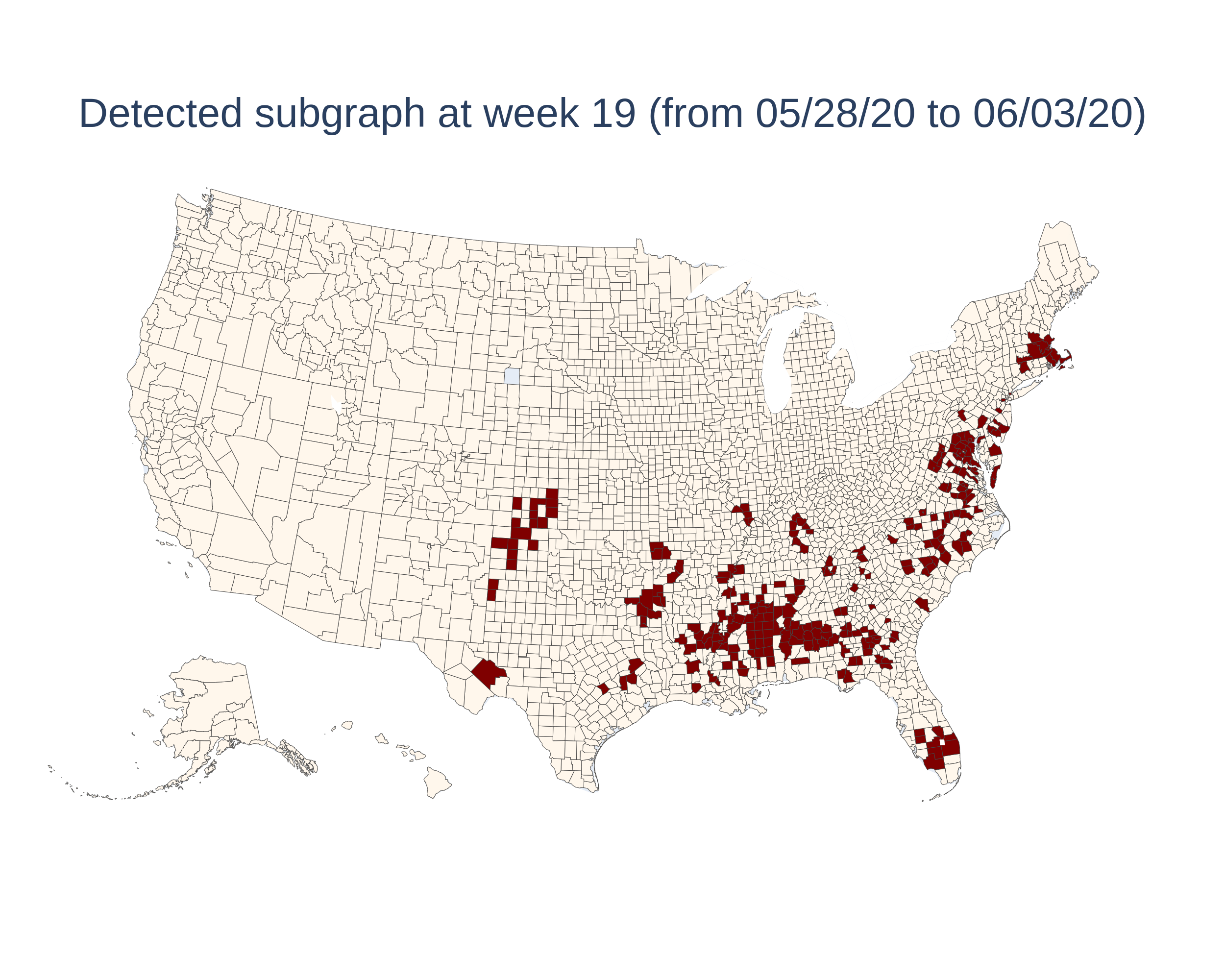}
      \end{subfigure}
      \begin{subfigure}{0.26\textwidth}
         \centering
          \includegraphics[width=1.\textwidth]{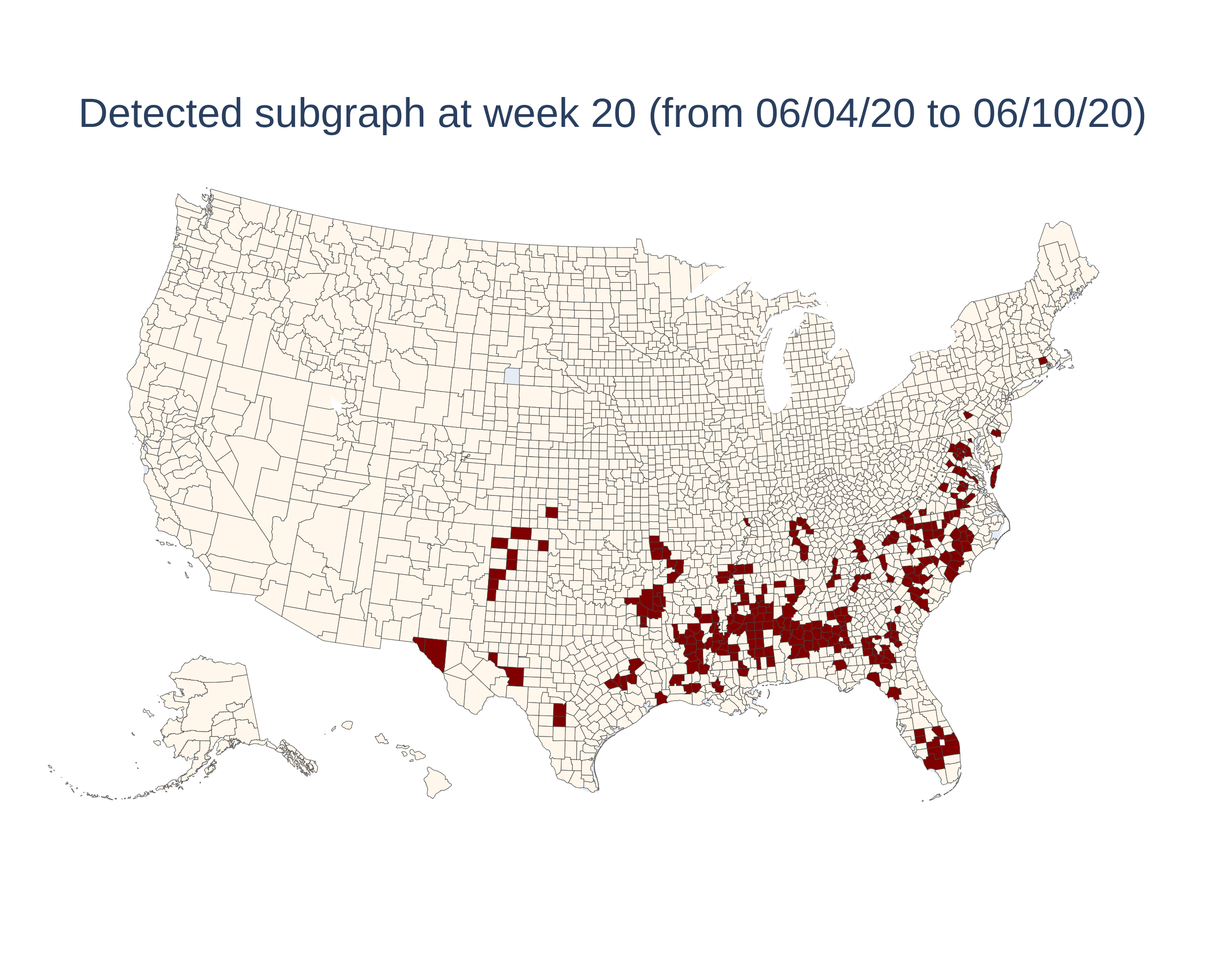}
     \end{subfigure}
     \newline
    
      \begin{subfigure}{0.26\textwidth}
          \centering
          \includegraphics[width=1.\textwidth]{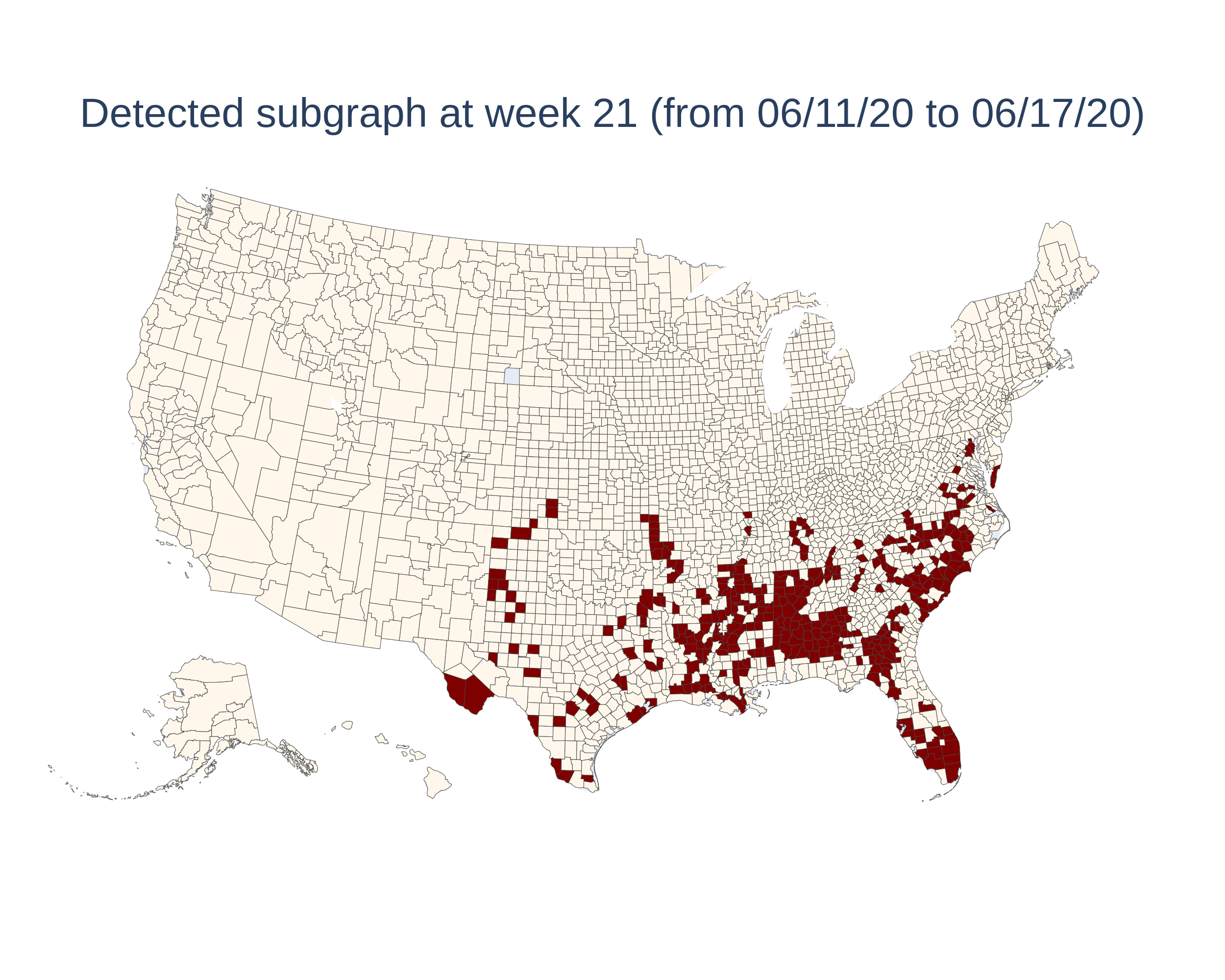}
     \end{subfigure}
      \begin{subfigure}{0.26\textwidth}
          \centering
          \includegraphics[width=1.\textwidth]{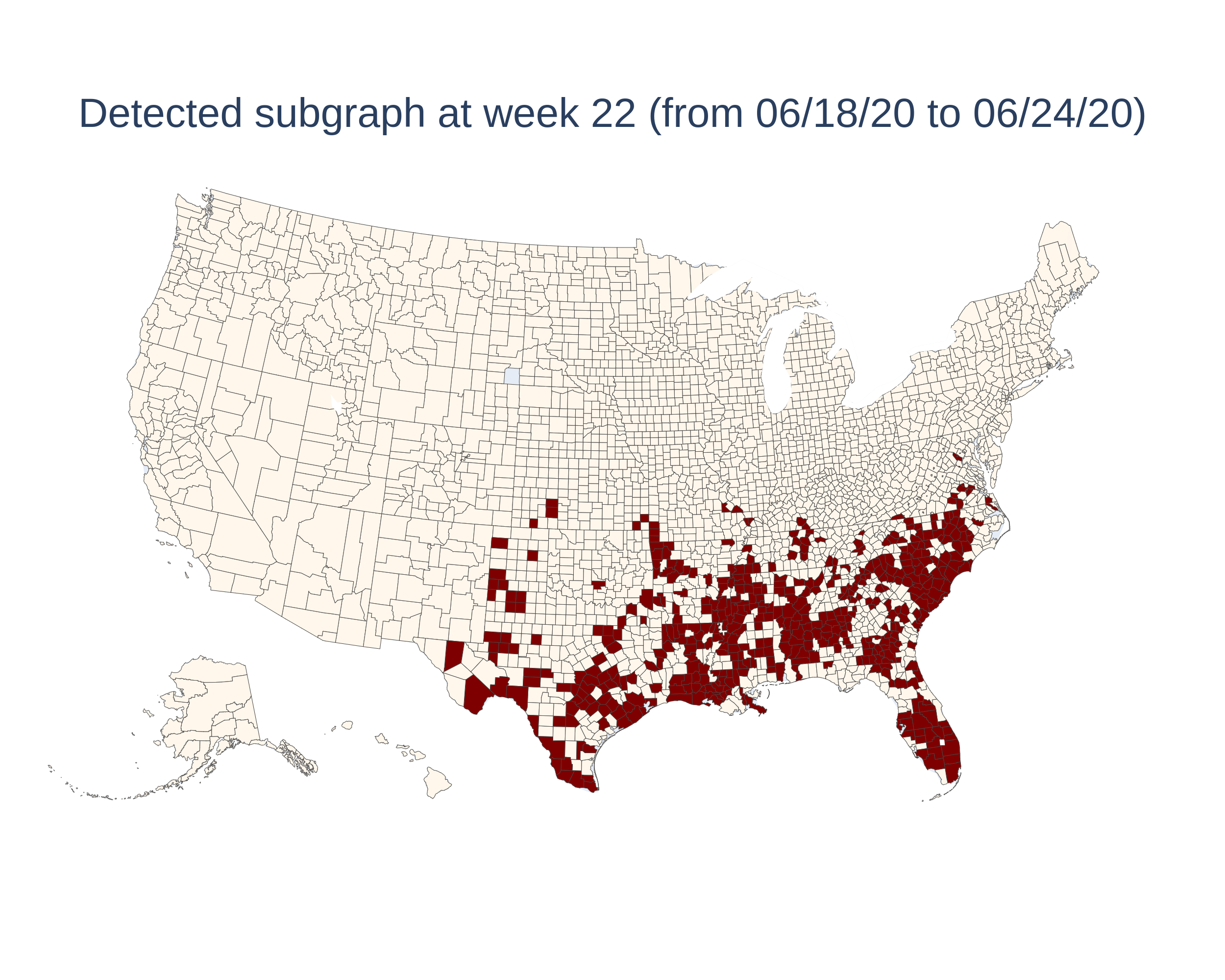}
     \end{subfigure}
      \begin{subfigure}{0.26\textwidth}
         \centering
          \includegraphics[width=1.\textwidth]{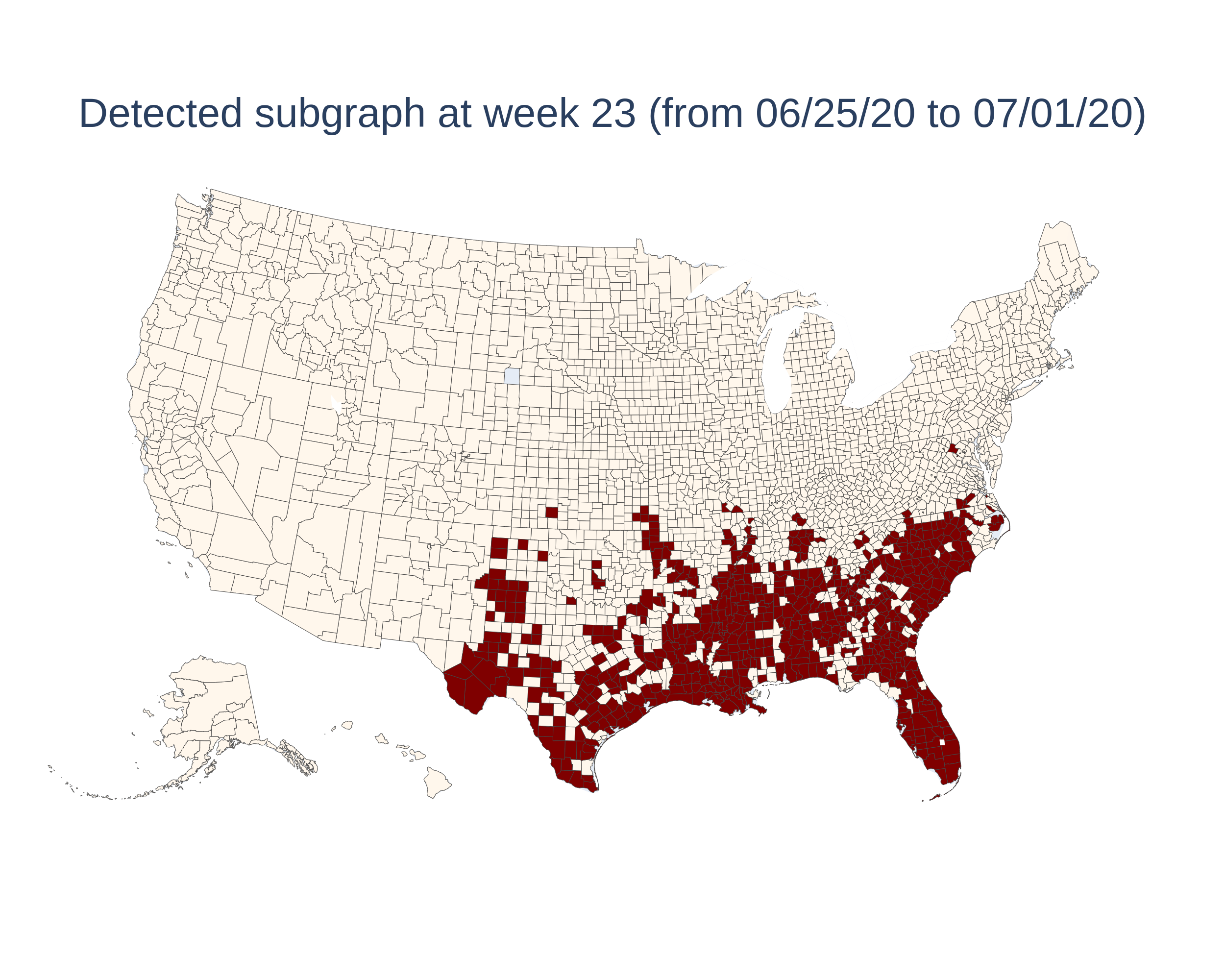}
     \end{subfigure}
      \newline
      \begin{subfigure}{0.26\textwidth}
         \centering
          \includegraphics[width=1.\textwidth]{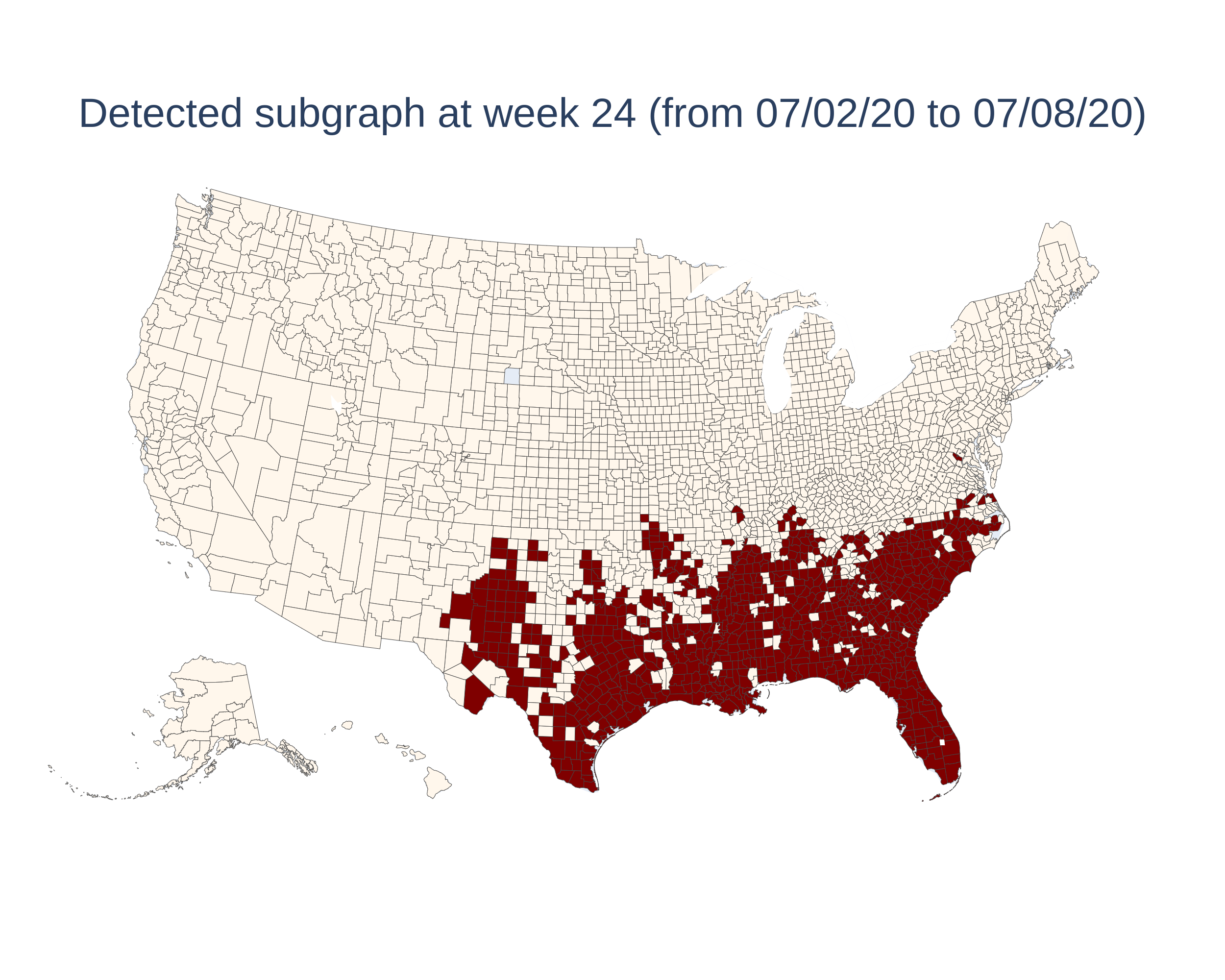}
      \end{subfigure}
      
 \caption{\texttt{CNSS} Top-1 Detected Spatial-Temporal Connected Subgraph on \texttt{COVID-19} Dataset}
\label{fig:covid19_cnss_detected_subgraph}
 \end{figure*}

\clearpage
 \begin{figure*}[!ht]
      \begin{subfigure}{0.26\textwidth}
          \centering
         \includegraphics[width=1.\textwidth]{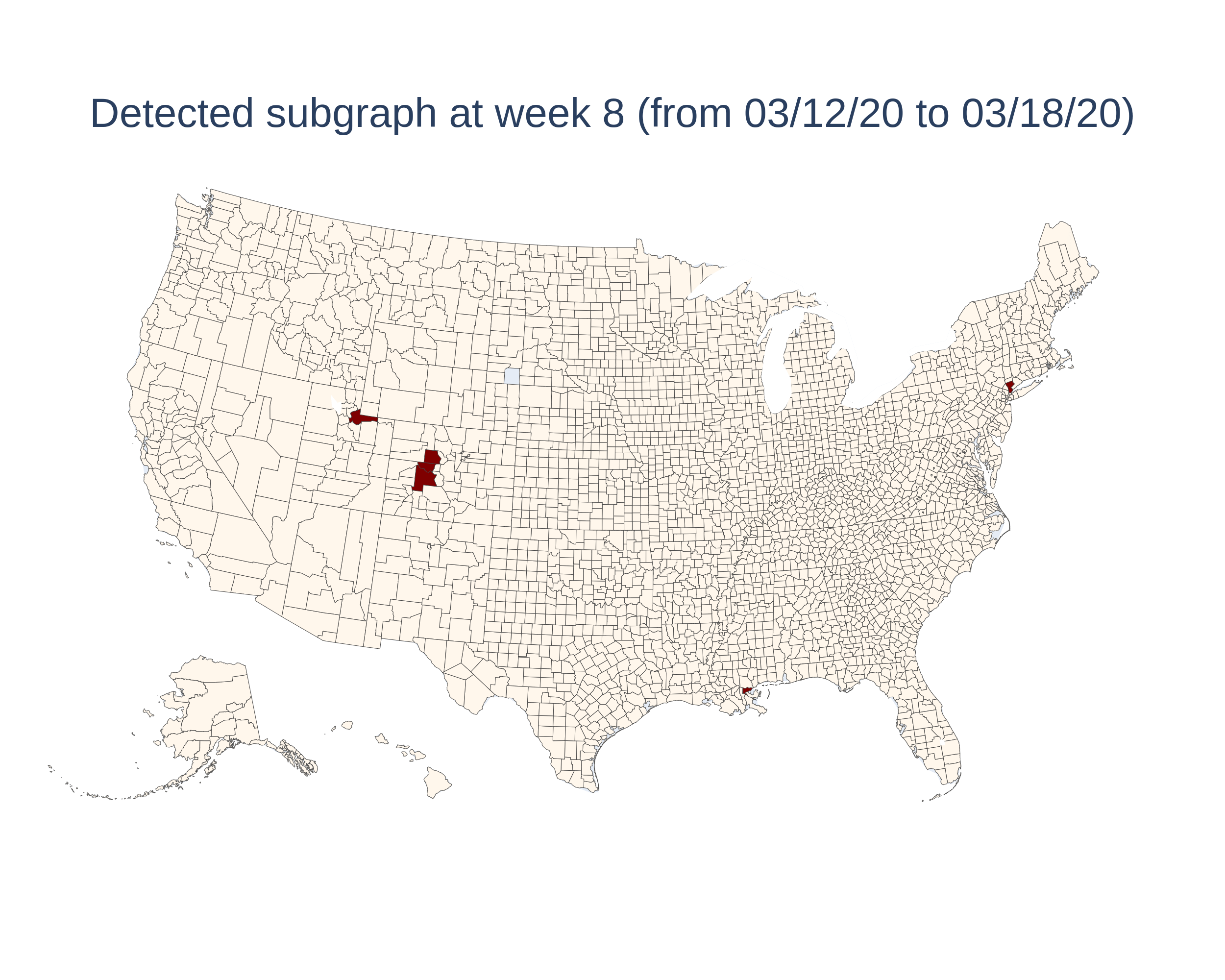}
      \end{subfigure}
      \begin{subfigure}{0.26\textwidth}
          \centering
         \includegraphics[width=1.\textwidth]{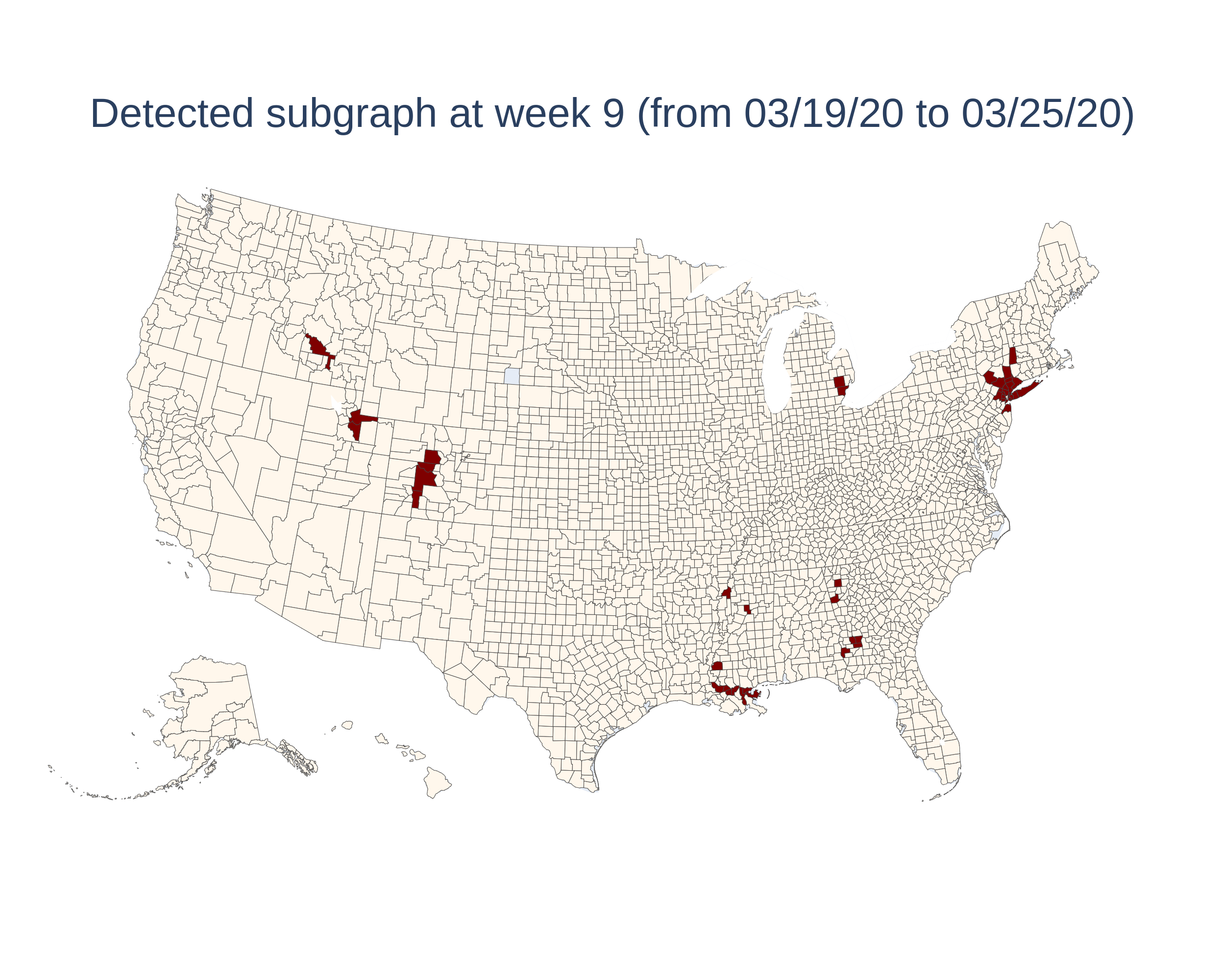}
      \end{subfigure}
      \begin{subfigure}{0.26\textwidth}
          \centering
         \includegraphics[width=1.\textwidth]{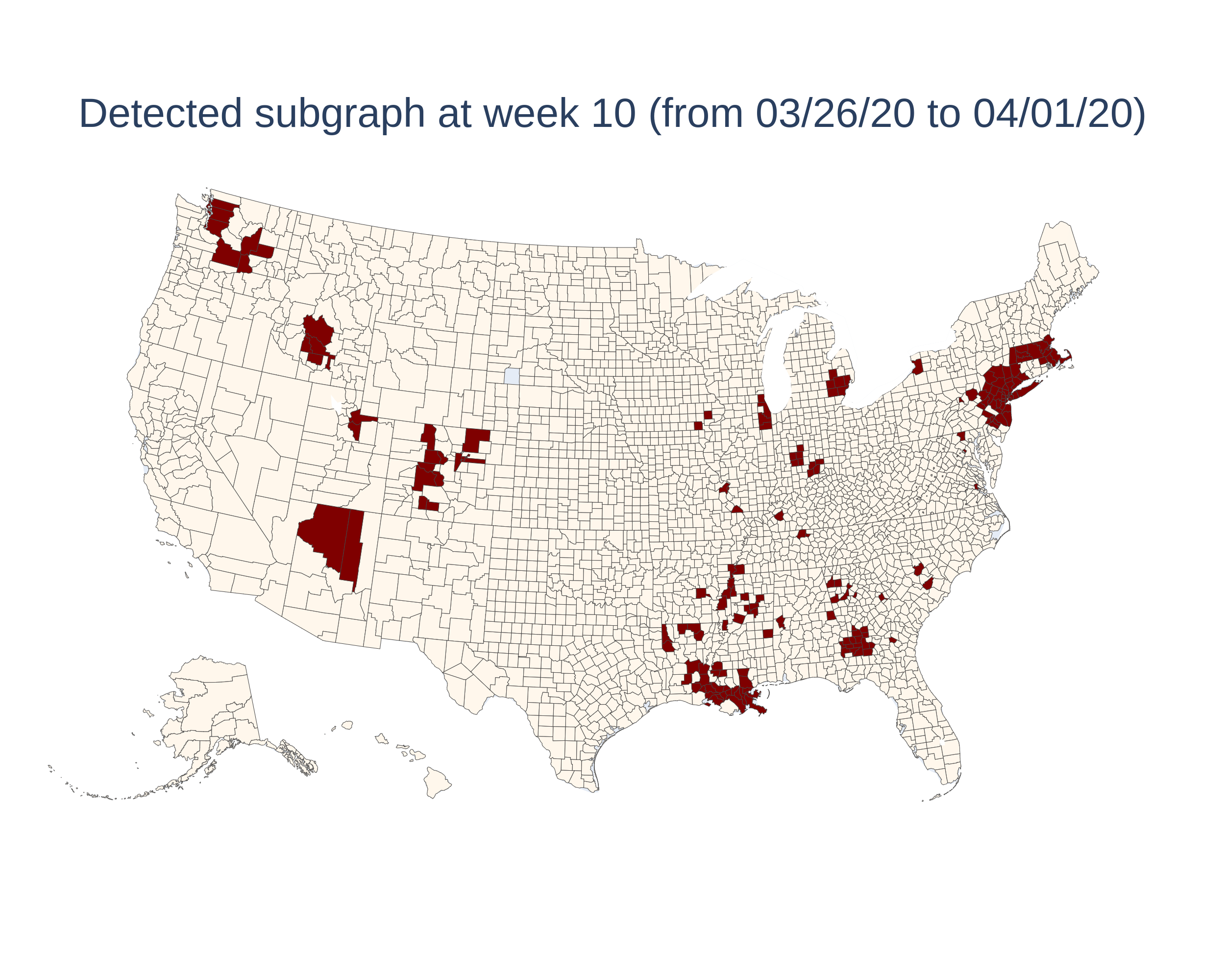}
     \end{subfigure}
      \newline
    
     \begin{subfigure}{0.26\textwidth}
          \centering
          \includegraphics[width=1.\textwidth]{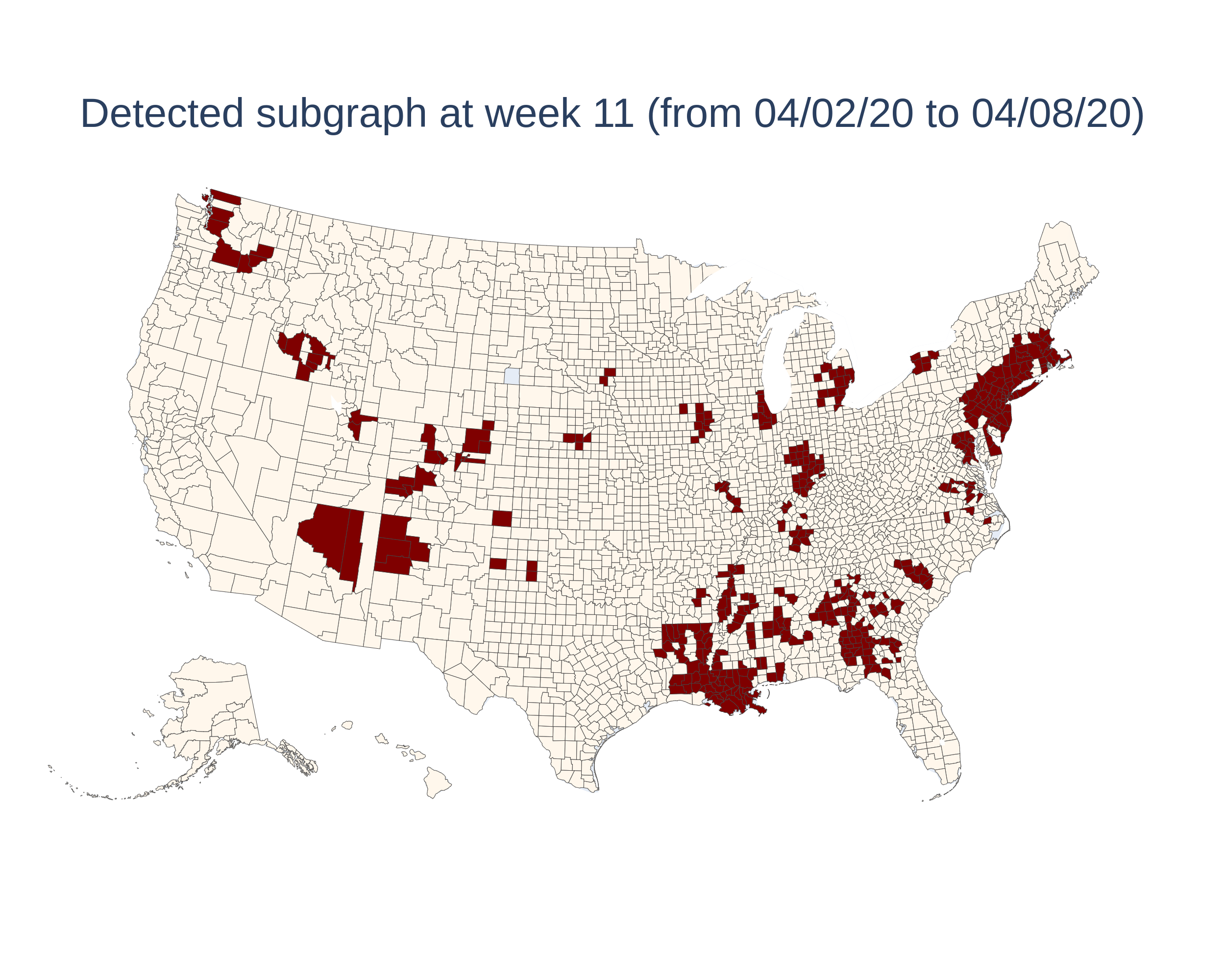}
      \end{subfigure}
      \begin{subfigure}{0.26\textwidth}
          \centering
         \includegraphics[width=1.\textwidth]{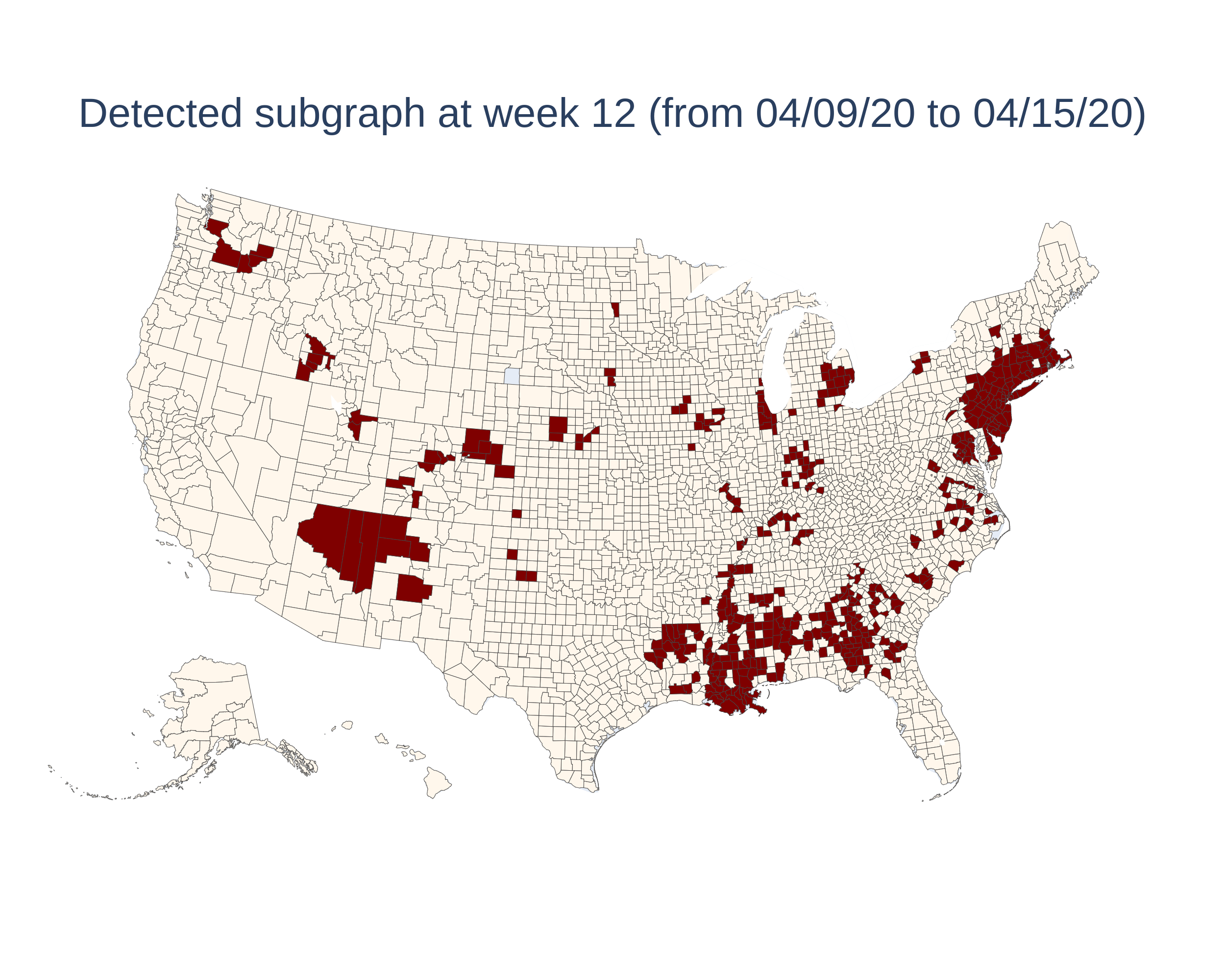}
      \end{subfigure}
      \begin{subfigure}{0.26\textwidth}
         \centering
         \includegraphics[width=1.\textwidth]{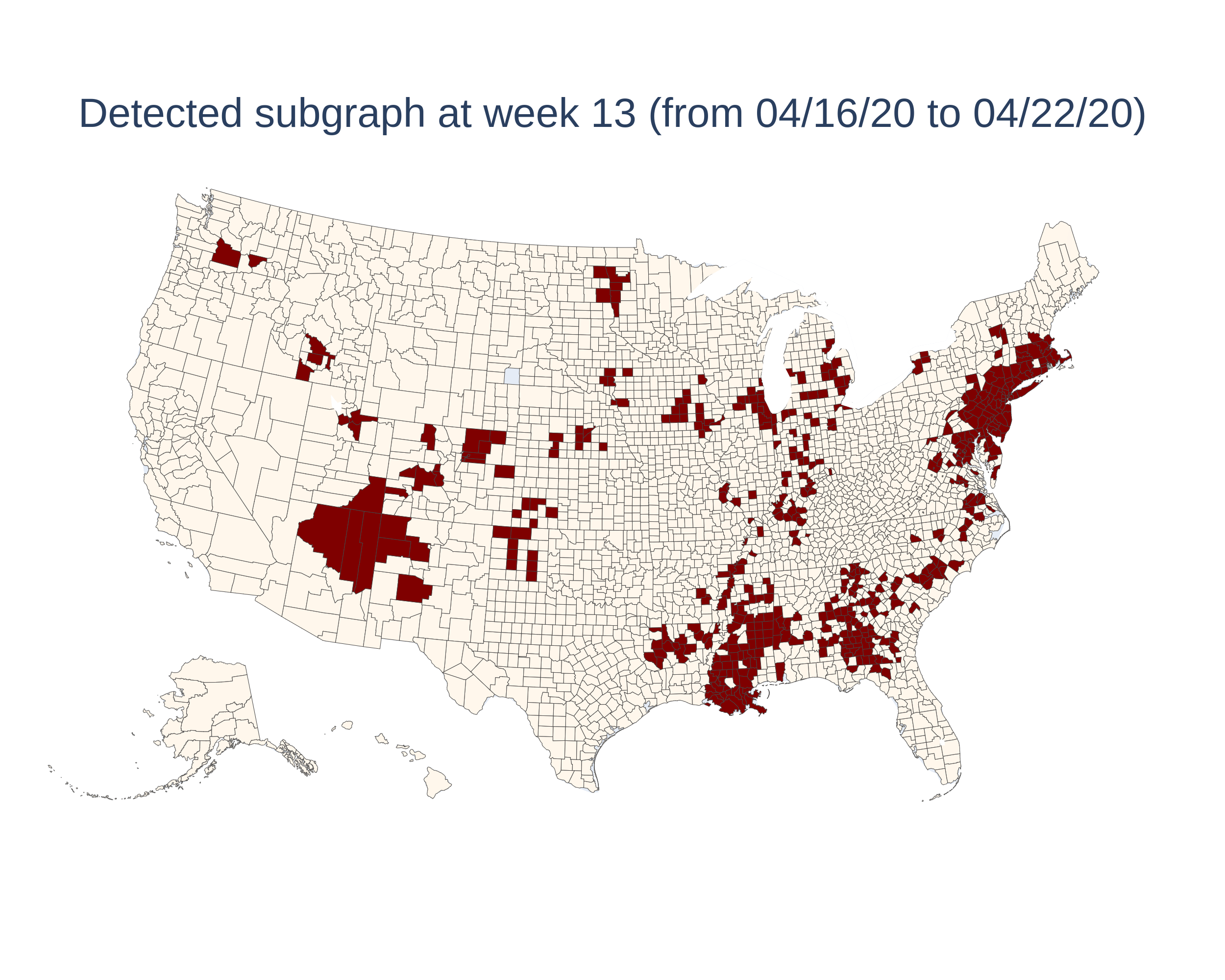}
      \end{subfigure}
      \newline
    
      \begin{subfigure}{0.26\textwidth}
          \centering
          \includegraphics[width=1.\textwidth]{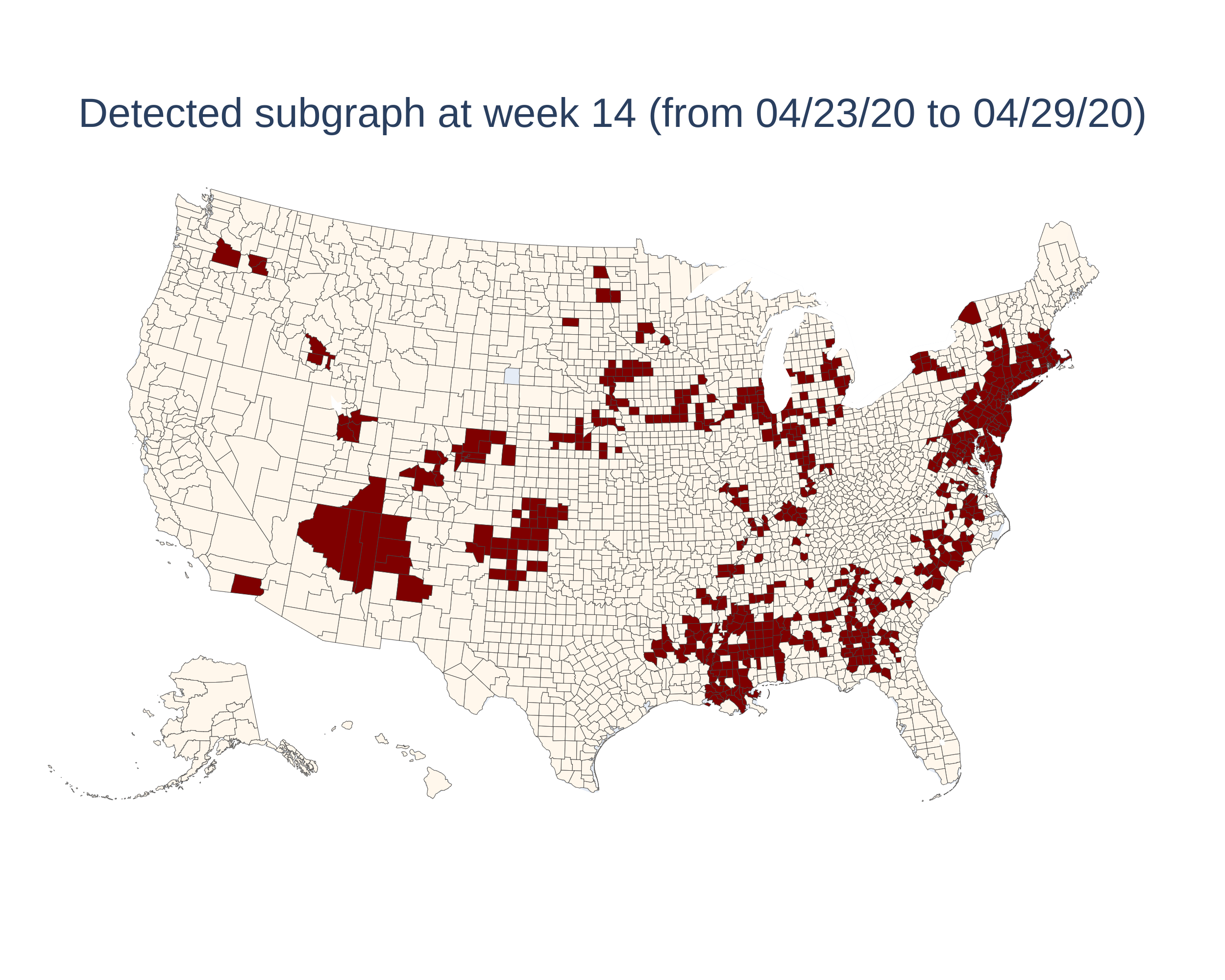}
      \end{subfigure}
      \begin{subfigure}{0.26\textwidth}
         \centering
         \includegraphics[width=1.\textwidth]{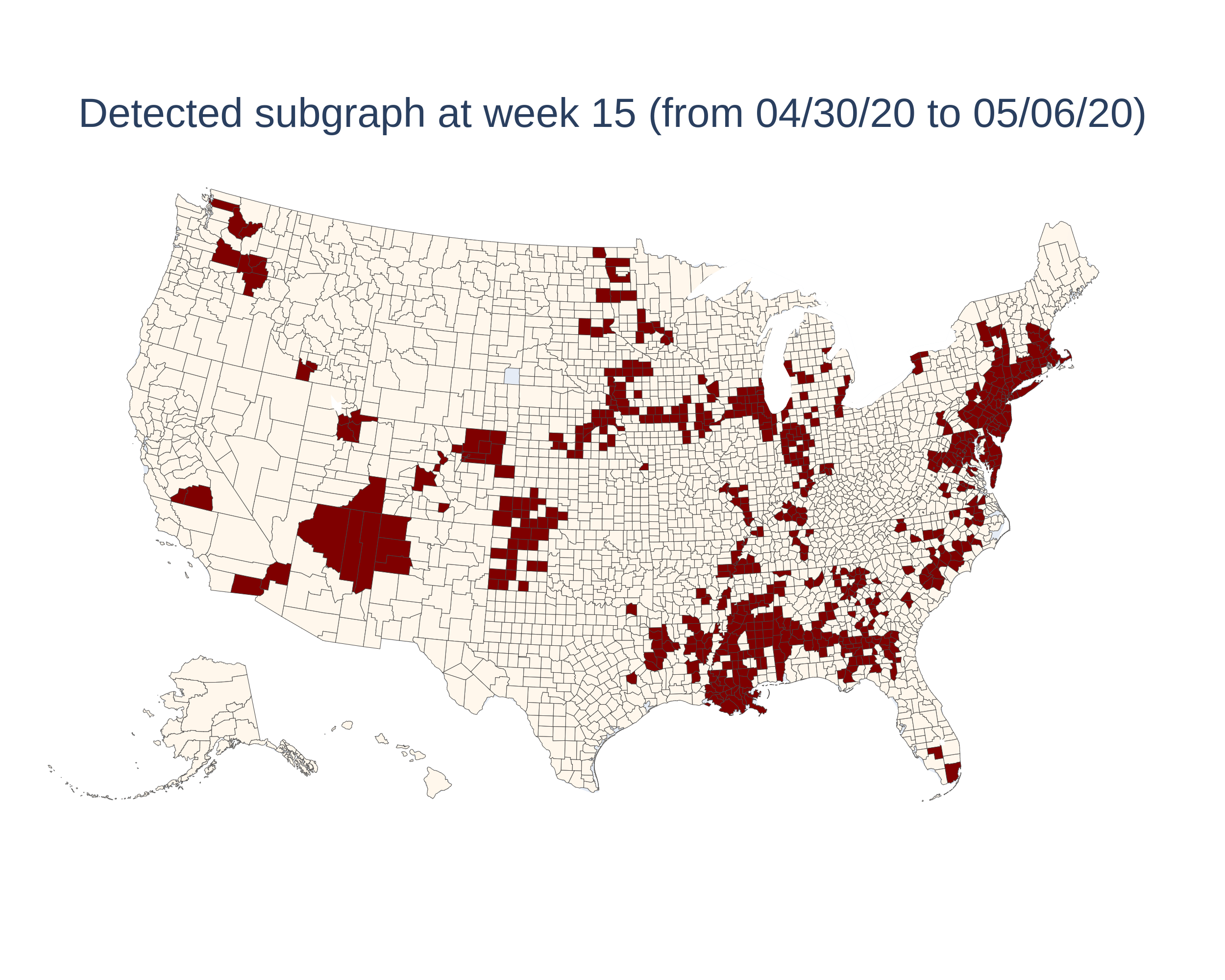}
      \end{subfigure}
      \begin{subfigure}{0.26\textwidth}
          \centering
          \includegraphics[width=1.\textwidth]{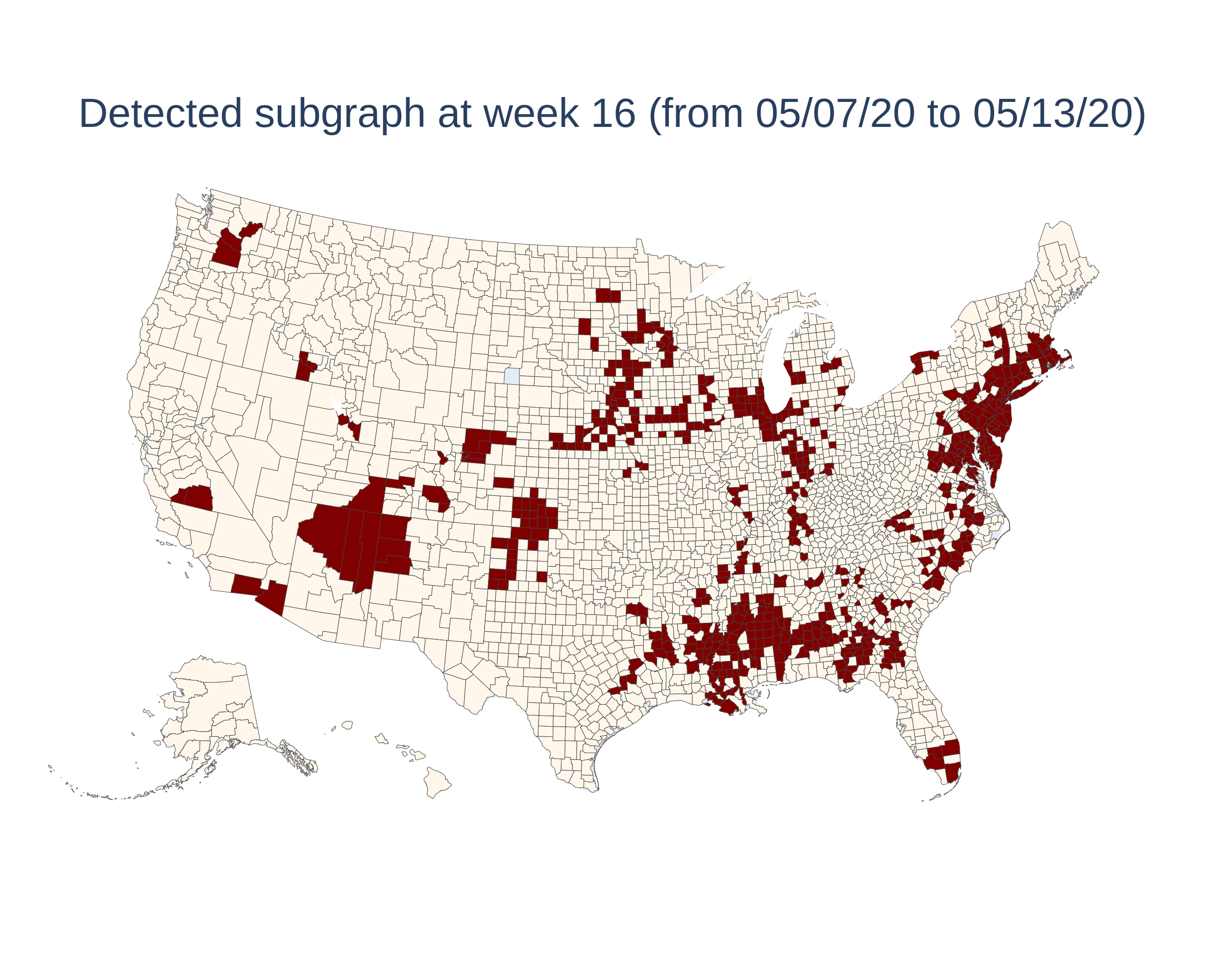}
      \end{subfigure}
     \newline
    
      \begin{subfigure}{0.26\textwidth}
         \centering
          \includegraphics[width=1.\textwidth]{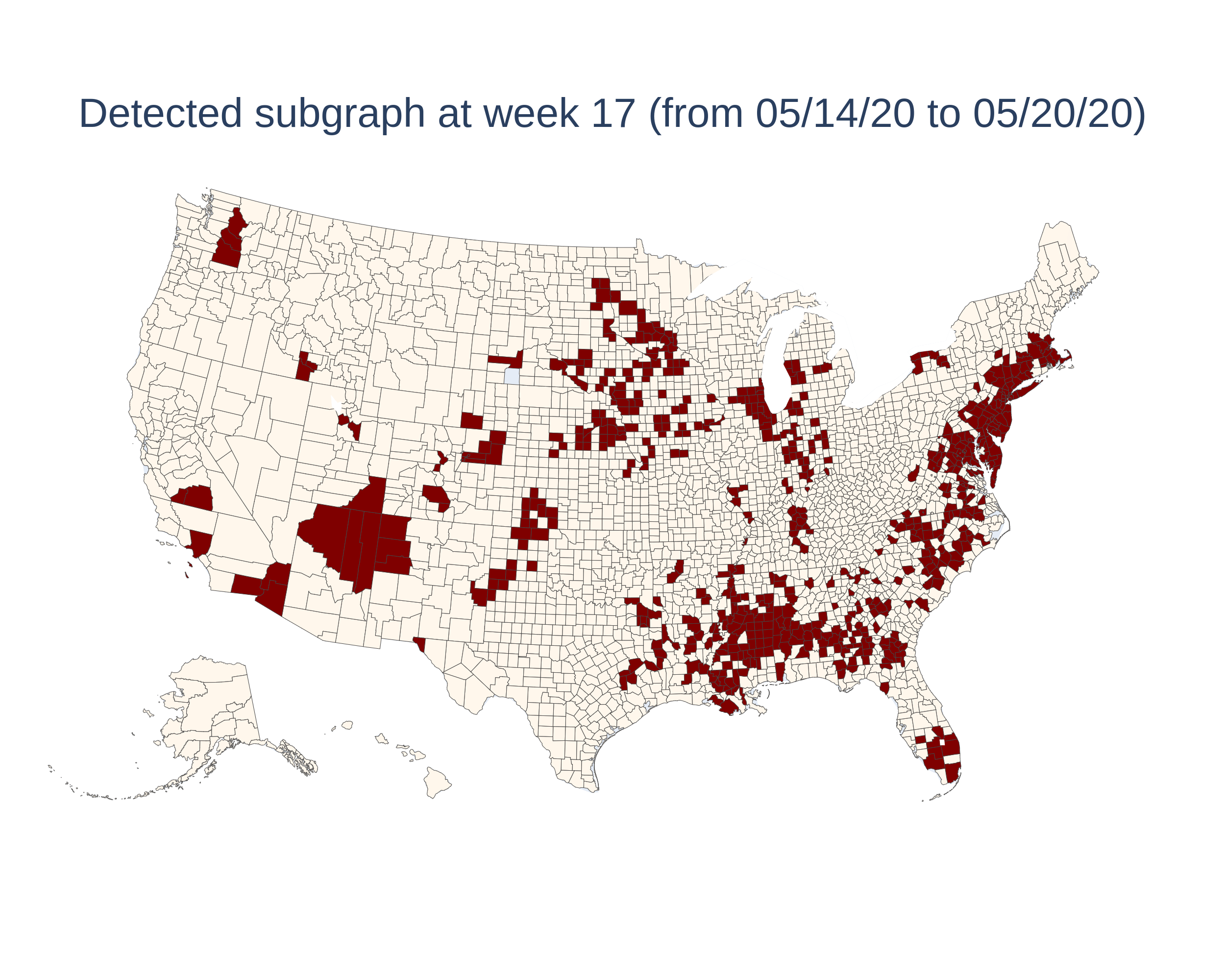}
      \end{subfigure}
      \begin{subfigure}{0.26\textwidth}
         \centering
          \includegraphics[width=1.\textwidth]{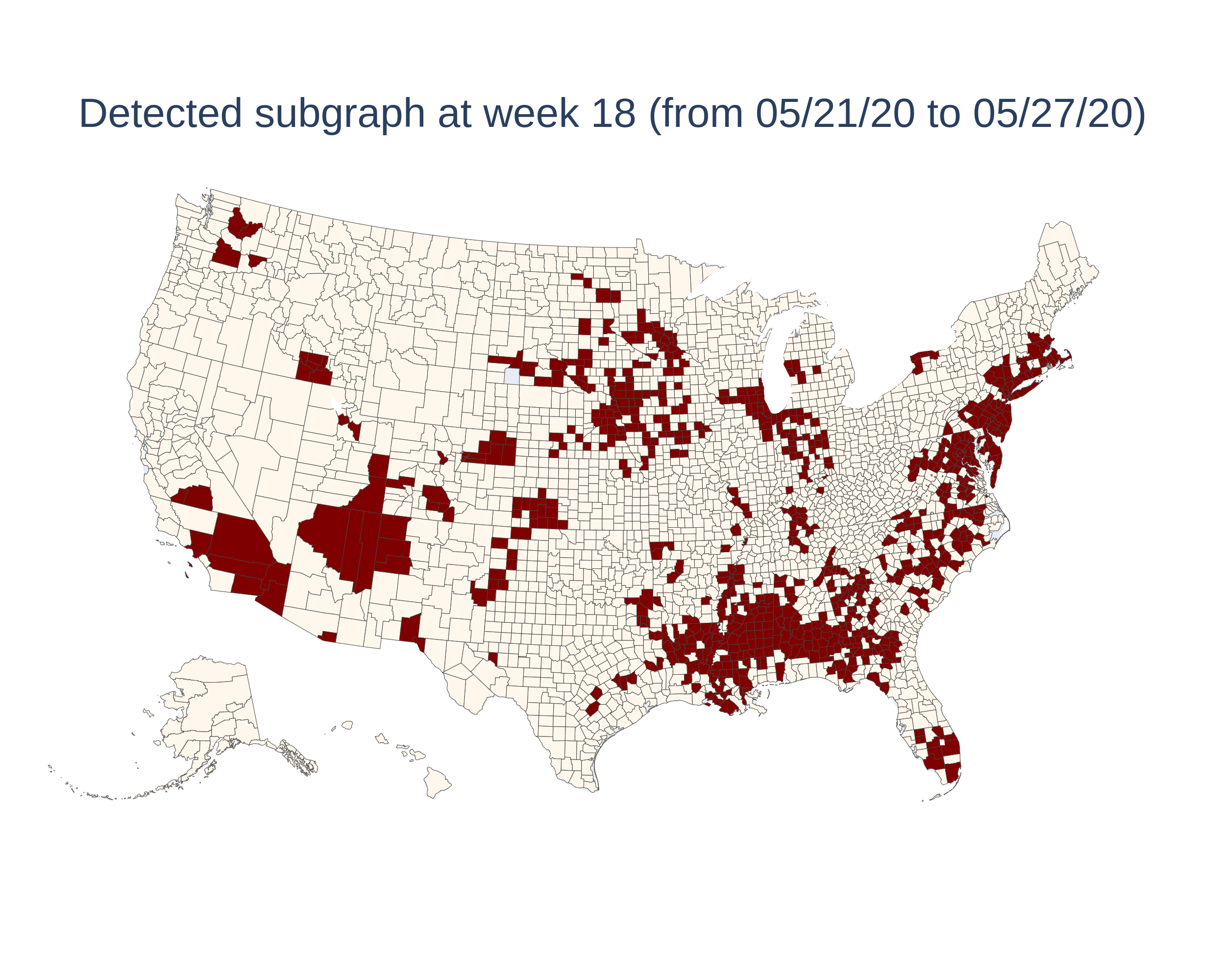}
     \end{subfigure}
      \begin{subfigure}{0.26\textwidth}
         \centering
          \includegraphics[width=1.\textwidth]{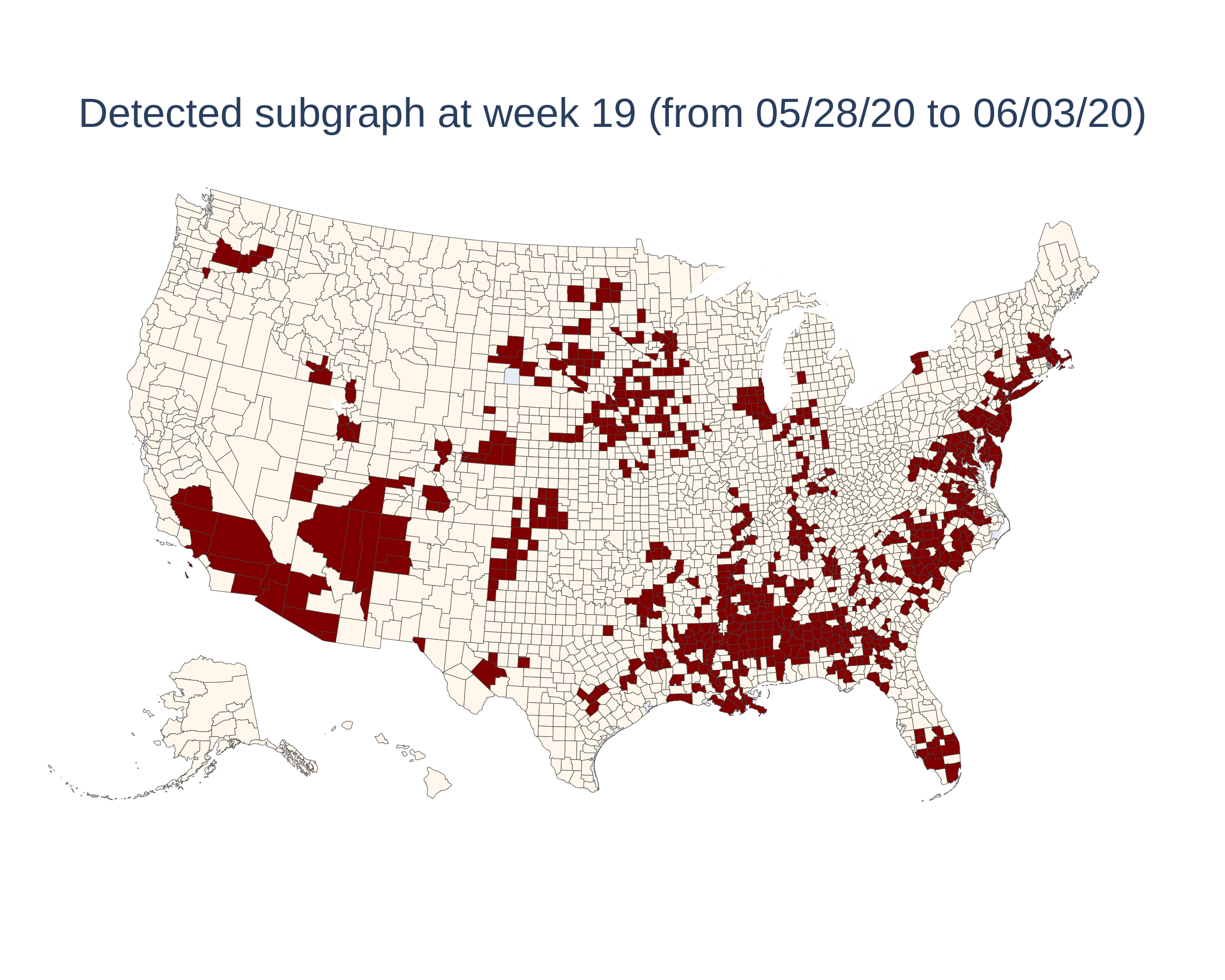}
      \end{subfigure}
      \newline
          \begin{subfigure}{0.26\textwidth}
     \centering
          \includegraphics[width=1.\textwidth]{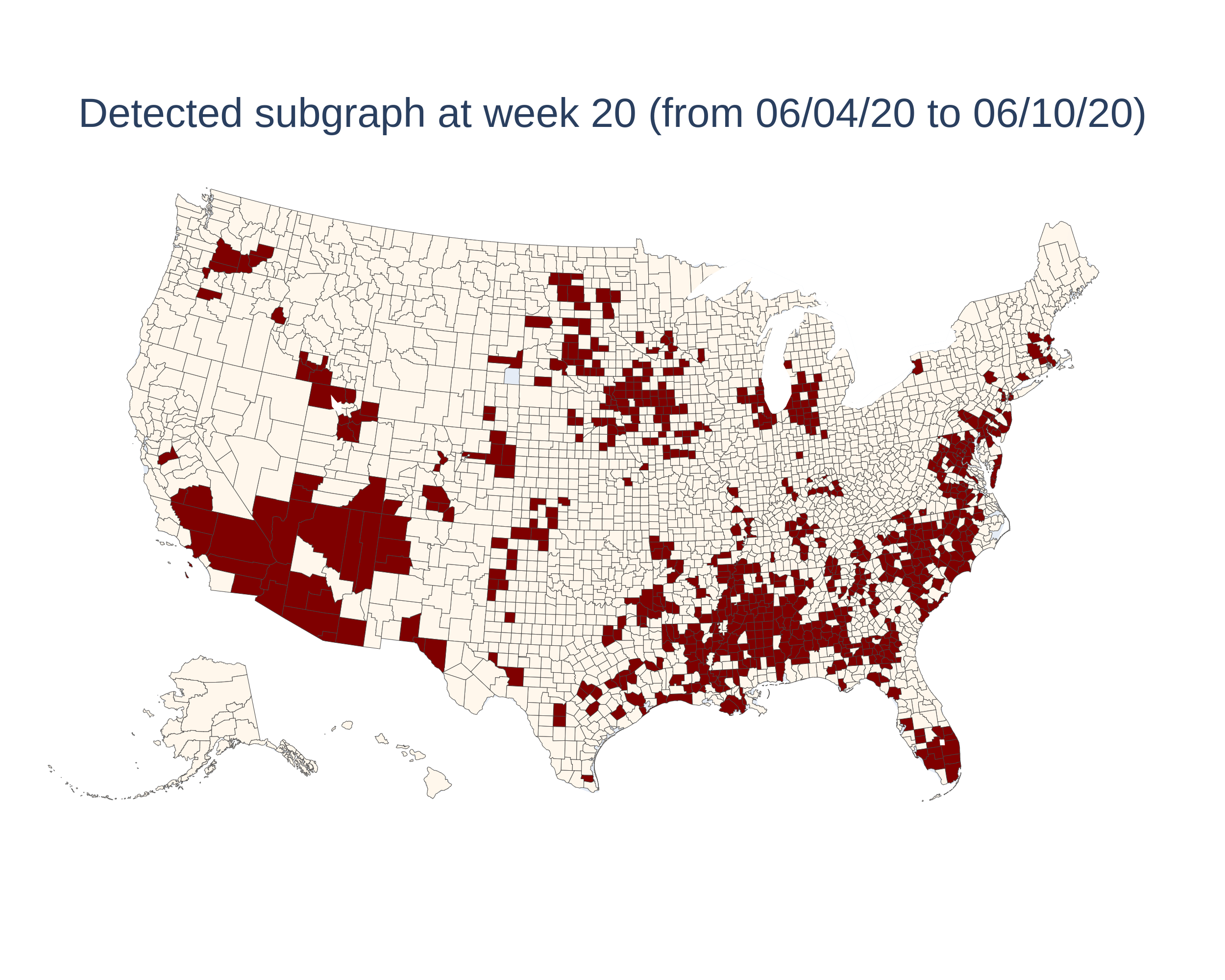}
      \end{subfigure}
      \begin{subfigure}{0.26\textwidth}
         \centering
          \includegraphics[width=1.\textwidth]{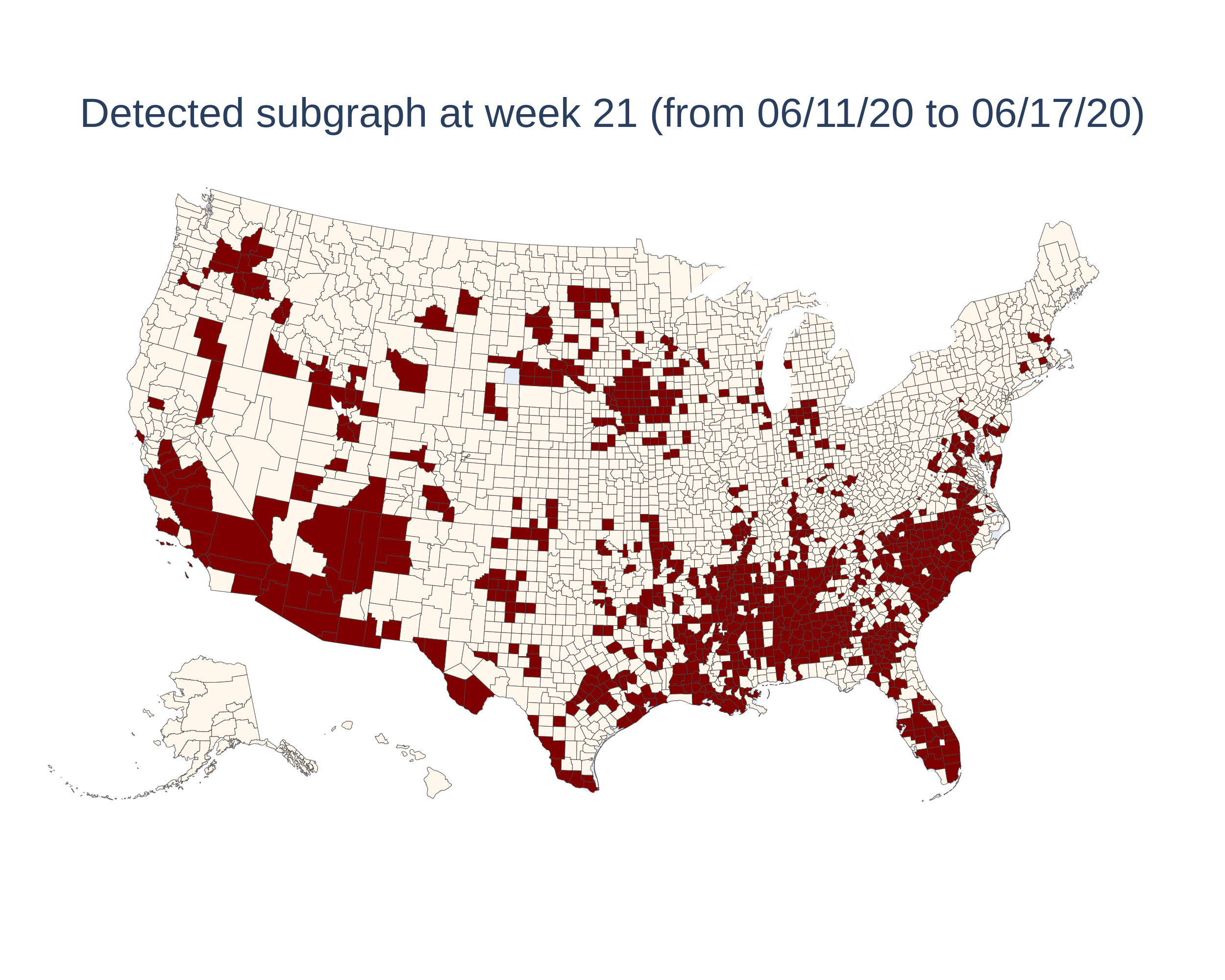}
      \end{subfigure}
      \begin{subfigure}{0.26\textwidth}
          \centering
         \includegraphics[width=1.\textwidth]{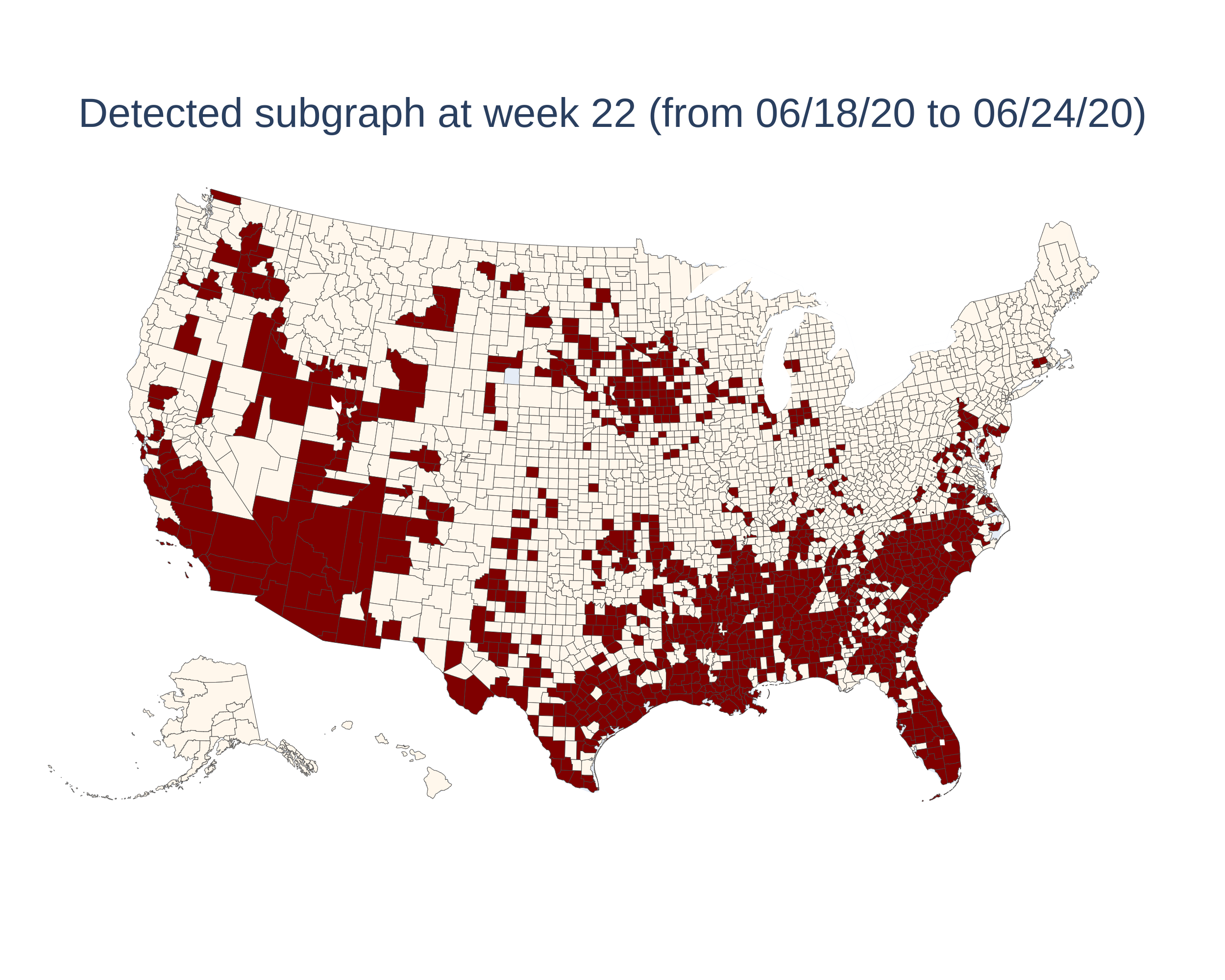}
      \end{subfigure}
      \newline
    
      \begin{subfigure}{0.26\textwidth}
          \centering
          \includegraphics[width=1.\textwidth]{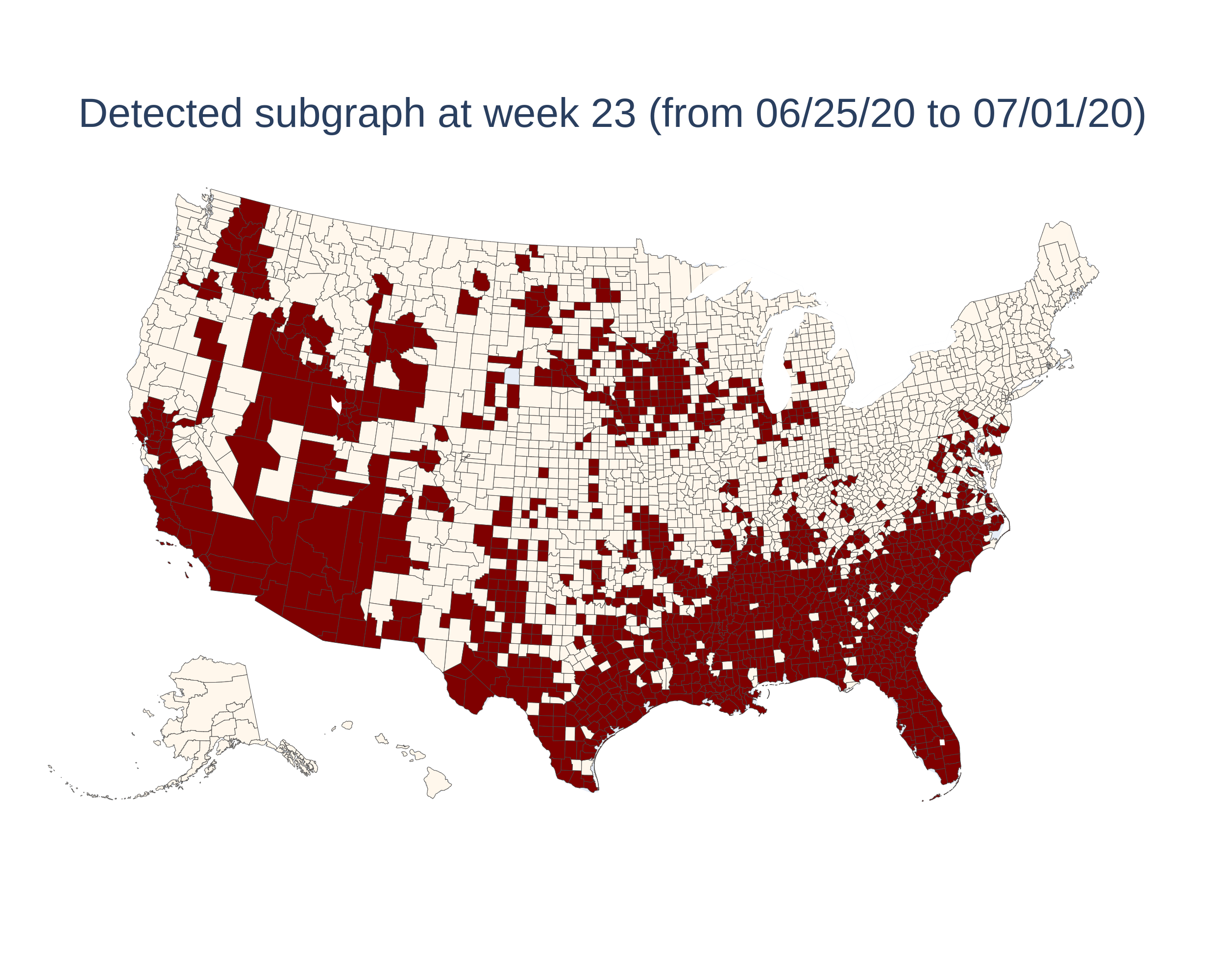}
      \end{subfigure}
      \begin{subfigure}{0.26\textwidth}
          \centering
          \includegraphics[width=1.\textwidth]{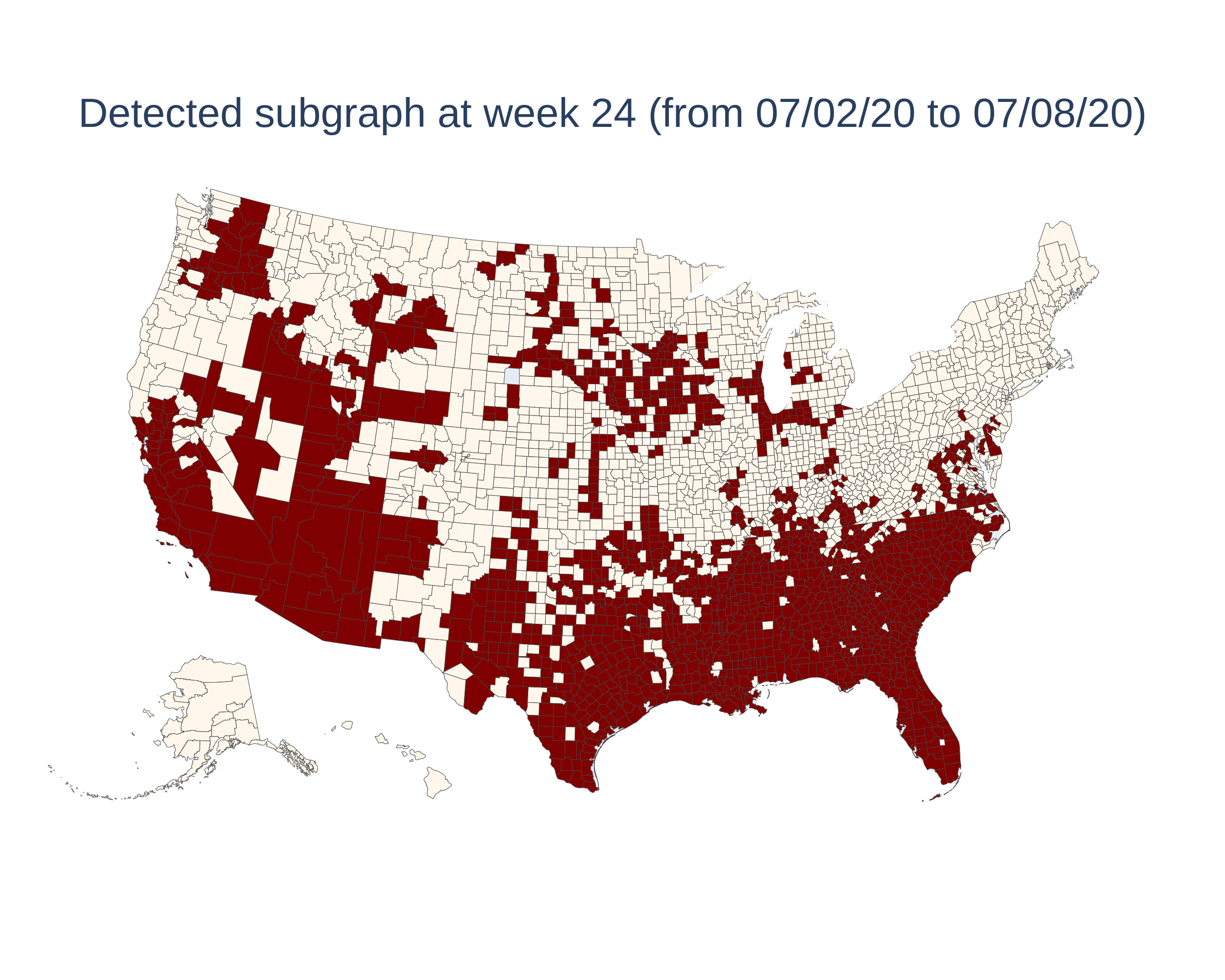}
      \end{subfigure}
      \caption{\texttt{LTSS} Top-1 Detected Spatial-Temporal Connected Subgraph on \texttt{COVID-19} Dataset}
\label{fig:covid19_ltss_detected_subgraph}
 \end{figure*}

\clearpage
\begin{figure*}[!ht]
      \begin{subfigure}{0.28\textwidth}
          \centering
         \includegraphics[width=1.\textwidth]{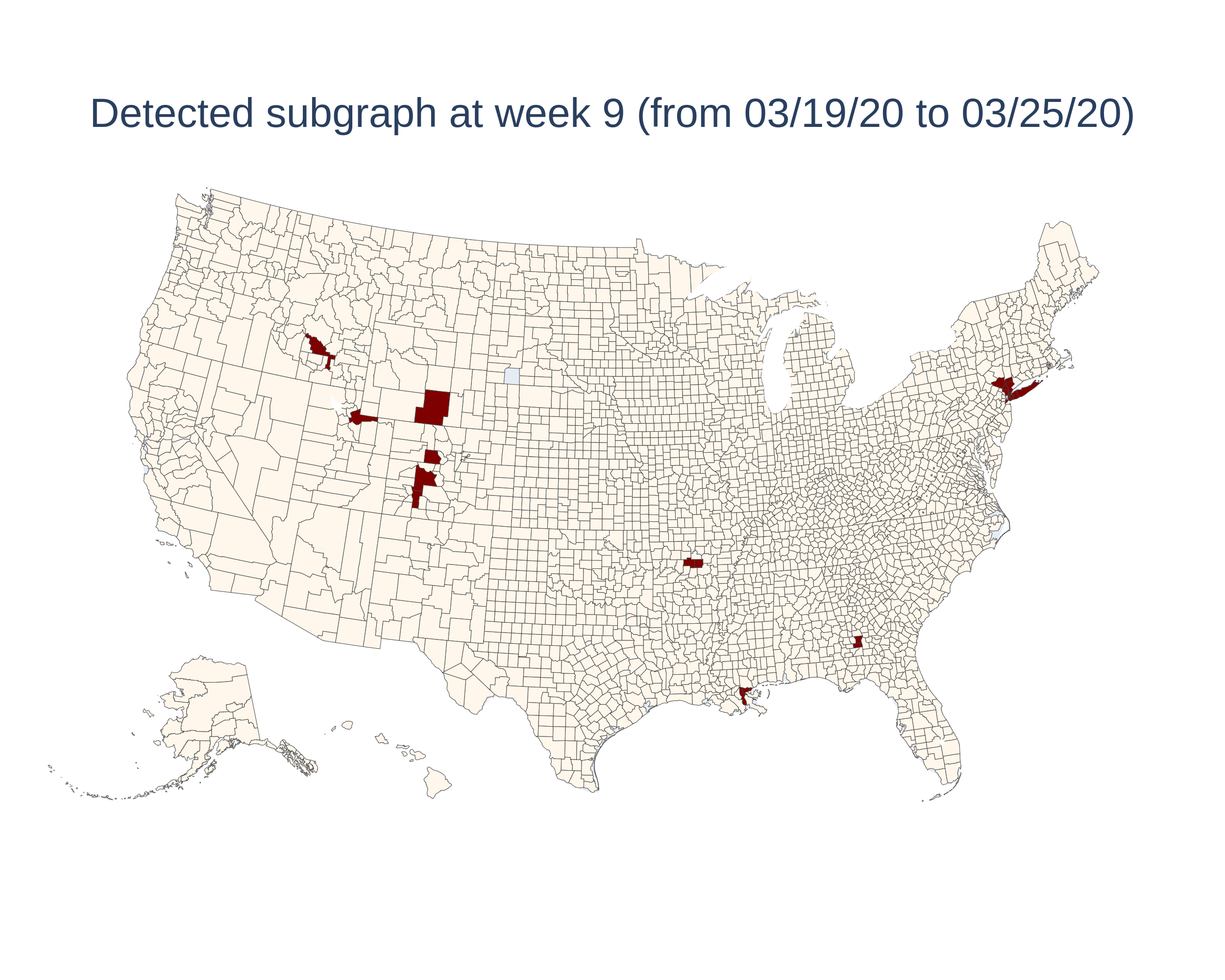}
      \end{subfigure}
      \begin{subfigure}{0.28\textwidth}
          \centering
         \includegraphics[width=1.\textwidth]{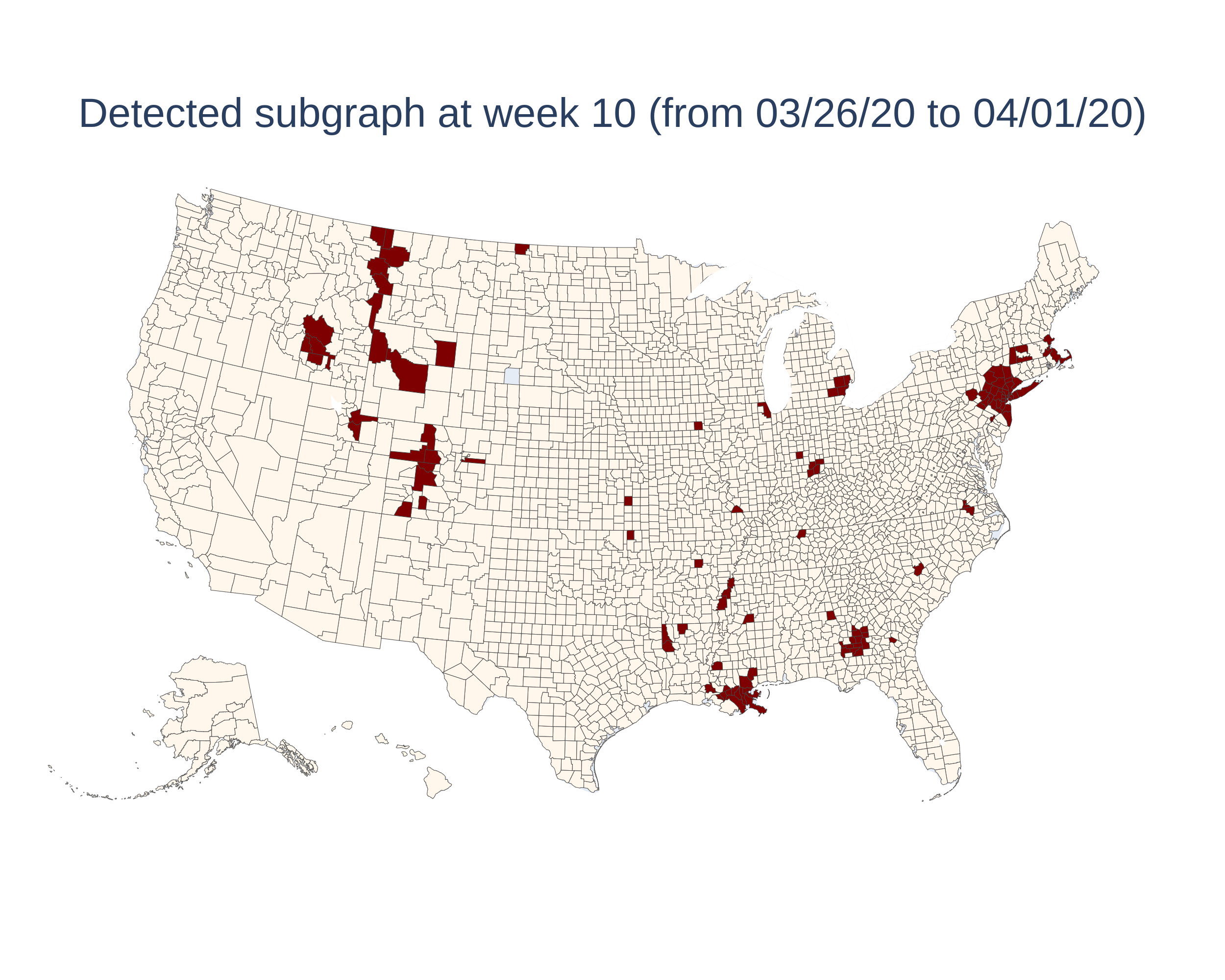}
      \end{subfigure}
      \begin{subfigure}{0.28\textwidth}
          \centering
         \includegraphics[width=1.\textwidth]{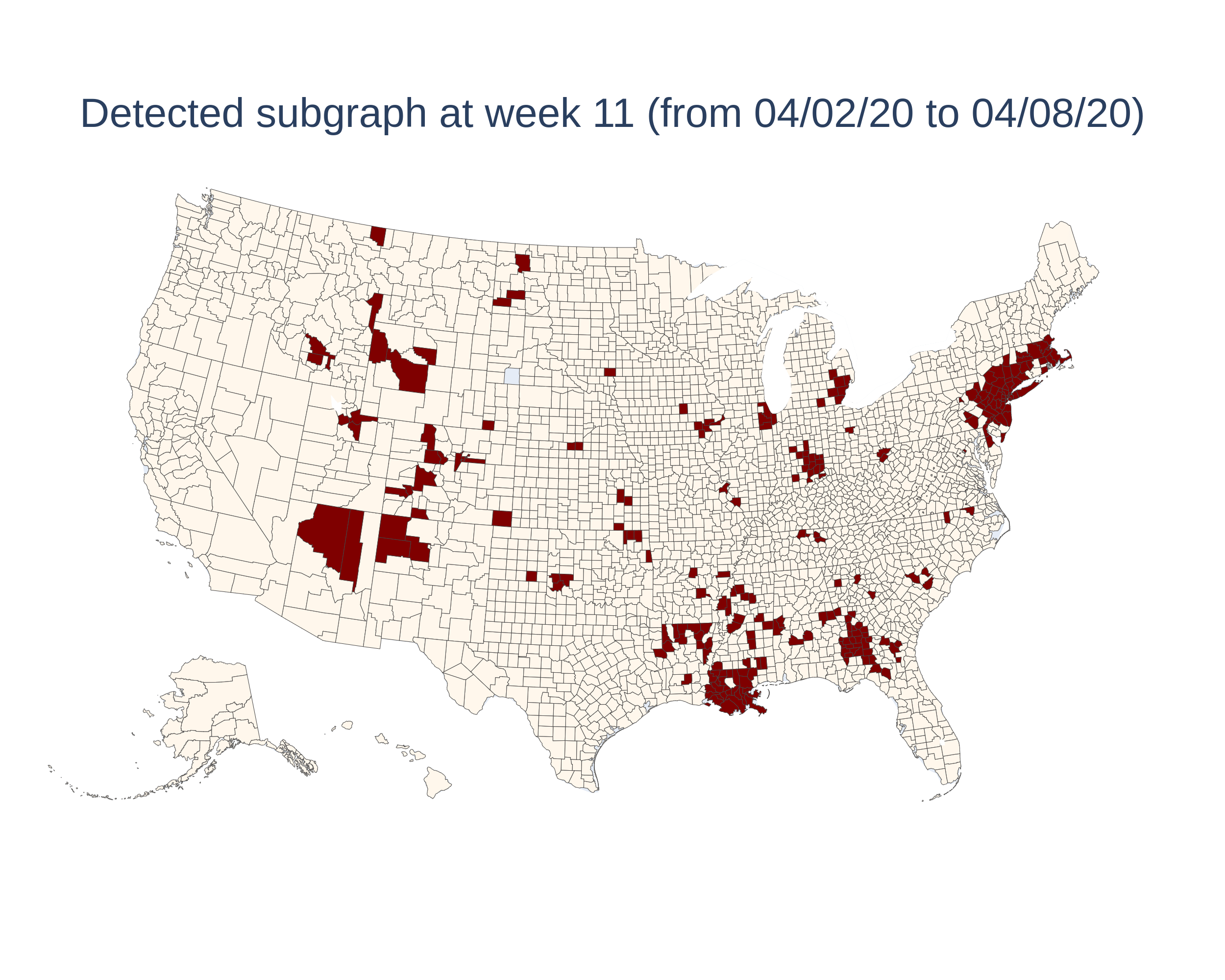}
      \end{subfigure}
      \newline
    
      \begin{subfigure}{0.28\textwidth}
          \centering
         \includegraphics[width=1.\textwidth]{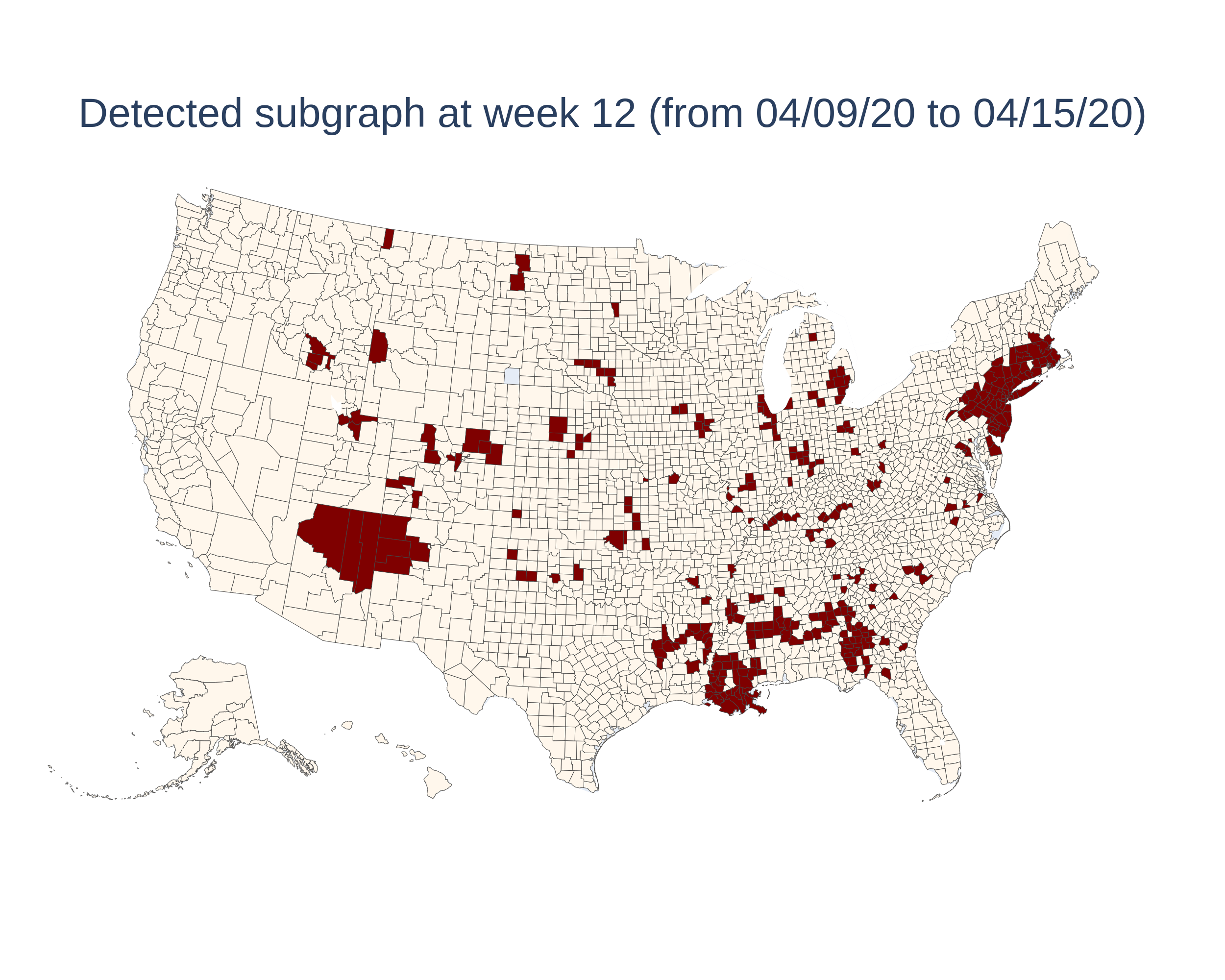}
      \end{subfigure}
      \begin{subfigure}{0.28\textwidth}
          \centering
         \includegraphics[width=1.\textwidth]{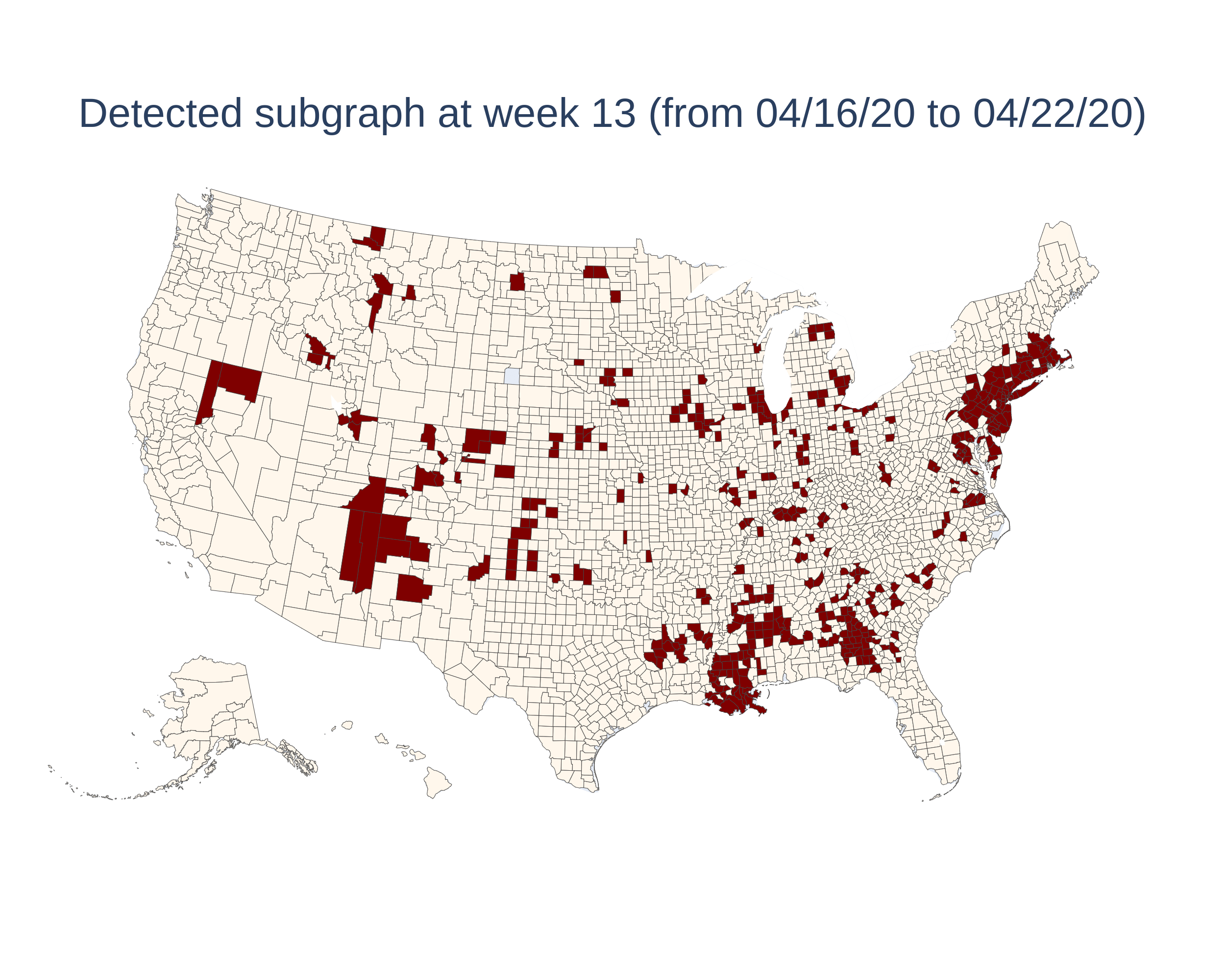}
      \end{subfigure}
      \begin{subfigure}{0.28\textwidth}
          \centering
         \includegraphics[width=1.\textwidth]{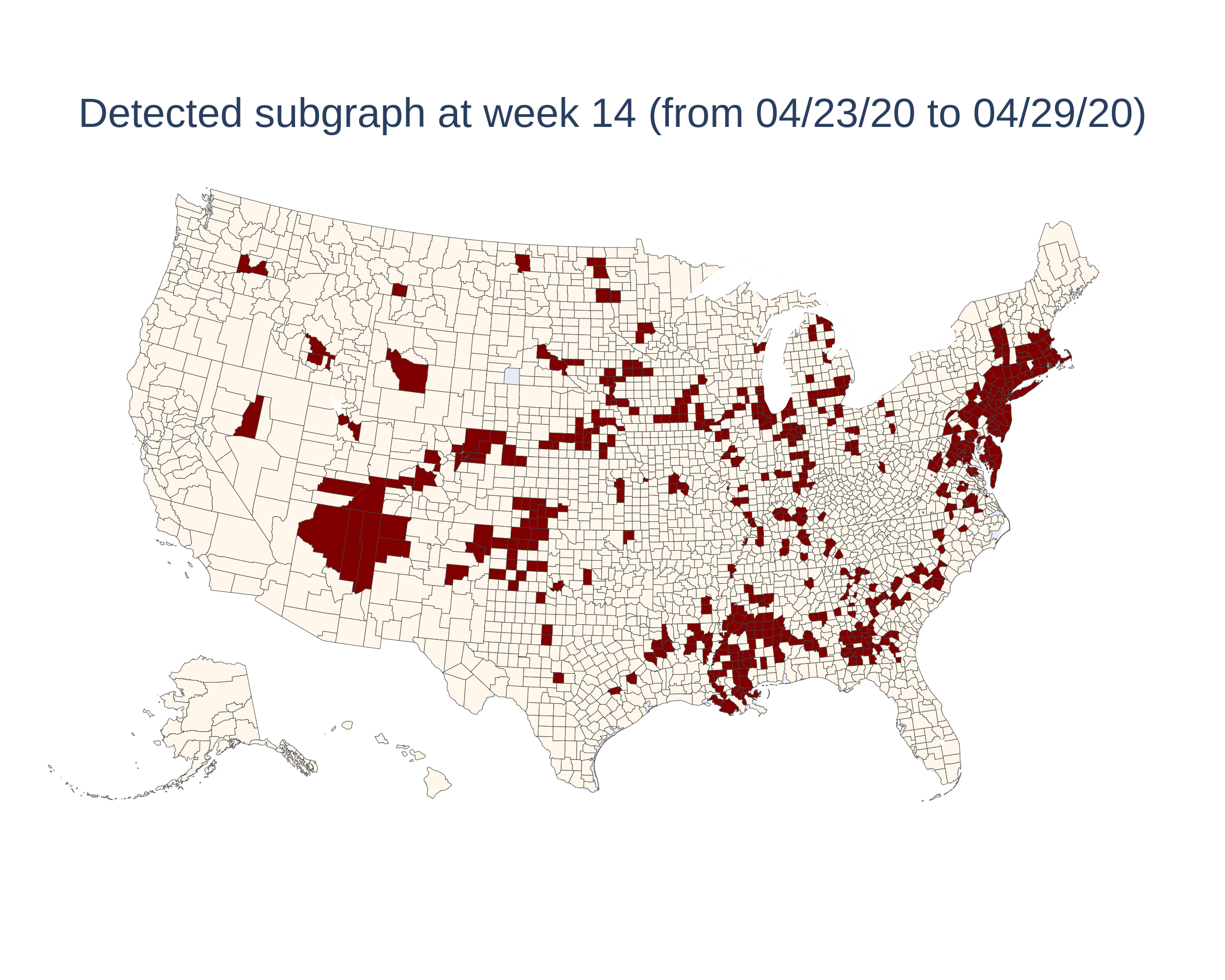}
      \end{subfigure}
      \newline
    
      \begin{subfigure}{0.28\textwidth}
          \centering
         \includegraphics[width=1.\textwidth]{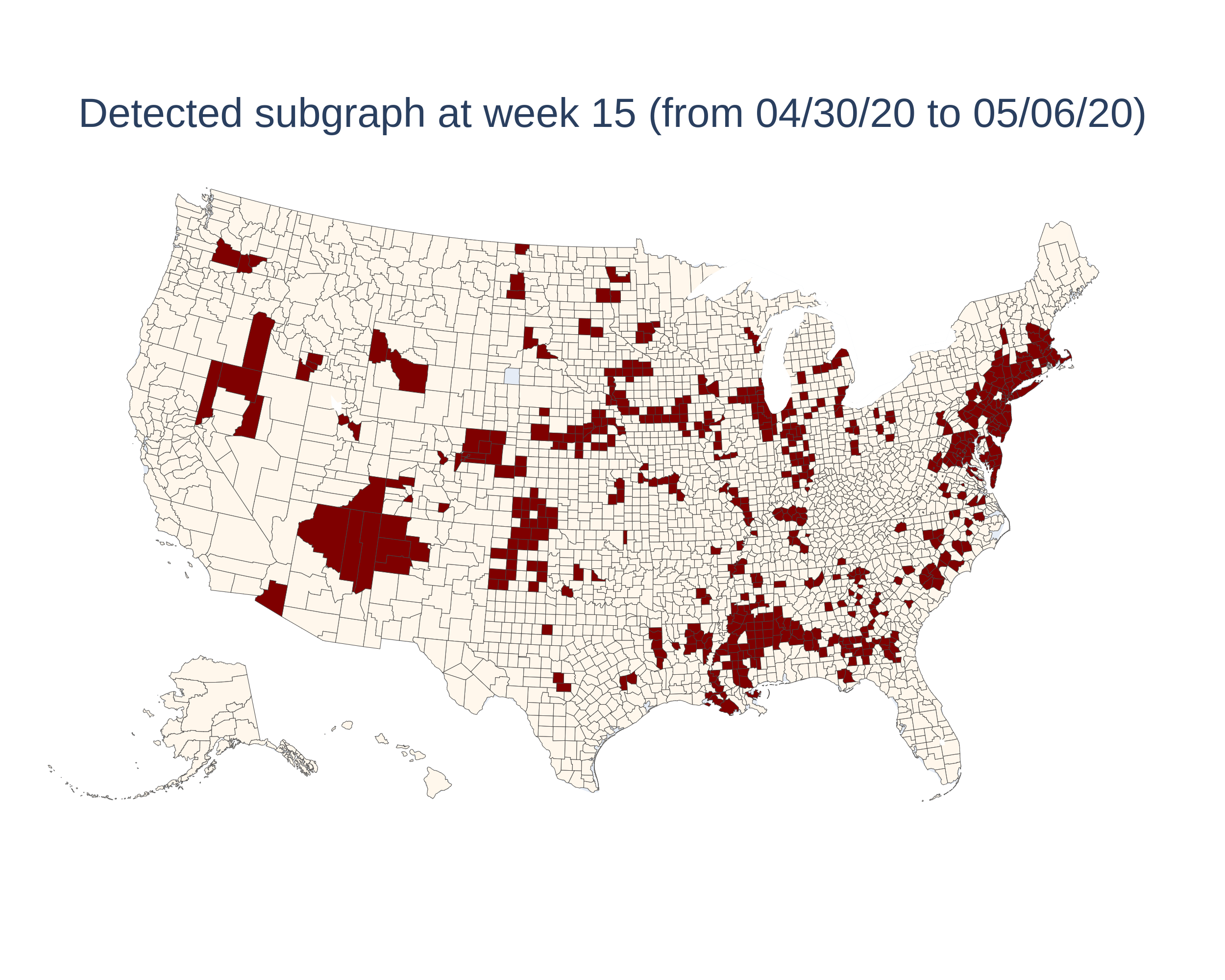}
      \end{subfigure}
      \begin{subfigure}{0.28\textwidth}
          \centering
         \includegraphics[width=1.\textwidth]{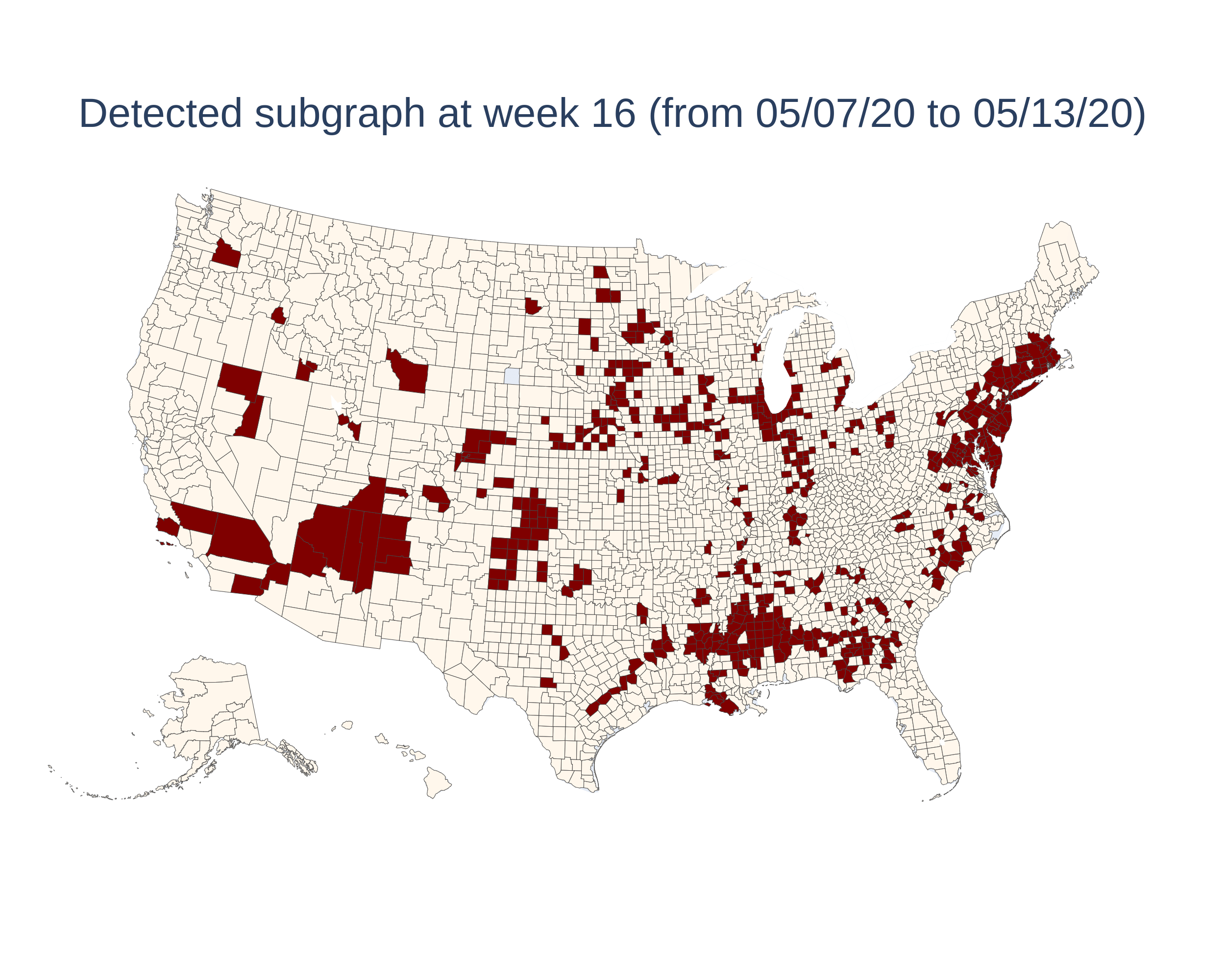}
      \end{subfigure}
      \begin{subfigure}{0.28\textwidth}
          \centering
         \includegraphics[width=1.\textwidth]{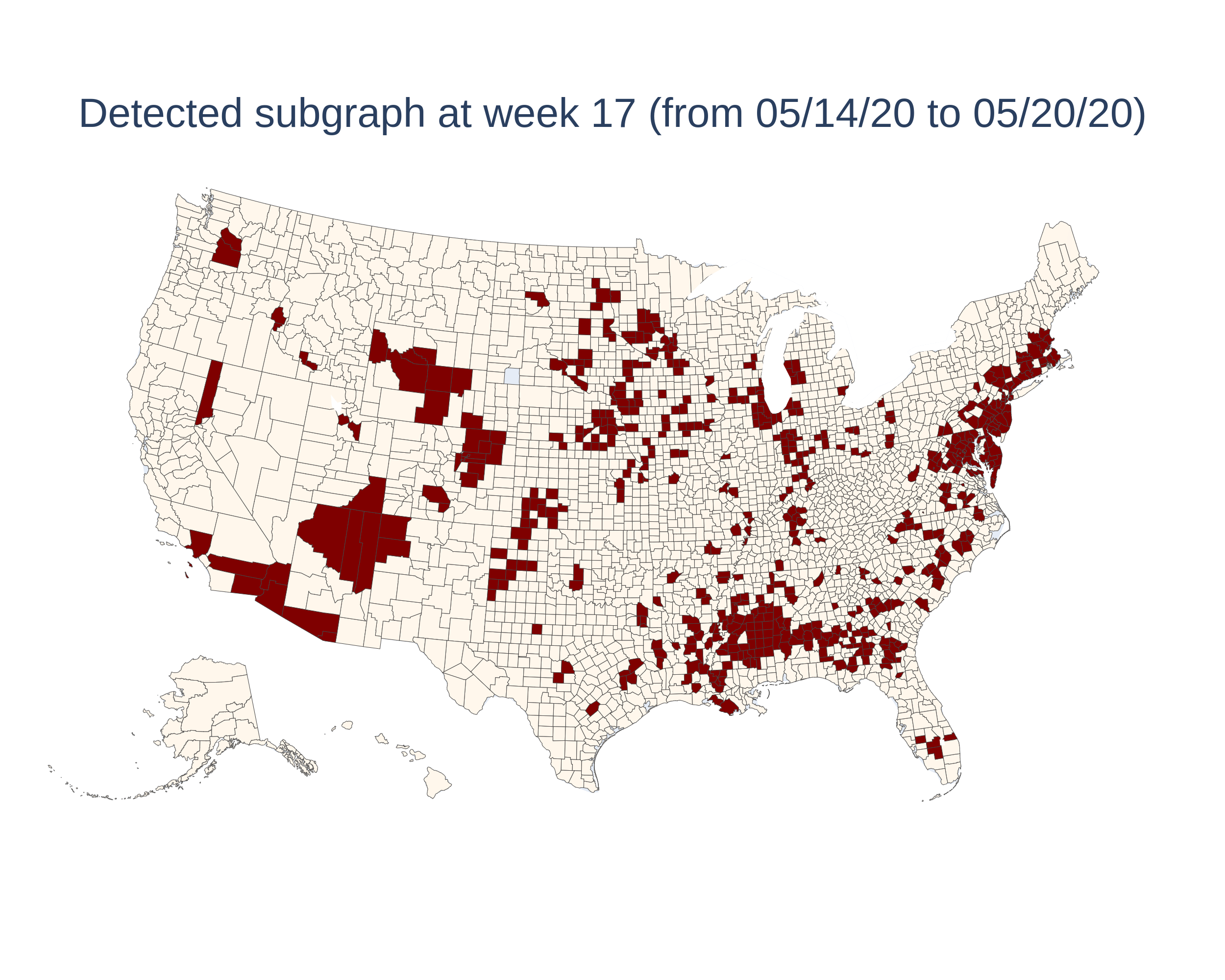}
      \end{subfigure}
      \newline
    
      \begin{subfigure}{0.28\textwidth}
          \centering
         \includegraphics[width=1.\textwidth]{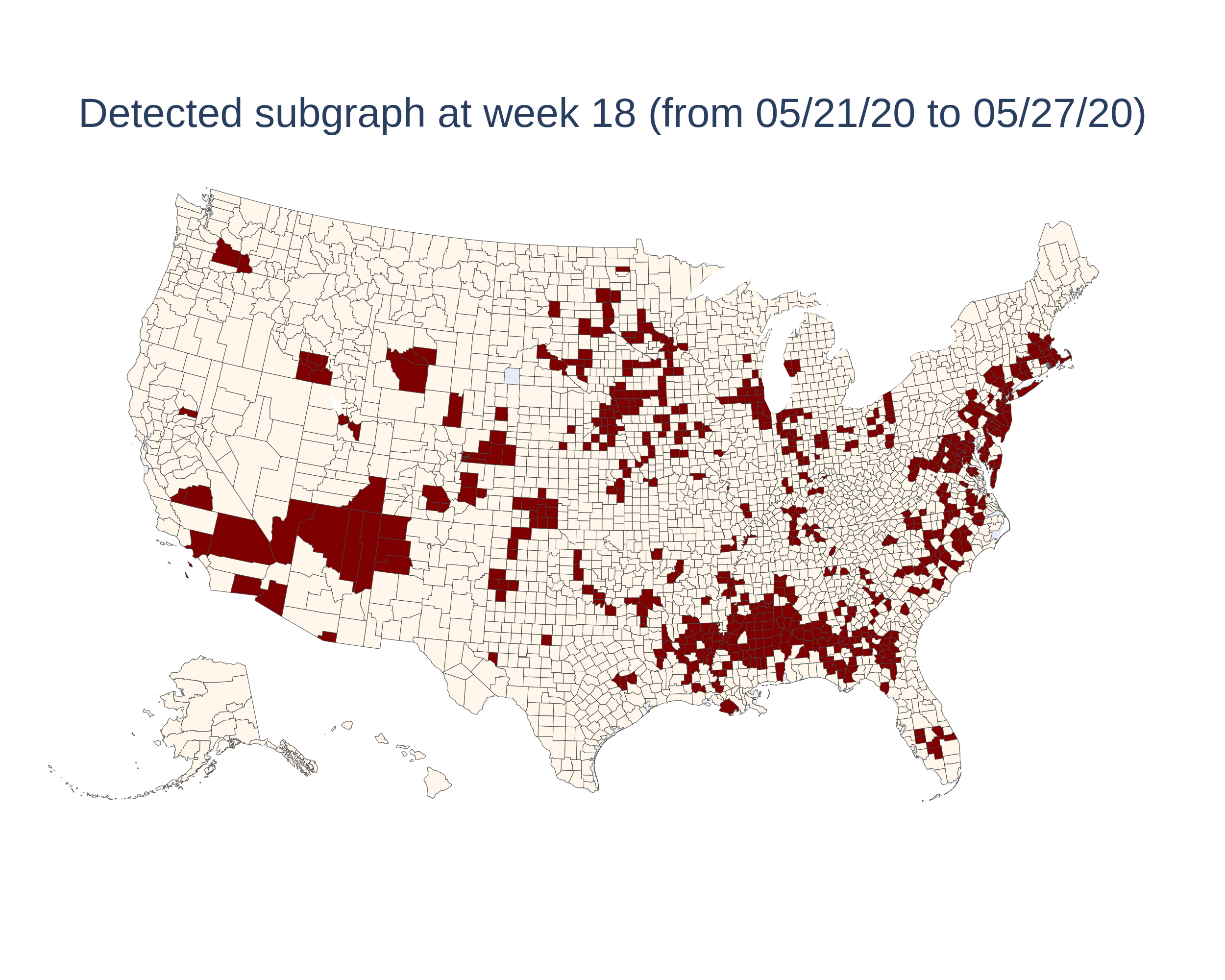}
      \end{subfigure}
      \begin{subfigure}{0.28\textwidth}
          \centering
         \includegraphics[width=1.\textwidth]{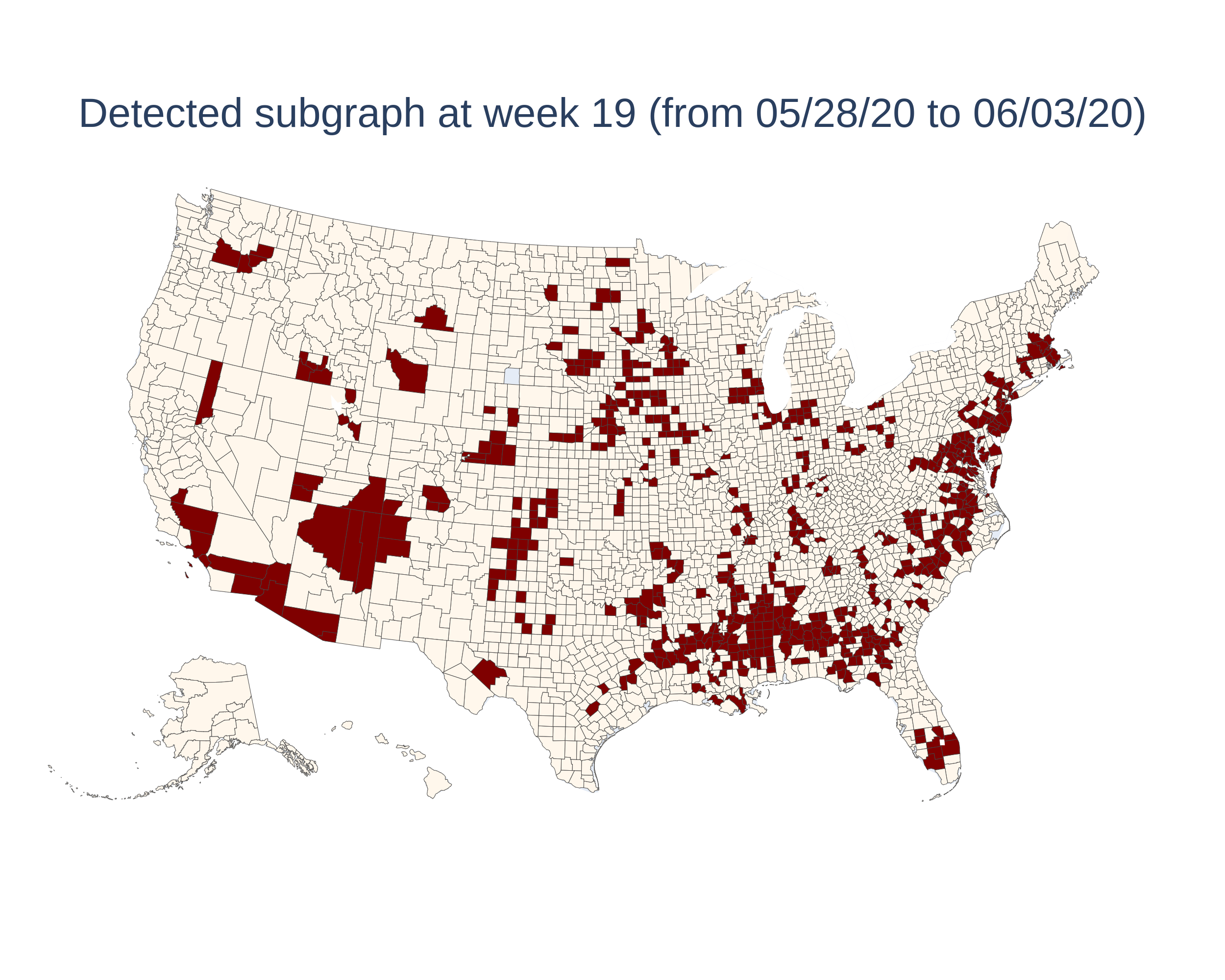}
      \end{subfigure}
      \begin{subfigure}{0.28\textwidth}
          \centering
         \includegraphics[width=1.\textwidth]{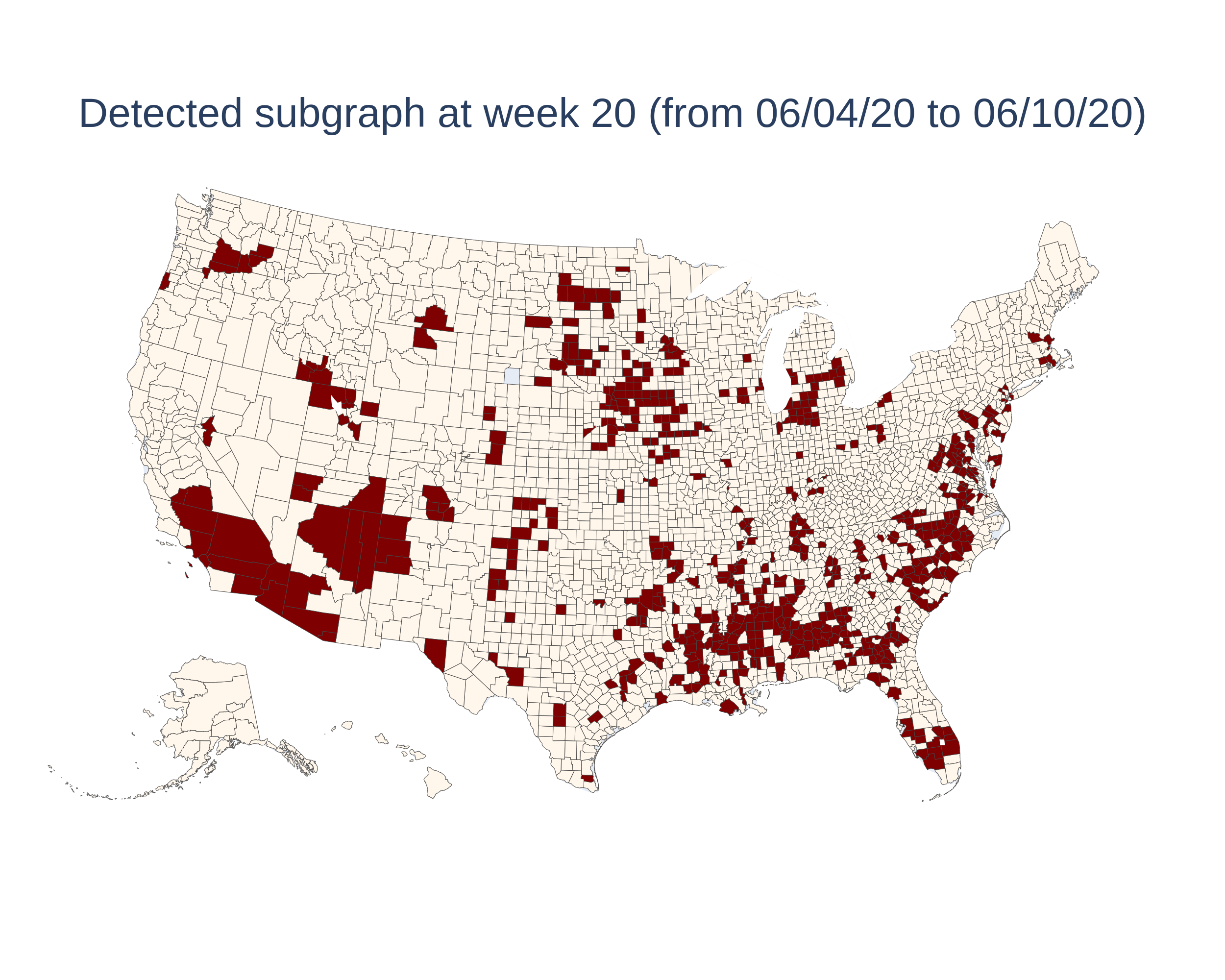}
      \end{subfigure}
      \newline
    
      \begin{subfigure}{0.28\textwidth}
          \centering
         \includegraphics[width=1.\textwidth]{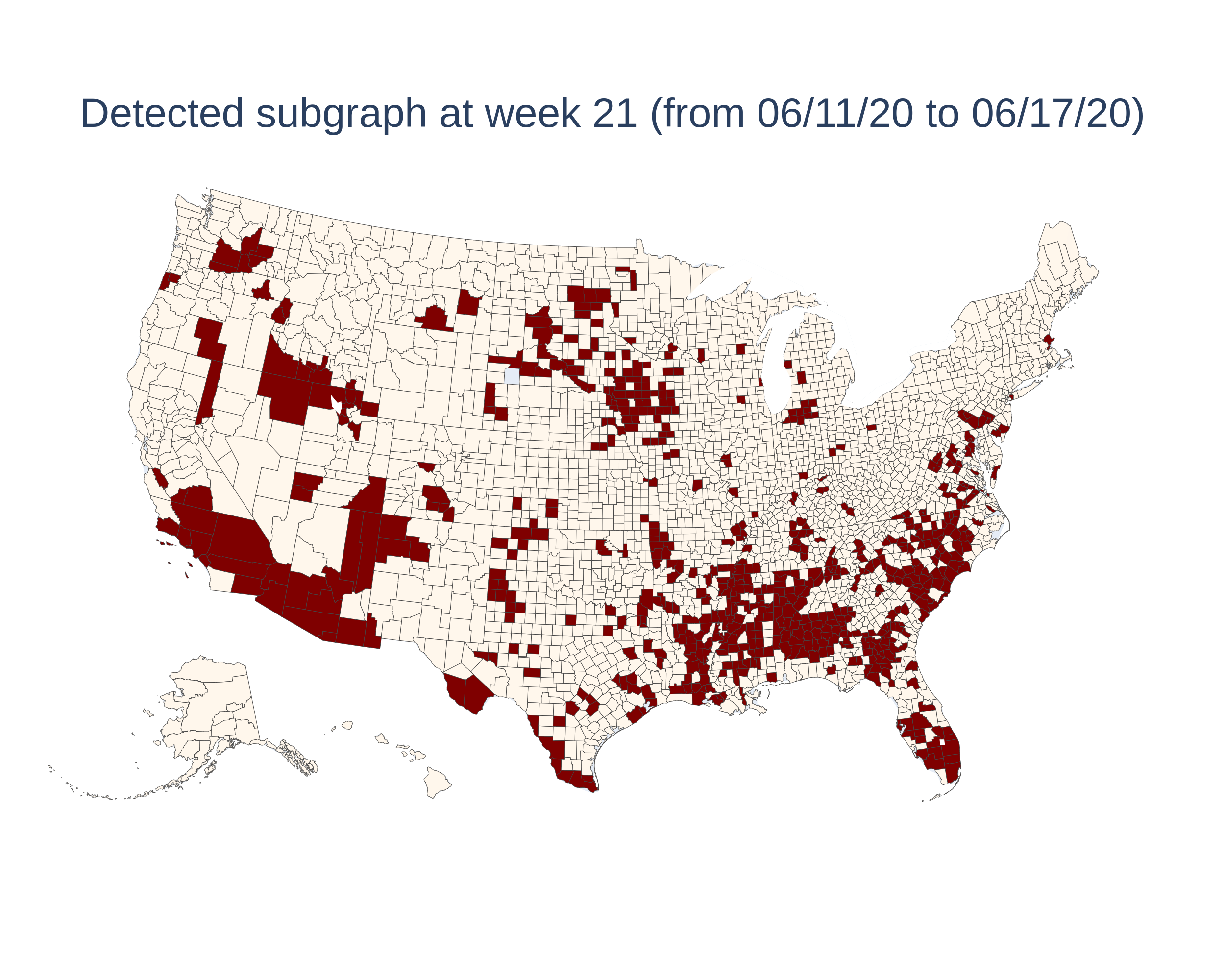}
      \end{subfigure}
      \begin{subfigure}{0.28\textwidth}
          \centering
         \includegraphics[width=1.\textwidth]{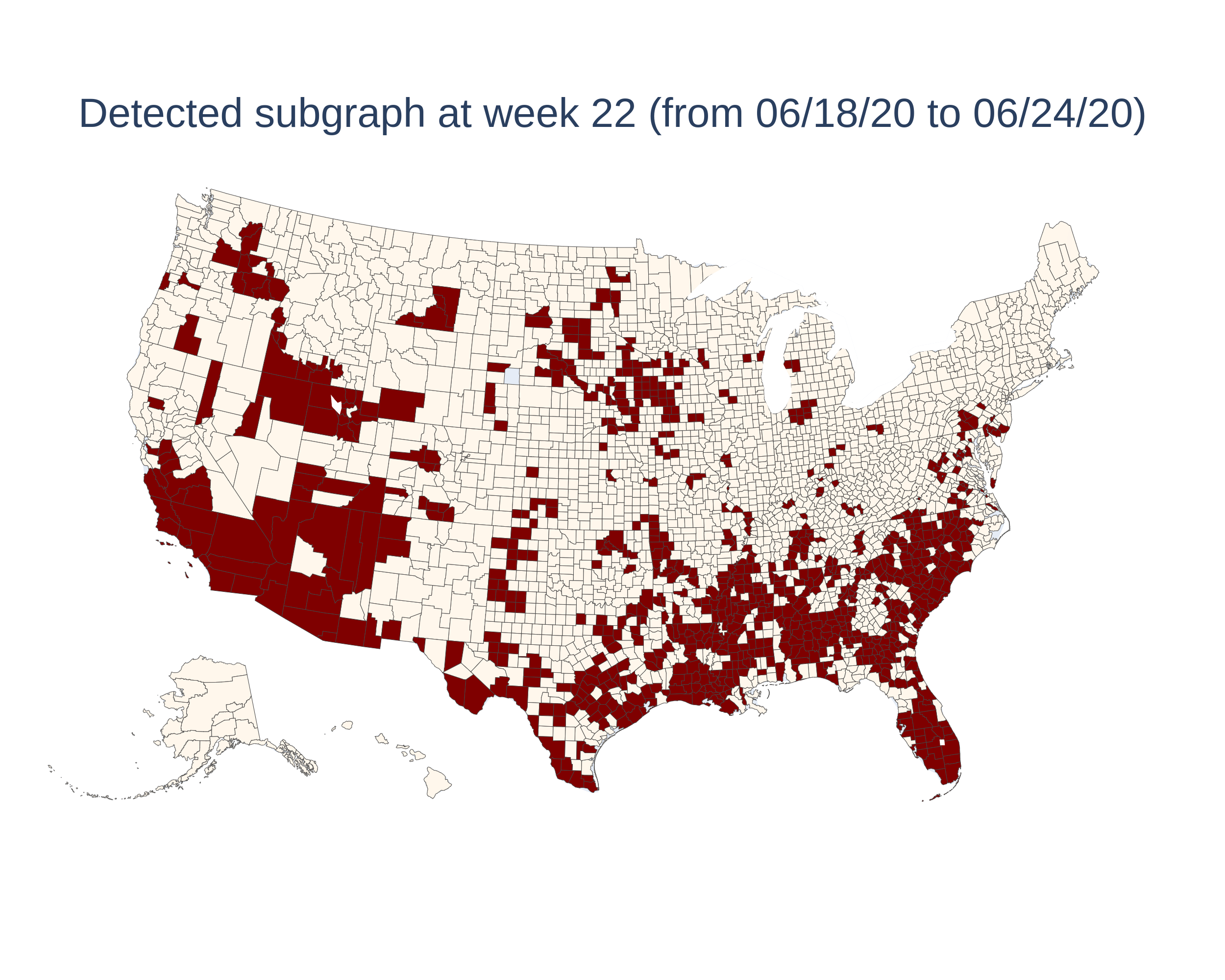}
      \end{subfigure}
      \begin{subfigure}{0.28\textwidth}
          \centering
         \includegraphics[width=1.\textwidth]{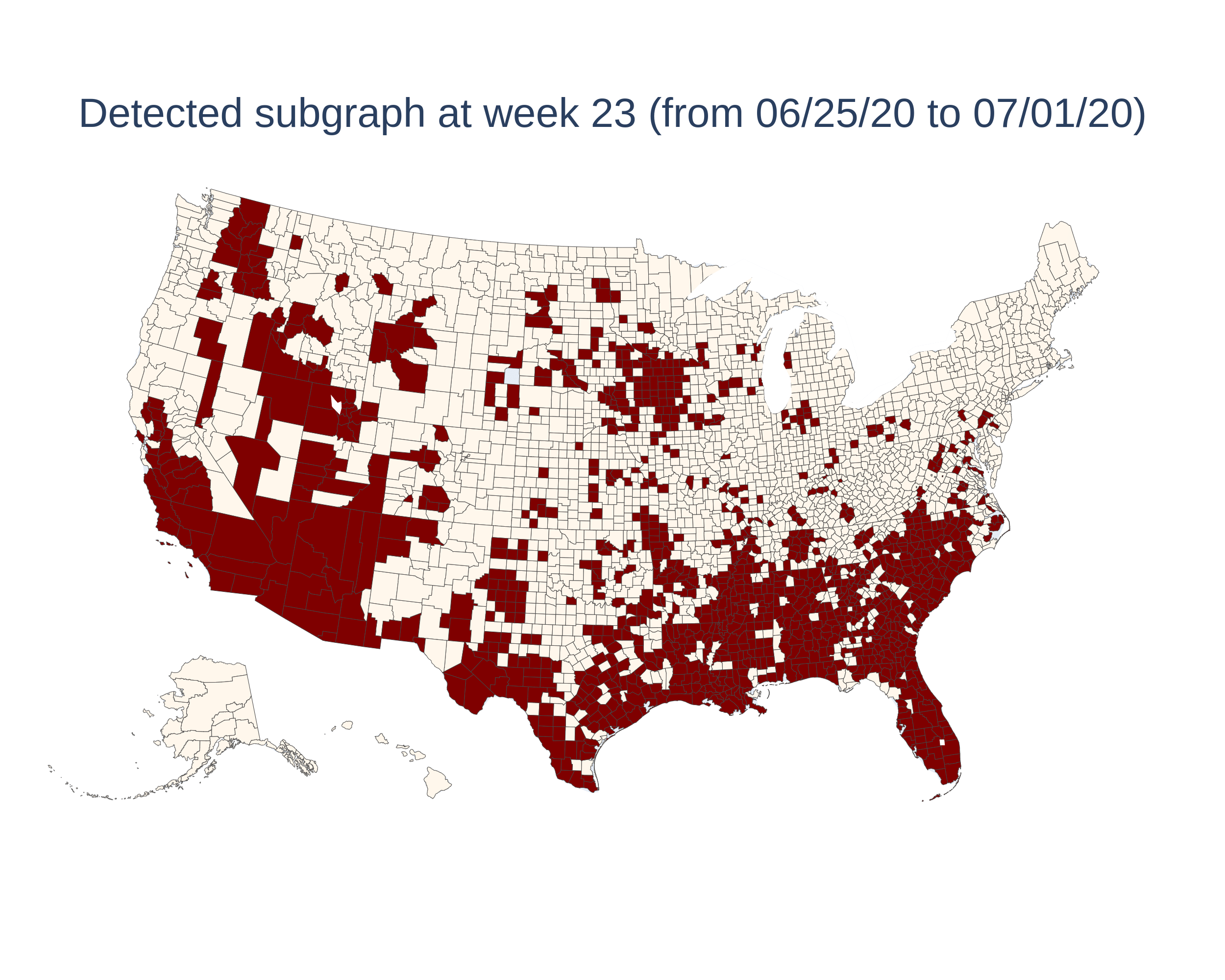}
      \end{subfigure}
      \caption{\texttt{EventTree} Top-1 Detected Spatial-Temporal Connected Subgraph on \texttt{COVID-19} Dataset}
\label{fig:covid19_eventtree_detected_subgraph}
 \end{figure*}

\end{document}